\documentclass[12pt]{article}
\usepackage[titletoc, title]{appendix}
\usepackage{tikz}
\usepackage{graphicx}

\usepackage[tbtags]{amsmath}
\usepackage{amsthm}
\usepackage{amssymb}
\usepackage{amsfonts}
\usepackage{mathtools}
\usepackage{bm}
\usepackage[english]{babel}

\usepackage{etoolbox}
\usepackage{caption} 
\usepackage{tabularx}
\usepackage{csquotes}
\usepackage{enumerate}
\usepackage{enumitem}
\usepackage{multirow}
\usepackage{multicol}
\usepackage{tabu}
\usepackage{float}
\usepackage{subfigure}
\usepackage[linesnumbered,ruled,vlined]{algorithm2e}

\usepackage{url} 
\usepackage[colorlinks=true,linkcolor=blue,urlcolor=blue,citecolor=blue]{hyperref}
\usepackage[authoryear]{natbib}
\newcommand{\blind}{1}

\graphicspath{ {./images/} }

\newtheorem{theorem}{Theorem}
\newtheorem{lemma}{Lemma}
\newtheorem{definition}{Definition}
\newtheorem{proposition}{Proposition}
\newtheorem{corollary}{Corollary}
\newtheorem{remark}{Remark}

\SetKwInput{KwInput}{Input}
\SetKwInput{KwOutput}{Output}
\SetKwInput{KwInit}{Initialization}

\DeclareMathOperator*{\argmin}{arg\,min}

\newtheorem{manualpropinner}{Proposition}
\newenvironment{manualprop}[1]{%
  \renewcommand\themanualpropinner{#1}%
  \manualpropinner
}{\endmanualpropinner}

\newtheorem{manualcorinner}{Corollary}
\newenvironment{manualcor}[1]{%
  \renewcommand\themanualcorinner{#1}%
  \manualcorinner
}{\endmanualcorinner}

\begin{document}

\def\spacingset#1{\renewcommand{\baselinestretch}%
{#1}\small\normalsize} \spacingset{1}

\if1\blind
{
  \title{\bf Multiple Change Point Detection in Reduced Rank High Dimensional Vector Autoregressive Models}
  \author{Peiliang Bai \thanks{Department of Statistics, University of Florida, Email: \texttt{baipl92@ufl.edu}, \texttt{a.safikhani@ufl.edu}, \texttt{gmichail@ufl.edu}} \and Abolfazl Safikhani \footnotemark[1] \and George Michailidis \footnotemark[1] \thanks{Computer and Information Science Engineering, University of Florida} \thanks{Informatics Institute, University of Florida.}}
  \maketitle
} \fi

\if0\blind
{
  \bigskip
  \bigskip
  \bigskip
  \begin{center}
    {\LARGE\bf Multiple Change Point Detection in Reduced Rank High Dimensional Vector Autoregressive Models}
  \end{center}
  \medskip
} \fi

\bigskip
\begin{abstract}
We study the problem of detecting and locating change points in high-dimensional Vector Autoregressive (VAR) models, whose transition matrices exhibit low rank plus sparse structure. We first address the problem of detecting a single change point using an exhaustive search algorithm and establish a finite sample error bound for its accuracy. Next, we extend the results to the case of multiple change points that can grow as a function of the sample size. Their detection is based on a two-step algorithm, wherein the first step, an exhaustive search for a candidate change point is employed for overlapping windows, and subsequently a backwards elimination procedure is used to screen out redundant candidates. The two-step strategy yields consistent estimates of the number and the locations of the change points. To reduce computation cost, we also investigate conditions under which a surrogate VAR model with a weakly sparse transition matrix can accurately estimate the change points and their locations for data generated by the original model. This work also addresses and resolves a number of novel technical challenges posed by the nature of the VAR models under consideration.
The effectiveness of the proposed algorithms and methodology is illustrated on both synthetic and two real data sets.
\end{abstract}

\noindent%
\vfill
\spacingset{1.4} 
\section{Introduction}\label{sec:intro}

High dimensional time series analysis and their applications have become increasingly important in diverse domains, including macroeconomics (\cite{kilian2017structural,stock2016dynamic}), financial economics (\cite{billio2012econometric,lin2017regularized}), molecular biology (\cite{michailidis2013autoregressive}) and neuroscience (\cite{friston2014granger,schroder2019fresped}). Such data are usually both cross-correlated and auto-correlated. There are two broad modeling paradigms for capturing these features in the data: (i) dynamic factor and latent models (\cite{bai2008large,stock2002forecasting,stock2016dynamic,lam2011estimation,li2014new}), and (ii) vector autoregressive (VAR) models (\cite{lutkepohl2013introduction,kilian2017structural}). The basic premise of models in (i) is that the common dynamics of a large number of time series are driven by a relatively small number of latent factors, the latter evolving over time. VAR models aim to capture the self and cross auto-correlation structure in the time series, but the number of parameters to be estimated grows quadratically in the number of time series under consideration. Various structural assumptions have been proposed in the literature to accommodate a large number of time series in the model, with that of 
\textit{sparsity} (\cite{basu2015regularized}) being a very popular one. However, in many applications the autoregressive dynamics of the time series exhibit also low dimensional structure, which gave rise to the introduction of reduced rank auto-regressive models (\cite{box1977canonical,velu1986reduced,ahn1988nested,wang2004forecasting}). For example, brain activity data (see Example 1 in Section~\ref{sec:application}) exhibit low dimensional structure (\cite{schroder2019fresped}) and so do macroeconomic data (\cite{stock2016dynamic}, Example 2 in Section~\ref{sec:application}). Reduced rank auto-regressive models for stationary high-dimensional data were studied in \cite{basu2019low}.
{
The key idea of such reduced rank models is that the lead-lagged relationships between the time series can not simply be described by a few sparse components, as is the case for sparse VAR models. Instead, all the time series influence these relationships and some of them are particularly pronounced (those in the sparse component). Applications in economics/finance, neuroimaging, and environmental science are important candidates for these models.
}

In many application areas including those mentioned above, nonstationary time series data are commonly observed. The simplest, but realistic departure from stationarity, that also leads to interpretable models for the underlying time series, is \textit{piecewise-stationarity}. Under this assumption, the time series data are modeled as approximately stationary between neighbouring change-points, whereas their distribution changes at these change points. The literature on change point analysis for the two classes of modeling paradigms previously mentioned is rather sparse. \cite{bardsley2017change} developed tests for the presence of change points in functional factor models motivated by modeling the yield curve of interest rates, while \cite{barigozzi2018simultaneous} employed the binary segmentation procedure for detecting and identifying the locations of multiple change points in factor models. Change point detection for \textit{sparse VAR} models has been investigated in \cite{wang2019localizing,safikhani2017joint}, and \cite{bai2020multiple}.

The objective of this study is to investigate the problem of change point detection in a reduced rank VAR model, whose transition matrices exhibit \textit{low-rank and sparse} structure. The problem poses a number of technical challenges that we address in the sequel. 

Formally, a piece-wise stationary VAR model of lag-1 (for introducing the basic issues related to it) for a $p$-dimensional time series $\{X_t\}$ with $m_0$ change points $1 \leq \tau^\star_1 < \tau^\star_2 < \cdots < \tau^\star_{m_0} \leq T$ is given by:
\begin{equation*}
    X_t = \sum_{j=1}^{m_0+1}\left(A_j^\star X_{t-1} + \epsilon_t^j \right)\mathbf{I}(\tau^\star_{j-1} \leq t < \tau^\star_j),\quad t=1,2,\dots, T,
\end{equation*}
where $A_j^\star$ is a $p\times p$ coefficient matrix for the $j$-th segment, $j=1,2,\dots, m_0+1$, $\mathbf{I}(\tau^\star_{j-1}\leq t < \tau_j^\star)$ presents the indicator function of the $j$-th interval, and $\epsilon_t^j$s are $m_0+1$ independent zero mean Gaussian noise processes. It is assumed that that the coefficient matrix $A_j^\star$ can be decomposed into a low-rank component plus a sparse component: namely, $A_j^\star = L_j^\star + S_j^\star$, where $L_j^\star$ is a low-rank matrix with rank $r^\star_j$ ($r^\star_j \ll p$), and $S_j^\star$ is a sparse matrix with $d_j^\star$ ($d_j^\star \ll p^2$) non-zero entries.  

{
The modeling framework differs vis-a-vis the one considered in \cite{bai2020multiple}, since in the current work, \textit{both the low rank and the sparse} components of the transition matrices are allowed to \textit{exhibit changes} at break points. This flexibility rules out the use of a fused lasso based detection algorithm that is suitable for the case wherein \textit{only} the sparse component is allowed to exhibit changes, which was the setting in \cite{bai2020multiple}. As a result, a novel rolling window detection algorithm is introduced and its theoretical properties studied in the current work.
}

Next, we outline novel technical challenges, not present in change point analysis of \textit{sparse VAR} (\cite{safikhani2017joint,wang2019localizing}) and other sparse high dimensional models (\cite{roy2017change}): \\
(i) The change in the transition matrix may be due to a change in the low-rank component, in the sparse component or in both. To that end, we introduce a novel \textit{sufficient identifiability condition} for both detecting a \textit{single} change point and decomposing the transition matrix into its low rank plus sparse components (Assumptions H1 and H2 in the sequel); then, it is \textit{extended} to the case of  \textit{multiple} change points (Assumptions H1' and H2'). \\
(ii) For the case of multiple change points, commonly used procedures, such as binary segmentation (\cite{cho2015multiple}) or fused type penalties (\cite{safikhani2017joint}) are not directly applicable due to the presence of the low rank component. Specifically, the former method would lead to effectively performing singular value decompositions on misspecified models involving mixtures of piece-wise low-rank and sparse models, which may lead to the imposition of very stringent conditions for ensuring detectability of the change points (see discussion on related issues in \cite{bhattacharjee2018change}). Further, it is unclear how to design fused penalties that accommodate low-rank matrices. {On the other hand, dynamic programming based algorithms are applicable. However, their time complexity is $\mathcal{O}(T^2C(T))$, where $C(T)$ indicates the computational cost of estimating the model parameters over the entire observation sequence. This is significantly higher complexity than the previously mentioned methods (which is $\mathcal{O}(TC(T))$, see numerical comparisons and discussion in Remark \ref{remark:8} and Appendix F.7).} 

To overcome these challenges, we develop a novel procedure based on \textit{rolling windows}, wherein a \textit{single candidate} change point is identified in each window and then only those exhibiting screened based on certain properties (see Section \ref{sec:multiplecp}) are retained. This allows to leverage the theoretical results developed for the single change point. The proposed procedure based on rolling windows is \textit{naturally parallelizable}, thus speeding up computations.

{Note that the developed rolling window strategy is applicable to any complex statistical model exhibiting multiple change points. One needs to establish consistency properties for a single change point in a time interval and then appropriately select the length of the rolling window, to ensure that at most a single change point falls within. Hence, this development is of general interest for change point analysis.}  \\
(iii) Note that the procedure of estimating change points in low-rank plus sparse VAR models is computationally expensive, even in the presence of a single change point, since it requires performing numerous singular value decompositions. We consider a surrogate model that comes with significant computational savings and under \textit{certain regularity conditions} exhibits \textit{similar accuracy} to the posited model. Specifically, we posit a lag-1 VAR model, wherein the transition matrices $A_j^\star$ are assumed to be \textit{weakly sparse} (see, e.g., \cite{negahban2012unified}), as an alternative modeling framework. The main reason is that the presence of low rank structure renders the auto-regressive parameters in the original model dense. The weak sparse assumption adequately accommodates dense structures under certain conditions and hence can prove useful in certain settings (carefully discussed in the sequel) for change point detection problems. Further, the theoretical properties of exhaustive-search based anomaly detection for weakly sparse VAR models have not been investigated in the literature, and hence this development is of independent interest. \\
 (iv) To establish non-asymptotic error bounds on the model parameters of stationary sparse models, one needs to verify that the commonly imposed (see, e.g. \cite{loh2012high}) restricted strong convexity and deviation bound conditions hold (see Propositions 4.2 and 4.3 in \cite{basu2015regularized}).
 Verifying these assumptions in the presence of change points in the posited reduced rank VAR model -which technically is equivalent to working with a misspecified model (see also discussion in \cite{roy2017change})- represents a non-trivial challenge. This issue is rigorously and successfully addressed in the sequel, together with the introduction of a new version of the deviation bound condition that allows working with misspecified models (technical details presented in Appendix~A). \\
 (v) Finally, obtaining consistent model parameters for each segment identified after detecting the change points requires some care, given the non-stationary nature of the posited model above. This is successfully addressed for the case of a single and multiple change points in Sections \ref{sec:model} and \ref{sec:multiplecp}, respectively, and for the surrogate model in Section \ref{sec:surrogate}.

The remainder of the paper is organized as follows. In Section~\ref{sec:model}, we formulate the model with a single change point, provide a detection procedure based on exhaustive search, and establish theoretical properties for the change point and model parameter estimates. Section~\ref{sec:multiplecp} discusses the case of multiple change points. It introduces a two-step detection algorithm and establishes consistency of the obtained estimates for the change points and model parameters, leveraging results from Section \ref{sec:model}. To reduce computations for detecting the change point(s) in the reduced rank VAR model, we introduce a weakly sparse surrogate model in Section~\ref{sec:surrogate} and establish that under certain regularity conditions on the structure of the transitions matrices $A_j^\star$ of the reduced rank model, the estimated change points from the surrogate model are consistent ones for data generated by the former. 
Section~\ref{sec:simu} presents a number of numerical experiments to illustrate and assess the performance of the estimates obtained from the single and multiple change points detection procedures. Two real data sets (one on EEG and the other on macroeconomics data) are analyzed using the proposed detection procedures in Section~\ref{sec:application}. Some concluding remarks are drawn in Section~\ref{sec:concluding remarks}. Additional technical conditions, proofs of the main results and additional numerical work are available in the Supplement.

\noindent
\textbf{Notation:} Throughout this paper, we denote with a superscript ``$\star$" the true value of the model parameters. For any $p \times p$ matrix, we use $\|\cdot\|_2$, $\|\cdot\|_F$, and $\|\cdot\|_*$ to represent the spectral, Frobenius, and nuclear norm, respectively. For any matrix $A$, $A^\prime$ denotes its transpose, and $A^\dagger$ denotes the conjugate transpose of $A$, while the $\ell_0$, $\ell_1$, and $\ell_\infty$ norms of the vectorized form of $A$ are denoted by: $\|A\|_0 = \text{Card}(\text{vec}(A))$, $\|A\|_1 = \|\text{vec}(A)\|_1$, and $\|A\|_\infty = \|\text{vec}(A)\|_\infty$, respectively. We use $\Lambda_{\max}(\mathbf{X})$ and $\Lambda_{\min}(\mathbf{X})$ to represent the maximum and minimum eigenvalue of the realization matrix $\mathbf{X}$. 


\section{Single Change Point Model Formulation and Detection Procedure}\label{sec:model}

We start by introducing a piece-wise stationary structured VAR(1) model that has a single change point. Suppose there is a $p$-dimensional time series $\{X_t\}$ observed at $T+1$ points: $t=0,1,\dots, T$. Further, there exists a change point, $0 < \tau^\star < T$, so that the available time series can be modeled according to the following two models in the time intervals $[0, \tau^\star)$ and $[\tau^\star+1, T)$, respectively:
\begin{equation}
    \label{eq:1}
    \begin{aligned}
    X_t &= A_1^\star X_{t-1} + \epsilon_t^1, \quad t=1,2,\dots, \tau^\star, \\
    X_t &= A_2^\star X_{t-1} + \epsilon_t^2, \quad t=\tau^\star+1, \dots, T,
    \end{aligned}
\end{equation}
where $X_t \in \mathbb{R}^{p}$ is a vector of observed time series at time $t$, and $A_1^\star$ and $A_2^\star$ are the $p\times p$ transition matrices for the corresponding models in the two time intervals, and the $p$ dimensional error processes  $\epsilon_t^1$ and $\epsilon_t^2$ are independent and identically drawn from Gaussian distributions with mean zero and covariance matrix $\sigma^2 I$ for some fixed $\sigma$.
It is further assumed that the transition matrices comprise of two time-varying components,
a low-rank and a sparse one:  
\begin{equation}
    \label{eq:2}
    A_1^\star = L_1^\star + S_1^\star \quad \text{and} \quad A_2^\star = L_2^\star + S_2^\star.
\end{equation}
The rank of the low-rank components and the density (number of non-zero elements) of the sparse components are denoted by $\text{rank}(L_1^\star) = r_1^\star$,  $\text{rank}(L_2^\star) = r_2^\star$, $d_1^\star = \|S_1^\star\|_0$ and $d_2^\star = \|S_2^\star\|_0$, respectively, and satisfy $r_1^\star, r_2^\star \ll p$, $d_1^\star, d_2^\star \ll p^2$.

\subsection{Detection Procedure}

Let $\lbrace X_0, X_1, \dots, X_T \rbrace$ be a sequence of observations generated from the VAR model posited in \eqref{eq:1} with the structure of the transition matrices given by \eqref{eq:2}. Then, for any time point $\tau\in\{1,\cdots,T\}$
the corresponding objective functions for estimating the model parameters in the intervals $[1,\tau)$ and $[\tau,T)$ are given by:
\begin{equation*}
    \ell(L_1, S_1; \mathbf{X}^{[1:\tau)}) \overset{\text{def}}{=} \frac{1}{\tau-1}\sum_{t=1}^{\tau-1} \|X_t - (L_1+S_1)X_{t-1}\|_2^2 + \lambda_1\|S_1\|_1 + \mu_1\|L_1\|_*,
\end{equation*}
\begin{equation*}
    \ell(L_2, S_2; \mathbf{X}^{[\tau:T)}) \overset{\text{def}}{=} \frac{1}{T-\tau}\sum_{t=\tau}^{T-1} \|X_t - (L_2+S_2)X_{t-1}\|_2^2 + \lambda_2\|S_2\|_1 + \mu_2\|L_2\|_*,
\end{equation*}
where $\mathbf{X}^{[b:e)}$ denotes the data $\{X_t\}$ from time points $b$ to $e$, and the non-negative tuning parameters $\lambda_1$, $\lambda_2$, $\mu_1$, and $\mu_2$ control the regularization of the sparse and the low-rank components in the corresponding transition matrices.

Next, we introduce the objective function with respect to the change point: for any time point $\tau \in \lbrace 1,2, \dots, T-1 \rbrace$, 
\begin{equation}
    \label{eq:3}
    \ell(\tau; L_1, L_2, S_1, S_2) \overset{\text{def}}{=} \frac{1}{T-1}\left( \sum_{t=1}^{\tau-1}\|X_t - (L_1+S_1)X_{t-1}\|_2^2 + \sum_{t=\tau}^{T-1}\|X_t - (L_2+S_2)X_{t-1}\|_2^2\right).
\end{equation}

The estimator $\widehat{\tau}$ of the change point $\tau^\star$ is given by:
\begin{equation}
    \label{eq:4}
    \widehat{\tau} \overset{\text{def}}{=} \argmin_{\tau \in \mathcal{T}}\ell(\tau; \widehat{L}_{1,\tau}, \widehat{L}_{2,\tau}, \widehat{S}_{1,\tau}, \widehat{S}_{2,\tau}),
\end{equation}
for the search domain $\mathcal{T} \subset \{1,2,\dots, T\}$, where, for each $\tau \in \mathcal{T}$, the estimators $\widehat{L}_{1,\tau}$, $\widehat{L}_{2,\tau}$, $\widehat{S}_{1,\tau}$, $\widehat{S}_{2,\tau}$ are derived from the optimization program (\ref{eq:4}) with tuning parameters $\mu_{1,\tau}$, $\mu_{2,\tau}$, $\lambda_{1,\tau}$, and $\lambda_{2,\tau}$, respectively. Algorithm 1 in Appendix B describes in detail the  key steps in estimating the change point $\tau^\star$ together with the model parameters. 


\subsection{Theoretical Properties}\label{sec:theory}

Next, we address the issue of identifiability of model parameters due to the posited decomposition of the transition matrices into low rank and sparse components. The key idea is to restrict the ``spikiness" of the low rank component, so that it can be distinguished from
the sparse component. \cite{agarwal2012noisy} introduced
the space $\Omega$ defined as
\begin{equation*}
    \Omega \overset{\text{def}}{=} \left\{ L_j^\star \in \mathbb{R}^{p \times p}: \|L_j^\star\|_\infty \leq \frac{\alpha_L}{p} \right\},\quad j=1,2,
\end{equation*}
wherein the universal parameter $\alpha_L$ defines the \textit{radius of nonidentifiability} that controls the degree of separating the sparse component from the low-rank one. Note that a larger $\alpha_L$ allows the low-rank component to absorb most of the signal, thus making it harder to identify the sparse component, and vice versa.

Thus, the estimators of the decomposition of the transition matrices $A_j$ are defined as follows, for any fixed time point $\tau$:
\begin{equation}
    \label{eq:5}
        (\widehat{L}_{1,\tau}, \widehat{S}_{1,\tau}) \overset{\text{def}}{=} \argmin_{\substack{L_1\in \Omega \\ L_1, S_1\in \mathbb{R}^{p\times p}}}\ell(L_1, S_1; \mathbf{X}^{[1:\tau)}),\quad
        (\widehat{L}_{2,\tau}, \widehat{S}_{2,\tau}) \overset{\text{def}}{=} \argmin_{\substack{L_2\in \Omega \\ L_2, S_2\in \mathbb{R}^{p\times p}}}\ell(L_2, S_2; \mathbf{X}^{[\tau:T)}).
\end{equation}
 
Next, we introduce an important quantity for future developments, 
the \textit{information ratio} that measures the relative strength of the maximum signal in the transition matrix $A_j^\star$ generated by the low-rank component vis-a-vis its sparse counterpart, defined as:
\begin{equation*}
    \gamma_j \overset{\text{def}}{=} \frac{\|L_j^\star\|_\infty}{\|S_j^\star\|_\infty},\quad j=1,2.
\end{equation*}
\begin{remark}
\label{remark:1}
Based on the definition of the information ratio, some algebra provides guidance on the identifiability conditions that need to be imposed on the transition matrices $A_j^\star$ and their constituent parts. Specifically, for the low rank component we obtain:
\begin{equation*}
    \begin{aligned}
    \|A_2^\star - A_1^\star\|_2 &= \|(L_2^\star - L_1^\star) + (S_2^\star - S_1^\star) \|_2 \geq \|L_2^\star - L_1^\star\|_2 - \|S_2^\star - S_1^\star\|_2 \\
    &\geq \|L_2^\star - L_1^\star\|_2 - p\left(\|S_2^\star\|_\infty + \|S_1^\star\|_\infty\right) \\
    &\geq \|L_2^\star - L_1^\star\|_2 - \alpha_L\left(\frac{1}{\gamma_2} + \frac{1}{\gamma_1}\right) \geq v_L - \frac{\alpha_L(\gamma_1+\gamma_2)}{\gamma_1\gamma_2}.
    \end{aligned}
\end{equation*}
Analogous derivations for the sparse component yield: $\|A_2^\star - A_1^\star\|_2 \geq \|S_2^\star - S_1^\star\|_2 - {2\alpha_L}/{p} \geq v_S - {2\alpha_L}/{p}$, where $v_L\equiv \|L_2^\star - L_1^\star\|_2 \geq 0 $, $v_S
\equiv \|S_2^\star - S_1^\star\|_2 \geq 0$ are  norm differences for the low-rank and the sparse components, respectively. 
\end{remark}
Based on Remark \ref{remark:1}, it can be seen that: (1) when $\gamma_1 \leq 1$ or $\gamma_2 \leq 1$, we have that ${(\gamma_1+\gamma_2)}/{\gamma_1\gamma_2} \geq 2 > {2}/{p}$ (since $p \gg 2$ in a high dimensional setting). The latter fact implies that in order for changes in the transition matrices $A_j^\star$ to be identifiable -and consequently $\tau^\star$- the difference in the $\ell_2$ norm of the low-rank components must significantly exceed that of the sparse components;
(2) when both $\gamma_1 > 1$ and $\gamma_2 > 1$, then the quantity ${(\gamma_1+\gamma_2)}/{\gamma_1\gamma_2}$ is strictly decreasing with respect to $\gamma_1$ and $\gamma_2$. Note that in case $1<\gamma_1 \leq p$ and $1<\gamma_2 \leq p$, ${(\gamma_1+\gamma_2)}/{\gamma_1\gamma_2} \geq {2}/{p}$. Combining these two cases leads to the conclusion that when $\gamma_1 \leq p$ and $\gamma_2 \leq p$, the difference in the $\ell_2$ norm $v_L$ between the low-rank components must be larger than $v_S$, the norm difference between the sparse components to guarantee that the change between the transition matrices is detectable. 

The following remark discusses an extreme case, wherein the signal in the low-rank components is dominant, but their $\ell_2$ norm difference is negligible.
\begin{remark}
\label{remark:2}
Suppose the low-rank components are dominant (i.e. $\gamma_1,\gamma_2 \geq 1$), but their $\ell_2$ norm difference change is small; i.e. $\|L_2^\star - L_1^\star\|_2 \leq \epsilon$, with $\epsilon>0$ being a small enough constant). Then, we have: 
\begin{equation*}
    \begin{aligned}
    \|A_2^\star - A_1^\star\|_2 & \geq \|S_2^\star - S_1^\star\|_\infty - \epsilon \geq \|S_2^\star\|_\infty - \|S_1^\star\|_\infty - \epsilon \geq \frac{1}{\gamma_2}\|L_2^\star\|_\infty - \frac{\alpha_L}{p\gamma_1} - \epsilon \\
    &= \frac{1}{\gamma_2}\left(\|L_2^\star\|_\infty - \frac{\alpha_L}{p}\frac{\gamma_2}{\gamma_1}\right) - \epsilon.
    \end{aligned}
\end{equation*}
Note that since the low rank components are constrained to be in the $\Omega$ space -$\|L_2^\star\|_\infty \leq \alpha_L / p$- it implies that the transition matrices are identifiable, only if $\gamma_2 < \gamma_1$ and $\|S_2^\star\|_\infty > \|S_1^\star\|_\infty$. The roles of $L_2^\star$ and $L_1^\star$ can be swapped to obtain that only if $\gamma_2 \neq \gamma_1$ and $\|S_2^\star\|_\infty \neq \|S_1^\star\|_\infty$, is the change in the full transition matrices $A_j^\star$ identifiable, which is intuitive.
\end{remark}
The derivations in the two Remarks provide insights into the necessary assumptions needed to establish the theoretical results, presented next.
\begin{itemize}[leftmargin=0pt]
    \itemsep.1em
    \item[(H1)]
    {
    There exists a positive constant $C_0 > 0$ such that
    \begin{equation*}
        \Delta_T (v_S^2+v_L^2) \geq C_0\left(d^\star_{\max}\log(p\vee T) + r^\star_{\max}(p\vee \log T)\right),
    \end{equation*}
    where $\Delta_T$ is the spacing between the change point $\tau^\star$ and the boundary, and $v_S$, $v_L$ are the jump sizes, defined as 
    \begin{equation*}
        \Delta_T = \min\{\tau^\star-1, T - \tau^\star\},\quad v_S = \|S_2^\star - S_1^\star\|_2,\quad v_L = \|L_2^\star - L_1^\star\|_2.
    \end{equation*}
    Further, at least one of $v_S, v_L$ is strictly positive.
    }
    {
    \item[(H2)] (Identifiability conditions)
    Consider low rank matrices $L_1^\star$, $L_2^\star$, and their corresponding Singular Value Decompositions: $L_j^\star = U_j^\star D_j^\star V_j^{\star^\prime}$, where $D_j^\star = \text{diag}(\sigma_1^j, \dots, \sigma_{r_j}^j, 0, \dots, 0)$, for $j=1,2$ and $U_j^\star, V_j^\star$ are orthonormal. Then,
    \begin{itemize}[leftmargin=*]
        \item[(1)] there exists a universal positive constant $M_S > 0$, such that for the sparse matrices $S_j^\star$, we have: $\|S_j^\star\|_\infty \leq M_S < +\infty$, $j=1,2$;
        \item[(2)] there exists a large enough constant $c>0$, such that the diagonal matrices $D_j^\star$ satisfy: $\max_{j=1,2}\|D_j^\star\|_\infty \leq c < +\infty$; further the orthonormal matrices $U_j^\star$ and $V_j^\star$ satisfy: $\max_{j=1,2} \left\{ \|U_j^\star\|_\infty, \|V_j^\star\|_\infty \right\} = \mathcal{O}\left(\sqrt{\frac{\alpha_L}{r_{\max}p}}\right)$, where $r_{\max} = \max\{r_1^\star, r_2^\star\}$. In addition, we assume that $\alpha_L = \mathcal{O}\left(p\sqrt{\frac{\log(pT)}{T}}\right)$.
        \item[(3)] the maximal sparsity level $d^\star_{\max} = \max\{ d_1^\star, d_2^\star\}$ satisfies:
        $
            d^\star_{\max} \leq \frac{1}{C_{\max}}\sqrt{\frac{T}{\log(pT)}}, 
        $ for a large enough positive constant $C_{\max}>0$.
    \end{itemize}}
    \item[(H3)] (Restrictions on the search domain $\mathcal{T}$) The change point $\tau^\star$ belongs to the search domain by $\mathcal{T} \subset \{1,2,\dots, T-1\}$ and denote the search domain $\mathcal{T} \overset{\text{def}}{=} [a, b]$. Assume that,
    $
        a = \left\lfloor (d_{\max}^\star + \sqrt{r_{\max}^\star})^{1+\eta}\right\rfloor \ \text{and}\ b = \left\lfloor T - (d_{\max}^\star+\sqrt{r_{\max}^\star})^{1+\eta}\right\rfloor,
    $
    and denote $|\mathcal{T}|$ as the length of the search domain, then:
    \begin{equation*}
        \frac{|\mathcal{T}|}{d_{\max}^\star\log(p\vee T) + r_{\max}^\star(p\vee \log T)} \to +\infty,
    \end{equation*}
    where $\eta>0$ is an arbitrarily small positive constant, $d_{\max}^\star = \max\{d_1^\star, d_2^\star\}$, and $r_{\max}^\star = \max\{r_1^\star, r_2^\star\}$.
\end{itemize}

\begin{remark}
\label{remark:3}
Assumption H1 specifies the relationship between the minimum spacing between the change point and the boundaries of the observation time period and the jump sizes for the low rank and sparse components, analogously to the signal-to-noise assumption in \cite{wang2019localizing}. 
Assumptions H2-(1) and H2-(2) define the restricted space for the low rank components $L_j^\star$: $\Omega \overset{\text{def}}{=} \left\{L: \|L_j^\star\|_\infty \leq \frac{\alpha_L}{p} \right\}$; see analogous definitions and discussion in \cite{agarwal2012noisy, basu2019low, bai2020multiple} for identifying low rank and sparse matrices. Assumptions H2-(1-3) are sufficient for satisfying the identifiability condition in \cite{hsu2011robust}, the latter implying that the decomposition $A_j^\star=L_j^\star+S_j^\star$ is \textit{unique}. {This condition is motivated by the so-called ``rank-sparsity" incoherence concept \citep{chandrasekaran2011rank}, with further refinements along the lines of results in \cite{hsu2011robust}. This assumption ensures identifiability of model parameters by putting certain conditions on the singular values, and  left/right orthonormal singular vectors of the low rank component. Specifically, the new assumption controls the maximum number of non-zeros in any row or column of the sparse component, while ensuring that the low rank part has singular vectors far from the coordinate bases. Note that the new conditions do not put any additional constrains on the dimensionality $p$ and further ensure the \textit{uniqueness} of the low rank plus sparse decomposition of the segment specific transition matrices.}
\end{remark}

Note that \cite{agarwal2012noisy} allow $\alpha_L$ to be any constant, whereas we {require $\alpha_L/p$ to be vanishing to obtain consistent estimates, due to the presence of misspecification, since the location of the change points is unknown.} Assumption H3 reflects the restrictions on the boundary of the search domain $\mathcal{T}$ and connects the estimation rate to the length of the search domain (see analogous condition in \cite{roy2017change}).

For any fixed time point $\tau$ in the search domain $\mathcal{T}$, let $(\lambda_{1,\tau}, \mu_{1,\tau})$ be the tuning parameters on $[1, \tau)$, and $(\lambda_{2,\tau}, \mu_{2,\tau})$  the tuning parameters on $[\tau,T)$, respectively. Then, the tuning parameters of the regularization terms are selected as follows:
\begin{equation}
    \label{eq:6}
    \begin{aligned}
        &(\lambda_{1,\tau}, \mu_{1,\tau}) = \left( 4c_0\sqrt{\frac{\log p + \log(\tau-1)}{\tau-1}},\ 4c_0^\prime\sqrt{\frac{p + \log(\tau-1)}{\tau-1}} \right), \\
        &(\lambda_{2,\tau}, \mu_{2,\tau}) = \left( 4c_0\sqrt{\frac{\log p + \log(T-\tau)}{T-\tau}},\ 4c_0^\prime\sqrt{\frac{p + \log(T-\tau)}{T-\tau}} \right),
    \end{aligned}
\end{equation}
for constants $c_0, c_0^\prime>0$.
\begin{theorem}
\label{thm:1}
    Suppose Assumptions H1-H3 hold, and select the tuning parameters according to \eqref{eq:6}. Then, as $T \to +\infty$, there exists a large enough constant $K_0 > 0$ such that
    \begin{equation*}
        \mathbb{P}\left( |\widehat{\tau} - \tau^\star| \leq K_0\frac{d_{\max}^\star\log(p \vee T) + r_{\max}^\star(p\vee \log T)}{v_S^2 + v_L^2} \right) \to 1.
    \end{equation*}
\end{theorem}
The proof of Theorem \ref{thm:1} is provided in Appendix E. Note that the Theorem provides an upper bound for the change point estimation error based on the total sparsity level and the total rank of the model. 

Next, we establish estimation consistency for the model parameters. First, given the estimated change point $\widehat{\tau}$, we remove it together with its $R$-radius neighborhoods $\mathcal{U}(\widehat{\tau}, R)$, to ensure that the remaining time points form two stationary segments. According to Theorem \ref{thm:1}, the radius $R$ can be of the order $d_{\max}^\star\log(p\vee T) + r_{\max}^\star(p\vee \log T)$.

Let $N_j$ be the length of the $j$-th segments after removing the $R$-radius neighborhoods; then, we select another pair of tuning parameters: 
\begin{equation}
    \label{eq:7}
    (\lambda_j, \mu_j) = \left( 4c_1\sqrt{\frac{\log p}{N_j}} + \frac{4c_1\alpha_L}{p},\  4c_1^\prime\sqrt{\frac{p}{N_j}} \right),\quad j=1,2,
\end{equation}
for constants $c_1$, $c_1^\prime$ that can selected using cross-validation. The procedure for selecting them, as well as $c_0,c_0^\prime$ in \eqref{eq:6}, is provided in Section~\ref{sec:simu}.

Note that the tuning parameters provided in \eqref{eq:7} are different from the tuning parameters in \eqref{eq:6}; the $\log T$ terms are eliminated, since on the selected stationary segments the optimal tuning parameters are always feasible. Based on analogous results in \cite{agarwal2012noisy} and \cite{basu2019low} for models whose parameters admit a low rank and sparse decomposition, the optimal tuning parameters in \eqref{eq:7} lead to the optimal estimation rate given in the next Theorem.
\begin{theorem}
\label{thm:2}
Suppose Assumptions H1-H3 hold, and select the tuning parameters according to \eqref{eq:7}. Then, as $T\to +\infty$, there exist universal positive constants $C_1, C_2>0$, so that the optimal solution of \eqref{eq:5} satisfies
\begin{equation*}
    \|\widehat{L}_j - L_j^\star\|_F^2 + \|\widehat{S}_j - S_j^\star\|_F^2 \leq C_1\left( \frac{d_j^\star\log p + r_j^\star p}{N_j} \right) + C_2\frac{d_j^\star\alpha_L^2}{p^2},\quad j=1,2.
\end{equation*}
\end{theorem}
The proof of Theorem \ref{thm:2} is provided in Appendix E. 
\begin{remark}
\label{remark:6}
    Notice that Theorem \ref{thm:2} provides the joint estimation rate for the low-rank and the sparse component. It comprises of two terms, wherein the first one involves the dimensions of the model parameters and converges to zero as the sample size increases, whereas the second term represents the error due to possible unidentifiability of the model parameters. However, in conjunction with Assumption H2 that restricts the space for the low rank component, the second term also converges to zero as the sample size (and hence the dimensionality of the model) increases.
\end{remark}

\section{The Case of Multiple Change Points}\label{sec:multiplecp}

Section \ref{sec:theory} introduced the technical framework and established the
consistency rate for detecting a single change point. Next, these technical developments are leveraged to address the more relevant in practice problem of detecting multiple change points consistently.

We start by formulating the piece-wise VAR model with multiple change points.
Consider the $p$-dimensional VAR(1) process $\{X_t\}$ with $m_0$ change points $1=\tau^\star_0<\tau^\star_1<\cdots<\tau^\star_{m_0}<\tau^\star_{m_0+1}=T$; then, the model under consideration is written as:
\begin{equation}
\label{eq:8}
    X_t = \sum_{j=1}^{m_0+1}\left(A_j^\star X_{t-1} + \epsilon_t^j \right)\mathbf{I}(\tau^\star_{j-1} \leq t < \tau^\star_j),\quad t=1,2,\dots, T,
\end{equation}
where $L_j^\star$ and $S_j^\star$ represent the decomposition of the $j$-th transition matrix into its low-rank and sparse components, and $\mathbf{I}(\tau^\star_{j-1}\leq t < \tau^\star_j)$ denotes the indicator function for the $j$-th stationary segment. Analogously to the single change point case, we define the sparsity level $d_j^\star = \|S_j^\star\|_0$ and rank $r_j^\star = \text{rank}(L_j^\star)$ for the components in each segment, wherein $d_j^\star \ll p^2$ and $r_j^\star \ll p$, (i.e., $d_j^\star = o(p^2)$ and $r_j^\star = o(p)$). Finally, $\epsilon_t^j$'s are independent and independently distributed zero mean Gaussian noise processes with covariance matrices $\sigma^2 I$, $j=1, \ldots, m_0+1$. 

For detecting the change points and estimating the model parameters consistently,
the following minor modifications to Assumptions H1-H3 are required:
\begin{itemize}[leftmargin=0pt]
    \itemsep.1em
    \item[(H1')]
    {
    There exists a positive constant $C_0>0$ such that
    \begin{equation*}
        \Delta_T\min_{1\leq j\leq m_0}\{ v_{j,S}^2+v_{j,L}^2 \} \geq C_0(d_{\max}^\star\log(p\vee T) + r^\star_{\max}(p\vee \log T)),
    \end{equation*}
    where $\Delta_T$ is the minimum spacing defined as $\Delta_T \overset{\text{def}}{=} \min_{1\leq j \leq m_0}|\tau^\star_{j+1} - \tau^\star_{j}|$, and the minimum norm differences (jump sizes) between two consecutive segments are defined as: $v_{j, S} \overset{\text{def}}{=} \|S_{j+1}^\star - S_{j}^\star\|_2$, and $v_{j, L} \overset{\text{def}}{=} \|L_{j+1}^\star - L_{j}^\star\|_2$.
    }
    \item[(H2')] 
    {
    Consider low rank matrices $L_j^\star$, and their corresponding Singular Value Decompositions: $L_j^\star = U_j^\star D_j^\star V_j^{\star^\prime}$, where $D_j^\star = \text{diag}(\sigma_1^j, \dots, \sigma_{r_j}^j, 0, \dots, 0)$, for $j=1,2,\dots,m_0+1$. Then,
    \begin{itemize}
        \item[(1)] there exists a universal positive constant $M_S > 0$, such that for the sparse matrices $S_j^\star$, we have: $\|S_j^\star\|_\infty \leq M_S < +\infty$, $j=1,\dots,m_0+1$;
        \item[(2)] there exists a large enough constant $c>0$, such that the diagonal matrices $D_j^\star$ satisfy: $\max_{j=1,2}\|D_j^\star\|_\infty \leq c < +\infty$, and the orthonormal matrices $U_j^\star$ and $V_j^\star$ such that: $\max_{1\leq j \leq m_0+1} \left\{ \|U_j^\star\|_\infty, \|V_j^\star\|_\infty \right\} = \mathcal{O}\left(\sqrt{\frac{\alpha_L}{r_{\max}p}}\right)$, where $r_{\max} = \max_{1\leq j\leq m_0+1} r_j^\star$. In addition, we assume that $\alpha_L = \mathcal{O}\left(p\sqrt{\frac{\log(pT)}{T}}\right)$.
        \item[(3)] the maximal sparsity level $d^\star_{\max} = \max_{1\leq j\leq m_0+1} d_j^\star$ satisfies:
        $
            d^\star_{\max} \leq \frac{1}{C_{\max}}\sqrt{\frac{T}{\log(pT)}},
        $
        for a large enough positive constant $C_{\max}>0$.
    \end{itemize}
    }
    \item[(H3')] There exists a vanishing positive sequence $\{\xi_T\}$ such that, as $T \to +\infty$,
    \begin{equation*}
        \begin{aligned}
            &\frac{\Delta_T}{T\xi_T(d_{\max}^{\star^3}+r_{\max}^{\star^2})} \to +\infty,\quad  d_{\max}^{\star^2}\sqrt{\frac{\log p}{T\xi_T}} \to 0,\quad r_{\max}^{\star^{\frac{3}{2}}}\sqrt{\frac{p}{T\xi_T}} \to 0, \\
            &\frac{\Delta_T(d_{\max}^\star\log p + r_{\max}^\star p)}{(T\xi_T)^2(d_{\max}^{\star^3} + r_{\max}^{\star^2})} \to C \geq 1,
        \end{aligned}
    \end{equation*}
    for a positive constant $C>0$.
\end{itemize}
Assumptions H1' and H2' are direct extensions of Assumptions H1 and H2 to the multiple change points setting. Assumption H3' provides a minimum distance requirement on the consecutive change points and connects the estimation rate and the minimum spacing between change points.  

Our detection algorithm will leverage results from the single change point case, and thus, we introduce additional assumptions next. As mentioned in the introduction, the use of fused type penalties is not applicable to the low-rank component and hence an entire different detection procedure is required.

\subsection{A Two-step Algorithm for Detecting Multiple Change Points and its Asymptotic Properties}\label{sec:twostep}

\begin{itemize}[leftmargin=0pt]
\itemsep.1em
\item {\bf Step 1:} It is based on Algorithm 1 provided in Appendix B that detects a single change point, additionally equipped with a \textit{rolling window} mechanism to select \textit{candidate} change points. We start by selecting an interval $[b_1,e_1) \subset \{1,2,\dots, T\}, b_1=1$, of length $h$ and employ on it the exhaustive search Algorithm 1 to obtain a candidate change point $\widetilde{\tau}_1$. Next, we shift the interval to the right by $l$ time points and obtain a new interval $[b_2, e_2)$, wherein $b_2=b_1+l$ and $e_2=e_1+l$. The application of Algorithm 1 to $[b_2,e_2)$ yields another candidate change point $\widetilde{\tau}_2$.
This procedure continues until the last interval that can be formed, namely $[b_{\widetilde{m}},e_{\widetilde{m}})$, where $e_{\widetilde{m}}=T$ and $\widetilde{m}$ denotes the number of windows of size $h$ that can be formed. The following Figure \ref{fig:rolling-window} depicts this rolling-window mechanism. 
The blue lines represent the boundaries of each window, awhile the green dashed lines represent the candidate change point in each window. Note that the basic assumption for Algorithm 1 is that there exists a single change point in the given time series. However, it can easily be seen in Figure \ref{fig:rolling-window} that \textbf{not} every window includes a single change point.
\begin{figure}[ht!]
    \centering
    \begin{tikzpicture}[scale=.85]
    \centering
        \draw[->](-8.2,0)--(8,0) node[pos=1,below]{$T$};
        \foreach \x/ \xtext in {-8/$b_1$, -7/$b_2$, -5/$e_1$, -4/$e_2$, 0/$b_j$, 3/$e_j$, 4/$b_{\widetilde{m}}$, 7/$e_{\widetilde{m}}$}{
            \draw (\x cm,-2pt) -- (\x cm,2pt) node[label={[label distance =0.1em]below:{\xtext}}]{};
            \draw[blue,thick] (\x cm,-0.2) -- (\x cm, 4);
        }
        \foreach \truept in {-4.5, 1, 5}{
            \draw(\truept,0)circle[radius=1.5pt];
        }
        \node[at={(-4.5,0)},label={[label distance=0.1em]above:{$\tau^\star_1$}}]{};
        \node[at={(1,0)},label={[label distance=0.1em]above:{$\tau^\star_2$}}]{};
        \node[at={(5,0)},label={[label distance=0.1em]above:{$\tau^\star_3$}}]{};
        
        \fill[red](-4.5,0)circle[radius=1.5pt];
        \fill[red](1,0)circle[radius=1.5pt];
        \fill[red](5,0)circle[radius=1.5pt];
        
        \foreach \estpt in {-6.44, -4.3, 0.8, 4.8}{
            \draw(\estpt,0)circle[radius=1.5pt];
            \draw[dashed,green,thick] (\estpt, -0.2) -- (\estpt, 4);
        }
        \node[at={(-6.44,0)},label={[label distance=0.1em]below:{$\widetilde{\tau}_1$}}]{};
        \node[at={(-4.3,0)},label={[label distance=0.1em]below:{$\widetilde{\tau}_2$}}]{};
        \node[at={(0.8,0)},label={[label distance=0.1em]below:{$\widetilde{\tau}_j$}}]{};
        \node[at={(4.8,0)},label={[label distance=0.1em]below:{$\widetilde{\tau}_{\widetilde{m}}$}}]{};
        
        \fill[green](-6.44,0)circle[radius=1.5pt];
        \fill[green](-4.3,0)circle[radius=1.5pt];
        \fill[green](0.8,0)circle[radius=1.5pt];
        \fill[green](4.8,0)circle[radius=1.5pt];
        
        \draw[ultra thick,->](-3.5,3) -- node[above]{rolling windows} (-0.5,3);
        
        \draw[<->,red,thick](-8,1) -- node[above]{$h$}(-5,1);
        \draw[->,red,thick](-8,3) -- node[above]{$l$}(-7,3);
        
    \end{tikzpicture}
    
    \caption{Depiction of the rolling windows strategy. There are three true change points: $\tau_1^\star$, $\tau_2^\star$, and $\tau_3^\star$ (red dots); the boundaries of the rolling-window are represented in blue lines; the estimated change points in each window are plotted in green dashed lines, where the subscript indicates the index of the window used to obtain it.}
    
    \label{fig:rolling-window}
\end{figure}
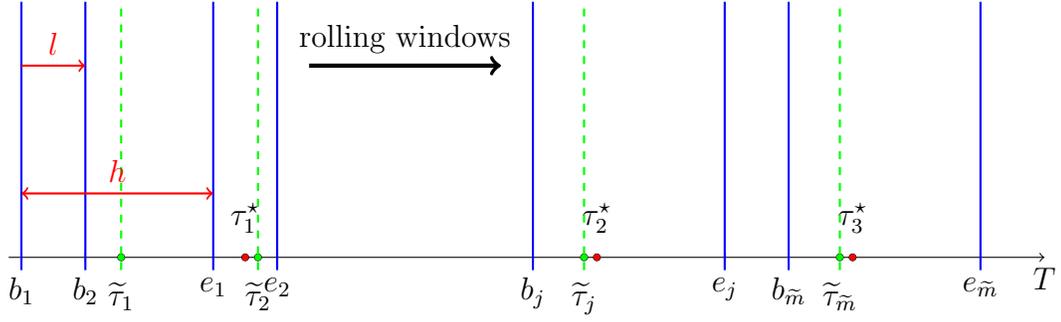

To showcase the last point, we compare the behavior of Algorithm 1 on an interval with and without a change point based on data generated from a low-rank plus sparse VAR process $\{X_t\}$ with $p=20$. We select two windows of length $h=200$, one containing a change point at $t=100$ and another not containing a change point. Plots of the objective function \eqref{eq:3} used in Algorithm 1 for these two windows are depicted in the left and right panels of Figure \ref{fig:interval_cp}, respectively. 
\begin{figure}[htb]
    \centering
    \resizebox{5in}{1.67in}{%
    \begin{subfigure}
        \centering
        \includegraphics[scale=.5]{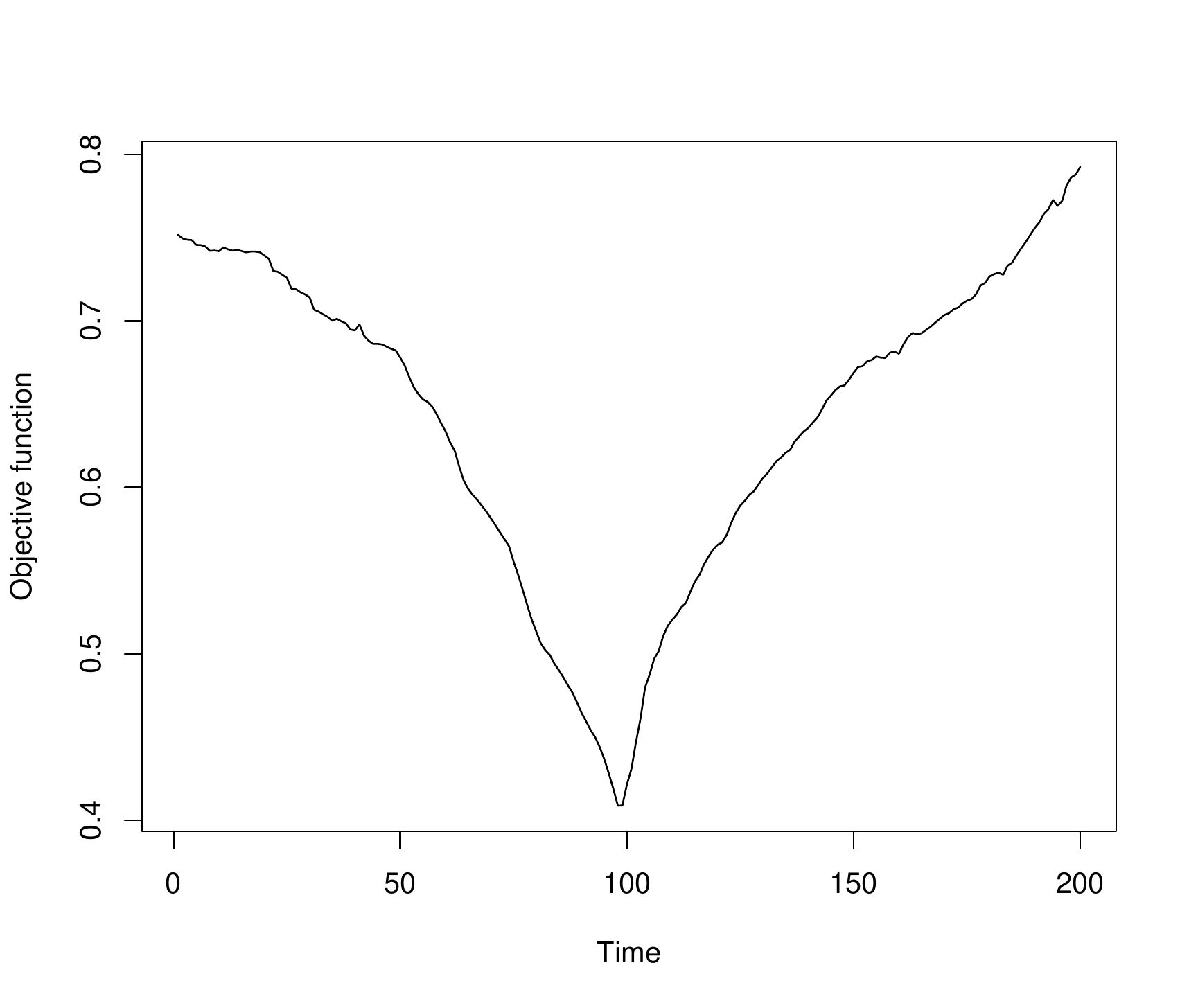}%
        \quad\quad
        \includegraphics[scale=.5]{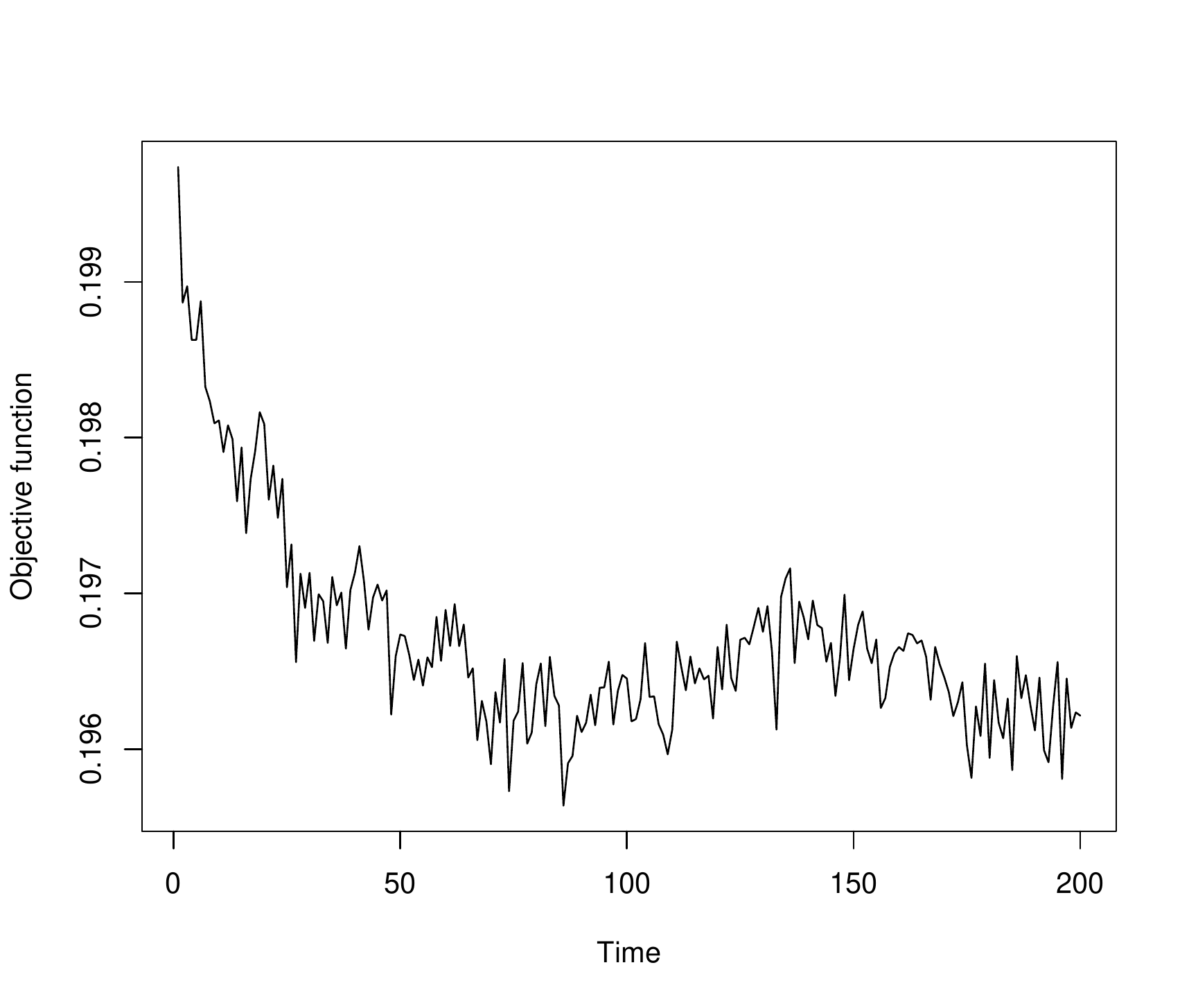}
    \end{subfigure}}
    
    \caption{Plots of the objective functions obtained by an application of Algorithm 1, in the presence (left panel) and absence (right panel) of a true change point.}
    
    \label{fig:interval_cp}
\end{figure}
It can be seen that in the presence of a change point, a clearly identified minimum close to the true change point exists. Contrary, in the absence of a change point, the objective function is mostly flat without a clearly identified minimum. Next, we introduce an assumption on the size of the window $h$ used in the detection procedure:
\begin{itemize}[leftmargin=0pt]
    \item[(H4)] Let $h$ denote the length of the window in the rolling window algorithm. Further, the minimum spacing $\Delta_T$ and the vanishing sequence $\{\xi_T\}$ are defined as in Assumption H3', and let $l$ denote the length by which the window is shifted to the right; it is assumed that:
    \begin{equation*}
        0 < l \leq \max\{\frac{h}{2}, 1\},\ \limsup_{T \to +\infty} \frac{h}{\Delta_T} < 1,\ \text{and}\ \liminf_{T \to +\infty} \frac{h}{T\xi_T} \geq 2.
    \end{equation*}
\end{itemize}
Assumption H4 restricts $h$, so that asymptotically can not include more than a single true change point and also is not too small, so that the deviation bound and restricted eigenvalue conditions used for establishing theoretical properties of the estimates of the model parameters hold for each time segment (see Appendix A). Further, this assumption places an upper bound on the shift $l$, to ensure that no true break point close to the boundary of windows would be missed by the proposed algorithm. The shift size can vary in $[1,h/2]$; a small $l$ helps reduce the finite sample estimation error for locating the break points, while a large $l$ speeds up the detection procedure, by considering fewer rolling windows.

Next, we establish theoretical guarantees for Step 1 of the proposed detection procedure. Denote by $\widetilde{\mathcal{S}}$ the set of candidate change points and by $\mathcal{S}^\star$ the set of true change points. Specifically, $\widetilde{\mathcal{S}}$ is defined as:
\begin{equation*}
    \widetilde{\mathcal{S}} \overset{\text{def}}{=} \left\{ \widetilde{t}_i \in [b_i, e_i): \widetilde{t}_i = \argmin_{\tau \in [b_i, e_i)}\ell(\tau; \widehat{L}_{1,\tau}, \widehat{L}_{2,\tau}, \widehat{S}_{1,\tau}, \widehat{S}_{2,\tau}),\quad  i = 1,2,\dots, \widetilde{m}  \right\},
\end{equation*}
where $[b_i, e_i)$ is the $i$-th rolling-window. Following \cite{chan2014group}, we define the Hausdorff distance between two countable sets on the real line as:
\begin{equation*}
    d_H(A, B) \overset{\text{def}}{=} \max_{b \in B} \min_{a \in A} |b - a|.
\end{equation*}
Next, we extend Theorem \ref{thm:1} to the multiple change points scenario:
\begin{proposition}
\label{prop:1}
Suppose Assumptions H1'-H3' and H4 hold, and select the tuning parameters for each rolling window according to (\ref{eq:6}). Then, as $T\to +\infty$, there exists a large enough constant $K>0$ such that
\begin{equation*}
    \mathbb{P}\left( d_H(\widetilde{\mathcal{S}}, \mathcal{S}^\star) \leq K\frac{d_{\max}^\star\log(p\vee h) + r_{\max}^\star(p\vee \log h)}{\min_{1\leq j\leq m_0}\{v_{j,S}^2+v_{j,L}^2\}}\right) \to 1.
\end{equation*}
\end{proposition}
Proposition~\ref{prop:1} shows that the number of candidate change points identified in Step 1
of the algorithm is an overestimate of the true number of change points. Hence,
a second \textit{screening step} is required to remove the redundant ones.
\item {\bf Step 2:} Let the candidate change points from Step 1 be denoted by $\{s_j\}$, $j=1,2,\cdots, \widetilde{m}$. Then, model \eqref{eq:8} can be rewritten in the following form:
\begin{equation*}
    X_t = \sum_{i=1}^{\widetilde{m}+1}\left((L_{(s_{i-1},s_i)}+S_{(s_{i-1},s_i)})X_{t-1} + \epsilon_t^i\right)\mathbf{I}(s_{i-1} \leq t < s_i),\quad t=1,2,\dots,T,
\end{equation*}
where $L_{(s_{i-1}, s_i)}$ and $S_{(s_{i-1}, s_i)}$ denote for the low-rank and sparse components of the transition matrix in the interval $[s_{i-1}, s_i)$. We define  {\small$0=s_0 < s_1 < s_2 < \dots < s_{\widetilde{m}} < s_{\widetilde{m}+1} = T$} and for ease of presentation use $L_i$ and $S_i$ instead of $L_{(s_{i-1}, s_i)}$ and $S_{(s_{i-1}, s_i)}$ for $i=1,2,\dots, m+1$. We also define matrices $\mathbf{L} \overset{\text{def}}{=} [L_1^\prime, L_2^\prime, \dots, L_{\widetilde{m}+1}^\prime]^\prime$ and $\mathbf{S} \overset{\text{def}}{=} [S_1^\prime, S_2^\prime, \dots, S_{\widetilde{m}+1}^\prime]^\prime$. Estimates for $\mathbf{L}$ and $\mathbf{S}$ are obtained as the solution to the following regularized regression problem:
\begin{equation*}
    (\widehat{\mathbf{L}}, \widehat{\mathbf{S}}) = \argmin_{L_i,S_i, 1\leq i \leq \widetilde{m}+1} \sum_{i=1}^{\widetilde{m}+1} \left\{
    \frac{1}{s_i - s_{i-1}}\sum_{t=s_{i-1}}^{s_i-1}\|X_t - (L_i+S_i)X_{t-1}\|_2^2 + \lambda_i\|S_i\|_1 + \mu_i \|L_i\|_*\right\},
\end{equation*}
with tuning parameters $(\bm{\lambda}, \bm{\mu}) = \{(\lambda_i, \mu_i)\}_{i=1}^{\widetilde{m}+1}$. Next, we define the objective function with respect to $(s_1,s_2, \dots, s_m)$:
\begin{equation}
    \label{eq:9}
    \mathcal{L}_T(s_1,s_2,\dots,s_m; \bm{\lambda}, \bm{\mu}) \overset{\text{def}}{=} \sum_{i=1}^{\widetilde{m}+1} \left\{\sum_{t=s_{i-1}}^{s_i-1}\|X_t - (\widehat{L}_i+\widehat{S}_i)X_{t-1}\|_2^2 + \lambda_i\|\widehat{S}_i\|_1 + \mu_i \|\widehat{L}_i\|_*\right\}.
\end{equation}
Then, for a suitably selected penalty sequence $\omega_T$, specified in the upcoming Assumption H5, we consider the following \textit{information criterion} defined as:
\begin{equation}
    \label{eq:10}
    \text{IC}(s_1, s_2,\dots, s_m; \bm{\lambda}, \bm{\mu}, \omega_T) \overset{\text{def}}{=} \mathcal{L}_T(s_1,\dots,s_m; \bm{\lambda}, \bm{\mu}) + m\omega_T.
\end{equation}
The second step selects a subset of initial $\widetilde{m}$ change points from the first step by solving:
\begin{equation*}
    (\widehat{m}, \widehat{\tau}_i, i=1,2,\dots, \widehat{m}) = \argmin_{0\leq m \leq \widetilde{m}, (s_1,\dots, s_m)} \text{IC}(s_1,\dots,s_m; \bm{\lambda}, \bm{\mu}, \omega_T).
\end{equation*}
\end{itemize}
Algorithm 2 in Appendix B describes in detail the key steps for screening the candidate change points by minimizing the information criterion.

The following two additional assumptions on the minimum spacing $\Delta_T$ and the selection of tuning parameters are required to establish the main theoretical results.
\begin{itemize}[leftmargin=0pt]
    \item[(H5)] 
    Assume that $m_0T\xi_T(d_{\max}^{\star^2} + r_{\max}^{\star^{\frac{3}{2}}})/ \omega_T \to 0$ and $m_0\omega_T / \Delta_T \to 0$ as $n \to +\infty$.
    \item[(H6)] Suppose $(s_1, \dots, s_m)$ are a set of change points obtained from the Step 1, we consider the following scenarios: (a) if $|s_i - s_{i-1}| \leq T\xi_T$, select $\lambda_i = c\sqrt{T\xi_T \log p}$ and $\mu_i = c\sqrt{T\xi_T p}$, for $i=1,2,\dots, m$; (b) if there exist two true change points $\tau_j^\star$ and $\tau_{j+1}^\star$ such that $|s_{i-1} - \tau_j^\star| \leq T\xi_T$ and $|s_i - \tau_{j+1}^\star| \leq T\xi_T$, select $\lambda_i = 4\left(c\sqrt{\frac{\log p}{s_i - s_{i-1}}} + M_Sd_{\max}^\star\frac{T\xi_T}{s_i - s_{i-1}}\right)$ and $\mu_i = 4\left(c\sqrt{\frac{p}{s_i - s_{i-1}}} + \alpha_L \sqrt{r_{\max}^\star}\frac{T\xi_T}{s_i - s_{i-1}}\right)$; (c) otherwise, select $\lambda_i = 4c\sqrt{\frac{\log p + \log(s_i-s_{i-1})}{s_i - s_{i-1}}}$ and $\mu_i = 4c\sqrt{\frac{p + \log(s_i-s_{i-1})}{s_i - s_{i-1}}}$, for some large constant $c$.
\end{itemize}

Assumption H5 connects the screening penalty term $\omega_T$, defined with the information criterion \eqref{eq:10}, and the minimum spacing $\Delta_T$ allowed between the change points. Assumption H6 provides the specific rate of the tuning parameters used in the regularized optimization problem formulated in \eqref{eq:9}. Note that Assumption H6 is required even in standard lasso regression problems for independent and identically distributed data and in the absence of change points \citep{zhang2008sparsity}. In the literature on change points analysis with misspecified models, a more complex selection of the tuning parameters is needed \citep{chan2014group, roy2017change}. Then, the following Theorem establishes the main result of estimating consistently the number of change points and their locations.
\begin{theorem}
\label{thm:3}
Suppose Assumptions H1'--H3', and H4--H6 hold. As $T \to +\infty$, the minimizer $(\widehat{\tau}_1, \dots, \widehat{\tau}_{\widehat{m}})$ of (\ref{eq:10}) satisfies: $\mathbb{P}(\widehat{m} = m_0) \to 1$. Further, there exists a large enough positive constant $B>0$ so that
\begin{equation*}
    \mathbb{P}\left( \max_{1\leq j \leq m_0}|\widehat{\tau}_j - \tau_j^\star| \leq Bm_0 T\xi_T \frac{d_{\max}^{\star^2} + r_{\max}^{\star^{\frac{3}{2}}}}{\min_{1\leq j\leq m_0}\{v_{j,S}^2 + v_{j,L}^2\}} \right) \to 1.
\end{equation*}
\end{theorem}

\begin{remark}
\label{remark:7}
For a finite number of change points $m_0$, the sequence $\{\xi_T\}$ can be selected as ${\left(d_{\max}^\star\log(p\vee T) + r_{\max}^\star(p\vee \log T)\right)^{1+\frac{\rho}{2}}}/{T}$
for some small $\rho>0$. Assuming that the maximum rank among all the low-rank components and the maximum sparsity level among all the sparse components satisfy $d_{\max}^{\star^2} + r_{\max}^{\star^{\frac{3}{2}}} = o\bigg(\left(d_{\max}^\star\log(p\vee T) + r_{\max}^\star(p\vee \log T)\right)^{\frac{\rho}{2}}\bigg)$, then the order of detecting the relative location -$\tau^\star_j/T$- becomes ${\left(d_{\max}^\star\log(p\vee T) + r_{\max}^\star(p\vee \log T)\right)^{1+\rho}}/{T}$ in Theorem \ref{thm:3}. Finally, one can choose the penalty tuning parameter $\omega_T$ to be of order $\left(d_{\max}^\star\log(p\vee T) + r_{\max}^\star(p\vee \log T)\right)^{1+2\rho}$ in this setting, and the minimum spacing $\Delta_T$ to be at least of order $\left(d_{\max}^\star\log(p\vee T) + r_{\max}^\star(p\vee \log T)\right)^{2+\rho}$ in accordance to Assumption H3'. Comparing the consistency rates provided in Theorem \ref{thm:3} with those in \cite{safikhani2017joint}, the additional term $r^\star_{\max}(p\vee\log T)$ reflects the complexity of estimating the low-rank components in the model.
\end{remark}
\begin{remark}
\label{remark:8}
{\bf Computational cost of the rolling windows strategy.} For the proposed $p$-dimensional VAR model with $T$ observations and window size $h=\mathcal{O}(T^\delta)$, where $\delta \in (0, 1]$, the computational complexity of the first step is of order $\mathcal{O}(TC(T))$, and the second screening step is of order $\mathcal{O}(T^{1-\delta}C(T))$, where $C(T)$ is the computational cost for model parameters estimation for every search. Hence, the overall complexity is $\mathcal{O}(TC(T))$.
\end{remark}
The following corollary provides the error bound for consistent estimation of the low-rank and the sparse components, which is directly extended from Theorem \ref{thm:2} to the multiple change points scenario. To obtain the stationary time series for each segments, we employ the exact same technique of removing $R$-radius neighborhoods for every estimated change point. In accordance to Theorem \ref{thm:3}, the radius $R$ should be at least of order $Bm_0T\xi_T(d_{\max}^{\star^2} + r_{\max}^{\star^{\frac{3}{2}}})$ for some large constant $B>0$. Denote the length of the $j$-th stationary segment by $N_j$, after removing the $R$-radius neighborhoods for each estimated change point.
\begin{corollary}
\label{cor:1}
Given the estimated change points: $1=\widehat{\tau}_0<\widehat{\tau}_1<\cdots<\widehat{\tau}_{\widehat{m}}<\widehat{\tau}_{\widehat{m}+1}=T$, let Assumptions H1'-H3' and H4 hold and remove the $R$-radius neighborhoods for each $\widehat{\tau}_j$ for $j=1,2,\dots,\widehat{m}+1$. Further, by using the following tuning parameters: $(\lambda_j, \mu_j) = \left( 4c_1\sqrt{\frac{\log p}{N_j}} + \frac{4c_1\alpha_L}{p},\  4c_1^\prime\sqrt{\frac{p}{N_j}} \right)$, where $c_1, c_1^\prime$ are positive constants. For $T \to +\infty$, there exist universal positive constants $C_1^\prime, C_2^\prime>0$ such that for each selected segment, the estimated low-rank and the sparse components satisfy
\begin{equation*}
    \|\widehat{L}_j - L_j^\star\|_F^2 + \|\widehat{S}_j - S_j^\star\|_F^2 \leq C_1^\prime\left( \frac{d_j^\star\log p + r_j^\star p}{N_j} \right) + C_2^\prime\frac{d_j^\star\alpha_L^2}{p^2}.
\end{equation*}
\end{corollary}

\begin{itemize}[leftmargin=0pt]
    \item {\bf Step 3 (Optional):} After the second Step, the results in Theorem \ref{thm:3} and Corollary \ref{cor:1} ensure accurate estimation of the number of change points and their locations, as well as of the underlying model parameters across the stationary segments. However, a further refinement and hence a tighter bound on the result provided in Theorem \ref{thm:3} can be obtained through the following re-estimation procedure (see also discussion on this point in \cite{wang2019localizing}). Specifically, the conclusions in Theorem~\ref{thm:3} ensure that $\widehat{m} = m_0$ almost surely and also provide good estimates of the boundaries of the stationary segments. Then, for an estimated change point $\widehat{\tau}_j$, consider a ``refined" interval $(s_j,e_j) \overset{\text{def}}{=} (2\widehat{\tau}_{j-1}/3+\widehat{\tau}_j/3, 2\widehat{\tau}_{j}/3+\widehat{\tau}_{j+1}/3)$ for $j=1,2,\dots,\widehat{m}$, where $\tau_0=0$. Then, we define the objective function:
    \begin{equation*}
        \ell(\tau; s_j, e_j, A_{j,1}, A_{j,2}) \overset{\text{def}}{=} \frac{1}{e_j-s_j}\left(\sum_{\tau=s_j}^{\tau-1}\|X_t - A_{j,1}X_{t-1}\|_2^2 + \sum_{t=\tau}^{e_j}\|X_t - A_{j,2}X_{t-1}\|_2^2\right),
    \end{equation*}
    and a ``refined" change point together with the refitted model parameters corresponds to:
    \begin{equation}
    \label{eq:11}
        (\widetilde{\tau}_j, \widetilde{A}_{j,1}, \widetilde{A}_{j,2}) = \argmin_{\tau \in (s_j, e_j)}\ell(\tau; s_j, e_j, A_{j,1}, A_{j,2})
    \end{equation}
\end{itemize}
According to the proposed refinement, we derive the following corollary:

\begin{corollary}
\label{cor:2}
Suppose Assumptions H1'--H3', and H4--H6 hold. As $T \to +\infty$, the minimizer $(\widetilde{\tau}_1, \dots, \widetilde{\tau}_{\widehat{m}})$ of \eqref{eq:11} satisfies:
\begin{equation*}
    \mathbb{P}\left(\max_{1\leq j\leq m_0}|\widetilde{\tau}_j - \tau_j^\star| \leq K\frac{d_{\max}^\star\log(p\vee h) + r_{\max}^\star(p\vee \log h)}{\min_{1\leq j\leq m_0}\{v_{j,S}^2+v_{j,L}^2\}}\right) \to 1.
\end{equation*}
\end{corollary}

\begin{remark}
\label{remark:9}
Note that in the bound of Corollary \ref{cor:2}, the maximum density across all sparse components $d_{\max}^{\star}$ appears as a linear term, instead of a quadratic one in Theorem \ref{thm:3}. This refinement is primarily of theoretical interest, since as the numerical work in Section \ref{sec:multi-performance} indicates the detection procedure based on Steps 1 and 2 achieves very accurate estimates of the change points and the model parameters.
\end{remark}

{
\begin{remark}
\label{remark:new}
Corollary \ref{cor:2} indicates that the high probability finite sample bound on the estimation error depends on the maximum sparsity level $d_{\max}^\star$ among the sparse components, the maximum rank $r^\star_{\max}$ among the low rank components, the dimension $p$, and the signal strength $v_S$, $v_L$ of the sparse and low rank components. Note that the issue of obtaining asymptotic distributions for the estimated change points is a rather complicated task and has not been addressed in the literature even for much simpler models, including sparse mean shift models.
\end{remark}
}


\section{A Fast Procedure Based on a Surrogate Model}\label{sec:surrogate}

Remark \ref{remark:8} shows that identifying multiple change points in a low-rank and sparse VAR model is computationally expensive, due to the presence of the nuclear norm and the need for selecting the tuning parameters through a 2-dimensional grid search.

The question addressed next is whether there are settings wherein the nature of the signal in the norm difference $||A_j^\star-A_{j+1}^\star||_2$ is such that it can be \textit{adequately captured} by a less computationally demanding surrogate model. For example, if the norm difference is primarily due to a large enough change in the sparse component, it is reasonable to expect that a \textit{surrogate} VAR model with a \textit{sparse} transition matrix may prove adequate under certain regularity conditions. However, if the norm difference is due to a change in the low-rank component, which by construction is dense, a pure sparse VAR model will not be adequate; however, a \textit{weakly sparse} model may be sufficient. Indeed, some numerical evidence suggests that this is the case. Figure \ref{fig:compare-curves} presents plots of the objective functions of the original and the surrogate weakly sparse model under the same experimental setting for a low-rank plus sparse VAR process $\{X_t\}$ with $p=20$, $T=200$, and a single change point at $\tau^\star=100$ with \textit{changes in both the low-rank and sparse} components.
\begin{figure}[htb]
    \centering
    \resizebox{5in}{1.67in}{%
    \begin{subfigure}
        \centering
        \includegraphics[scale=.45]{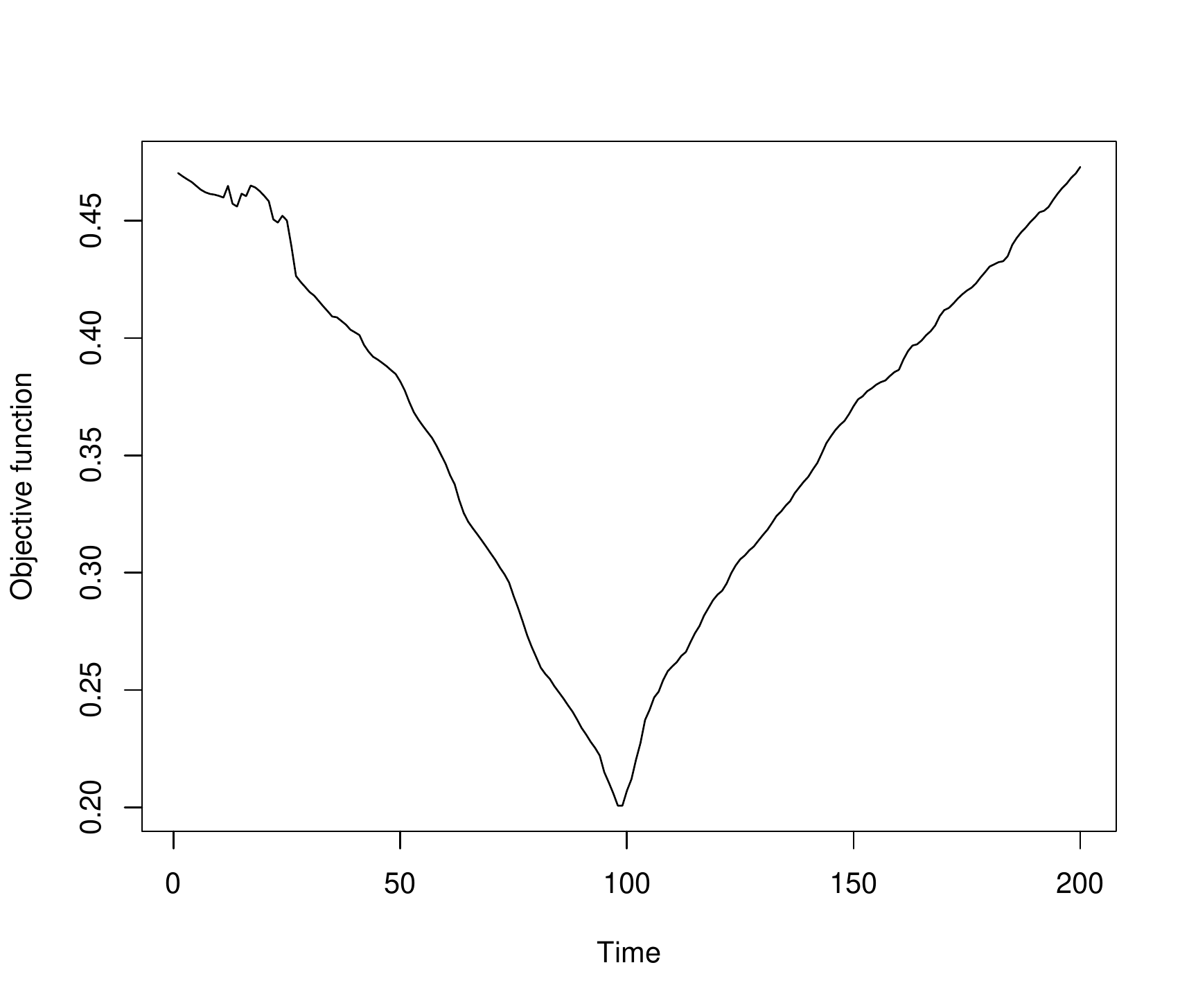}%
        \quad\quad
        \includegraphics[scale=.45]{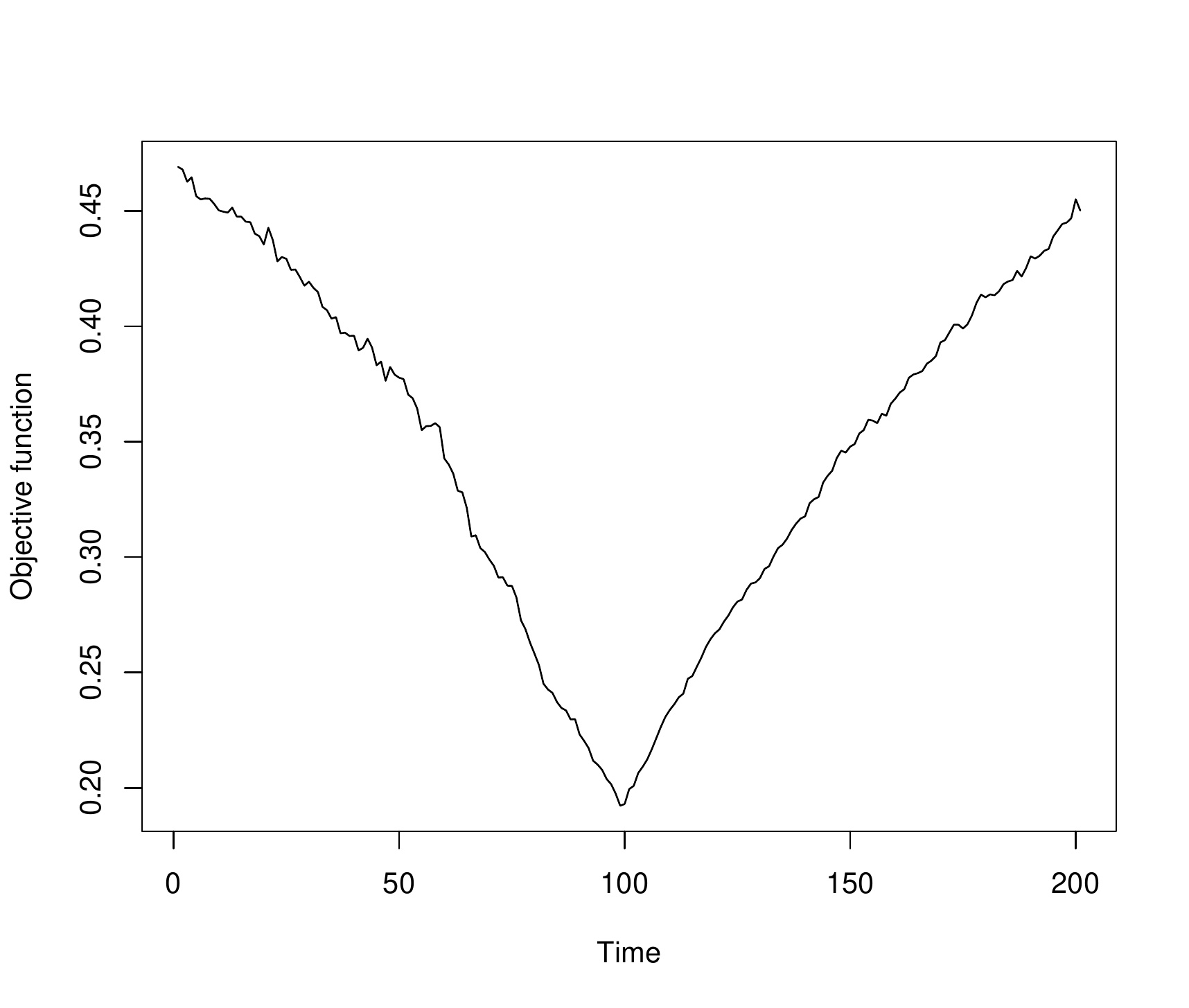}
    \end{subfigure}}
    
    \caption{Left: the curve of the objective function of the full low-rank plus sparse model; Right: the curve of the objective function of the alternative weakly sparse model.}
    
    \label{fig:compare-curves}
\end{figure}

As can be seen, the plot for the surrogate weakly sparse model shares a similar pattern to that of the true model. However, in practice, we can not a priori guarantee a change both in the low-rank and the sparse component, simultaneously. Therefore, an extra assumption is required to ensure the detectability of the change points. Before we state it, we first introduce formally the surrogate piece-wise weakly sparse VAR model.

\subsection{Formulation of the Surrogate Weakly Sparse VAR Model}\label{sec:weakformulate}

A $p\times p$ real matrix $A$ is weakly sparse, if it satisfies
\begin{equation}
    \label{eq:12}
    \mathbb{B}_q(R_q) := \bigg\{ A \in \mathbb{R}^{p \times p}:  \sum_{i=1}^p\sum_{j=1}^p |a_{ij}|^q  \leq R_q \bigg\},
\end{equation}
for some $q \in (0, 1)$; namely, its entries are restricted in an $\ell_q$ ball of radius $R_q$ \citep{negahban2012unified}. Note that when $q \to 0^+$, this set converges to an exact sparse model, that is, $A \in \mathbb{B}_0(R_0)$, if and only if $A$ has at most $R_0$ nonzero elements. When $q \in (0,1)$, the set $\mathbb{B}_q(R_q)$ enforces a certain rate of decay on the ordered absolute values of $A$.

We focus the discussion on detecting a single change point and establish under what conditions the change point can be estimated consistently based on the weakly sparse surrogate model. Subsequently, we extend the result to the case of multiple change points using the proposed rolling window strategy.

{
Since the focus is on the weakly sparse VAR model, the detection procedure provided in Section~\ref{sec:model} requires some modification, whose details are given in Appendix C.
}

We assume that $(A_1^\star,A_2^\star) \in \mathbb{B}_q(R_q)$, for some $q\in (0,1)$ and $R_q>0$.
We also introduce a modification on the Assumptions made in Sections \ref{sec:model} and \ref{sec:theory}. Based on Remark \ref{remark:1} and using the same notation as in the results in Sections \ref{sec:theory} and \ref{sec:multiplecp}, the counterpart of Assumption H1 becomes: 
{
\begin{itemize}[leftmargin=0pt]
    \item[(W1)] The weakly sparse assumption on the $A_j^\star$'s singles out \textit{spiky} entries. Hence, one of the following needs to hold: 
    \begin{itemize}[leftmargin=*]
        \itemsep.1em
        \item[(1)] If $\gamma_1, \gamma_2 \geq p$, then we require the minimum spacing $\Delta_T$ and the jump size $v_A = \|A_2^\star - A_1^\star\|_2$ satisfy:
        \begin{equation*}
            \Delta_Tv_A^2 \geq C_0^w\left(T^{\frac{q}{2}}R_q(\log(p \vee T))^{1-\frac{q}{2}}\right);
        \end{equation*}
        \item[(2)] Otherwise, the change point is identifiable as long as: 
        \begin{equation*}
            \Delta_Tv_S^2 \geq C_0^w\left(T^{\frac{q}{2}}R_q(\log(p \vee T))^{1-\frac{q}{2}}\right).
        \end{equation*}
    \end{itemize}
\end{itemize}}
\begin{remark}
\label{remark:10}
Assumption W1 is based on Remark \ref{remark:1}. Note that if the low-rank components dominate the signal, then an adequate change in them is required to identify the change point; otherwise, we need different information ratios together with distinct spiky entries in the sparse components. The latter sufficient condition indicates that the changes in the spiky entries play an important role in identifying the change points. For the second case, if the low-rank components are not dominant in both segments, then an adequately large change in the sparse components is sufficient to determine the change point.
\end{remark}

\subsection{Theoretical Properties}

The following proposition provides a lower bound for the radius $R_q$, so that the true transition matrices $(A_1^\star,A_2^\star)$ that admit a low-rank plus sparse decomposition do belong to the above defined $\ell_q$ ball. We only discuss the case $0 < \gamma_1, \gamma_2 \leq p$. Analogous results for the other cases can be derived in a similar manner. 
\begin{proposition}
\label{prop:2}
Let $q \in (0,1)$ be fixed and $R_q>0$ be the radius of $\mathbb{B}_q(R_q)$ defined in \eqref{eq:12}. Further, the transition matrices for the data generating model satisfy the following decomposition: $A_1^\star = L_1^\star + S_1^\star$ and $A_2^\star = L_2^\star + S_2^\star$, where $L_1^\star$, $L_2^\star$, $S_1^\star$, and $S_2^\star$ are the corresponding low-rank and sparse components. Then, $A_1^\star, A_2^\star$ belong to $\mathbb{B}_q(R_q)$ if $R_q$ satisfies:
\begin{equation*}
    R_q \geq  d^\star_{\max}\left(\left(\frac{\alpha_L}{p}\right)^q + M_S^q\right) + (p^2 - d^\star_{\max})|\sigma_{\max}|^q,
\end{equation*}
where $\sigma_{\max} = \max \{\|L_1^\star\|_2, \|L_2^\star\|_2\}$ and $d_{\max}^\star = \max\{d_1^\star, d_2^\star\}$. 
\end{proposition}

Before we extend Theorem \ref{thm:1} to the surrogate weakly sparse model, a modification to the selection of tuning parameters is required. Recall that \eqref{eq:6} identifies the tuning parameters for the low-rank plus sparse model, while for the surrogate weakly sparse model, the only parameter is the transition matrix $A_j^\star$ for $j=1,2$. Along with the notation defined in \eqref{eq:6}, the tuning parameters are given by:
\begin{equation}
    \label{eq:13}
    \lambda_{1,\tau}^w = 4c_0^w\sqrt{\frac{\log p + \log(\tau-1)}{\tau - 1}},\quad
    \lambda_{2,\tau}^w = 4c_0^{w^\prime}\sqrt{\frac{\log p + \log(T-\tau)}{T - \tau}},
\end{equation}
where $c_0^w, c_0^{w^\prime}>0$ are some positive constants selected by the similar method as $c_0$ and $c_0^\prime$ in \eqref{eq:6}, the selection procedure is provided in the next section.
Since we employ the same exhaustive search algorithm in Algorithm 1, a similar assumption as H3 on the search domain $\mathcal{T}^w$ is required.
\begin{itemize}[leftmargin=0pt]
    \item[(W2)] Using similar definitions to Assumption H3, denote the search domain by $\mathcal{T}^w \overset{\text{def}}{=} [a^w,b^w]$, and let $|\mathcal{T}^w|$ to be the length of $\mathcal{T}^w$. Then, we assume that,
\begin{equation*}
    a^w = \left\lfloor R_q\left(\frac{\log(p\vee T)}{T}\right)^{-\frac{q}{2}} \right\rfloor,\ b^w = \left\lfloor T - R_q\left(\frac{\log(p\vee T)}{T}\right)^{-\frac{q}{2}} \right\rfloor,\ \frac{|\mathcal{T}^w|}{T^{\frac{q}{2}}R_q(\log(p\vee T))^{1-\frac{q}{2}}} \to +\infty.
\end{equation*}
\end{itemize}
We are now in a position to extend the result in Theorem \ref{thm:1} in the following proposition, whose proof is provided in Appendix E.
\begin{proposition}
\label{prop:3}
Suppose Assumptions W1 and W2 hold and the transition matrices $A_1^\star$ and $A_2^\star$ in \eqref{eq:1} belong to the set $\mathbb{B}_q(R_q)$ for some fixed constant $q\in (0,1)$ and radius $R_q>0$, such that $c_1\sqrt{R_q}\left(\frac{\log p + \log T}{T}\right)^{\frac{1}{2} - \frac{q}{4}} \leq 1$ for some constant $c_1>0$. Then, by employing Algorithm 1 and using the tuning parameters as in \eqref{eq:13}, there exists a large enough constant $K^w_0 > 0$ such that, with respect to the jump size $v_A = \|A_2^\star - A_1^\star\|_2$, as $T \to +\infty$
\begin{equation*}
    \mathbb{P}\left( |\widehat{\tau} - \tau^\star| \leq K^w_0 \frac{ T^{\frac{q}{2}}R_q\left(\log (p\vee T)\right)^{1-\frac{q}{2}}}{v_A^2} \right) \to 1.
\end{equation*}
\end{proposition}

The following Proposition extends the above result to the case of multiple change points based on the rolling window strategy previously described. The window size $h$ can be selected by substituting the vanishing sequence $\{\xi_T\}$ in Assumption H4 by the vanishing sequence $\{\xi_T^w\}$ defined in Assumption W3 below, for the weakly sparse model.

\begin{proposition}
\label{prop:4}
Suppose Assumptions W1 and W2 hold and the transition matrices $A_j^\star$, $j=1,\dots, m_0+1$ belong to the set $\mathbb{B}_q(R_q)$ for some fixed constant $q\in (0,1)$ and the $\ell_q$-ball radius $R_q>0$ satisfies that $\sqrt{R_q}\left(\frac{\log p + \log h}{h}\right)^{\frac{1}{2} - \frac{q}{4}} \leq 1$. Then, by employing the rolling window strategy, we obtain the candidate change points set $\widetilde{\mathcal{S}}_w=\{\widetilde{\tau}_1,\dots, \widetilde{\tau}_{\widetilde{m}}\}$. Then, as $T\to +\infty$, there exists a large enough constant $K^w_1 > 0$ such that, 
\begin{equation*}
    \mathbb{P}\left( d_H(\widetilde{\mathcal{S}}_w, \mathcal{S}^\star) \leq K^w_1 \frac{ h^{\frac{q}{2}}R_q\left(\log(p\vee h)\right)^{1-\frac{q}{2}}}{\min_{1\leq j\leq m_0}v_{j,A}^2} \right) \to 1,
\end{equation*}
where $v_{j,A} = \|A_{j+1}^\star - A_j^\star\|_2$.
\end{proposition}
Recall that the rolling-window mechanism will result in a number of \textit{redundant} candidate change points. By using the surrogate weakly sparse model, we obtain a few redundant candidate change points as well. Therefore, we need to remove those redundant change points by using a similar screening step as introduced in the two-step algorithm in Section \ref{sec:twostep}. Similarly, we also extend Assumptions H3', H5 and H6 to the weakly sparse scenario -Assumptions W3 and W4 given in Appendix C- in order to formally introduce the theoretical results for the surrogate model. Employing the selected tuning parameters as detailed in Assumptions W3 and W4, we can establish consistent estimation of the change points.
\begin{proposition}
\label{prop:5}
Suppose Assumptions W1--W4 hold and denote the minimizer of (7) in Appendix C by $(\widehat{\tau}_1^w, \dots, \widehat{\tau}_{\widehat{m}^w}^w)$. Then, as $T \to +\infty$, there exists a large enough positive constant $B^w > 0$ such that
 \begin{equation*}
    \mathbb{P}\left( \max_{1\leq j \leq m_0}|\widehat{\tau}_j^w - \tau^\star_j| \leq B^w m_0T\xi_T^w \frac{R_q^2\left({\log(p\vee T)}/{T}\right)^{-q}}{\min_{1\leq j\leq m_0}v_{j,A}^2} \right) \to 1.
    \end{equation*}
\end{proposition}

\begin{remark}\label{remark:11}
Proposition \ref{prop:5} provides the consistency rate of the final estimated change points obtained by the surrogate weakly sparse model. In the case of $m_0$ being finite, we select the vanishing sequence $\{\xi_T^w\}$ to be of order $R_q^{2} \left(\log (p\vee T)\right)^{(1+\rho+q)}/T$ for some arbitrarily small constant $\rho>0$. Therefore, the consistency rate in Proposition \ref{prop:5} becomes $B^\prime m_0T^{q} R_q^{4} \left(\log (p\vee T)\right)^{(1+\rho)}$. According to Assumption W3, the penalty term $\omega_T^w$ can be selected to be of the order $T^{1+q} \xi_T^w R_q^{2}\left(\log(p\vee T)\right)^{\rho-q}$ and the minimum spacing in the weakly sparse model $\Delta_T$ must be at least $T^{1+q} \xi_T^w R_q^{2}\left(\log(p\vee T)\right)^{2\rho-q}$.
\end{remark}

An analogue of Corollary \ref{cor:1} and a comparison of the error bounds established in Theorem \ref{thm:3} and Proposition \ref{prop:5} are given in Appendix C.

\section{Performance Evaluation}\label{sec:simu}

We start by investigating the performance of the exhaustive search algorithm for a single change point detection for the low-rank plus sparse VAR model and its surrogate counterpart and the two-step algorithm for detecting multiple change points for these models.
\begin{itemize}[leftmargin=0pt]
    \itemsep.1em
    \item {\bf Data generation:} (1) We generate the time series data $\{X_t\}$ with a \textbf{single} change point at $\tau^\star = \lfloor T/2 \rfloor$ from model \eqref{eq:1}. We set the true ranks $r^\star_1 = \lfloor p/15 \rfloor$, $r^\star_2 = \lfloor p/15 \rfloor + 1$, and the information ratio $\gamma_1=\gamma_2$ for most of the cases considered, unless otherwise specified. The low-rank components $L_1^\star$ and $L_2^\star$ are designed by randomly generating an orthonormal matrix $U$ and singular values $\sigma_1, \dots, \sigma_p$ to obtain $L_1^\star = \sum_{l=1}^{r_1^\star}\sigma_l\mathbf{u}_l\mathbf{u}_l^\prime$, and $L_2^\star = \sum_{l=1}^{r_2^\star}\sigma_l\mathbf{u}_l\mathbf{u}_l^\prime$, where $\mathbf{u}_l$ represents the $l$-th column of matrix U. Then, the sparse components share the same 1-off diagonal structure with values $-\|L_1^\star\|_\infty/\gamma_1$ and $\|L_2^\star\|_\infty/\gamma_2$, respectively. The error term $\{\epsilon_t\}$ is normally distributed from $\mathcal{N}_p(\mathbf{0}, 0.01\mathbf{I}_p)$. (2) In the \textbf{multiple} change points case, we create the time series data $\{X_t\}$ from model \eqref{eq:8} with $m_0$ change points, the true ranks $r_j^\star$ are randomly chosen from: $\lfloor p/10 \rfloor-1, \lfloor p/10 \rfloor, \lfloor p/10 \rfloor+1$ unless otherwise specified, and the information ratios are fixed to $\gamma_j=0.25$. The low-rank components are designed in a similar way as the single change point case, and the $j$-th sparse components are generated by $(-1)^j\|L_j^\star\|_\infty / \gamma_j$.
    
    \item {
    {\bf Tuning parameter selection:} To select the tuning parameters related to optimization problem \eqref{eq:3}, we can use the theoretical values of $\lambda_j$ and $\mu_j$ provided in \eqref{eq:6} and \eqref{eq:7}, and select the constants $c_0$ and $c_0^\prime$ by using a grid search as follows:
    \begin{itemize}[leftmargin=*]
        \item[(1)] Choose an equally spaced sequence within $[0.001, 10]$ as the range for constants $c_0$ and $c_0^\prime$ to construct the grid $\mathcal{G}(\lambda, \mu)$;
        \item[(2)] Next, extract a time point every $k$ time points (we set $k=5$ in all numerical settings) to construct the testing set $\mathcal{T}_{\text{test}}$, and use the remaining time points as the training set $\mathcal{T}_{\text{train}}$, and denote the corresponding estimated transition matrix $\widehat{A}_{(\lambda, \mu)}$ with respect to the tuning parameters $(\lambda, \mu)$;
        \item[(3)] Select the tuning parameters $(\widehat{\lambda}, \widehat{\mu})$ satisfying:
        \begin{equation*}
            (\widehat{\lambda}, \widehat{\mu}) = \argmin_{(\lambda, \mu)\in \mathcal{G}(\lambda,\mu)}\left\{ \frac{1}{|\mathcal{T}_{\text{test}}|}\sum_{t\in \mathcal{T}_{\text{test}}}\|X_{t+1} - \widehat{A}_{(\lambda, \mu)}X_t\|_2^2 \right\}.
        \end{equation*} 
    \end{itemize}}
    
    \item {\bf Window size selection:} The width of the rolling window plays an important role in the multiple change points scenario. In practice, we can manually select a suitable window-size, or we may use the following strategy. In Assumption H4, we provided conditions on the window size $h$ and rolling step size $l$. Next, we discuss an iterative procedure for determining these two parameters in practice.\\
    (1) Start with $h=cT^\delta$, and $l={h}/{4}$, where $\delta$ is selected from 1 to 0.5 (equally spaced) and $0<c<1$ is a constant; (2)
    For a given $\delta$, apply Algorithm 2 and obtain the final set of change points $\{\widehat{\tau}_1, \dots, \widehat{\tau}_m\}$; (3)
    Repeat (2) until the number of the final set of change points does not change. Return the corresponding window size $\widehat{h}$. 
    
    \item {\bf Model evaluation:} We evaluate the performance of our algorithm by using the mean and standard deviation of the estimated change point locations relative to the number of observations as well as the boxplots for the estimated change point for each case. We use estimated rank, sensitivity (SEN), specificity (SPC), and relative error (RE) for the whole transition matrices and the low-rank and the sparse components as additional metrics to evaluate the performance of model.
    \begin{equation*}
        \text{SEN} = \frac{\text{TP}}{\text{TP+FN}},\ \text{SPC} = \frac{\text{TN}}{\text{FN} + \text{TN}},\ \text{RE} = \frac{\|\text{Est.} - \text{Truth}\|_F}{\|\text{Truth}\|_F}.
    \end{equation*}
    For multiple change points settings, we also measure the \textit{selection rate}. Specifically, a detected change point $\widehat{t}_j$ is counted as a \textit{success} for the true change point $t_j^\star$, if and only if $\widehat{t}_j \in [t_j^\star - \frac{1}{10}(t_j^\star - t_{j-1}^\star), t_j^\star + \frac{1}{10}(t_{j+1}^\star - t_j^\star)]$. Then, the selection rate is defined by calculating the percentage of simulation replications with successes.
\end{itemize}
All numerical experiments are run in \textsc{R} 3.6.0 on the uf HiPerGator Computing platform with 4 Intel E5 2.30 GHz Cores and 16 GB memory. The code and scripts for simulation examples and applications are available at \url{https://github.com/peiliangbai92/LSVAR_cpd}.

\subsection{Performance for Detecting A Single Change Point}\label{sec:single-performance}

We investigate the following factors: the dimension of the model $p$, the sample size $T$, the differences in the $\ell_2$ norm, $v_L$ and $v_S$ of the two low-rank and sparse components, respectively and the information ratio $\gamma$. The following parameters settings are considered in our investigation. A full summary is provided in the form of a Table in Appendix F.1.
\begin{itemize}[leftmargin=0pt]
    \itemsep.05em
    \item[(A)] In the first setting, we consider the case that the low-rank component exhibits a very small change while the sparse one a large change. Further, the ``total signal" in the transition matrix comes mostly from the sparse component and therefore, $\gamma_j < 1, j=1,2$.
    \item[(B)] This setting is similar in structure to A: the low-rank components exhibit very small change, while the sparse components change by a significant amount, but the ``total signal" in the transition matrix comes mostly from the former; i.e., $\gamma_j \geq 1$ for $j=1,2$.
    \item[(C)] The structure of this setting is as in B, but different values of $\gamma_j$ are considered.
    \item[(D)] This setting is the reverse of B, wherein the low-rank components exhibit a large change, while the sparse ones a very small ones, and further $\gamma_j \geq 1, j=1,2$.
    \item[(E)] This setting is similar in structure to C, but the information ratio $\gamma_j < 1, j=1,2$.
    \item[(F)] The setting is similar to E, but an increasing $|\gamma_1-\gamma_2|$ is considered.
\end{itemize}
The results for these settings over 50 replications are given in Table \ref{tab:2}. The first two columns record the mean and standard deviation of the estimated change point location, the third and fourth columns are the estimated ranks for the low-rank components, the fifth and sixth columns give the sensitivity and specificity of the estimated sparse components, and finally the last column shows the relative norm error of the estimated transition matrix $\widehat{A}$ to the truth $A^\star$, and we also provide the relative error of the estimated sparse components (low-rank components) $\widehat{S}$ (or $\widehat{L}$) to the truth $S^\star$ (or $L^\star$).
\begin{table}[ht]
    \spacingset{1}
    \centering
    \caption{Performance of the L+S model under different simulation settings.}
    \label{tab:2}
    \resizebox{\textwidth}{!}{%
    \begin{tabular}{c|c|c|c|c|c|c|c}
        \hline\hline
          & mean & sd & $\widehat{r}_1$ & $\widehat{r}_2$ & SEN & SPC & Total RE/ Sparse RE / Low-rank RE  \\
        \hline
        A.1 & 0.498 & 0.002 & $1.020$ & $2.900$ & $(1.000, 1.000)$ & $(0.909, 0.976)$ & $(0.186, 0.237)/(0.172, 0.220)/(0.582, 0.648)$ \\
        A.2 & 0.499 & 0.002 & $1.020$ & $2.820$ & $(1.000, 1.000)$ & $(0.910, 0.974)$ & $(0.186, 0.241)/(0.172, 0.217)/(0.582,0.759)$ \\
        A.3 & 0.499 & 0.002 & $1.020$ & $2.960$ & $(1.000, 1.000)$ & $(0.909, 0.979)$ & $(0.186, 0.249)/(0.172,0.225)/(0.582,0.749)$ \\ 
        \hline
        B.1 & 0.530 & 0.090 & $1.000$ & $1.340$ & $(0.166, 0.108)$ & $(0.947, 0.980)$ & $(0.590, 0.579)/(1.140, 1.006)/(0.482,0.413)$ \\
        B.2 & 0.532 & 0.089 & $1.000$ & $1.340$ & $(0.166, 0.109)$ & $(0.947, 0.979)$ & $(0.590, 0.580)/(1.139,1.006)/(0.482,0.414)$ \\
        B.3 & 0.534 & 0.089 & $1.000$ & $1.330$ & $(0.165, 0.109)$ & $(0.947, 0.980)$ & $(0.591, 0.580)/(1.140,1.006)/(0.482,0.413)$ \\
        \hline
        C.1 & 0.522 & 0.056 & $1.000$ & $1.350$ & $(0.237, 0.103)$ & $(0.944, 0.978)$ & $(0.592, 0.569)/(1.070, 1.015)/(0.459, 0.384)$ \\
        C.2 & 0.497 & 0.005 & $1.000$ & $1.300$ & $(0.400, 0.120)$ & $(0.948, 0.979)$ & $(0.645, 0.575)/(0.953, 1.006)/(0.482, 0.397)$ \\
        C.3 & 0.502 & 0.031 & $1.000$ & $1.320$ & $(0.629, 0.109)$ & $(0.947, 0.978)$ & $(0.646, 0.570)/(0.858, 1.007)/(0.499, 0.389)$ \\
        C.4 & 0.497 & 0.005 & $1.000$ & $1.300$ & $(1.000, 0.132)$ & $(0.927, 0.977)$ & $(0.357, 0.559)/(0.381, 1.002)/(0.499, 0.381)$ \\
        \hline
        D.1 & 0.494 & 0.011 & $1.000$ & $1.500$ & $(0.301, 0.207)$ & $(0.948, 0.978)$ & $(0.654, 0.581)/(1.036,0.969)/(0.543,0.455)$ \\
        D.2 & 0.494 & 0.008 & $1.000$ & $1.920$ & $(0.305, 0.325)$ & $(0.948, 0.975)$ & $(0.654, 0.639)/(1.037,0.934)/(0.544,0.478)$ \\
        D.3 & 0.495 & 0.007 & $1.000$ & $2.080$ & $(0.307, 0.485)$ & $(0.948, 0.972)$ & $(0.653, 0.558)/(1.031,0.878)/(0.544,0.444)$ \\
        \hline
        E.1 & 0.477 & 0.048 & $1.200$ & $3.060$ & $(1.000, 1.000)$ & $(0.727, 0.739)$ & $(0.171, 0.193)/(0.160,0.176)/(0.563,0.674)$ \\
        E.2 & 0.478 & 0.026 & $1.000$ & $3.040$ & $(1.000, 1.000)$ & $(0.836, 0.932)$ & $(0.185, 0.216)/(0.168,0.191)/(0.673,0.633)$ \\
        E.3 & 0.496 & 0.015 & $1.000$ & $3.000$ & $(1.000, 1.000)$ & $(0.917, 0.729)$ & $(0.204, 0.254)/(0.180,0.250)/(0.674,0.776)$ \\
        \hline
        F.1 & 0.495 & 0.053 & $1.000$ & $2.880$ & $(1.000, 1.000)$ & $(0.924, 0.958)$ & $(0.405, 0.330)/(0.429,0.330)/(0.603,0.482)$ \\
        F.2 & 0.487 & 0.039 & $1.000$ & $3.520$ & $(1.000, 0.996)$ & $(0.925, 0.964)$ & $(0.411, 0.415)/(0.437,0.486)/(0.602,0.429)$ \\
        F.3 & 0.495 & 0.023 & $1.000$ & $2.640$ & $(1.000, 0.895)$ & $(0.924, 0.970)$ & $(0.405, 0.539)/(0.429,0.688)/(0.602,0.484)$ \\
      \hline\hline
    \end{tabular}}
    
\end{table}

For settings A and D, where the dominant components change significantly, the algorithm identifies the change point extremely accurately, as evidenced by the mean estimate over 50 replicates and the very small standard deviation recorded. Further, the ranks of $L_j$ are accurately estimated under setting A, and the specificity and sensitivity of $S_j$ is close to 1. Under setting D, there is deterioration in the estimation of the rank of $L_2$, as well as in the sensitivity of both $S_1$ and $S_2$.
In settings B and E, where there is a small change in the dominant component, the estimates of the change point deteriorate and also exhibit larger variability (especially in setting B). Under setting B, estimation of the rank of $L_2$ is also off, as is the sensitivity for the sparse components. Note that all estimated model parameters under setting E are very accurate, with a small deterioration in the specificity of the $S_j$'s.
In settings C and F, we examine how the behavior of the information ratio influences the accuracy of the change point detection. As the difference between $\gamma_1$ and $\gamma_2$ increases, the estimation accuracy improves of the change point improves markedly. The same happens for the model parameters under setting F. 
Note that the results for settings C and F are in accordance with Remark \ref{remark:1} that discusses how the detectability of the full transition matrix is controlled by the information ratio. 
We provide the performance of single change point detection based on the surrogate model in Table 4 in Appendix F.1.  

Figure \ref{fig:boxplots} depicts boxplots based on 50 replicates of the distance between the location of the true change point and its estimate, i.e., $|\widehat{\tau} - \tau^\star|$. The yellow bars correspond to the full low-rank plus sparse model, while the orange ones to the surrogate model. In accordance to previous findings, under settings A and C, the results are comparable, as well as certain cases for setting E. On the other hand, under settings B, D and F, the full model clearly outperforms the surrogate one, even though in settings F2 and F3 the differences become smaller as the corresponding differences in the information ratios increase.
\begin{figure}[!ht]
    \centering
    \resizebox{6.4in}{1.85in}{%
    \begin{subfigure}
        \centering
        \includegraphics[scale=0.42]{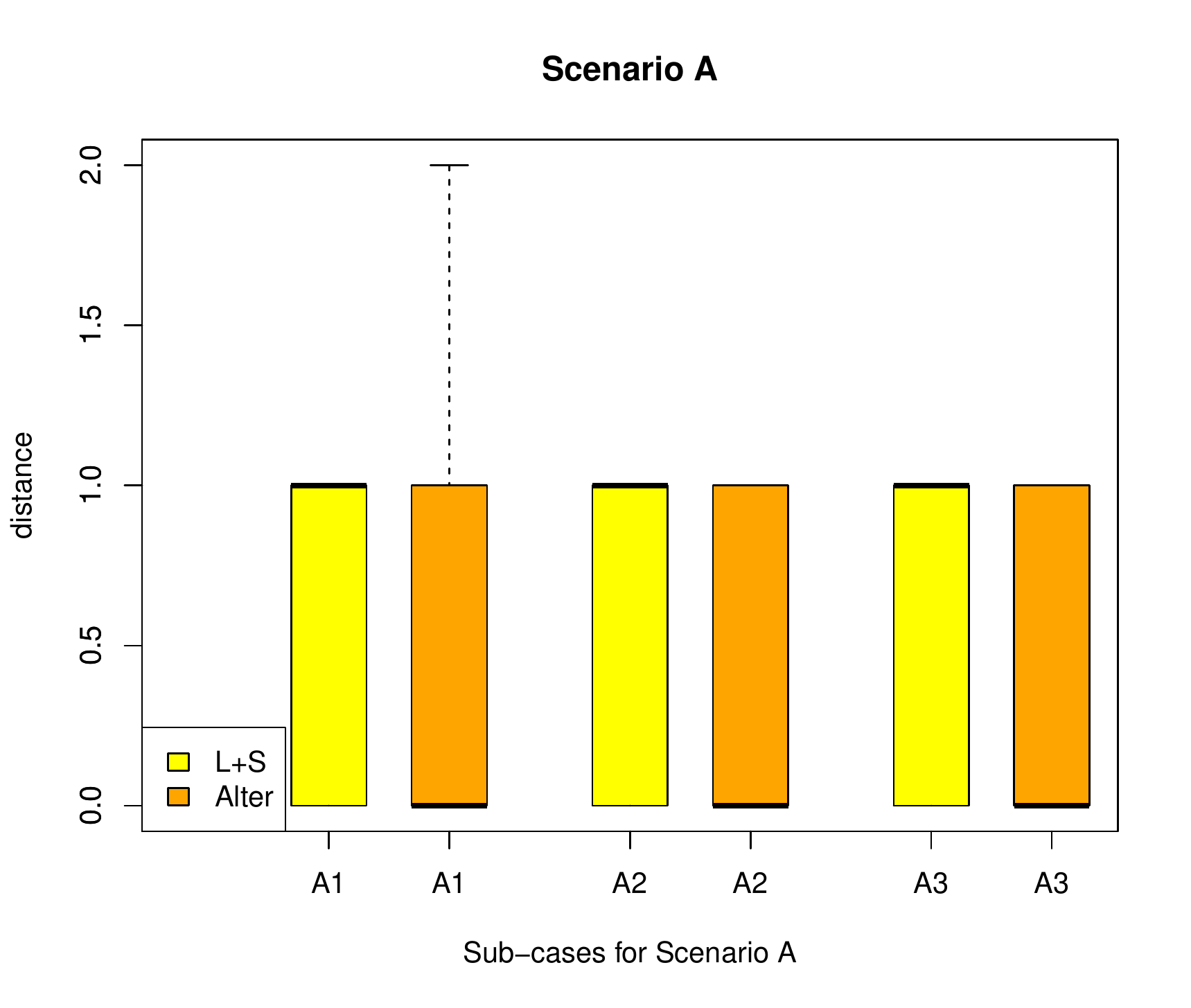}
    \end{subfigure}
    \begin{subfigure}
        \centering 
        \includegraphics[scale=0.42]{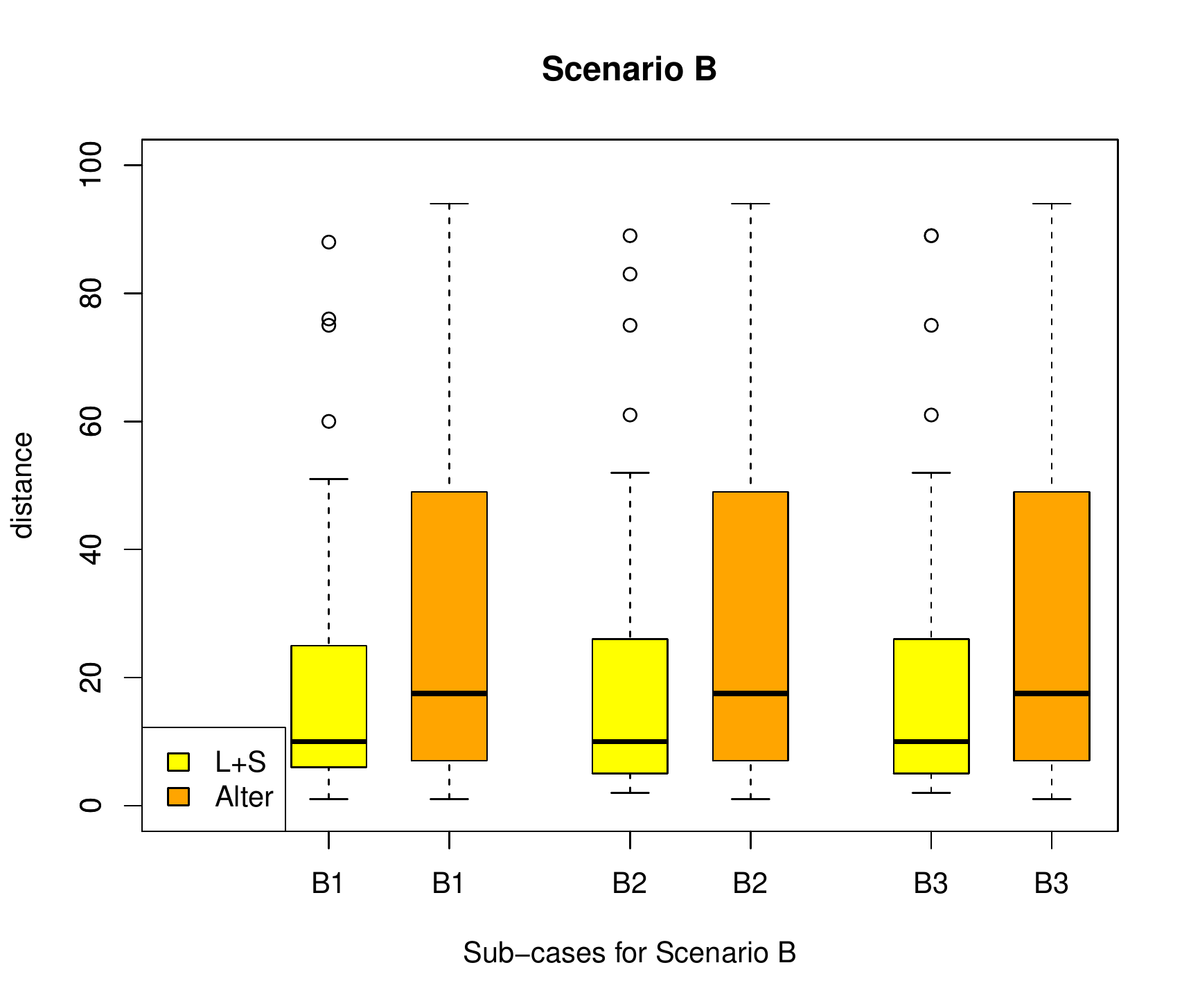}
    \end{subfigure}
    \begin{subfigure}
        \centering
        \includegraphics[scale=0.42]{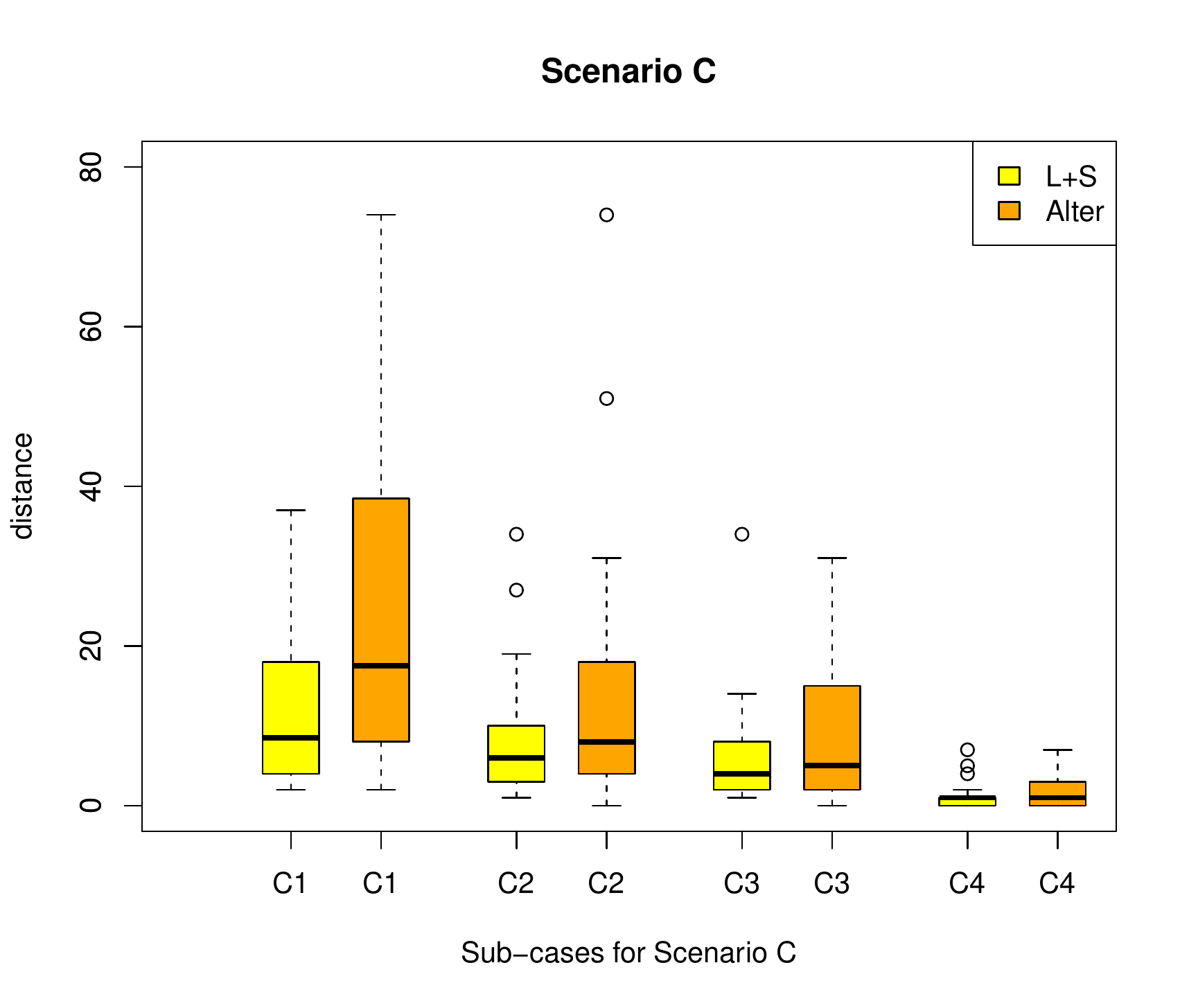}
    \end{subfigure}
    }
    \resizebox{6.4in}{1.85in}{%
    \begin{subfigure}
        \centering 
        \includegraphics[scale=0.42]{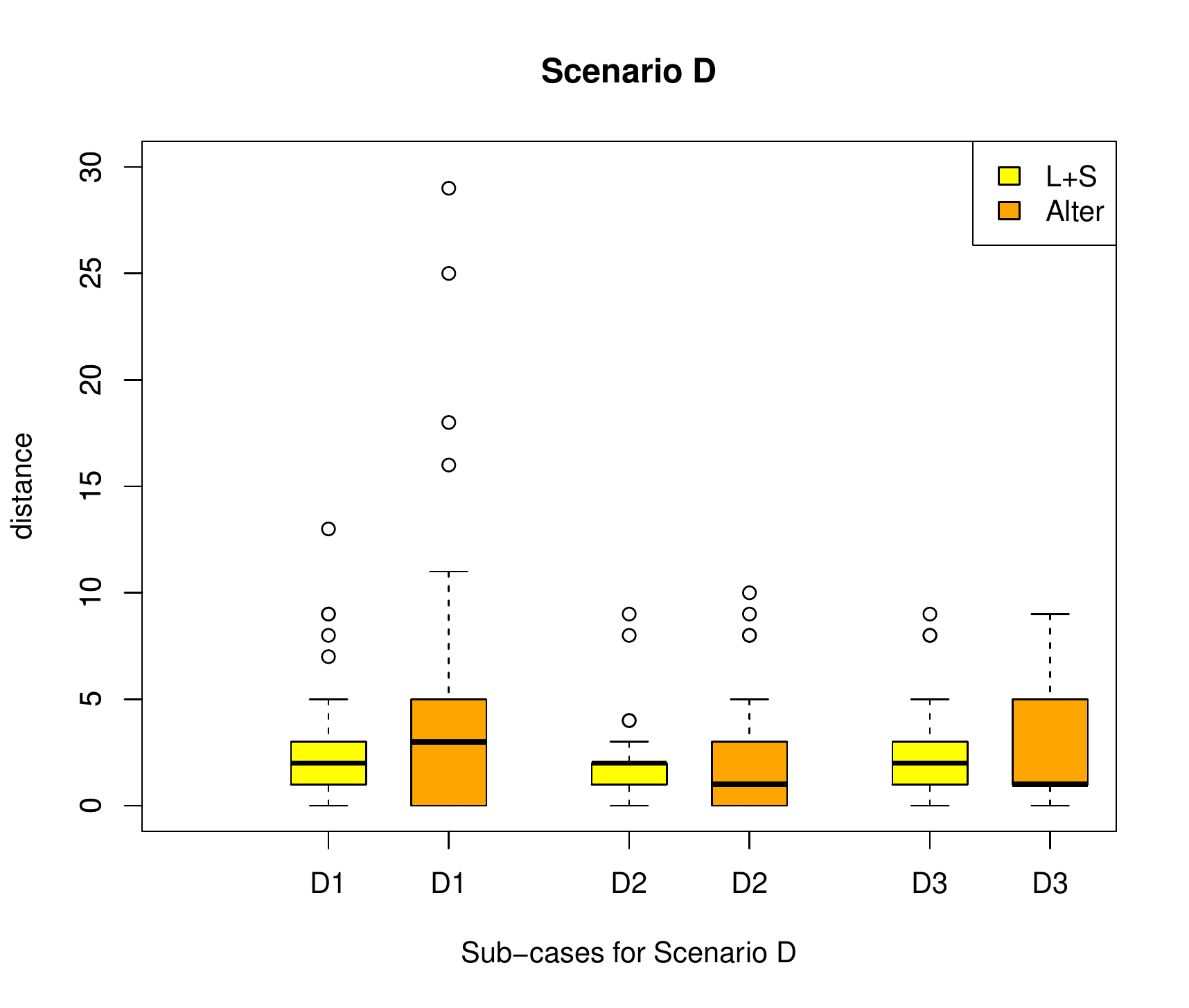}
    \end{subfigure}
    \begin{subfigure}
        \centering 
        \includegraphics[scale=0.42]{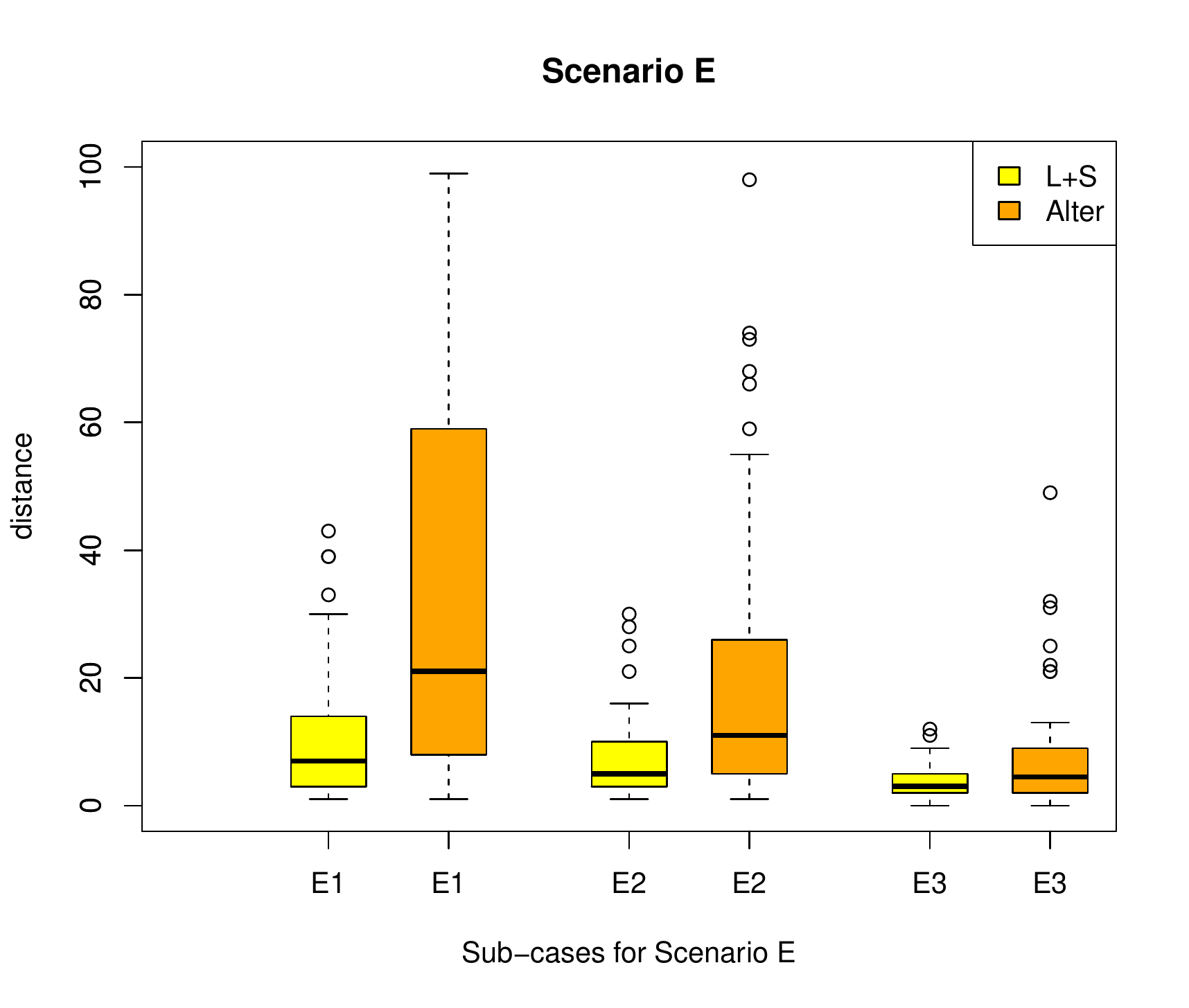}
    \end{subfigure}
    \begin{subfigure}
        \centering
        \includegraphics[scale=0.42]{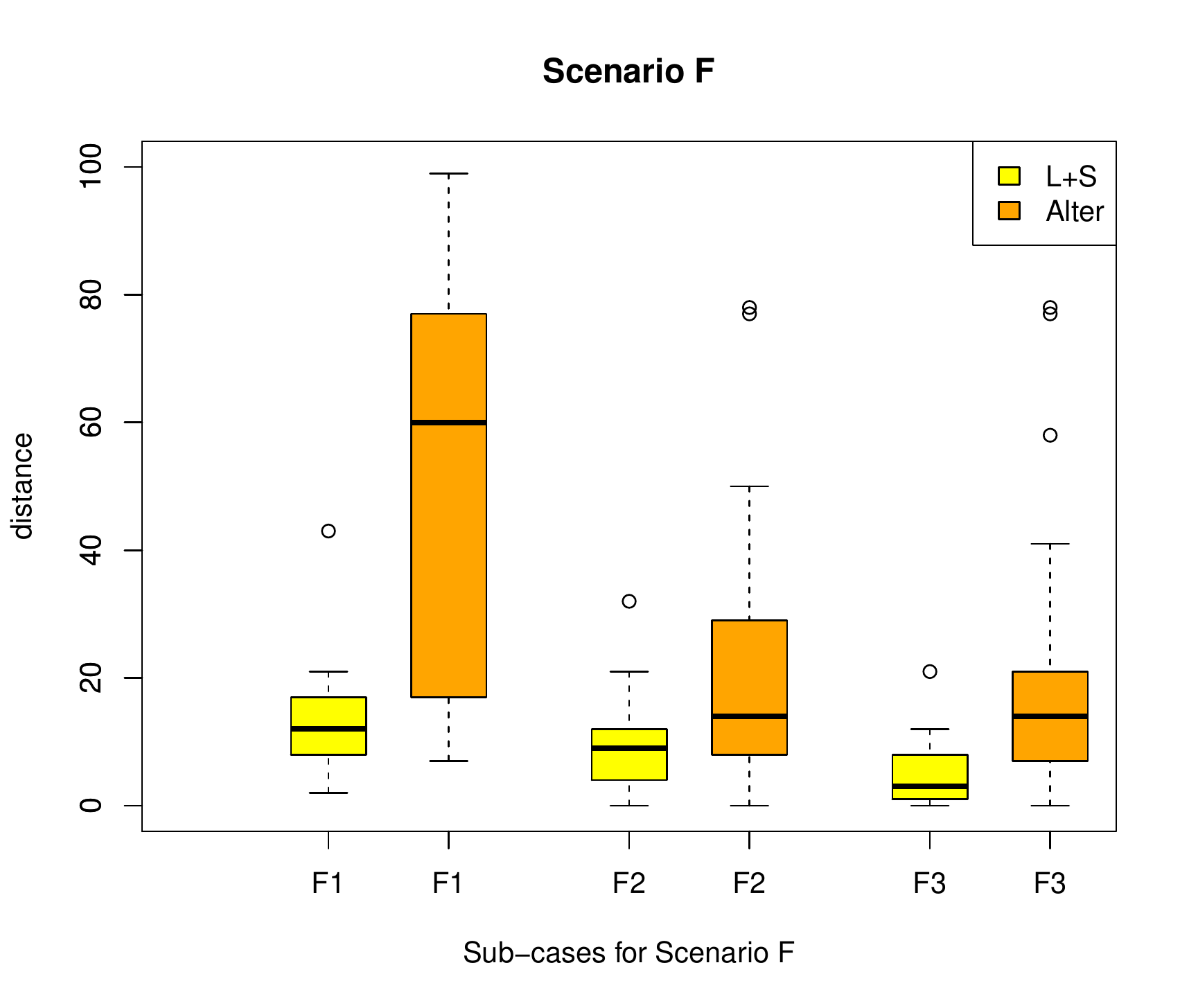}
    \end{subfigure}
    }
    
    \caption{Boxplots for $|\widehat{\tau} - \tau^\star|$ under settings A--F with the full model and the surrogate weakly sparse model.} 
    
    \label{fig:boxplots}
\end{figure}

\subsection{Performance for Detecting Multiple Change Points}\label{sec:multi-performance}

We consider the same settings for each change point, as in case A in Section \ref{sec:single-performance} with modified $T$ and $p$, respectively. The specific scenarios under consideration are as follows:
\begin{itemize}[leftmargin=0pt]
    \itemsep.1em
    \item[(L)] In the first case, we consider settings with different number of change points. Specifically, we investigate the following three cases: (1) $T = 1200$ with $\tau^\star_1 = \lfloor T/6 \rfloor$, $\tau^\star_2 = \lfloor T/3 \rfloor$, $\tau^\star_3 = \lfloor T/2 \rfloor$, $\tau^\star_4 = \lfloor 2T/3 \rfloor$, and $\tau^\star_5 = \lfloor 5T/6 \rfloor$; (2) $T=1800$ with $\tau^\star_1 = \lfloor T/10 \rfloor$, $\tau^\star_2 = \lfloor 3T/10 \rfloor$, $\tau^\star_3 = \lfloor T/2 \rfloor$, $\tau^\star_4 = \lfloor 7T/10 \rfloor$, and $\tau^\star_9 = \lfloor 9T/10 \rfloor$; (3) $T=2400$ with $\tau^\star_1 = \lfloor T/10 \rfloor$, $\tau^\star_2 = \lfloor T/4 \rfloor$, $\tau^\star_3 = \lfloor 2T/5 \rfloor$, $\tau^\star_4 = \lfloor 3T/5 \rfloor$, and $\tau^\star_5 = \lfloor 4T/5 \rfloor$.
    \item[(M)] In the second case, we consider $p$ large enough to satisfy $p^2 > T$ with two change points: $\tau^\star_1 = \lfloor T/3 \rfloor$ and $\tau^\star_2 = \lfloor 2T/3 \rfloor$.
    \item[(N)] In the last scenario, the change in sparsity patterns is considered. We consider a different sparsity pattern rather than the 1-off diagonal structure in the sparse components. 
\end{itemize}

The detailed model parameters are listed in the Table 5 in the Appendix F.2. 

Table \ref{tab:5} presents the mean and standard deviation of the estimated locations of the change points, relative to the sample size $T$, together with the selection rate, as defined at the beginning of the current section. For all cases under settings L and M, the two-step algorithm obtains very accurate results, also exhibiting little variability. The complex random sparse pattern considered in setting N leads to a small deterioration in the selection rate. The locations of the estimated change points together with box plots of $|\widehat{\tau}_j - \tau_j^\star|$ for scenario N over 50 replicates are depicted in the Appendix F.2. 
\begin{table}[!ht]
    \spacingset{1}
    \centering
    \caption{Results for multiple change point selection by full L+S model.}
    \label{tab:5}
    \resizebox{\textwidth}{!}{%
    \begin{tabular}{c|c|c|c|c|c|c|c|c|c|c|c}
        \hline\hline
                           & points & truth & mean & sd & selection rate & & points & truth & mean & sd & selection rate \\
        \cline{1-12}
        \multirow{5}*{L.1} & 1 & 0.1667 & 0.1667 & 0.0004 & 1.00 & \multirow{2}*{M.1} & 1 & 0.3333 & 0.3331 & 0.0005 & {1.00} \\
                          & 2 & 0.3333 & 0.3333 & 0.0003 & 1.00 & & 2 & 0.6667 & 0.6665 & 0.0004 & 1.00 \\
                          \cline{7-12}
                          & 3 & 0.5000 & 0.4999 & 0.0003 & 1.00 & \multirow{2}*{M.2} & 1 & 0.3333 & 0.3329 & 0.0003 & 1.00 \\
                          & 4 & 0.6667 & 0.6665 & 0.0004 & 1.00 & & 2 & 0.6667 & 0.6667 & 0.0006 & 1.00 \\
                          \cline{7-12}
                          & 5 & 0.8333 & 0.8335 & 0.0004 & 1.00 & \multirow{2}*{N.1} & 1 & 0.3333 & 0.3311 & 0.0125 & {0.94} \\
        \cline{1-6}
        \multirow{5}*{L.2} & 1 & 0.1000 & 0.0999 & 0.0002 & 1.00 & & 2 & 0.6667 & 0.6656 & 0.0056 & {0.98} \\
        \cline{7-12}
                          & 2 & 0.2500 & 0.2500 & 0.0000 & 1.00 & \multirow{2}*{N.2} & 1 & 0.1667 & 0.1683 & 0.0115 & {0.92} \\
                          & 3 & 0.4000 & 0.3999 & 0.0002 & 1.00 & & 2 & 0.8333 & 0.8267 & 0.0181 & {0.94} \\
                          \cline{7-12}
                          & 4 & 0.6000 & 0.6000 & 0.0000 & 1.00 & \multirow{2}*{N.3} & 1 & 0.3333 & 0.3302 & 0.0121 & 0.98 \\
                          & 5 & 0.8000 & 0.7999 & 0.0001 & 1.00 & & 2 & 0.6667 & 0.6655 & 0.0119 & 0.98 \\
        \cline{1-12}
        \multirow{5}*{L.3} & 1 & 0.1000 & 0.1000 & 0.0000 & 1.00 & & & & & & \\
                          & 2 & 0.3000 & 0.3000 & 0.0000 & 1.00 & & & & & & \\
                          & 3 & 0.5000 & 0.5000 & 0.0000 & 1.00 & & & & & & \\
                          & 4 & 0.7000 & 0.6999 & 0.0002 & 1.00 & & & & & & \\
                          & 5 & 0.9000 & 0.8998 & 0.0002 & 1.00 & & & & & & \\
        \hline\hline
    \end{tabular}}
\end{table}

\subsection{A Simulation Scenario Based on a EEG Data Set}

{
For this scenario, the sparsity structure is extracted from the EEG data set analyzed in Section \ref{sec:eeg}. Specifically, the setting under consideration is as follows: $T=300$, $p=21$, with two change points located at $\lfloor T/3 \rfloor$ and $\lfloor 2T/3 \rfloor$, respectively. The structure of the transition matrices is obtained by using the results presented in the application section (see Figure 6 in the Section G.2. in the supplement). We keep the non-zero elements (see Figure \ref{fig:brain-simu}) and set their magnitudes at random to 0.4, -0.6, and 0.4, respectively. The low rank components are generated by using the spectral decomposition with ranks equal to 1, 3, and 1. The estimated \emph{sparse} and \emph{low rank} structures are illustrated in Figure \ref{fig:brain-simu}:
\begin{figure}[!ht]
    \centering
    \includegraphics[trim={0 4cm 0 4cm}, clip, scale=.35]{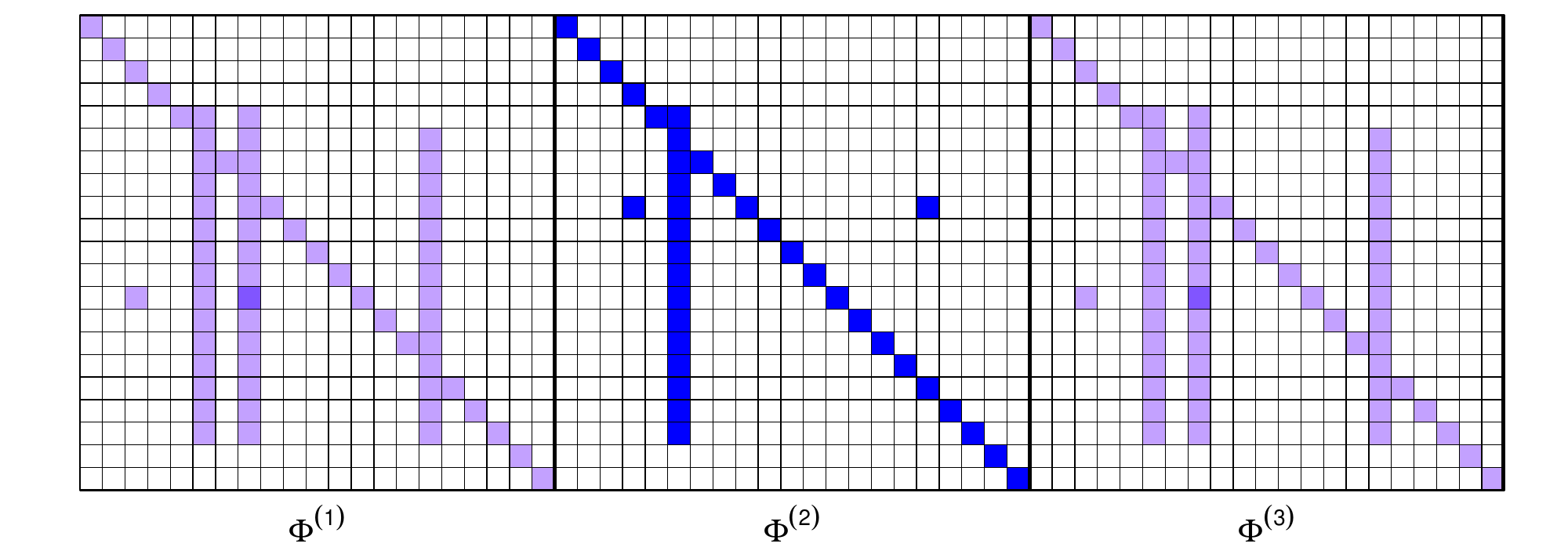}%
    \includegraphics[trim={0 4cm 0 4cm}, clip, scale=.35]{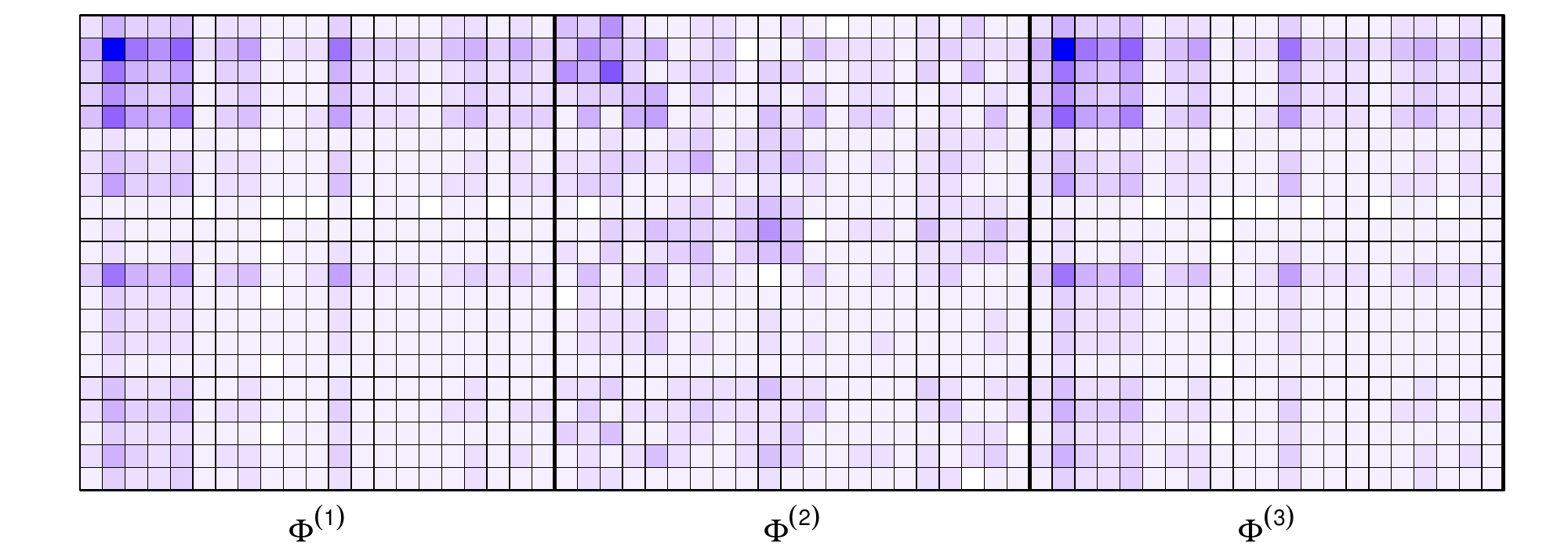}
    \caption{Left: Estimated sparse brain connectivity structure; Right: Estimated low rank brain connectivity structure.}
    \label{fig:brain-simu}
\end{figure}

The results are summarized in Table \ref{tab:general}.
\begin{table}[!ht]
    \centering
    \spacingset{1}
    \caption{Results of simulation scenario based on an EEG data set.}
    \label{tab:general}
    \begin{tabular}{c|c|c|c|c|c}
    \hline\hline
                & points & truth & mean & sd & selection rate \\
    \hline
        \multirow{2}*{General sparsity pattern} & 1 & 0.3333 & 0.3328 & 0.002 & 1.00 \\
                                    & 2 & 0.6667 & 0.6663 & 0.007 & 1.00 \\
    \hline\hline
    \end{tabular}
\end{table}
It can be seen that based on a low rank and sparse structure motivated by real data, the proposed algorithm exhibits a very satisfactory performance.
}

\subsection{Impact of the Signal-to-noise Ratio on the Detection Rate}

{
The signal-to-noise ratio (SNR) is defined as (see also \cite{wang2020univariate,rinaldo2021localizing}):
\begin{equation*}
    \text{SNR} = \frac{\Delta_Tv}{T},
\end{equation*}
wherein $v \overset{\text{def}}{=} \min_j v_j$ the minimum jump size, and $\Delta_T$ is the minimum spacing, i.e. $\Delta_T = \min_{1\leq j \leq m_0}|\tau_j^\star - \tau_{j+1}^\star|$. We set $T=300$ and $p=20$ with two change points located at $\lfloor T/3\rfloor=100$ and $\lfloor 2T/3 \rfloor=200$, respectively. Further, we set the minimum jump size to $v$ = 0.8, 1.0, and 1.6, and the resulting SNR takes the values 0.27, 0.33, and 0.53. The results are given Table \ref{tab:snr}.
\begin{table}[!ht]
    \centering
    \spacingset{1}
    \caption{Extra simulation performance for different signal-to-noise ratios}
    \label{tab:snr}
    \begin{tabular}{c|c|c|c|c|c}
    \hline\hline
      SNR & points & truth & mean & sd & selection rate \\
    \hline
      \multirow{2}*{0.27} & 1 & 0.3333 & 0.3412 & 0.017 & 0.90 \\
                          & 2 & 0.6667 & 0.6702 & 0.012 & 0.94 \\
    \hline
      \multirow{2}*{0.33} & 1 & 0.3333 & 0.3330 & 0.002 & 1.00 \\
                         & 2 & 0.6667 & 0.6687 & 0.004 & 1.00 \\
    \hline
      \multirow{2}*{0.57} & 1 & 0.3333 & 0.3332 & 0.002 & 1.00 \\
                        & 2 & 0.6667 & 0.6665 & 0.001 & 1.00 \\
    \hline\hline
    \end{tabular}
\end{table}

As expected, for small SNR the detection accuracy deteriorates, both in terms of the selection rate of change points, as well as their locations. However, for SNR around or greater than 1, it becomes very satisfactory. Additional results are provided in Section F.3 in the Supplement. 
}

{
\begin{remark}[Additional numerical results and comparisons] Additional numerical results including (i) for the surrogate model, (ii) for additional scenarios for multiple change points, (iii) for run times between the low rank plus sparse and the surrogate models, (iv) with a factor model exhibiting change points, (v) between a factor and the low rank plus sparse models under a misspecified data generating mechanism, (vi) comparison between the proposed two-step algorithm and the TSP algorithm in \cite{bai2020multiple}, and (vi) between the two-step rolling window strategy and a dynamic programming algorithm are presented in Appendices F.1-F.7, respectively.
\end{remark}
}

\section{Applications}\label{sec:application}

\subsection{Change Point Detection in EEG Signals}\label{sec:eeg}

There has been work in the literature on analyzing EEG data using low-rank models for task related signals, since the latter exhibit low-rank structure  \citep{liu2018estimating,jao2018using}. 
Next, we employ the full low-rank plus sparse model to detect change points in data from \cite{trujillo2017effect}. This data set recorded 72 channels of continuous EEG signals by using active electrodes. The sampling frequency is 256Hz and the total number of time points per EEG electrode is 122880 over 480 seconds. The stimulus procedure is that after a resting state (eliminated from the data set) lasting 8 mins, the subject alternates between a 1-min period with eyes open followed by a 1-min period with eyes closed, repeated four times.
Hence, we expect that the employed model captures the low-rank structure associated with the task at hand (open/closed eyes), while the sparse component
can capture idiosyncratic behavior across repetitions of the task.

To illustrate the proposed methodology, two subjects are selected; differences in the EEG signals over time are visible for the first subject, but not for the second one. 
The data are de-trended, by calculating the moving average of each EEG signal and removing it. Specifically, the period average, which is an unbiased estimator of trend, is given by $\hat{m}_l = \frac{1}{d}\sum_{t=1}^{d} X_{l+t}$; we select $d=256$ in accordance to the frequency of the data, and we obtain the de-trended time series by removing the period average. In this work, we use 21 selected EEG channels and $T=67952$ time points in the middle of the whole time series. According to the experiments described in \cite{trujillo2017effect}, there are five open/closed eyes segments in the selected time period with four change points approximately at locations: $\tau_1^\star \cong 11650$, $\tau_2^\star \cong 27750$, $\tau_3^\star \cong 44000$, and $\tau_4^\star \cong 60000$. The data are plotted in Figure 2 in Appendix G.2. Selection of the tuning parameters is based on the guidelines given in Appendix G.1. Note that to separate adequately the sparse component from the low-rank one, we set $\alpha_L$ based on its theoretical values provided in Assumption H2.

The change points estimated by the two-step algorithm are $\widehat{\tau}_1 = 9633$, $\widehat{\tau}_2 = 28529$, $\widehat{\tau}_3 = 43361$ and $\widehat{\tau}_4 = 60209$. The estimated change points are close to those identified based on the designed experiment. In order to quantify the differences among the estimated components across segments, we use the Hamming distance for both sparse and low-rank ones. The results are shown in Figure \ref{fig:hamming-dist} in the form of a heat map that confirms the high degree of similarity between all ``eyes closed" segments (1, 3, 5) and all ``eyes open" segments (2, 4), thus further confirming the accuracy of the methodology. We also provide the estimated low-rank and the sparse patterns for 5 segments in Figure 3, and the correlation networks for the sparse components in Figure 4 in Appendix G.2.
\begin{figure}[!ht]
    \centering
    \resizebox{.7\textwidth}{!}{
    \begin{subfigure}
        \centering
        \includegraphics[scale=.4]{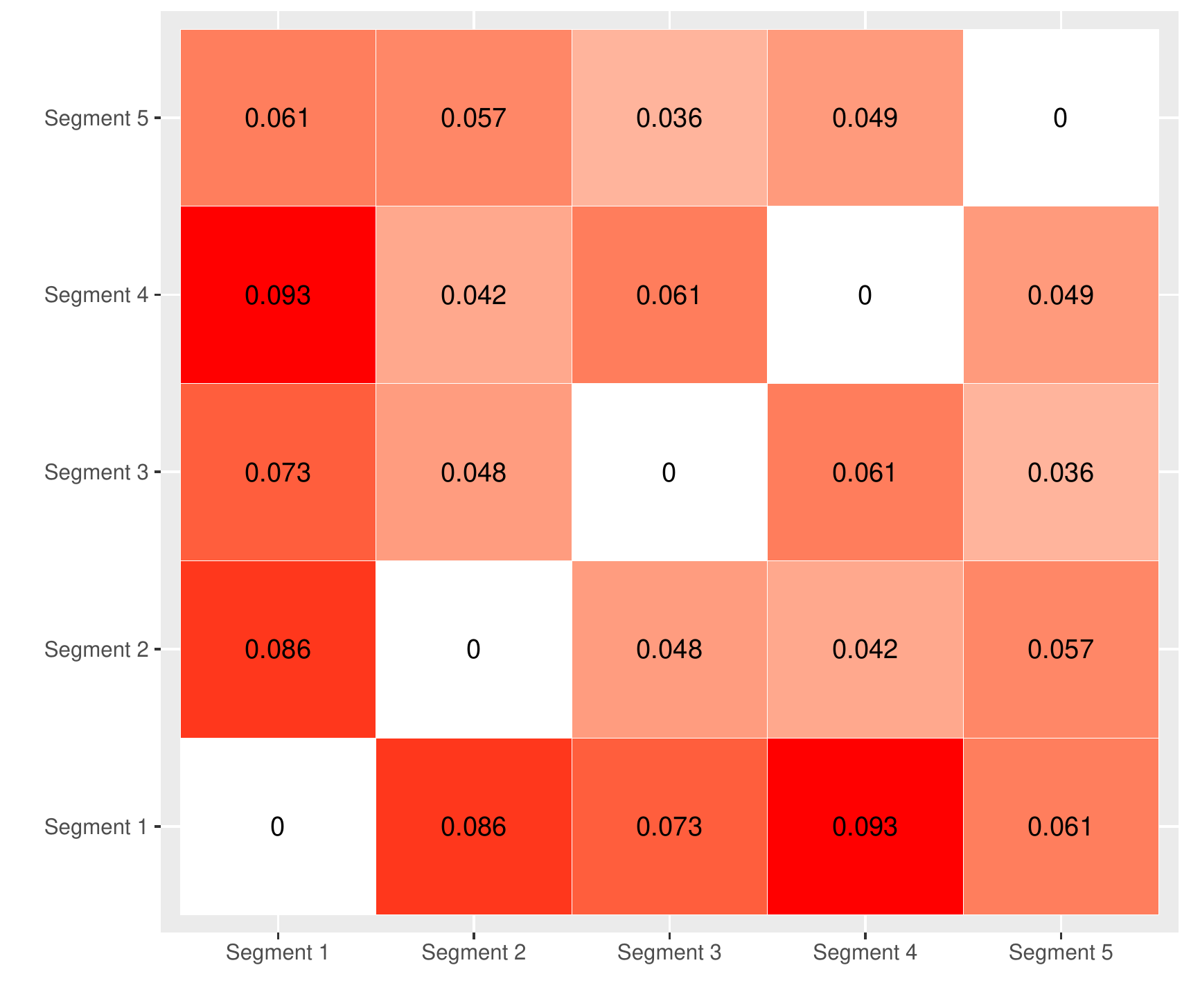}%
        \quad\quad\quad
        \includegraphics[scale=.4]{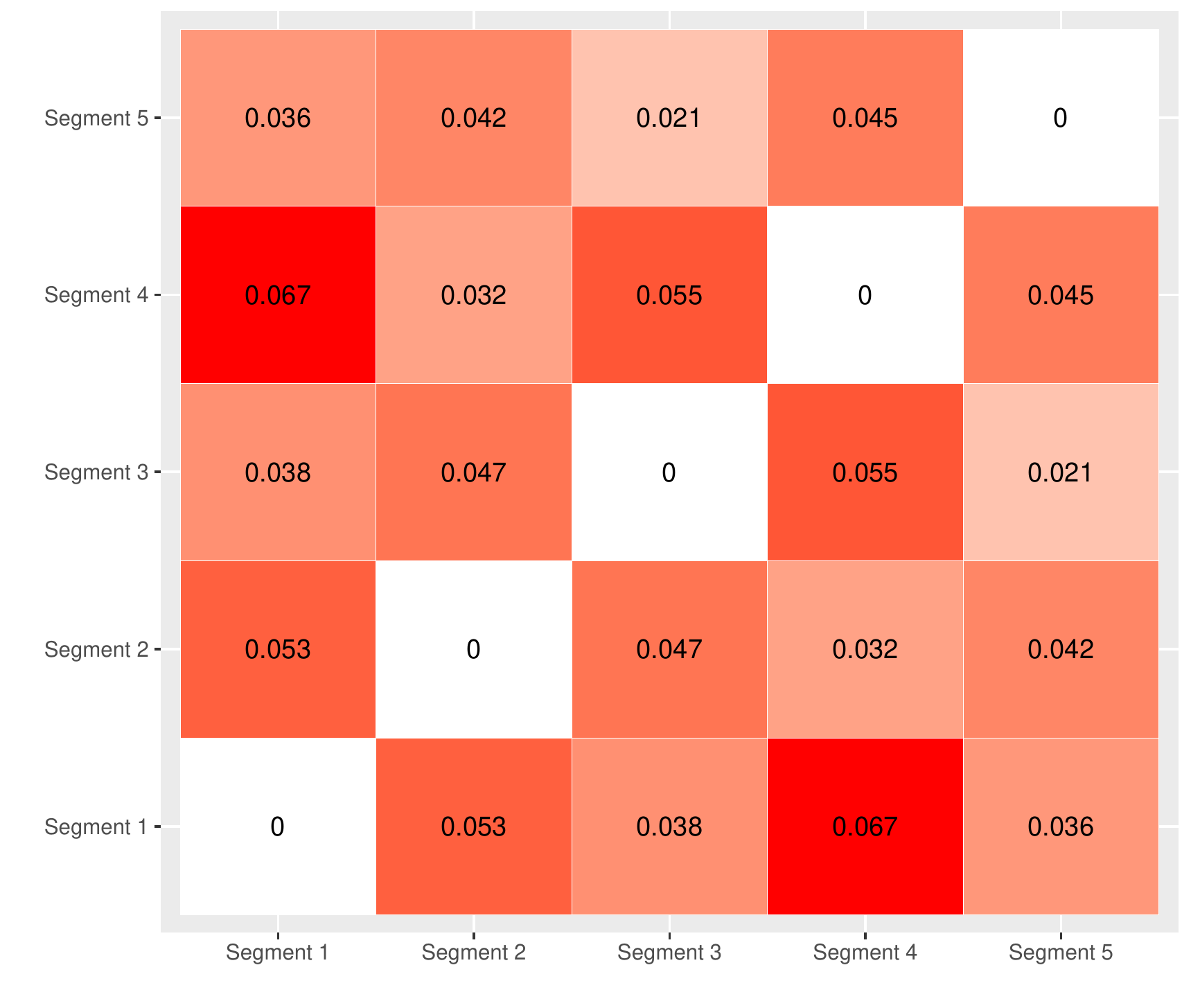}
    \end{subfigure}}
    
    \caption{Left: heat map of Hamming distances between the estimated low-rank components; Right: heat map of Hamming distances between the estimated sparse components.}
    
    \label{fig:hamming-dist}
\end{figure}

\subsection{An Application to Macroeconomics Data}\label{sec:macro}

We consider the macroeconomics data obtained from the FRED database \cite{doi:10.1080/07350015.2015.1086655}. This data set comprises of 19 key macroeconomic variables, corresponding to the ``Medium" model analyzed in \cite{banbura2010large} and covering the 1959--2019 period (723 observations). The original time series data are non-stationary and we de-trend them by taking first differences.  

To select the tuning parameters $(\lambda, \mu)$, we employ a 2-dimensional grid search procedure. In our analysis, we set $\alpha_L$ based on its theoretical value in Assumption H2 to ensure identifiability of the sparse component from the low-rank one. The estimated change points are listed in Table \ref{tab:cp-events}, while the sparsity levels and ranks for each segment are plotted in Figure \ref{fig:result-pts}. 
{
The selected change points are presented in Figure 5 in Appendix G.3. A detailed discussion (due to space constraints) of related events is also provided in Appendix G.3.
}
\begin{table}[!ht]
    \spacingset{1}
    \centering
    \caption{Estimated Change Points and Candidate Related Events.}
    \label{tab:cp-events}
    \resizebox{0.8\textwidth}{!}{%
    \begin{tabular}{c|l}
    \hline\hline
        Date (mm/dd/yyyy) & Candidate Related Events\\
        \hline
        02/01/1975 & Aftermath of 1973 oil crisis \\
        \hline
        04/01/1977 & Rapid build-up of inflation expectations \\
        \hline
        12/01/1980 & Rapid increase of interest rates by the Volcker Fed \\
        \hline
        01/01/1994 & Multiple events - see Appendix G.3 \\
        \hline
        09/01/2008 & Recession following collapse of Lehman Brothers\\
        \hline
        05/01/2010 & Recovery from the Great Financial crisis of 2008  \\
        \hline\hline
    \end{tabular}}
\end{table}
\begin{figure}[!ht]
    \centering
    \resizebox{5in}{1.67in}{%
    \begin{subfigure}
        \centering
        \includegraphics[scale=.5]{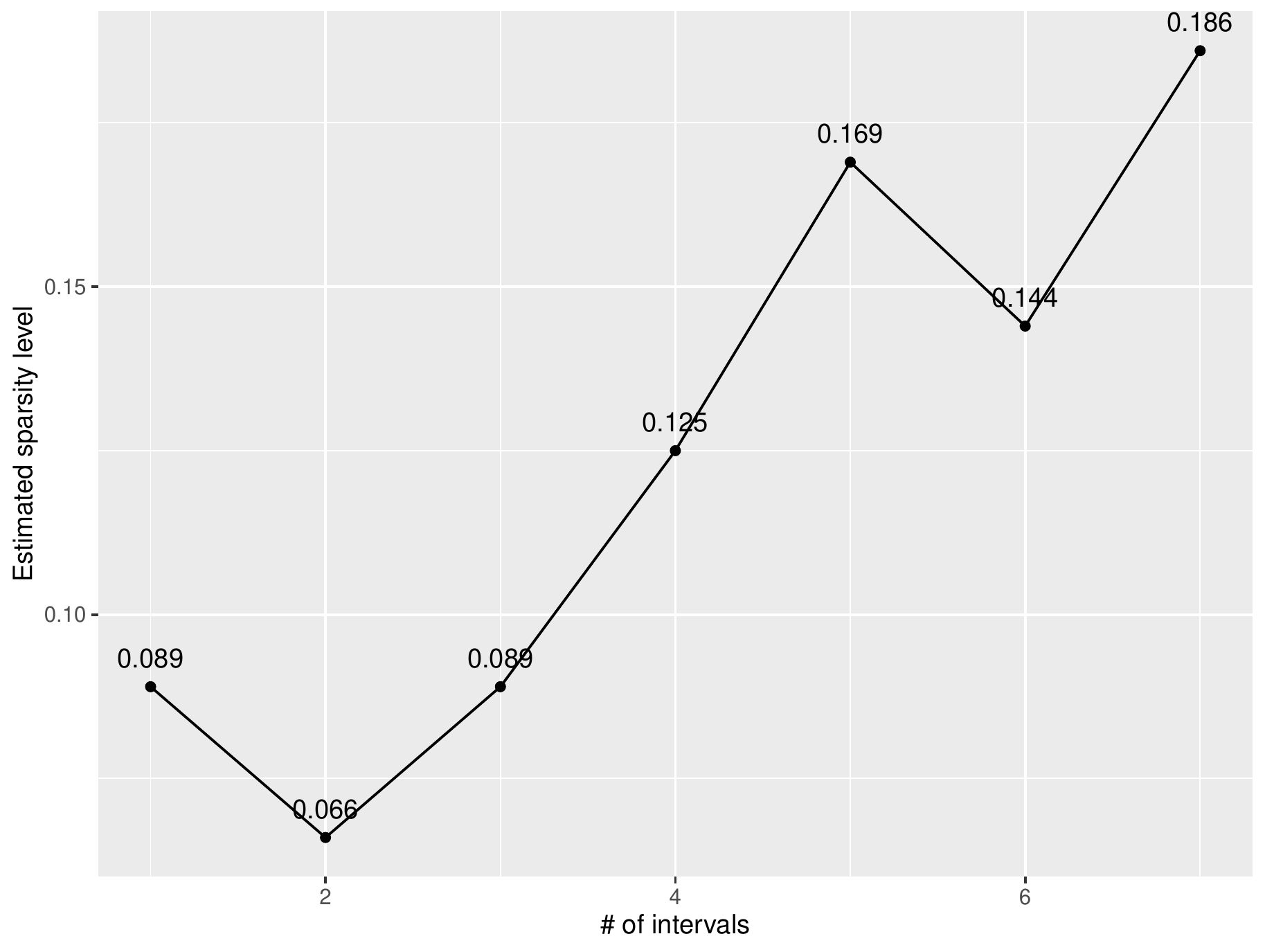}%
        \quad\quad
        \includegraphics[scale=.5]{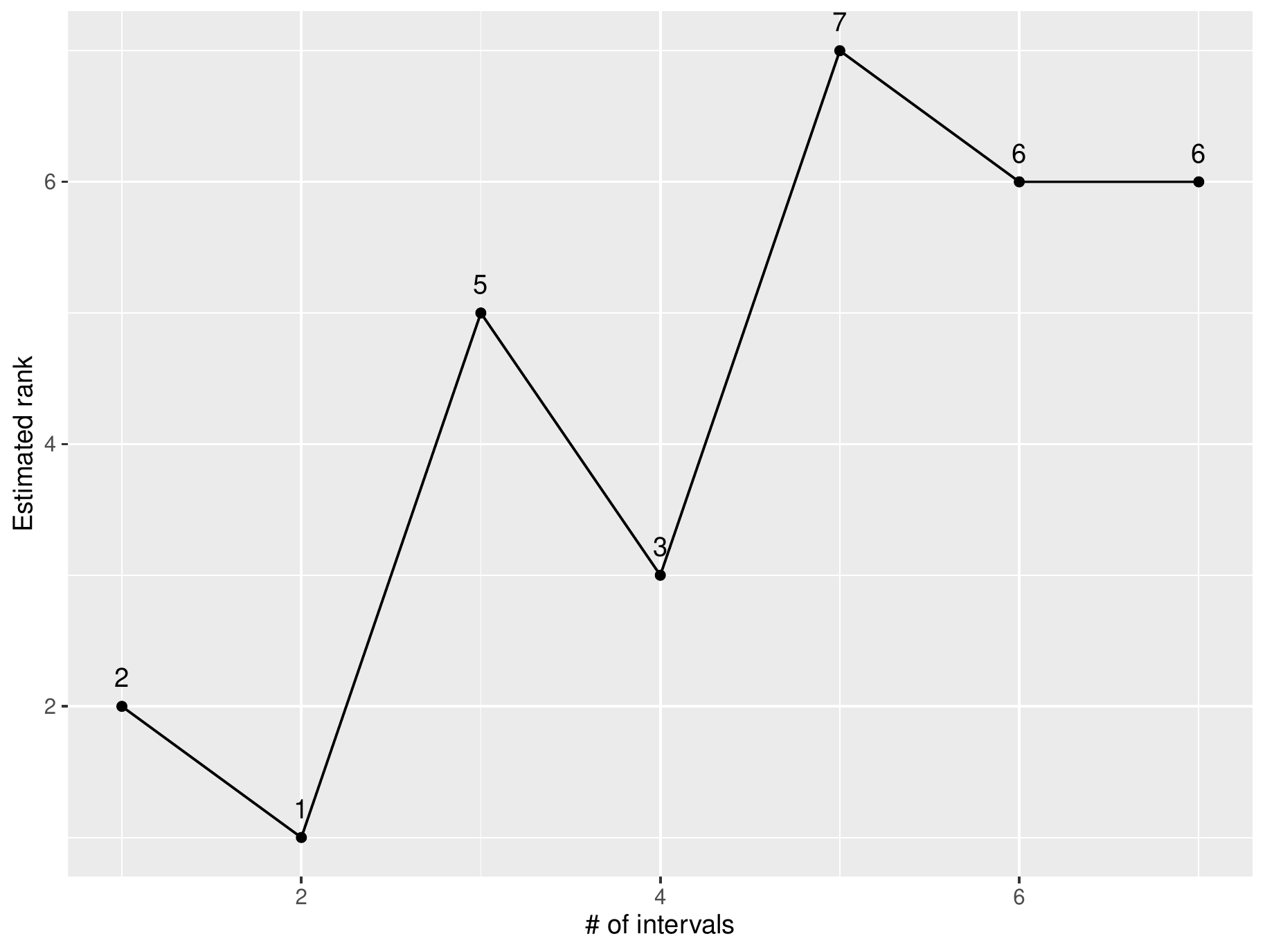}
    \end{subfigure}}
    \caption{Left panel: Estimated sparsity level for each selected interval; Right panel: Estimated rank for each selected interval.}
    \label{fig:result-pts}
\end{figure}

We also compare the results using the detection strategy based on the static factor model in \cite{barigozzi2018simultaneous}. According to \cite{fama1996multifactor}, we set the maximum number of factors to three and the estimated change points are listed in Table~\ref{tab:factorcpt}.
\begin{table}[!ht]
    \spacingset{1}
    \centering
    \caption{Estimated Change Points by the Detection Strategy based on a Factor Model.}
    \label{tab:factorcpt}
    \begin{tabular}{c|c}
        \hline\hline
        Date (mm/dd/yyyy) & Candidate Related Events \\
        \hline
        12/01/1979 & Rapid increase of interest rates by the Volcker Fed \\
        \hline
        01/01/1985 & Multiple events \\
        \hline
        11/01/1993 & Multiple events \\
        \hline
        04/01/2008 & Prequel to the Great Financial Crisis \\
        \hline\hline
    \end{tabular}
\end{table}

The factor model misses important events, including the economic recovery following the Financial Crisis of 2008 and the recession following the first oil crisis of 1973. Further, it identifies a change point in early April of 2008, even though most of the macroeconomic (as opposed to financial market) indices started deteriorating in the summer of 2008 and tumbled in the 3rd quarter, following the collapse of Lehman Brothers in mid-September.

\section{Concluding Remarks}\label{sec:concluding remarks}

The paper addressed the problem of multiple change point detection in reduced rank VAR models. The key innovation is the development of a two-step strategy that obtains consistent estimates of the change points and the model parameters. Other strategies for detecting multiple change points in high-dimensional models, such as fused penalties or binary segmentation type of procedures, either require very stringent conditions or are not directly applicable. Further, dynamic programming entails a quadratic computational cost in the number of time points compared to a linear cost for the proposed strategy. To enhance computational efficiency, we introduced a surrogate weakly sparse model and identified sufficient conditions under which the aforementioned 2-step strategy detects change points in low-rank and sparse VAR models as accurately as using the correctly specified model, but at significant computational gains. 

{
In the algorithmic and technical results presented, similar to the case of a sparse VAR model with change points (\cite{wang2019localizing}), we assume a simple structure on the error terms, i.e., in segment $j$, $\epsilon_t^j \sim \mathcal{N}(0,\sigma^2I)$, where $\sigma$ is a fixed constant independent of $j$. Such a simple structure on the covariance matrices of error terms ensures the identifiability of change points, since a change in the transition matrices would imply that the second order structure (the auto-correlation function) of the stochastic process before and after the change points have changed, thus the definition of change points becomes meaningful. It is of interest to investigate in future work a general covariance matrix $\Sigma_E$, or even segment specific ones $\Sigma_E^j$, including conditions that lead to changes in the segment specific auto-correlation function of the process.
}

Further, the proposed strategy is directly applicable to other forms of structured sparsity in the transition matrix of the VAR model, including low-rank plus structured sparse, or structured sparse plus sparse, as discussed for stationary models in \cite{basu2019low}.

Finally, the presentation focused on a VAR model with a single lag, but both the modeling framework and the developed 2-step detection strategy can be extended to $\mbox{VAR}(d)$ processes with $d>1$ in a similar manner, as presented in \cite{basu2019low}.

\bibliography{main}

@article{friedrich2008complexity,
  title={Complexity penalized M-estimation: fast computation},
  author={Friedrich, Felix and Kempe, Angela and Liebscher, Volkmar and Winkler, Gerhard},
  journal={Journal of Computational and Graphical Statistics},
  volume={17},
  number={1},
  pages={201--224},
  year={2008},
  publisher={Taylor \& Francis}
}

@inproceedings{rinaldo2021localizing,
  title={Localizing changes in high-dimensional regression models},
  author={Rinaldo, Alessandro and Wang, Daren and Wen, Qin and Willett, Rebecca and Yu, Yi},
  booktitle={International Conference on Artificial Intelligence and Statistics},
  pages={2089--2097},
  year={2021},
  organization={PMLR}
}

@article{davis2006structural,
  title={Structural break estimation for nonstationary time series models},
  author={Davis, Richard A and Lee, Thomas C M and Rodriguez-Yam, Gabriel A},
  journal={Journal of the American Statistical Association},
  volume={101},
  number={473},
  pages={223--239},
  year={2006},
  publisher={Taylor \& Francis}
}

@article{wang2004forecasting,
  title={Forecasting performance of multivariate time series models with full and reduced rank: An empirical examination},
  author={Wang, Zijun and Bessler, David A},
  journal={International Journal of Forecasting},
  volume={20},
  number={4},
  pages={683--695},
  year={2004},
  publisher={Elsevier}
}

@article{wang2020univariate,
  title={Univariate mean change point detection: Penalization, cusum and optimality},
  author={Wang, Daren and Yu, Yi and Rinaldo, Alessandro},
  journal={Electronic Journal of Statistics},
  volume={14},
  number={1},
  pages={1917--1961},
  year={2020},
  publisher={The Institute of Mathematical Statistics and the Bernoulli Society}
}

@article{hartigan1979algorithm,
  title={Algorithm AS 136: A k-means clustering algorithm},
  author={Hartigan, John A and Wong, Manchek A},
  journal={Journal of the Royal Statistical Society. Series C (Applied Statistics)},
  volume={28},
  number={1},
  pages={100--108},
  year={1979},
  publisher={JSTOR}
}

@article{bai2020multiple,
  title={Multiple Change Points Detection in Low Rank and Sparse High Dimensional Vector Autoregressive Models},
  author={Bai, Peiliang and Safikhani, Abolfazl and Michailidis, George},
  journal={IEEE Transactions on Signal Processing},
  volume={68},
  pages={3074--3089},
  year={2020},
  publisher={IEEE}
}

@article{wang2019localizing,
  title={Localizing Changes in High-Dimensional Vector Autoregressive Processes},
  author={Wang, Daren and Yu, Yi and Rinaldo, Alessandro and Willett, Rebecca},
  journal={arXiv preprint arXiv:1909.06359},
  year={2019}
}

@book{csorgo1997limit,
  title={Limit theorems in change-point analysis},
  author={Cs{\"o}rg{\"o}, Mikl{\'o}s and Horv{\'a}th, Lajos},
  volume={18},
  year={1997},
  publisher={John Wiley \& Sons Inc}
}

@article{chan2014group,
  title={Group LASSO for structural break time series},
  author={Chan, Ngai Hang and Yau, Chun Yip and Zhang, Rong-Mao},
  journal={Journal of the American Statistical Association},
  volume={109},
  number={506},
  pages={590--599},
  year={2014},
  publisher={Taylor \& Francis}
}

@article{fama1996multifactor,
  title={Multifactor explanations of asset pricing anomalies},
  author={Fama, Eugene F and French, Kenneth R},
  journal={The journal of finance},
  volume={51},
  number={1},
  pages={55--84},
  year={1996},
  publisher={Wiley Online Library}
}

@article{safikhani2017joint,
  title={Joint structural break detection and parameter estimation in high-dimensional non-stationary VAR models},
  author={Safikhani, Abolfazl and Shojaie, Ali},
  journal={Journal of the American Statistical Association (Theory and Methods), to appear},
  year={2020}
}

@article{basu2015regularized,
  title={Regularized estimation in sparse high-dimensional time series models},
  author={Basu, Sumanta and Michailidis, George},
  journal={The Annals of Statistics},
  volume={43},
  number={4},
  pages={1535--1567},
  year={2015},
  publisher={Institute of Mathematical Statistics}
}

@article{basu2019low,
  title={Low rank and structured modeling of high-dimensional vector autoregressions},
  author={Basu, Sumanta and Li, Xianqi and Michailidis, George},
  journal={IEEE Transactions on Signal Processing},
  volume={67},
  number={5},
  pages={1207--1222},
  year={2019},
  publisher={IEEE}
}

@article{agarwal2012noisy,
  title={Noisy matrix decomposition via convex relaxation: Optimal rates in high dimensions},
  author={Agarwal, Alekh and Negahban, Sahand and Wainwright, Martin J and others},
  journal={The Annals of Statistics},
  volume={40},
  number={2},
  pages={1171--1197},
  year={2012},
  publisher={Institute of Mathematical Statistics}
}

@article{negahban2012unified,
  title={A unified framework for high-dimensional analysis of $ M $-estimators with decomposable regularizers},
  author={Negahban, Sahand N and Ravikumar, Pradeep and Wainwright, Martin J and Yu, Bin},
  journal={Statistical Science},
  volume={27},
  number={4},
  pages={538--557},
  year={2012},
  publisher={Institute of Mathematical Statistics}
}

@ARTICLE{trujillo2017effect,
  title={The effect of electroencephalogram (EEG) reference choice on information-theoretic measures of the complexity and integration of EEG signals},
  author={Trujillo, Logan T and Stanfield, Candice T and Vela, Ruben D},
  journal={Frontiers in neuroscience},
  volume={11},
  pages={425},
  year={2017},
  publisher={Frontiers}
}

@article{cho2015multiple,
  title={Multiple-change-point detection for high dimensional time series via sparsified binary segmentation},
  author={Cho, Haeran and Fryzlewicz, Piotr},
  journal={Journal of the Royal Statistical Society: Series B (Statistical Methodology)},
  volume={77},
  number={2},
  pages={475--507},
  year={2015},
  publisher={Wiley Online Library}
}

@article{chandrasekaran2011rank,
  title={Rank-sparsity incoherence for matrix decomposition},
  author={Chandrasekaran, Venkat and Sanghavi, Sujay and Parrilo, Pablo A and Willsky, Alan S},
  journal={SIAM Journal on Optimization},
  volume={21},
  number={2},
  pages={572--596},
  year={2011},
  publisher={SIAM}
}

@article{banbura2010large,
  title={Large Bayesian vector auto regressions},
  author={Ba{\'n}bura, Marta and Giannone, Domenico and Reichlin, Lucrezia},
  journal={Journal of applied Econometrics},
  volume={25},
  number={1},
  pages={71--92},
  year={2010},
  publisher={Wiley Online Library}
}

@article{doi:10.1080/07350015.2015.1086655,
author = {Michael W. McCracken and Serena Ng},
title = {FRED-MD: A Monthly Database for Macroeconomic Research},
journal = {Journal of Business \& Economic Statistics},
volume = {34},
number = {4},
pages = {574-589},
year  = {2016},
publisher = {Taylor & Francis},
doi = {10.1080/07350015.2015.1086655},
URL = { 
        https://doi.org/10.1080/07350015.2015.1086655
},
eprint = { 
        https://doi.org/10.1080/07350015.2015.1086655
}
}

@article{michailidis2013autoregressive,
  title={Autoregressive models for gene regulatory network inference: Sparsity, stability and causality issues},
  author={Michailidis, George and d’Alch{\'e}-Buc, Florence},
  journal={Mathematical biosciences},
  volume={246},
  number={2},
  pages={326--334},
  year={2013},
  publisher={Elsevier}
}

@incollection{stock2016dynamic,
  title={Dynamic factor models, factor-augmented vector autoregressions, and structural vector autoregressions in macroeconomics},
  author={Stock, James H and Watson, Mark W},
  booktitle={Handbook of macroeconomics},
  volume={2},
  pages={415--525},
  year={2016},
  publisher={Elsevier}
}

@article{lin2017regularized,
  title={Regularized estimation and testing for high-dimensional multi-block vector-autoregressive models},
  author={Lin, Jiahe and Michailidis, George},
  journal={The Journal of Machine Learning Research},
  volume={18},
  number={1},
  pages={4188--4236},
  year={2017},
  publisher={JMLR. org}
}

@book{kilian2017structural,
  title={Structural vector autoregressive analysis},
  author={Kilian, Lutz and L{\"u}tkepohl, Helmut},
  year={2017},
  publisher={Cambridge University Press}
}

@article{friston2014granger,
  title={Granger causality revisited},
  author={Friston, Karl J and Bastos, Andr{\'e} M and Oswal, Ashwini and van Wijk, Bernadette and Richter, Craig and Litvak, Vladimir},
  journal={Neuroimage},
  volume={101},
  pages={796--808},
  year={2014},
  publisher={Elsevier}
}

@article{schroder2019fresped,
  title={FreSpeD: Frequency-specific change-point detection in epileptic seizure multi-channel EEG data},
  author={Schr{\"o}der, Anna Louise and Ombao, Hernando},
  journal={Journal of the American Statistical Association},
  volume={114},
  number={525},
  pages={115--128},
  year={2019},
  publisher={Taylor \& Francis}
}

@article{bardsley2017change,
  title={Change point tests in functional factor models with application to yield curves},
  author={Bardsley, Patrick and Horv{\'a}th, Lajos and Kokoszka, Piotr and Young, Gabriel},
  journal={The Econometrics Journal},
  volume={20},
  number={1},
  pages={86--117},
  year={2017},
  publisher={Oxford University Press Oxford, UK}
}

@article{barigozzi2018simultaneous,
  title={Simultaneous multiple change-point and factor analysis for high-dimensional time series},
  author={Barigozzi, Matteo and Cho, Haeran and Fryzlewicz, Piotr},
  journal={Journal of Econometrics},
  volume={206},
  number={1},
  pages={187--225},
  year={2018},
  publisher={Elsevier}
}

@article{billio2012econometric,
  title={Econometric measures of connectedness and systemic risk in the finance and insurance sectors},
  author={Billio, Monica and Getmansky, Mila and Lo, Andrew W and Pelizzon, Loriana},
  journal={Journal of financial economics},
  volume={104},
  number={3},
  pages={535--559},
  year={2012},
  publisher={Elsevier}
}

@article{loh2012high,
  title={High-dimensional regression with noisy and missing data: provable guarantees with nonconvexity},
  author={Loh, Po-Ling and Wainwright, Martin J},
  journal={The Annals of Statistics},
  pages={1637--1664},
  year={2012},
  publisher={JSTOR}
}

@article{velu1986reduced,
  title={Reduced rank models for multiple time series},
  author={Velu, Raja P and Reinsel, Gregory C and Wichern, Dean W},
  journal={Biometrika},
  volume={73},
  number={1},
  pages={105--118},
  year={1986},
  publisher={Oxford University Press}
}

@article{ahn1988nested,
  title={Nested reduced-rank autoregressive models for multiple time series},
  author={Ahn, Sung K and Reinsel, Gregory C},
  journal={Journal of the American Statistical Association},
  volume={83},
  number={403},
  pages={849--856},
  year={1988},
  publisher={Taylor \& Francis Group}
}

@article{lam2011estimation,
  title={Estimation of latent factors for high-dimensional time series},
  author={Lam, Clifford and Yao, Qiwei and Bathia, Neil},
  journal={Biometrika},
  volume={98},
  number={4},
  pages={901--918},
  year={2011},
  publisher={Oxford University Press}
}

@article{li2014new,
  title={A new method of dynamic latent-variable modeling for process monitoring},
  author={Li, Gang and Qin, S Joe and Zhou, Donghua},
  journal={IEEE Transactions on Industrial Electronics},
  volume={61},
  number={11},
  pages={6438--6445},
  year={2014},
  publisher={IEEE}
}

@article{stock2002forecasting,
	title={Forecasting using principal components from a large number of predictors},
	author={Stock, James H and Watson, Mark W},
	journal={Journal of the American Statistical Association},
	volume={97},
	number={460},
	pages={1167--1179},
	year={2002},
	publisher={Taylor \& Francis}
}

@article{bai2008large,
	title={Large dimensional factor analysis},
	author={Bai, Jushan and Ng, Serena},
	journal={Foundations and Trends{\textregistered} in Econometrics},
	volume={3},
	number={2},
	pages={89--163},
	year={2008},
	publisher={Now Publishers, Inc.}
}

@inproceedings{liu2018estimating,
  title={Estimating Latent Brain Sources with Low-Rank Representation and Graph Regularization},
  author={Liu, Feng and Wang, Shouyi and Qin, Jing and Lou, Yifei and Rosenberger, Jay},
  booktitle={International Conference on Brain Informatics},
  pages={304--316},
  year={2018},
  organization={Springer}
}

@inproceedings{jao2018using,
  title={Using robust principal component analysis to reduce EEG intra-trial variability},
  author={Jao, Ping-Keng and Chavarriaga, Ricardo and Mill{\'a}n, Jos{\'e} del R},
  booktitle={40th Annual Conference of the IEEE Engineering in Medicine and Biology Society (EMBC)},
  pages={1956--1959},
  year={2018},
  organization={IEEE}
}

@article{roy2017change,
  title={Change point estimation in high dimensional Markov random-field models},
  author={Roy, Sandipan and Atchad{\'e}, Yves and Michailidis, George},
  journal={Journal of the Royal Statistical Society: Series B (Statistical Methodology)},
  volume={79},
  number={4},
  pages={1187--1206},
  year={2017},
  publisher={Wiley Online Library}
}

@article{bhattacharjee2018change,
  title={Change point estimation in a dynamic stochastic block model},
  author={Bhattacharjee, Monika and Banerjee, Moulinath and Michailidis, George},
  journal={Journal of Machine Learning Research},
  volume={21},
  number={107},
  pages={1--59},
  year={2020}
}

@book{lutkepohl2013introduction,
  title={Introduction to multiple time series analysis},
  author={L{\"u}tkepohl, Helmut},
  year={2013},
  publisher={Springer Science \& Business Media}
}

@article{zhang2008sparsity,
  title={The sparsity and bias of the lasso selection in high-dimensional linear regression},
  author={Zhang, Cun-Hui and Huang, Jian},
  journal={The Annals of Statistics},
  volume={36},
  number={4},
  pages={1567--1594},
  year={2008},
  publisher={Institute of Mathematical Statistics}
}

@book{bordo2007retrospective,
  title={A retrospective on the Bretton Woods system: lessons for international monetary reform},
  author={Bordo, Michael D and Eichengreen, Barry},
  year={2007},
  publisher={University of Chicago Press}
}

@article{kareken1978inflation,
  title={Inflation: an extreme view},
  author={Kareken, John H},
  journal={Quarterly Review},
  number={win},
  year={1978},
  publisher={Federal Reserve Bank of Minneapolis}
}

@article{orphanides2004monetary,
  title={Monetary policy rules, macroeconomic stability, and inflation: A view from the trenches},
  author={Orphanides, Athanasios},
  journal={Journal of Money, Credit and Banking},
  pages={151--175},
  year={2004},
  publisher={JSTOR}
}

@book{eichengreen2014hall,
  title={Hall of mirrors: The great depression, the great recession, and the uses-and misuses-of history},
  author={Eichengreen, Barry},
  year={2014},
  publisher={Oxford University Press}
}

@article{hsu2011robust,
  title={Robust matrix decomposition with sparse corruptions},
  author={Hsu, Daniel and Kakade, Sham M and Zhang, Tong},
  journal={IEEE Transactions on Information Theory},
  volume={57},
  number={11},
  pages={7221--7234},
  year={2011},
  publisher={IEEE}
}

@article{box1977canonical,
  title={A canonical analysis of multiple time series},
  author={Box, George EP and Tiao, George C},
  journal={Biometrika},
  volume={64},
  number={2},
  pages={355--365},
  year={1977},
  publisher={Oxford University Press}
}



\bigskip
\bigskip
\bigskip
\begin{center}
    {\LARGE \bf Supplementary Material for Multiple Change Point Detection in Reduced Rank High Dimensional Vector Autoregressive Models}
\end{center}

In the Supplement, we first present conditions related to the underlying VAR processes in Appendix~\ref{appendix:A}, the detection algorithms are provided in Appendix~\ref{appendix:B}, a comparison of the true and the surrogate model is provided in Appendix~\ref{appendix:C}, then introduce auxiliary lemmas with their proofs in Appendix~\ref{appendix:D}, followed by the proofs of then main Theorems in Appendix~\ref{appendix:E}, the results for additional numerical experiments in Appendix~\ref{appendix:F}, and the supplementary results for real-data applications in Appendix~\ref{appendix:G}.

\begin{appendices}

\section{Conditions on VAR Processes}\label{appendix:A}
Recall that the piece-wise stationary structured VAR(1) model is given by:
\begin{equation}
    \label{eq:14}
    \begin{aligned}
    X_t &= (L_1^\star + S_1^\star)X_{t-1} + \epsilon_t^1, \quad t=1,2,\dots, \tau^\star, \\
    X_t &= (L_2^\star + S_2^\star)X_{t-1} + \epsilon_t^2, \quad t=\tau^\star+1, \dots, T,
    \end{aligned}
\end{equation}
where $X_t \in \mathbb{R}^{p}$ is a vector of observed time series at time $t$, $L_1^\star$, $L_2^\star$ are the $p\times p$ low-rank components and $S_1^\star$, $S_2^\star$ are the $p\times p$ sparse matrices for the corresponding models in the two time intervals. The $p$-dimensional noise process $\epsilon_t^1$ and $\epsilon_t^2$ are independent and identically distributed from Gaussian distributions; {$\epsilon_t^j \overset{\text{iid}}{\sim} \mathcal{N}(0, \sigma^2I)$, where $I$ denotes the identity matrix and $\sigma$ is a fixed constant.} In this modeling formulation, there are two independent VAR processes over the whole time interval $[1,T]$; the first process is observed up to time $\tau^\star$ and the second VAR process is observed afterwards. In other words, there are two stationary VAR processes concatenated together at the (unknown) time $\tau^\star$. The independence assumption on these two VAR processes makes it possible to write down the density function of the stochastic process $\lbrace X_t \rbrace_{t=1}^T $ as the product of the density functions of the two stationary processes, thus properly defining the infinite-dimensional distribution of the stochastic process $X_t$.


\textit{Stability.} To show consistency of both change points and model parameters, we assume that the VAR(1) model in \eqref{eq:14} is piece-wisely stable. Formally, we assume that the characteristic polynomial with respect to the $j$-th segment $\mathcal{A}_j(z) \overset{\text{def}}{=} I_p - Az$ satisfies $\text{det}(\mathcal{A}_j(z)) \neq 0$ for $z \in \{z \in \mathbb{C}: |z| = 1\}$. This is a common assumption for stationary VAR processes \citep{lutkepohl2013introduction, basu2015regularized}. This assumption also indicates that the spectral density function for the $j$-th segment of the VAR model:
\begin{equation*}
    f_j(\theta)= \frac{1}{2\pi}\left(\mathcal{A}^{-1}_j(e^{i\theta})\right)\Sigma_j\left(\mathcal{A}_j^{-1}(e^{i\theta})\right)^{\dagger},\quad \theta \in [-\pi, \pi],
\end{equation*}
is upper bounded in the spectral norm. In addition, we also assume that \\
(i) the maximum eigenvalue is bounded:
\begin{equation*}
    \mathcal{M} \overset{\text{def}}{=} \max_{1\leq j \leq m_0+1}\mathcal{M}(f_j) = \max_{1\leq j\leq m_0+1}\sup_{\theta \in [-\pi, \pi]}\Lambda_{\max}(f_j(\theta)) < +\infty;
\end{equation*}
(ii) the minimum eigenvalue is bounded away from zero:
\begin{equation*}
    \mathfrak{m} \overset{\text{def}}{=} \min_{1\leq j \leq m_0+1}\mathfrak{m}(f_j) = \min_{1\leq j \leq m_0+1}\inf_{\theta \in [-\pi, \pi]}\Lambda_{\min}(f_j(\theta)) > 0.
\end{equation*}
As discussed in \cite{basu2015regularized}, it is known that $\mathcal{M}(f_j)$ and $\mathfrak{m}(f_j)$ reflect the stability of a VAR process and they are related to $\mu_{\max}(\mathcal{A}_j)$ and $\mu_{\min}(\mathcal{A}_j)$ as follows:
\begin{equation*}
    \mathfrak{m}(f_j) \geq \frac{1}{2\pi}\frac{\Lambda_{\min}(\Sigma_j)}{\mu_{\max}(\mathcal{A})}\ \text{and}\  \mathcal{M}(f_j) \leq \frac{1}{2\pi}\frac{\Lambda_{\max}(\Sigma_j)}{\mu_{\min}(\mathcal{A})},
\end{equation*}
where {$\mu_{\max}(\mathcal{A}_j) = \max_{|z|=1}\Lambda_{\max}\left(\mathcal{A}^{\dagger}_j(z)\mathcal{A}_j(z)\right)$} and { $\mu_{\min}(\mathcal{A}_j) = \min_{|z|=1}\Lambda_{\min}\left(\mathcal{A}^{\dagger}_j(z)\mathcal{A}_j(z)\right)$}, respectively. For the specific low-rank plus sparse structure for the underlying VAR process, we can verify that for each segment, an upper bound of $\mu_{\max}(\mathcal{A})$ is given by:
\begin{equation*}
    \mu_{\max}(\mathcal{A}_j) = \|I_p - (L_j^\star+S_j^\star)z\|^2_{\text{op}} \leq \|I_p\|^2_{\text{op}} + \|L_j^\star\|^2_{\text{op}} + \|S_j^\star\|^2_{\text{op}} \leq \left(1 + \sigma_{\max}(L_j^\star) + M_Sd_j^\star\right)^2,
\end{equation*}
where $\sigma_{\max}(L_j^\star)$ is the maximum singular value of $L_j^\star$.

Next, we introduce two key assumptions, required in our theoretical proofs 
that are commonly made for high-dimensional regularized estimation problems.
\begin{itemize}[leftmargin=0pt]
    \item \textbf{Restricted strong convexity.} We start by introducing the \textit{weighted regularizer} $\mathcal{Q}$ that combines the penalties for the low-rank and the sparse components. Specifically, for any pair $(\lambda, \mu)$ of positive numbers, we define the weighted regularizer $\mathcal{Q}$ with respect to the low-rank matrix $L$ and the sparse matrix $S$ as:
    \begin{equation*}
        \mathcal{Q}(L, S) \overset{\text{def}}{=} \|L\|_* + \frac{\lambda}{\mu}\|S\|_1,
    \end{equation*}
    and the associated norm $\Phi$ is given by:
    \begin{equation*}
        \Phi(\Delta) \overset{\text{def}}{=} \inf_{\Delta = L+S}\mathcal{Q}(L,S).
    \end{equation*}
    Then, the restricted strong convexity (RSC) condition becomes:
    \begin{definition}[{\bf Restricted Strong Convexity (RSC)}]
        \label{def:1}
        A generic linear operator $\mathfrak{X}: \mathbb{R}^{p\times p} \to \mathbb{R}^{p\times T}$ satisfies the RSC condition with respect to the associated norm $\Phi$ with curvature constant $\nu_{\text{RSC}}>0$ and tolerance constant $\kappa>0$ if:
        \begin{equation*}
            \frac{1}{2T}\|\mathfrak{X}(\Delta)\|_F^2 \geq \frac{\nu_{\text{RSC}}}{2}\|\Delta\|_F^2 - \kappa\Phi^2(\Delta).
        \end{equation*}
    \end{definition}
    Note that the RSC condition is equivalent to the \textit{restricted eigenvalues} (RE) condition \citep{loh2012high, basu2015regularized}. In our work, the RSC condition is required in the proofs of Lemma \ref{lemma:5} case (b), and Theorem 2.
    \item \textbf{Deviation bounds.} 
    Given a time interval $I$, according to \cite{loh2012high,basu2015regularized,basu2019low}, we need to upper bound the following quantities for the sparse and the low-rank components respectively:
    \begin{equation*}
        \text{(for sparse)}\ \left\| \sum_{t\in I}X_{t-1}\epsilon_t^\prime \right\|_\infty,\ \text{and}\ \text{(for low-rank)}\ \left\| \sum_{t\in I}X_{t-1}\epsilon_t^\prime \right\|_{\text{op}}.
    \end{equation*}
    Unlike the deviations in \cite{basu2015regularized}, in our model, as the time interval $I$ changes, the misspecified intervals are needed to be considered in both Algorithms~\ref{algo:1} and \ref{algo:2}. In order to upper bound the misspecified terms, a sufficiently large deviation bound is required:
    \begin{definition}[{\bf Deviation Bounds with Misspecifications}]
    \label{def:2}
        For the given time interval $I$ with $|I| \geq T\xi_T$, and the realization $\{X_t\}$ generated from the piecewisely stationary VAR(1) model in (1) or (8), it is assumed that there exist deterministic functions $\mathbb{Q}(A^\star, \Sigma_\epsilon)$, $\mathbb{Q}^\prime_\infty(A^\star, \Sigma_\epsilon)$ and $\mathbb{Q}^\prime_{\text{op}}(A^\star, \Sigma_\epsilon)$ such that the realizations $\{X_t\}_{t\in I}$, and the noise terms $\{\epsilon_t\}_{t\in I}$ satisfy:\\ 

        (a) For $I \subset [\tau^\star_{j-1}, \tau^\star_j)$:
        \begin{equation*}
            \frac{1}{|I|}\left\| \sum_{t\in I}X_{t-1}\epsilon_t^\prime \right\|_\infty \leq 4\mathbb{Q}(A^\star_j,\Sigma_\epsilon)\sqrt{\frac{\log p}{|I|}},\ \text{and}\ \frac{1}{|I|}\left\| \sum_{t\in I}X_{t-1}\epsilon_t^\prime \right\|_{\text{op}} \leq 4\mathbb{Q}(A^\star_j,\Sigma_\epsilon)\sqrt{\frac{p}{|I|}}.
        \end{equation*}
        (b) For $I = [b, e)$, $\tau^\star_{j-1} < b \leq \tau^\star_{j} < e < \tau^\star_{j+1}$, and $e-\tau^\star_j \leq \tau^\star_j-b$, with the misspecified interval $[\tau^\star_j, e)$:
        \begin{equation*}
            \begin{aligned}
                &\frac{1}{|I|}\left\|\sum_{t\in I}X_{t-1}\epsilon_t^\prime \right\|_\infty \leq 4\mathbb{Q}(A^\star_j,\Sigma_\epsilon)\sqrt{\frac{\log p + \log |I|}{|I|}} + \frac{e-\tau_j^\star}{|I|}\mathbb{Q}^\prime_\infty(A^\star_{j+1}-A^\star_{j}, \Sigma_{\epsilon}), \\
                &\frac{1}{|I|}\left\|\sum_{t\in I}X_{t-1}\epsilon_t^\prime \right\|_{\text{op}} \leq 4\mathbb{Q}(A^\star_j,\Sigma_\epsilon)\sqrt{\frac{p + \log |I|}{|I|}} + \frac{e-\tau^\star_j}{|I|}\mathbb{Q}^\prime_{\text{op}}(A^\star_{j+1}-A^\star_{j}, \Sigma_{\epsilon}),
            \end{aligned}
        \end{equation*}
        where $\mathbb{Q}^\prime_\infty$ and $\mathbb{Q}^\prime_{\text{op}}$ are deterministic functions to upper bound the misspecified terms:
        \begin{equation*}
            \frac{1}{e-\tau^\star}\sum_{t=\tau^\star_j}^{e-1}X_{t-1}X_{t-1}^\prime(A_2^\star-A_1^\star)^\prime
        \end{equation*}
        with respect to the $\ell_\infty$ and operator norms, respectively.
    \end{definition}
    The verification of these two types of deviation bounds and also for the surrogate model are provided in Lemmas \ref{lemma:1}, \ref{lemma:3} and \ref{lemma:7}. In our analysis, the modified deviation bounds are essential in the proofs of Lemma~\ref{lemma:6}, Theorem 1, Theorem 3, Proposition 3, Proposition 5.
\end{itemize}
\newpage
\section{Detection Algorithms}\label{appendix:B}
In this section, we provide schematic representations for both the exhaustive search and the backwards elimination algorithms. 

First, we exhibit the main procedures for the exhaustive search to detect a single change point in Algorithm \ref{algo:1}.
\begin{algorithm}[!ht]
    \DontPrintSemicolon
    \KwInput{\small Time series data $\{X_t\},\ t = 0,1,\dots,n$; search domain $\mathcal{T} \subset \{1,2,\dots, T\}$; }
    \While{$\tau \in \mathcal{T}$}{
        \ Estimate the low rank and sparse components on the sub-intervals $[1,\tau)$ and $[\tau, T)$, respectively:
        \begin{equation*}
            \begin{aligned}
                &(\widehat{L}_{1,\tau}, \widehat{S}_{1,\tau}) := \argmin_{\substack{L_1\in \Omega \\ L_1, S_1\in \mathbb{R}^{p\times p}}} \left\{ \frac{1}{\tau-1}\sum_{t=1}^{\tau-1}\|X_t - (L_1+S_1)X_{t-1}\|_2^2 + \lambda_{1,\tau}\|S_1\|_1 + \mu_{1,\tau}\|L_1\|_* \right\}, \\
                &(\widehat{L}_{2,\tau}, \widehat{S}_{2,\tau}) := \argmin_{\substack{L_2\in \Omega \\ L_2, S_2\in \mathbb{R}^{p\times p}}} \left\{ \frac{1}{T-\tau}\sum_{t=\tau}^{T-1}\|X_t - (L_2+S_2)X_{t-1}\|_2^2 + \lambda_{2,\tau}\|S_2\|_1 + \mu_{2,\tau}\|L_2\|_* \right\},
            \end{aligned}
        \end{equation*}
        
        \ Estimate the change point $\widetilde{\tau}$:
        \begin{equation*}
        \widehat{\tau} := \argmin_{\tau \in \mathcal{T}}\left\{\frac{1}{T-1}\left( \sum_{t=1}^{\tau-1} \|X_t - (\widehat{S}_{1,\tau}+\widehat{L}_{1,\tau})X_{t-1}\|_2^2 + \sum_{t=\tau}^{T-1} \|X_t - (\widehat{S}_{2,\tau}+\widehat{L}_{2,\tau})X_{t-1}\|_2^2\right) \right\}.
        \end{equation*}
        
        \ Updating the time point $\tau$ by $\tau+1$
    }    
    \KwOutput{ The estimated change point $\widehat{\tau}$ and model parameters $\widehat{L}_{j,\widehat{\tau}}$, $\widehat{S}_{j,\widehat{\tau}}$, $j=1,2$. }
    \caption{\textbf{Single Change Point Detection via Exhaustive Search}}
    \label{algo:1}
\end{algorithm}

Next, we present the backwards elimination algorithm for screening the redundant candidate change points in Algorithm \ref{algo:2}. 
\begin{algorithm}
    \DontPrintSemicolon
    \KwInput{Time series data $\{X_t\}$, $t=1,2,\dots, n$; candidate change points $\{\widetilde{t}_j\}$ for $j=1,2,\dots, \widetilde{m}$. }
    \KwInit{Define the interval partition of time axis based on candidate change points:
    $\mathcal{P} \overset{\text{def}}{=} \left\{ \{1,\dots, \widetilde{t}_1\}, \{\widetilde{t}_1+1, \dots, \widetilde{t}_2\}, \dots, \{\widetilde{t}_{\widetilde{m}}+1, \dots, n\} \right\}$. Set the initial value of information criterion is $W_0=0$ and the number of final selected change points $m=\widetilde{m}$.}
    \While{$W_{m-1} \leq W_m$ and $m \neq 1$}{
        \ Let $\widetilde{\mathbf{t}} \overset{\text{def}}{=} \{\widetilde{t}_1, \dots, \widetilde{t}_{m}\}$ be the screened change points and define $W_m^* = \text{IC}(\widetilde{\mathbf{t}}; \bm{\lambda}, \bm{\mu}, \omega_n)$;
        
        \ For each $j=1,2,\dots, m$, we calculate $W_{m,-j} = \text{IC}(\widetilde{t}/\{\widetilde{t}_j\}; \bm{\lambda}, \bm{\mu}, \omega_n)$, and define $W_{m-1} = \min_j W_{m,-j}$;
        
        \ There are three cases:
        \begin{itemize}
            \item[(a)] If $W_{m-1} > W_m$, then no further step is needed. Return the current partition $\widehat{\mathcal{P}}$;
            \item[(b)] If $W_{m-1} \leq W_m$ and $m > 1$, set $j^* = \argmin_j W_{m,-j}$, then we update candidate \\ 
            change points vector $\widetilde{\mathbf{t}} \leftarrow \widetilde{\mathbf{t}} / \widetilde{t}_{j^*}$ and $m \leftarrow m-1$;
            \item[(c)] If $W_{m-1} \leq W_m$ and $m = 1$, return an empty set. 
        \end{itemize}
    }
    \KwOutput{The final set of screened change points $\{\widehat{t}_j\}$, for $j=1,2,\dots, \widehat{m}$. }
    \caption{\textbf{Screening via a Backwards Elimination Algorithm}}
    \label{algo:2}
\end{algorithm}

Prompted by a comment by a reviewer, we also consider a dynamic programming (DP) based algorithm to detect multiple change points. The key steps are listed in Algorithm~\ref{algo:3}. As mentioned in Section \ref{sec:dp}, the time complexity of the DP based algorithm is {$\mathcal{O}(T^2C(T))$}.

First, we introduce some additional notation. For any given two time points: $0 \leq l < u \leq T$, we define the following regularized regression problem:
\begin{equation*}
    \begin{aligned}
     &(\widehat{L}_{(l,u)}, \widehat{S}_{(l,u)}) \\
     &= \argmin_{(L_{(l,u)}, S_{(l,u)})}\left\{\frac{1}{u - l}\sum_{t=l}^{u-1}\|X_t - (L_{(l,u)}+S_{(l,u)})X_{t-1}\|_2^2 + \lambda_i\|S_{(l,u)}\|_1 + \mu_i \|L_{(l,u)}\|_*\right\},
    \end{aligned}
\end{equation*}
{
and the corresponding objective function $\mathcal{L}$ is given by:
\begin{equation*}
     \mathcal{L}(l, u) \overset{\text{def}}{=} \sum_{t=l}^{u-1}\|X_t - (\widehat{L}_{(l,u)}+\widehat{S}_{(l,u)})X_{t-1}\|_2^2.
\end{equation*}

According to the \textit{optimal partitioning} problem introduced in \cite{friedrich2008complexity}, we obtain the following algorithm based on Dynamic Programming (DP) as outlined in Algorithm \ref{algo:3}.
\begin{algorithm}[!ht]
    \DontPrintSemicolon
    \KwInput{Time series data $\{X_t\}$, $t=1,2,\dots, T$, the tuning parameter $\gamma$ which controls the partition, a list of empty sets to store estimated change points $\mathcal{C}(\tau)$ for $\tau=1,2,\dots, T$.}
    
    \KwInit{Let $F$ be a list with length $T$, and set $F(0) = -\gamma$, $\mathcal{C}(\tau) = \emptyset$ for $\tau=1,2,\dots, T$.}
    
    \For{$t=1,2,\dots, T$}{
        \ Let $F(t)= \min_{1\leq s\leq t}\left\{ F(s) + \mathcal{L}(s,t) + \gamma \right\}$;
        
        \ Denote $t^\prime = \argmin_{1\leq s \leq t}\left\{ F(s) + \mathcal{L}(s,t) + \gamma \right\}$;
        
        \ $\mathcal{C}(t) = \mathcal{C}(t^\prime) \cup \{t^\prime\}$;
    }
    
    \KwOutput{The set of estimated change points $\mathcal{C}(T)$.}
    \caption{{\bf {Penalized Dynamic Programming Algorithm.}}}
    \label{algo:3}
\end{algorithm}

The tuning parameter $\gamma$ aims to control the issue of over-partitioning. Intuitively, the larger the value of $\gamma$, the fewer partitions are created.  
}

\newpage
\section{Additional Technical Developments for the Weakly Sparse Surrogate Model}\label{appendix:C}
{
For the weakly sparse model, the detection procedure provided in Section 2.1 requires some modification. Specifically, for the single change point problem, the corresponding objective functions for estimating the model parameters in the intervals $[1,\tau)$ and $[\tau, T)$ for any time point $\tau \in \{1,\dots, T\}$ are given by:
\begin{equation*}
    \ell^w(A_1; \mathbf{X}^{[1:\tau)}) \overset{\text{def}}{=} \frac{1}{\tau-1}\sum_{t=1}^{\tau-1}\|X_t - A_1X_{t-1}\|_2^2,\quad \text{s.t.}\ A_1 \in \mathbb{B}_q(R_q);
\end{equation*}
\begin{equation*}
    \ell^w(A_2; \mathbf{X}^{[\tau:T)}) \overset{\text{def}}{=} \frac{1}{T-\tau}\sum_{t=\tau}^{T-1}\|X_t - A_2X_{t-1}\|_2^2,\quad \text{s.t.}\ A_2 \in \mathbb{B}_q(R_q);
\end{equation*}
where $\mathbb{B}_q(R_q)$ is the $\ell_q$ ball with radius $R_q$ introduced in \eqref{eq:12}. Therefore, analogously to \eqref{eq:3}, the objective function with respect to the change point becomes:
\begin{equation*}
    \ell^w(\tau; A_1, A_2) \overset{\text{def}}{=} \frac{1}{T-1}\left(\sum_{t=1}^{\tau-1}\|X_t - A_1X_{t-1}\|_2^2 + \sum_{t=\tau}^{T-1}\|X_t - A_2X_{t-1}\|_2^2\right).
\end{equation*}
Then, the estimator $\widehat{\tau}$ is obtained by:
\begin{equation*}
   \widehat{\tau} \overset{\text{def}}{=} \argmin_{\tau \in \mathcal{T}^w}\ell^w(\tau; \widehat{A}_1, \widehat{A}_2),
\end{equation*}
for the search domain $\mathcal{T}^w$, which is precisely specified in Assumption W2. Further, the estimators $\widehat{A}_1$ and $\widehat{A}_2$ are obtained by the above defined optimization problems.
}

Next and in order to formulate the theoretical properties of the surrogate model, we present several useful definitions and results. Specifically, for a chosen threshold $\eta_j > 0$, we firstly define the thresholded subset:
\begin{equation*}
    \mathcal{J}(\eta_j) \overset{\text{def}}{=} \left\{ (k,l) \in \{1,2,\dots, p\}^2: |A_j^\star(k,l)| > \eta_j \right\}.
\end{equation*}
Recalling the $\ell_1$ decomposition with respect to $\mathcal{J}(\eta_j)$, then we derive the upper bound of the cardinality of $\mathcal{J}(\eta_j)$ in terms of the threshold $\eta_j$ and $\ell_q$-ball radius $R_q$. Note that we have:
\begin{equation*}
    R_q \geq \sum_{k,l}|A_j^\star(k,l)|^q \geq \sum_{(k,l) \in \mathcal{J}(\eta_j)}|A_j^\star(k,l)|^q \geq \eta_j^q|\mathcal{J}(\eta_j)|,
\end{equation*}
hence, $|\mathcal{J}(\eta_j)| \leq R_q \eta_j^{-q}$ for any $\eta_j > 0$. Here, we set $\eta_j \propto \lambda_{j,\tau}^w$, for $j=1,2$, and denote  $\widetilde{\Delta}_{1,\tau} \overset{\text{def}}{=} \widetilde{A}_{1,\tau}^w - A_1^\star$, $\widetilde{\Delta}_{2,\tau} \overset{\text{def}}{=} \widetilde{A}_{2,\tau}^w - A_2^\star$, and $\widetilde{\Delta}_{1/2,\tau} \overset{\text{def}}{=} \widetilde{A}_{1,\tau}^w - A_2^\star$, where $\widetilde{A}_{j,\tau}^w$ represent the estimated transition matrices by the weakly sparse model with respect to time point $\tau$. Based on the decomposition of the $\ell_1$ norm and the discussion in Section 4.3 in \cite{negahban2012unified}, we obtain that:
\begin{equation*}
    \begin{dcases}
        \|\widetilde{\Delta}_{1,\tau}\|_1 \leq 4\sqrt{R_q}\eta_1^{-\frac{q}{2}}\|\widetilde{\Delta}_{1,\tau}\|_2 + 4R_q\eta_1^{1-q}, \\
        \|\widetilde{\Delta}_{2,\tau}\|_1 \leq 4\sqrt{R_q}\eta_2^{-\frac{q}{2}}\|\widetilde{\Delta}_{2,\tau}\|_2 + 4R_q\eta_2^{1-q}, \\
        \|\widetilde{\Delta}_{1/2,\tau}\|_1 \leq 4\sqrt{R_q}\eta_2^{-\frac{q}{2}}\|\widetilde{\Delta}_{1/2,\tau}\|_2 + 4R_q\eta_2^{1-q}.
    \end{dcases}
\end{equation*}

Based on step 2 of Algorithm~\ref{algo:2}, denote by $s_1, s_2, \dots, s_{\widetilde{m}}$ the candidate change points obtained from the rolling-window step. We analogously formulate the model as (8) in the main paper. Then, we estimate $A_{(s_{i-1}, s_i)}$ by solving the following regularized problem:
\begin{equation*}
    \widehat{A}^w_{(s_{i-1}, s_i)} = \argmin_{A \in \mathbb{B}_q(R_q)}\frac{1}{s_i - s_{i-1}}\sum_{t=s_{i-1}}^{s_i-1}\|X_t - AX_{t-1}\|_2^2.
\end{equation*}
Further, define the tuning parameter vector $\bm{\lambda}^w \overset{\text{def}}{=} (\lambda_1^w, \dots, \lambda_{\widetilde{m}}^w)$ to obtain
\begin{equation*}
    \mathcal{L}^w_T(s_1,s_2,\dots,s_m; \bm{\lambda}^w) \overset{\text{def}}{=} \sum_{i=1}^{\widetilde{m}+1} \left\{\sum_{t=s_{i-1}}^{s_i-1}\|X_t - \widehat{A}_i^wX_{t-1}\|_2^2 + \lambda_i^w\|\widehat{A}^w_i\|_1\right\}.
\end{equation*}
Then, we define the \textit{information criterion} for the weakly sparse model as follows:
\begin{equation}
    \label{eq:15}
    \text{IC}^w(s_1, s_2,\dots, s_m; \bm{\lambda}^w, \omega_T^w) \overset{\text{def}}{=} \mathcal{L}^w_T(s_1,\dots,s_m; \bm{\lambda}^w) + m\omega_T^w.
\end{equation}
The final selected change points are given by:
\begin{equation*}
    (\widehat{m}^w, \widehat{\tau}^w_i, i=1,2,\dots, \widehat{m}^w) = \argmin_{0\leq m \leq \widetilde{m}, (s_1,\dots, s_m)} \text{IC}^w(s_1,\dots,s_m; \bm{\lambda}^w, \omega_T^w).
\end{equation*}
Then, we can use the exact same backward elimination algorithm as proposed in Algorithm 2 to screen the redundant candidate change points by substituting the information criterion function with the newly defined $\text{IC}^w$ in \eqref{eq:15}. 

\noindent
\textbf{Additional Assumptions for the Multiple Change Points Problem for the Surrogate Model:}
Recall that employing the rolling-window mechanism in the multiple change points scenario will result in a number of \textit{redundant} candidate change points. By using the surrogate weakly sparse model, we obtain a few redundant candidate change points as well. Therefore, we need to remove those redundant change points by using a similar screening step as introduced in the two-step algorithm in Section 3. Similarly, we also extend Assumptions H3', H5 and H6 to the weakly sparse scenario in order to formally introduce the theoretical results for the surrogate model.
\begin{itemize}[leftmargin=*]
    \item[(W3)] Let $\Delta_T \overset{\text{def}}{=} \min_{1\leq j \leq m_0}|\tau_{j+1}^\star - \tau_j^\star|$ denote the minimum spacing between consecutive change points, there exists a vanishing positive sequence $\{\xi_T^w\}$ such that, as $T \to +\infty$, 
    \begin{equation*}
        \frac{\Delta_T}{T\xi_T} \to +\infty,\ R_q\eta_{\min}^{-q}\sqrt{\frac{\log p}{T\xi_T^w}} \to 0,\ \frac{m_0T^{1+q}\xi_TR_q^2\left(\log (p\vee T)\right)^{-q}}{\omega_T^w} \to 0,\ \text{and}\ \frac{\Delta_T}{m_0\omega_T^w} \to +\infty,
    \end{equation*}
    where $\eta_{\min} = \min_{1\leq j \leq m_0+1}\eta_j$, $\eta_j \propto \lambda^w_j$, and the definition of $\eta_j$'s are provided in Appendix \ref{appendix:C}.
    \item[(W4)] Suppose $(s_1, \dots, s_m)$ are a set of change points obtained from Step 1 of the rolling window strategy. Then, we consider the following scenarios: (a) if $|s_i - s_{i-1}| \leq T\xi_T^w$, select $\lambda_i^w = c\sqrt{T\xi_T^w \log p}$ for $i=1,2,\dots, m$; (b) if there exist two true change points $\tau_j^\star$ and $\tau_{j+1}^\star$ such that $|s_{i-1} - \tau_j^\star| \leq T\xi_T^w$ and $|s_i - \tau_{j+1}^\star| \leq T\xi_T^w$, select $\lambda_i^w = 4\left(c\sqrt{\frac{\log p}{s_i - s_{i-1}}} + M_SR_q\left(\frac{\log p}{s_i-s_{i-1}}\right)^{-\frac{q}{2}}\frac{T\xi_T^w}{s_i - s_{i-1}}\right)$; (c) otherwise, select $\lambda_i^w = 4c\sqrt{\frac{\log p + \log(s_i-s_{i-1})}{s_i - s_{i-1}}}$.
\end{itemize}
Assumption W3 is a direct extension of Assumptions H3' and H5 to the weakly sparse model. It reflects the connections among the minimum spacing $\Delta_T^w$, the radius of $\ell_q$-ball $R_q$, and the vanishing sequence $\xi_T^w$. Similar to Assumption H6, Assumption W4 specifies the selection of the tuning parameters to solve a lasso regression problem given in Appendix \ref{appendix:C}. Note that these complex tuning parameters are due to the misspecified models, per the discussion ensuing Assumption H6. 


Next, an analogue of Corollary 1 in the main text is established. We select the radius $R_w$ as $B^\prime m_0T^{1+q}\xi_T^wR_q^2\left(\log(p\vee T)\right)^{-q}$ for some large constant $B^\prime$. Then, for each estimated change point, we remove its $R_w$-radius neighborhood to establish:
\begin{manualcor}{3}
\label{cor:3}
For the estimated change points $\widehat{\tau}^w_1, \dots, \widehat{\tau}^w_{\widehat{m}^w}$ and for the $j$-th interval of length $N_j^w$ obtained after removing $R_w$-radius neighborhoods around them, by selecting the tuning parameter $\lambda_j^w = 4c_0^\prime\sqrt{\frac{\log p}{N_j^w}}$, we can establish that the error bound for the estimated model parameters is given by:
\begin{equation*}
    \|\widehat{A}^w_j - A^\star_j\|_F^2 \leq C_0R_q\left(\frac{\log p}{N_j^w}\right)^{1-\frac{q}{2}},
\end{equation*}
for some large enough universal constant $C_0>0$.
\end{manualcor}
{
Next, we compare the error bounds obtained from the low-rank plus sparse model with the surrogate weakly sparse model in the following proposition. Before we state the result, we need to clarify the following further assumptions:
\begin{itemize}[leftmargin=*]
    \item[(W5a)] The window size is set to $h = c_0\log T \left( d_{\max}^\star\log p + r_{\max}^\star p\right)$, where $c_0>0$ is a positive constant. The information ratios are $0 < \gamma_j < p$ for $j=1,2,\dots, m_0+1$, that is, all segments are sparse dominating.
    \item[(W5b)] $q$ is restricted in the range: $0 < q \leq \frac{1}{2}\frac{\log (d^\star_{\max} + r^\star_{\max})}{\log p} < 1$.
    \item[(W5c)] Let $\alpha_L \overset{\text{def}}{=}\max_j{\alpha_j}$ for each segment; then, the radius $R_q$ is upper bounded by: 
    \begin{equation*}
        d_{\max}^\star\left(\left(\frac{\alpha_L}{p}\right)^q + M_S^q\right) + (p^\star - d_{\max}^\star)|\sigma_{\max}|^q \leq R_q \leq p^{2-q}(d_{\max}^\star + r_{\max}^\star)^{1-\frac{q}{2}}\max\left\{\alpha_L, M_S\right\}^q.
    \end{equation*}
\end{itemize}
These assumptions are essential ingredients to analyze the asymptotic behaviour of the error bounds obtained from different models. Assumption W5a is designed to satisfy the Assumption H4 on the size of the rolling window $h$. Assumption W5b indicates that $q$ is not allowed to be too large, which is consistent with the preceding discussion in Section 4. Finally, Assumption W5c controls the spiky entries in the transition matrix and is in accordance with Assumption H2.
\begin{manualprop}{6}
\label{prop:6}
Suppose Assumptions (W5a)--(W5c) hold; then, the following result holds:
\begin{equation*}
    1 \leq \frac{d_H(\widetilde{\mathcal{S}}_w, \mathcal{S}^\star)}{d_H(\widetilde{\mathcal{S}}, \mathcal{S}^\star)} \leq c_0p^{2-q}(\log T)^{\frac{q}{2}}.
\end{equation*}
\end{manualprop}}

Proposition~\ref{prop:6} indicates that the error bound for the estimated change points obtained from the surrogate weakly sparse model is larger than the one obtained from the low-rank plus sparse model, while it can be asymptotically upper bounded by $\mathcal{O}(p^{2-q}(\log T)^{\frac{q}{2}})$. For the extreme case $q=0$, the surrogate model becomes a strictly sparse one and thus the upper bound is $\mathcal{O}(p^2)$ due to the dense low-rank components in the true model. The details of the proof and the required assumptions for Proposition \ref{prop:6} are provided in Appendix E.

\newpage
\section{Auxiliary Lemmas}\label{appendix:D}
\begin{lemma}
\label{lemma:1}
Given a VAR(1) series $\{X_t\}$ and a time point $s$, for any true change point $\tau_j^\star$, if $|s-\tau_j^\star| \geq T\xi_T$, and $\tau_{j-1}^\star < s < \tau_j^\star$, there exist constants $c_i > 0$ such that with probability at least $1-c_1\exp(-c_2\log p)$:
\begin{align*}
    &\sup_{1\leq j \leq m_0, |s-\tau^\star_j| \geq T\xi_T}\left\| (t_j^\star-s)^{-1}\left( \sum_{t= s}^{\tau_j^\star -1}X_{t-1}X_{t-1}^\prime - \Gamma_j(0) \right) \right\|_{\infty} \leq c_0\sqrt{\frac{\log p}{\tau_j^\star - s}}, \\
    &\sup_{1\leq j \leq m_0, |s-\tau^\star_j| \geq T\xi_T}\left\| (\tau_j^\star-s)^{-1}\sum_{t = s}^{\tau_j^\star-1}X_{t-1}
    ^\prime \epsilon_t \right\|_{\infty} \leq c_0\sqrt{\frac{\log p}{\tau_j^\star - s}}.
\end{align*}
Similarly, there exist constants $c_i^\prime > 0$, such that with probability at least $1-c_1\exp(-c_2 p)$:
\begin{align*}
    &\sup_{1\leq j \leq m_0, |s-\tau^\star_j| \geq T\xi_T}\left\| (\tau_j^\star-s)^{-1}\left( \sum_{t= s}^{\tau_j^\star -1}X_{t-1}X_{t-1}^\prime - \Gamma_j(0) \right) \right\|_{\text{op}} \leq c_0\sqrt{\frac{p}{\tau_j^\star - s}}, \\
    &\sup_{1\leq j \leq m_0, |s-\tau^\star_j| \geq T\xi_T}\left\| (\tau_j^\star-s)^{-1}\sum_{t = s}^{\tau_j^\star-1}X_{t-1}
    ^\prime \epsilon_t \right\|_{\text{op}} \leq c_0\sqrt{\frac{p}{\tau_j^\star - s}}.
\end{align*}
\end{lemma}

\begin{proof}[Proof of Lemma \ref{lemma:1}]
    The proof of the lemma follows along similar lines as that of Proposition 2.4 in \cite{basu2015regularized}, Proposition 3 in \cite{basu2019low} and Lemma 3 in \cite{safikhani2017joint} and thus is omitted.
\end{proof}

Consider the following two sets of subspaces $\{\mathcal{I}, \mathcal{I}^c\}$ and $\{A, B\}$ associated with some generic matrix $\Theta \in \mathbb{R}^{p \times p}$, in which the $\ell_1$ norm and the nuclear norm are decomposable, respectively \citep{negahban2012unified}. Specifically, let the singular value decomposition of $\Theta$ be $\Theta = U\Sigma V^\prime$ with $U$ and $V$ being orthonormal matrices. Let $r = \text{rank}(\Theta)$, and $U^r$ and $V^r$  denote the first $r$ columns of $U$ and $V$ associated with the first $r$ singular values of $\Theta$, respectively. Define:
\begin{align*}
    A &\overset{\text{def}}{=} \left\{ \Psi \in \mathbb{R}^{p \times p}: \text{row}(\Psi) \subseteq V^r \ \text{and}\ \text{col}(\Psi) \subseteq U^r \right\}, \\
    B &\overset{\text{def}}{=} \left\{ \Psi \in \mathbb{R}^{p \times p}: \text{row}(\Psi) \perp V^r \ \text{and}\ \text{col}(\Psi) \perp U^r \right\}.
\end{align*}
Let $\mathcal{J}$ be the set of indices in which $\Theta$ is nonzero. Analogously, we define
\begin{align*}
    \mathcal{I} &\overset{\text{def}}{=} \left\{ \Psi \in \mathbb{R}^{p \times p}: \Psi_{ij} = 0\  \text{for}\ (i,j) \notin \mathcal{J} \right\}, \\
    \mathcal{I}^c &\overset{\text{def}}{=} \left\{ \Psi \in \mathbb{R}^{p \times p}: \Psi_{ij} = 0\  \text{for}\ (i,j) \in \mathcal{J} \right\}.
\end{align*}
Therefore, we have $\|\Theta\|_* = \|\Theta\|_{*, A} + \|\Theta\|_{*, B}$ and $\|\Theta\|_1 = \|\Theta\|_{1,\mathcal{I}} + \|\Theta\|_{1,\mathcal{I}^c}$. 
\begin{lemma}
\label{lemma:2}
    Define the error matrices $\widehat{\Delta}^L = \widehat{L} - L^\star$ and $\widehat{\Delta}^S = \widehat{S} - S^\star$ associated with any positive parameters $\lambda$, $\mu$, and let the weighted regularizer $\mathcal{Q}$ be defined as:
    \begin{equation*}
        \mathcal{Q}(\widehat{\Delta}^L,\widehat{\Delta}^S) \overset{\text{def}}{=} \|\widehat{\Delta}^L\|_* + \frac{\lambda}{\mu}\|\widehat{\Delta}^S\|_1,
    \end{equation*}
    for the previously defined subspaces. Then, the following inequality holds:
    \begin{equation*}
        \mathcal{Q}(L^\star, S^\star) - \mathcal{Q}(\widehat{L}, \widehat{S}) \leq \mathcal{Q}(\widehat{\Delta}^L_A, \widehat{\Delta}^L_{\mathcal{I}}) - \mathcal{Q}(\widehat{\Delta}^L_B, \widehat{\Delta}^L_{\mathcal{I}^c}).
    \end{equation*}
\end{lemma}

\begin{proof}[Proof of Lemma \ref{lemma:2}]
    Based on the definition of the subspaces, we immediately get that $L_B^\star = 0$ and $S_{\mathcal{I}^c}^\star = 0$. Then, we get:
    \begin{align*}
        \mathcal{Q}(\widehat{L}, \widehat{S})
        &= \mathcal{Q}(L^\star + \widehat{\Delta}^L, S^\star + \widehat{\Delta}^S) = \|L^\star_A + L^\star_B + \widehat{\Delta}^L_A + \widehat{\Delta}^L_B\|_* + \frac{\lambda}{\mu}\|S^\star_{\mathcal{I}} + S^\star_{\mathcal{I}^c} + \widehat{\Delta}^S_{\mathcal{I}} + \widehat{\Delta}^S_{\mathcal{I}^c}\|_1 \\
        &\geq \|L^\star_A + \widehat{\Delta}^L_B\|_* - \|\widehat{\Delta}^L_A\|_* + \frac{\lambda}{\mu}\left( \|S^\star_{\mathcal{I}} + \widehat{\Delta}^S_{\mathcal{I}}\|_1 - \|\widehat{\Delta}^S_{\mathcal{I}^c}\|_1 \right) \\
        &\geq \|L_A^\star\|_* + \|\widehat{\Delta}^L_B\|_* - \|\hat{\Delta}^L_A\|_* + \frac{\lambda}{\mu}\left( \|S^\star_{\mathcal{I}}\|_1 + \|\widehat{\Delta}^S_{\mathcal{I}}\|_1 - \|\widehat{\Delta}^S_{\mathcal{I}^c}\|_1 \right).
    \end{align*}
    Therefore, it follows that,
    \begin{align*}
        &\mathcal{Q}(L^\star, S^\star) - \mathcal{Q}(\widehat{L}, \widehat{S})
        = \left( \|L_A^\star\|_* + \frac{\lambda}{\mu}\|S_{\mathcal{I}}^\star\|_1 \right) - \mathcal{Q}(\widehat{L}, \widehat{S}) \\
        &\leq \|\widehat{\Delta}^L_B\|_* + \frac{\lambda}{\mu}\|\widehat{\Delta}^S_{\mathcal{I}^c}\|_1 - \left(\|\widehat{\Delta}^L_A\|_* + \frac{\lambda}{\mu}\|\widehat{\Delta}^S_{\mathcal{I}}\|_1 \right) = \mathcal{Q}(\widehat{\Delta}^L_B, \widehat{\Delta}^S_{\mathcal{I}^c}) - \mathcal{Q}(\widehat{\Delta}^L_A, \widehat{\Delta}^S_{\mathcal{I}}).
    \end{align*}
\end{proof}

Recall that the exhaustive search algorithm requires examining every time point in the search domain $\mathcal{T}$. It can then be seen that for the fixed true change point $\tau^\star$ solving optimization problems (4) in the main paper on $[1,\tau)$ and $[\tau,T)$ includes a portion of time where the underlying model is \textit{misspecified}. For example, assuming that $\tau > \tau^\star$, then solving (4) in the main paper cannot reach the optimal estimation error rate. Therefore, we require the following lemma to select the tuning parameters for intervals involving misspecified models.
\begin{lemma}\label{lemma:3}
Under the condition of Theorem 1 with $\tau > \tau^\star$, consider the interval $[1,\tau)$ where the model is misspecified and further select tuning parameters
\begin{equation*}
    \lambda_{1,\tau} = 4c\sqrt{\frac{\log p + \log(\tau-1)}{\tau-1}},\quad \mu_{1,\tau} = 4c\sqrt{\frac{p + \log(\tau-1)}{\tau-1}}.
\end{equation*}
Suppose that the search domain $\mathcal{T}$ satisfies Assumption H3; then, the following hold: \\
(1) for $T \succsim \log p$, with probability at least $1-c_1p^{-1}$:
\begin{equation}\label{eq:16}
    \left\| \frac{1}{\tau-1}\sum_{t=1}^{\tau-1}X_{t-1}(X_t - (L_1^\star+S_1^\star)X_{t-1})^\prime \right\|_\infty \leq \frac{\lambda_{1,\tau}}{2} + c_0\frac{(\tau-\tau^\star)_+}{\tau-1}(M_S\vee \alpha_L)(d_{\max}^\star+\sqrt{r_{\max}^\star}),
\end{equation}
(2) for $T \succsim p$, with probability at least $1-c_1^\prime\exp(-c_2^\prime p)T^{1-c_3^\prime}$:
\begin{equation}\label{eq:17}
    \left\| \frac{1}{\tau-1}\sum_{t=1}^{\tau-1}X_{t-1}(X_t - (L_1^\star+S_1^\star)X_{t-1})^\prime \right\|_{\text{op}} \leq \frac{\mu_{1,\tau}}{2} + c_0\frac{(\tau-\tau^\star)_+}{\tau-1}(M_S\vee \alpha_L),
\end{equation}
where $c_0, c_1, c_1^\prime, c_2^\prime, c_3^\prime$ are some generic large enough positive constants. Symmetrically, we can obtain the following deviation bounds for the other side of interval $[\tau, T)$:\\
(1) with probability at least $1-c_1p^{-1}$:
\begin{equation}
    \left\| \frac{1}{T-\tau}\sum_{t=\tau}^{T-1}X_{t-1}(X_t - (L_2^\star+S_2^\star)X_{t-1})^\prime \right\|_\infty \leq \frac{\lambda_{2,\tau}}{2} + c_0\frac{(\tau^\star-\tau)_+}{T-\tau}(M_S\vee \alpha_L)(d_{\max}^\star+\sqrt{r_{\max}^\star}),\tag{3'}\label{eq:16'}
\end{equation}
(2) with probability at least $1-c_1^\prime\exp(-c_2^\prime p)T^{1-c_3^\prime}$:
\begin{equation}
    \left\| \frac{1}{T-\tau}\sum_{t=\tau}^{T-1}X_{t-1}(X_t - (L_2^\star+S_2^\star)X_{t-1})^\prime \right\|_{\text{op}} \leq \frac{\mu_{2,\tau}}{2} + c_0\frac{(\tau^\star-\tau)_+}{T-\tau}(M_S\vee \alpha_L),\tag{4'}\label{eq:17'}
\end{equation}
where $c_0, c_1, c_1^\prime, c_2^\prime, c_3^\prime$ are some large enough positive constants, and the tuning parameters $\lambda_{2,\tau}$ and $\mu_{2,\tau}$ are given by:
\begin{equation*}
    \lambda_{2,\tau} = 4c\sqrt{\frac{\log p + \log(T-\tau)}{T-\tau}},\quad \mu_{2,\tau} = 4c\sqrt{\frac{p + \log(T-\tau)}{T-\tau}}.
\end{equation*}
\end{lemma}

\begin{proof}[Proof of Lemma \ref{lemma:3}]
    First, we present the details for \eqref{eq:16}. Fix $t\in \mathcal{T}$, and define $Y_t \overset{\text{def}}{=} X_t - (L_1^\star+S_1^\star)X_{t-1}$ for $t=2,\dots, \tau$. Therefore, we obtain that $\mathbb{E}(Y_t) = 0$, while $\text{cov}(X_t, Y_t) \neq 0$. Then, by setting $\lambda_{1,\tau} = \frac{A}{\sqrt{\tau-1}}$ and $\mu_{1,\tau} = \frac{A^\prime}{\sqrt{\tau-1}}$, with $A\overset{\text{def}}{=} 4c\sqrt{\log p + \log(\tau-1)}$ and $A^\prime\overset{\text{def}}{=} 4c\sqrt{p + \log(\tau-1)}$, and for some large enough constant $c>0$, we obtain using a union bound the following:
    \begin{align*}
        &\mathbb{P}\left( \max_{\tau \in \mathcal{T}} \lambda_{1,\tau}^{-1}\left\| \frac{1}{\tau-1}\sum_{t=1}^{\tau-1}X_{t-1}Y_t^\prime \right\|_\infty - c_0\lambda_{1,\tau}^{-1}\frac{(\tau-\tau^\star)_+}{\tau-1}(M_S\vee \alpha_L)(d_{\max}^\star+\sqrt{r_{\max}^\star}) > \frac{1}{2} \right) \nonumber \\
        \leq &\sum_{\tau \in \mathcal{T}}\mathbb{P}\left(\left\|\frac{1}{\tau-1}\sum_{t=1}^{\tau-1}X_{t-1}\epsilon_t^\prime \right\|_\infty > \frac{A}{2\sqrt{\tau-1}}\right) \overset{\text{(i)}}{\leq} 6\sum_{\tau\in \mathcal{T}}\exp\left(-\frac{c_0^\prime A^2}{4}\right) \leq \frac{6c_0^{\prime\prime}}{p} \to 0,
    \end{align*}
    where $c_0^\prime$ and $c^{\prime\prime}$ are some large constants, and inequality (i) holds based on the results of Proposition 3 in \cite{basu2019low}. 
    
    Next, to see \eqref{eq:17}, we use the same defined random process $Y_t = X_t - (L_1^\star+S_1^\star)X_{t-1}$ and the notations $A$ and $A^\prime$ above. Then, we obtain:
    \begin{equation*}
        \begin{aligned}
            &\mathbb{P}\left(\max_{\tau \in \mathcal{T}}\mu_{1,\tau}^{-1}\left\| \frac{1}{\tau-1}\sum_{t=1}^{\tau-1}X_{t-1}Y_t^\prime \right\|_{\text{op}} - \mu_{1,\tau}^{-1}c_0\frac{(\tau-\tau^\star)_+}{\tau-1}(M_S\vee \alpha_L) > \frac{1}{2}\right) \\
            \leq &\sum_{\tau \in \mathcal{T}}\mathbb{P}\left(\left\| \frac{1}{\tau-1}\sum_{t=1}^{\tau-1}X_{t-1}\epsilon_t^\prime \right\|_{\text{op}} > \frac{A^\prime}{2\sqrt{\tau-1}}\right) \overset{\text{(i)}}{\leq} 6\sum_{\tau \in \mathcal{T}}\exp\left(-\frac{c_1A^{\prime^2}}{4}\right) \leq \frac{6}{e^{c_1p}T^{c_1-1}} \to 0,
        \end{aligned}
    \end{equation*}
    for some large enough constant $c_1>0$. (i) is a direct application of the result of Proposition 3 in \cite{basu2019low} to this inequality with the choice of $\eta = \frac{A^\prime}{2\sqrt{\tau-1}}$. Note that $T \succsim p$ ensures that $\eta \leq 4c^\prime\frac{\log (\tau-1)}{\tau-1} < 1$, and then we can derive the anticipated results in \eqref{eq:16} and \eqref{eq:17}. By using a similar procedure, \eqref{eq:16'} and \eqref{eq:17'} also follow.
\end{proof}
{
Next, we provide a proof for the \textit{uniqueness} of the low rank and sparse decomposition. The main idea follows the \textit{rank-sparsity incoherence} condition introduced in \cite{chandrasekaran2011rank} followed by certain refinements in \cite{hsu2011robust} to characterize a decomposition of a matrix including a low rank component $L$ and a sparse component $S$. Before proving the following lemma, we require the following essential quantities:
\begin{itemize}
    \item[1.] \textit{Maximum number of non-zero entries in any row or column of $S$}:
    \begin{equation*}
        \alpha(\rho) \overset{\text{def}}{=} \max\left\{ \rho\|\text{sign}(S)\|_{1\to1}, \rho^{-1}\|\text{sign}(S)\|_{\infty \to \infty} \right\},
    \end{equation*}
    \item[2.] \textit{Sparseness of the singular vectors of $L$}: let $L=UDV$, $U$ and $V$ are matrices of left and right orthonormal singular vectors corresponding to the non-zero singular values of $L$, and the rank of $L$ is $r$. Define
    \begin{equation*}
        \beta(\rho) \overset{\text{def}}{=} \rho^{-1}\|UU^\prime\|_{\infty} + \rho\|VV^\prime\|_{\infty} + \|U\|_{2\to\infty}\|V\|_{2\to\infty},
    \end{equation*}
    where 
    \begin{equation*}
        \text{sign}(M)_{i,j} = 
        \begin{dcases}
            -1, \quad &\text{if }M_{i,j} < 0 \\
            0, \quad &\text{if }M_{i,j} = 0 \\
            +1, \quad &\text{if }M_{i,j} > 0
        \end{dcases},
    \end{equation*}
    and further define the induced norm $\|M\|_{p\to q} \overset{\text{def}}{=} \max\left\{ \|Mv\|_q: v \in \mathbb{R}^n, \|v\|_p \leq 1 \right\}$.
\end{itemize}
Additionally, we define two subspaces:
\begin{equation*}
    \Omega = \Omega(S) \overset{\text{def}}{=} \left\{ X \in \mathbb{R}^{p\times p}: \text{supp}(X) \subset \text{supp}(S) \right\},
\end{equation*}
be the space of matrices whose supports are subsets of the support of $S$, and let
\begin{equation*}
    T = T(L) \overset{\text{def}}{=} \left\{ X_1 + X_2 \in \mathbb{R}^{p\times p}: \text{range}(X_1) \subset \text{range}(L),\  \text{range}(X_2^\prime) \subset \text{range}(L^\prime) \right\}
\end{equation*}}
{
\begin{lemma}
\label{lemma:4}
Suppose Assumption~H2 in the case of a single change point or Assumption~H2' in the case of multiple change points is satisfied. Then, the low rank plus sparse decomposition of all transition matrices $A_j^*$'s for $j=1, \ldots, m_0+1$ are unique and further the restricted space condition proposed in \cite{agarwal2012noisy} is satisfied.
\end{lemma}}

\begin{proof}[Proof of Lemma~\ref{lemma:4}]
    {
    Without loss of generality, we consider the multiple change points case (i.e., we investigate Assumption H2'). Assuming that at the $j$-th stationary segment, $A_j^\star = S_j^\star + L_j^\star$, and then an application of the singular value decomposition (SVD) on $L_j^\star$ yields: $L_j^\star = U_jD_jV_j^\prime$, where $D_j = \text{diag}(\sigma_1^j, \dots, \sigma_{r_j}^j, 0, \dots, 0)$, $\sigma_i^j$ is the $i$-th singular value for $L_j^\star$, for $j=1,2,\cdots, m_0+1$. Next, we consider the Assumption H2'-(1)-(3).}
    
    {
    First, based on Assumptions H2'-(1)-(2), we get that $\|L_j^\star\|_\infty \leq \frac{\alpha_L}{p}$ for $j=1,\cdots, m_0+1$, which coincides with the constrained space condition proposed in \cite{agarwal2012noisy} and \cite{basu2019low}.}
    
    {
    Then, according to the definition of functions $\alpha(\rho)$ and $\beta(\rho)$, we derive that
    \begin{equation}
        \label{eq:18}
        \alpha(\rho) = \max\left\{\rho, \frac{d_{\max}^\star}{\rho}\right\},\ \beta(\rho) = \frac{\alpha_L}{p} (1 + \rho + \frac{1}{\rho} ).
    \end{equation}
    Thus, by using \eqref{eq:18} together with Assumption H2'-(3), we get for $\rho=1$:
    \begin{equation*}
        \alpha(1)\beta(1) = 3d_{\max}^\star\frac{\alpha_L}{p} = \mathcal{O}\left( d_{\max}^\star\sqrt{\frac{\log(pT)}{T}} \right).
    \end{equation*}
    Hence, with the newly proposed Assumption H2'-(3), we obtain that $\alpha(1)\beta(1) < 1$, which satisfies the sufficient condition of uniqueness of decomposition in Theorem 1 in \cite{hsu2011robust}.}
\end{proof}

\begin{lemma}
\label{lemma:5}
Suppose that the Assumptions of Theorem 1 hold, and use the weighted regularizer $\mathcal{Q}$. Further, tor a fixed $\tau \in \mathcal{T}$, define $\widehat{\Delta}^L_{1,\tau} = \widehat{L}_{1,\tau} - L_1^\star$, $\widehat{\Delta}^S_{1,\tau} = \widehat{S}_{1,\tau} - S_1^\star$, $\widehat{\Delta}^L_{2,\tau} = \widehat{L}_{2,\tau} - L_2^\star$, and $\widehat{\Delta}^S_{2,\tau} = \widehat{S}_{2,\tau} - S_2^\star$ for two intervals $[1,\tau)$ and $[\tau,T)$, and the misspecified error terms $\widetilde{\Delta}^L_{1/2,\tau} = \widehat{L}_1 - L_2^\star$, $\widetilde{\Delta}^S_{1/2,\tau} = \widehat{S}_1 - S_2^\star$ for the interval $[\tau^\star, \tau)$, respectively. Then, for the tuning parameters $(\lambda_{1,\tau}, \mu_{1,\tau})$ and $(\lambda_{2,\tau}, \mu_{2,\tau})$ proposed in (10) of Theorem 1, we obtain that:
\begin{equation*}
    \mathcal{Q}\left( \widehat{\Delta}^L_{1,\tau}, \widehat{\Delta}^S_{1,\tau} \right) \leq 4\mathcal{Q}\left( \widehat{\Delta}^L_{1,\tau}\vert_A, \widehat{\Delta}^S_{1,\tau}\vert_{\mathcal{I}} \right), \quad \mathcal{Q}\left( \widehat{\Delta}^L_{2,\tau}, \widehat{\Delta}^S_{2,\tau} \right) \leq 4\mathcal{Q}\left( \widehat{\Delta}^L_{2,\tau}\vert_A, \widehat{\Delta}^S_{2,\tau}\vert_{\mathcal{I}} \right),
\end{equation*}
and
\begin{equation*}
    \mathcal{Q}\left( \widetilde{\Delta}^L_{1/2,\tau}, \widetilde{\Delta}^S_{1/2,\tau} \right) \leq 4\mathcal{Q}\left( \widetilde{\Delta}^L_{1/2,\tau}\vert_A, \widetilde{\Delta}^S_{1/2,\tau}\vert_{\mathcal{I}} \right).
\end{equation*}
\end{lemma}

\begin{proof}[Proof of Lemma \ref{lemma:5}]
    Assuming that $\tau > \tau^\star$, we investigate the behavior of the misspecified model in the interval $[1,\tau)$ and the non-misspecified model in the interval $[\tau, T)$, separately. Therefore, for the interval $[\tau, T)$ and $(\widehat{L}_{2,\tau}, \widehat{S}_{2,\tau})$, according to the defined objective functions $\ell(L_1, S_1; \mathbf{X}^{[1:\tau)})$ and $\ell(L_2, S_2; \mathbf{X}^{[\tau:T)})$ in the main text, we obtain that for minimizing $\ell(L_2, S_2; \mathbf{X}^{[\tau:T)})$, we derive:
    \begin{equation*}
        \begin{aligned}
        &\frac{1}{T - \tau}\sum_{t=\tau}^{T-1}\|X_t - (\widehat{L}_2 + \widehat{S}_2)X_{t-1}\|_2^2 + \lambda_{2,\tau}\|\widehat{S}_2\|_1 + \mu_{2,\tau}\|\widehat{L}_2\|_* \\
        \leq &\frac{1}{T - \tau}\sum_{t=\tau}^{T-1}\|X_t - (L_2^\star + S_2^\star)X_{t-1}\|_2^2 + \lambda_{2,\tau}\|S_{2}^\star\|_1 + \mu_{2,\tau}\|L_{2}^\star\|_*.
        \end{aligned}
    \end{equation*}
    After some algebraic rearrangements, and due to the nature of the decomposition spaces $(A, B)$ for the low-rank components and the corresponding decomposable support sets $(\mathcal{I}, \mathcal{I}^c)$ for the sparse components, we get:
    \begin{align*}
        0
        &\leq \frac{1}{T-\tau}\sum_{t=\tau}^{T-1}\|X_{t-1}(\widehat{\Delta}^L_{2,\tau} + \widehat{\Delta}^S_{2,\tau})\|_2^2 \nonumber \\
        &\leq \frac{2}{T-\tau}\sum_{t=\tau}^{T-1}X_{t-1}^\prime(\widehat{\Delta}^L_{2,\tau} + \widehat{\Delta}^S_{2,\tau})^\prime \epsilon_t + \lambda_{2,\tau}\left( \|S_2^\star\|_1 - \|\widehat{S}_{2,\tau}\|_1 \right) + \mu_{2,\tau}\left( \|L_2^\star\|_* - \|\widehat{L}_{2,\tau}\|_* \right) \nonumber \\
        &\overset{\text{(i)}}{\leq} 2c_0\sqrt{\frac{\log p + \log(T-\tau)}{T-\tau}}\|\widehat{\Delta}^S_{2,\tau}\|_1 + 2 c_0\sqrt{\frac{p + \log(T-\tau)}{T-\tau}}\|\widehat{\Delta}^L_{2,\tau}\|_* \\
        &+ \lambda_{2,\tau}\left( \|S_2^\star\|_1 - \|\widehat{S}_{2,\tau}\|_1 \right) + \mu_{2,\tau}\left( \|L_2^\star\|_* - \|\widehat{L}_{2,\tau}\|_* \right) \nonumber \\
        &\leq \frac{\lambda_{2,\tau}}{2}\|\widehat{\Delta}^S_{2,\tau}\|_1 + \frac{\mu_{2,\tau}}{2}\|\widehat{\Delta}^L_{2,\tau}\|_* + \lambda_{2,\tau}\left( \|S_2^\star\|_1 - \|\widehat{S}_{2,\tau}\|_1 \right) + \mu_{2,\tau}\left( \|L_2^\star\|_* - \|\widehat{L}_{2,\tau}\|_* \right) \nonumber \\
        &\leq \frac{3}{2}\mu_{2,\tau}\mathcal{Q}(\widehat{\Delta}^L_{2,\tau}\vert_A, \widehat{\Delta}^S_{2,\tau}\vert_{\mathcal{I}}) - \frac{1}{2}\mu_{2,\tau}\mathcal{Q}(\widehat{\Delta}^L_{2,\tau}\vert_B, \widehat{\Delta}^S_{2,\tau}\vert_{\mathcal{I}^c}), \nonumber 
    \end{align*}
    where inequality (i) holds because of the deviation bound derived in Lemma \ref{lemma:3}. 
    Therefore, we can further derive that:
    \begin{equation*}
        \mathcal{Q}(\widehat{\Delta}^L_{2,\tau}, \widehat{\Delta}^S_{2,\tau}) \leq 4\mathcal{Q}(\widehat{\Delta}^L_{2,\tau}\vert_A, \widehat{\Delta}^S_{2,\tau}\vert_{\mathcal{I}}).
    \end{equation*}
    
    On the other hand, for the misspecified model in the interval $[1,\tau)$, by minimizing the objective function $\ell(L_1, S_1; \mathbf{X}^{[1:\tau)})$ to the intervals $[1,\tau^\star)$ and $[\tau^\star, \tau)$ separately, we obtain:
    \begin{equation*}
        \begin{aligned}
        &\frac{1}{\tau^\star-1}\sum_{t=1}^{\tau^\star-1}\|X_t - (\widehat{L}_1 + \widehat{S}_1)X_{t-1}\|_2^2 + \lambda_{1,\tau}\|\widehat{S}_1\|_1 + \mu_{1,\tau}\|\widehat{L}_1\|_* \\
        \leq &\frac{1}{\tau^\star-1}\sum_{t=1}^{\tau^\star-1}\|X_t - (L_1^\star + S_1^\star)X_{t-1}\|_2^2 + \lambda_{1,\tau}\|S_1^\star\|_1 + \mu_{1,\tau}\|L_1^\star\|_*,
        \end{aligned}
    \end{equation*}
    and
    \begin{equation*}
        \begin{aligned}
        &\frac{1}{\tau-\tau^\star}\sum_{t=\tau^\star}^{\tau-1}\|X_t - (\widehat{L}_1 + \widehat{S}_1)X_{t-1}\|_2^2 + \lambda_{1,\tau}\|\widehat{S}_1\|_1 + \mu_{1,\tau}\|\widehat{L}_1\|_* \\
        \leq &\frac{1}{\tau-\tau^\star}\sum_{t=\tau^\star}^{\tau-1}\|X_t - (L_2^\star + S_2^\star)X_{t-1}\|_2^2 + \lambda_{1,\tau}\|S_2^\star\|_1 + \mu_{1,\tau}\|L_2^\star\|_*.
        \end{aligned}
    \end{equation*}
    Similarly for the first inequality, after some algebraic rearrangements, we derive that
    \begin{align*}
        0
        &\leq \frac{1}{\tau^\star-1}\sum_{t=1}^{\tau^\star-1}\|X_{t-1}(\widehat{\Delta}^L_{1,\tau} + \widehat{\Delta}^S_{1,\tau})\|_2^2 \nonumber \\
        &\leq \frac{2}{\tau^\star-1}\sum_{t=1}^{\tau^\star-1}X_{t-1}^\prime(\widehat{\Delta}^L_{1,\tau} + \widehat{\Delta}^S_{1,\tau})^\prime \epsilon_t + \lambda_{1,\tau}\left( \|S_1^\star\|_1 - \|\widehat{S}_{1,\tau}\|_1 \right) + \mu_{1,\tau}\left( \|L_1^\star\|_* - \|\widehat{L}_{1,\tau}\|_* \right) \nonumber \\
        &\leq 2c_0\sqrt{\frac{\log p + \log(\tau-1)}{\tau-1}}\|\widehat{\Delta}^S_{1,\tau}\|_1 + 2 c_0\sqrt{\frac{p + \log(\tau-1)}{\tau-1}}\|\widehat{\Delta}^L_{1,\tau}\|_* \\
        &+ \lambda_{1,\tau}\left( \|S_1^\star\|_1 - \|\widehat{S}_{1,\tau}\|_1 \right) + \mu_{1,\tau}\left( \|L_1^\star\|_* - \|\widehat{L}_{1,\tau}\|_* \right) \nonumber \\
        &\leq \frac{\lambda_{1,\tau}}{2}\|\widehat{\Delta}^S_{1,\tau}\|_1 + \frac{\mu_{1,\tau}}{2}\|\widehat{\Delta}^L_{1,\tau}\|_* + \lambda_{1,\tau}\left( \|S_1^\star\|_1 - \|\widehat{S}_{1,\tau}\|_1 \right) + \mu_{1,\tau}\left( \|L_1^\star\|_* - \|\widehat{L}_{1,\tau}\|_* \right) \nonumber \\
        &\leq \frac{3}{2}\mu_{1,\tau}\mathcal{Q}(\widehat{\Delta}^L_{1,\tau}\vert_A, \widehat{\Delta}^S_{1,\tau}\vert_{\mathcal{I}}) - \frac{1}{2}\mu_{1,\tau}\mathcal{Q}(\widehat{\Delta}^L_{1,\tau}\vert_B, \widehat{\Delta}^S_{1,\tau}\vert_{\mathcal{I}^c}), \nonumber 
    \end{align*}
    since the second inequality can be derived by the same procedure. Therefore, we conclude that
    \begin{equation*}
        \mathcal{Q}(\widehat{\Delta}^L_{1,\tau}, \widehat{\Delta}^S_{1,\tau}) \leq 4\mathcal{Q}(\widehat{\Delta}^L_{1,\tau}\vert_A, \widehat{\Delta}^S_{1,\tau}\vert_{\mathcal{I}})\ \text{and}\ \mathcal{Q}(\widetilde{\Delta}^L_{1/2,\tau}, \widetilde{\Delta}^S_{1/2,\tau}) \leq 4\mathcal{Q}(\widetilde{\Delta}^L_{1/2,\tau}\vert_A, \widetilde{\Delta}^S_{1/2,\tau}\vert_{\mathcal{I}}).
    \end{equation*}
\end{proof}
{
\begin{lemma}
\label{lemma:6}
Under Assumptions H1'-H5', for a set of estimated change points $(s_1, s_2, \cdots, s_m)$ with $m < m_0$, there exist universal positive constants $c_1, c_2 > 0$ such that:
\begin{equation*}
    \mathbb{P}\left( \min_{(s_1,\dots, s_m)}\mathcal{L}_n(s_1,\dots,s_m; \bm{\lambda}, \bm{\mu}) > \sum_{t=1}^T\|\epsilon_t\|_2^2 + c_1\widetilde{v}\Delta_T - c_2 mT\xi_T\left(d_{\max}^{\star^2} + r_{\max}^{\star^{\frac{3}{2}}}\right) \right) \to 1,
\end{equation*}
where $\widetilde{v} \overset{\text{def}}{=} \min_{1\leq j \leq m_0}\{v_{j,S}^2 + v_{j,L}^2\}$.
\end{lemma}}

\begin{proof}[Proof of Lemma \ref{lemma:6}]
    Suppose we obtain a set of candidate change points $(s_1, s_2, \cdots, s_m)$ obtained by the rolling-window strategy. Since $m < m_0$, there exists a true change point $t_j^\star$ satisfying $|s_i - t_j^\star| > {\Delta_T}/{4}$. In order to find a lower bound for $\mathcal{L}(s_1, \cdots, s_m; \bm{\lambda}, \bm{\mu})$, based on the vanishing sequence $\{\xi_T\}$ specified in Assumption H3', there are three different cases to consider: (a) $|s_i - s_{i-1}| \leq T\xi_T$, which implies that there is a negligibly small interval between two consecutive estimated change points $s_{i-1}$ and $s_i$; (b) there exist two true change points $\tau_j^\star$ and $\tau_{j+1}^\star$ such that $|s_{i-1} - \tau_j^\star| \leq T\xi_T$ and $|s_i - \tau_{j+1}^\star| \leq T\xi_T$; and (c) otherwise. 
    
    Next, we introduce some additional notation used in the sequel. Let $\widehat{\Delta}^L$ and $\widehat{\Delta}^S$ denote the difference between the true expression and its estimate; i.e., $\widehat{\Delta}^L = L^\star_{j+1} - \widehat{L}_i$ and $\widehat{\Delta}^S = S^\star_{j+1} - \widehat{S}_i$, respectively. We denote by $\widetilde{\Delta}^L$ and $\widetilde{\Delta}^S$ the difference between the true expression and its estimate in the misspecified time segments; i.e., $\widetilde{\Delta}^L = L^\star_j - \widehat{L}_i$ and $\widetilde{\Delta}^S = S^\star_j - \widehat{S}_i$. 
    
    For case (a), without loss of generality, we assume that $\tau_j^\star < s_{i-1} < s_i < \tau^\star_{j+1}$ to obtain:
    \begin{equation*}
        \begin{aligned}
            \sum_{t=s_{i-1}}^{s_i}\|X_t - (\widehat{L}_i + \widehat{S}_i)X_{t-1}\|_2^2 
            &= \sum_{t=s_{i-1}}^{s_i}\|\epsilon_t\|_2^2 + \sum_{t=s_{i-1}}^{s_i}\|X_{t-1}(\widehat{\Delta}^L + \widehat{\Delta}^S)\|_2^2 + 2\sum_{t=s_{i-1}}^{s_i}X_{t-1}^\prime(\widehat{\Delta}^L + \widehat{\Delta}^S)\epsilon_t \\
            &\geq \sum_{t=s_{i-1}}^{s_i}\|\epsilon_t\|_2^2 - 2\left| \sum_{t=s_{i-1}}^{s_i}\langle X_{t-1}^\prime \epsilon_t, \widehat{\Delta}^L \rangle \right| - 2\left| \sum_{t=s_{i-1}}^{s_i}\langle X_{t-1}^\prime \epsilon_t, \widehat{\Delta}^S \rangle \right| \\
            &\geq \sum_{t=s_{i-1}}^{s_i}\|\epsilon_t\|_2^2 - c\sqrt{T\xi_T p}\|\widehat{\Delta}^L\|_* - c^\prime \sqrt{T\xi_T \log p}\|\widehat{\Delta}^S\|_1.
        \end{aligned}
    \end{equation*}
    Based on Assumption H6 on the selection of the tuning parameters, we conclude that:
    \begin{align}
        \label{eq:19}
        &\sum_{t=s_{i-1}}^{s_i}\|X_t - (\widehat{L}_i + \widehat{S}_i)X_{t-1}\|_2^2 + \lambda_i\|\widehat{S}_i\|_1 + \mu_i\|\widehat{L}_i\|_* \nonumber \\
        \geq &\sum_{t=s_{i-1}}^{s_i}\|\epsilon_t\|_2^2 - c\sqrt{T\xi_T p}\|L^\star_{j+1}\|_* - c^\prime \sqrt{T\xi_T \log p}\|S_{j+1}^\star\|_1.
    \end{align}
    
    For case (b), we assume that $s_{i-1} < \tau_j^\star$, $s_i < \tau_{j+1}^\star$, $|s_{i-1} - \tau_j^\star| \leq T\xi_T$, and $|s_i - \tau_{j+1}^\star| \leq T\xi_T$. Since the estimates $\widehat{L}_i$ and $\widehat{S}_i$ are the minimizers to the objective function as (9) in the main paper, then we obtain:
    \begin{equation}
        \label{eq:20}
        \begin{aligned}
        &\frac{1}{s_i - s_{i-1}}\sum_{t=s_{i-1}}^{s_i-1}\|X_t - (\widehat{L}_i + \widehat{S}_i)X_{t-1}\|_2^2 + \lambda_i\|\widehat{S}_i\|_1 + \mu_i\|\widehat{L}_i\|_* \\
        &\leq \frac{1}{s_i - s_{i-1}}\sum_{t=s_{i-1}}^{s_i-1}\|X_t - (L_{j+1}^\star + S_{j+1}^\star)X_{t-1}\|_2^2 + \lambda_i\|S_{j+1}^\star\|_1 + \mu_i\|L_{j+1}^\star\|_*.
        \end{aligned}
    \end{equation}
     Some algebraic rearrangements and based on Assumption H6 on the selection of the tuning parameters, we obtain that:
    \begin{align}\label{eq:21}
        0 &\leq \frac{1}{s_i - s_{i-1}}\sum_{t=s_{i-1}}^{s_i}\|X_{t-1}(\widehat{\Delta}^L + \widehat{\Delta}^S)\|_2^2 \nonumber \\
         &\leq \frac{2}{s_i - s_{i-1}}\sum_{t=s_{i-1}}^{s_i}\langle X_{t-1}^\prime \epsilon_t, \widehat{\Delta}^L + \widehat{\Delta}^S \rangle + \frac{2}{s_i - s_{i-1}}\sum_{t=s_{i-1}}^{\tau_j^\star-1}\langle X_{t-1}^\prime\epsilon_t, L_j^\star-L_{j+1}^\star+S_j^\star-S_{j+1}^\star \rangle\nonumber \\
         &+ \lambda_i(\|S_{j+1}^\star\|_1 - \|\widehat{S}_i\|_1) + \mu_i(\|L_{j+1}^\star\|_* - \|\widehat{L}_i\|_*) \nonumber \\
         &\leq 2\left( c\sqrt{\frac{\log p}{s_i - s_{i-1}}} + M_Sd_{\max}^\star\frac{T\xi_T}{s_i - s_{i-1}} \right)\|\widehat{\Delta}^S\|_1 + \lambda_i(\|S_{j+1}^\star\|_1 - \|\widehat{S}_i\|_1) \nonumber \\
         &+ 2\left(c\sqrt{\frac{p}{s_i - s_{i-1}}} + \alpha_L\sqrt{r_{\max}^\star} \frac{T\xi_T}{s_i-s_{i-1}}\right)\|\widehat{\Delta}^L\|_* + \mu_i(\|L_{j+1}^\star\|_* - \|\widehat{L}_i\|_*) \\
         &\leq \frac{\lambda_i}{2}\|\widehat{\Delta}^S\|_1 + \lambda_i(\|S_{j+1}^\star\|_1 - \|\widehat{S}_i\|_1) + \frac{\mu_i}{2}\|\widehat{\Delta}^L\|_* + \mu_i(\|L_{j+1}^\star\|_* - \|\widehat{L}_i\|_*) \nonumber \\
         &= \frac{3}{2}\lambda_i\|\widehat{\Delta}^S\|_{1,\mathcal{I}} - \frac{1}{2}\lambda_i\|\widehat{\Delta}^S\|_{1,\mathcal{I}^c} + \frac{3}{2}\mu_i\|\widehat{\Delta}^L\|_{*,A} - \frac{1}{2}\mu_i\|\widehat{\Delta}^L\|_{1,B} \nonumber \\
         &= \frac{3}{2}\mu_i\mathcal{Q}(\widehat{\Delta}^L_A, \widehat{\Delta}^S_\mathcal{I}) - \frac{1}{2}\mu_i\mathcal{Q}(\widehat{\Delta}^L_B, \widehat{\Delta}^S_{\mathcal{I}^c}).\nonumber 
    \end{align}
    The properties of the weighted regularizer $\mathcal{Q}$ have been discussed in Lemma \ref{lemma:3}. Therefore, in accordance to equation \eqref{eq:21}, we obtain that:
    \begin{equation*}
        \mathcal{Q}(\widehat{\Delta}^L, \widehat{\Delta}^S) \leq 4\mathcal{Q}(\widehat{\Delta}^L_A, \widehat{\Delta}^S_\mathcal{I}).
    \end{equation*}
    Moreover, an application of the Cauchy-Schwarz inequality leads to the following upper bound for the weighted regularizer $\mathcal{Q}$ with respect to the support sets $(A, \mathcal{I})$ defined before Lemma \ref{lemma:2}: 
    \begin{equation}\label{eq:22}
        \mu_i\mathcal{Q}(\widehat{\Delta}^L_A, \widehat{\Delta}^S_\mathcal{I}) \leq \mu_i\sqrt{r_{\max}^\star}\|\widehat{\Delta}^L\|_F + \lambda_i\sqrt{d_{\max}^\star}\|\widehat{\Delta}^S\|_F \leq \sqrt{\lambda_i^2 d_{\max}^\star + \mu_i^2 r_{\max}^\star}\sqrt{\|\widehat{\Delta}^L\|_F^2 + \|\widehat{\Delta}^S\|_F^2}.
    \end{equation}
    Next, examining the first inequality in \eqref{eq:21}, we see that there exists a positive constant $c^\prime>0$ such that:
    \begin{equation}\label{eq:23}
        \frac{1}{s_i - s_{i-1}}\sum_{t=s_{i-1}}^{s_i}\|X_{t-1}(\widehat{\Delta}^L + \widehat{\Delta}^S)\|_2^2 \geq c^\prime\|\widehat{\Delta}^L + \widehat{\Delta}^S\|_F^2 \geq \nu(\|\widehat{\Delta}^L\|_F^2 + \|\widehat{\Delta}^S\|_F^2) - \frac{1}{2}\mu_i\mathcal{Q}(\widehat{\Delta}^L, \widehat{\Delta}^S),
    \end{equation}
    where $\nu>0$ is the curvature constant appearing in the RSC condition and the last inequality holds due to Lemma \ref{lemma:3}.
    By substituting \eqref{eq:23} into \eqref{eq:21}, we obtain:
    \begin{equation}\label{eq:24}
        \begin{aligned}
        &\nu(\|\widehat{\Delta}^L\|_F^2 + \|\widehat{\Delta}^S\|_F^2) - \frac{1}{2}\mu_i\mathcal{Q}(\widehat{\Delta}^L, \widehat{\Delta}^S) \leq \frac{3}{2}\mu_i\mathcal{Q}(\widehat{\Delta}^L_A, \widehat{\Delta}^S_\mathcal{I}) - \frac{1}{2}\mu_i\mathcal{Q}(\widehat{\Delta}^L_B, \widehat{\Delta}^S_{\mathcal{I}^c}) \\
        \implies &\nu(\|\widehat{\Delta}^L\|_F^2 + \|\widehat{\Delta}^S\|_F^2) \leq 2\mu_i\mathcal{Q}(\widehat{\Delta}^L, \widehat{\Delta}^S) \leq 8\sqrt{\lambda_i^2 d_{\max}^\star + \mu_i^2 r_{\max}^\star}\sqrt{\|\widehat{\Delta}^L\|_F^2 + \|\widehat{\Delta}^S\|_F^2} \\
        \implies &\|\widehat{\Delta}^L\|_F^2 + \|\widehat{\Delta}^S\|_F^2 \leq \frac{64}{\nu^2}(\lambda_i^2 d_{\max}^\star + \mu_i^2 r_{\max}^\star).
        \end{aligned}
    \end{equation}
    Following analogous derivations to \eqref{eq:25}, one can similarly conclude that the error bound of the estimates can still be verified in the interval $[s_{i-1}, s_i)$. Note that there is a misspecified model in the interval $[s_{i-1}, \tau^\star_j)$, which is discussed separately.
    
    First, consider the interval $[\tau_j^\star, s_i)$, for which we have:
    \begin{align}
        \label{eq:25}
        &\sum_{t=t_j^\star}^{s_i-1}\|X_t - (\widehat{L}_i + \widehat{S}_i)X_{t-1}\|_2^2  \nonumber \\
        \geq &\sum_{t=t_j^\star}^{s_i-1}\|\epsilon_t\|_2^2 + c|s_i - \tau_j^\star|\|\widehat{\Delta}^L + \widehat{\Delta}^S\|_F^2 - c^\prime\sqrt{|s_i - t_j^\star|\log p}\|\widehat{\Delta}^S\|_1 - c^\prime\sqrt{|s_i - t_j^\star|p}\|\widehat{\Delta}^L\|_*  \nonumber \\
        \geq &\sum_{t=t_j^\star}^{s_i-1}\|\epsilon_t\|_2^2 + c|s_i - \tau_j^\star|\left( \|\widehat{\Delta}^L\|_F^2 + \|\widehat{\Delta}^S\|_F^2 - \frac{1}{2}\mu_i\mathcal{Q}(\widehat{\Delta}^L, \widehat{\Delta}^S) \right) \nonumber \\
        &- c^\prime\sqrt{|s_i - \tau_j^\star|\log p}\|\widehat{\Delta}^S\|_1 - c^\prime\sqrt{|s_i - t_j^\star|p}\|\widehat{\Delta}^L\|_*  \nonumber \\
        \overset{\text{(i)}}{\geq} &\sum_{t=t_j^\star}^{s_i-1}\|\epsilon_t\|_2^2 + c|s_i - \tau_j^\star|\left( \|\widehat{\Delta}^L\|_F^2 + \|\widehat{\Delta}^S\|_F^2 - (\frac{1}{2}+\frac{1}{c})\mu_i\mathcal{Q}(\widehat{\Delta}^L, \widehat{\Delta}^S) \right) \\
        \overset{\text{(ii)}}{\geq} &\sum_{t=t_j^\star}^{s_i-1}\|\epsilon_t\|_2^2 + c|s_i - \tau_j^\star|\left( \|\widehat{\Delta}^L\|_F^2 + \|\widehat{\Delta}^S\|_F^2 - (\frac{1}{2}+\frac{1}{c})\sqrt{\lambda_i^2 d_{j+1}^\star + \mu_i^2 r_{j+1}^\star}\sqrt{\|\widehat{\Delta}^L\|_F^2 + \|\widehat{\Delta}^S\|_F^2}\right)  \nonumber \\
        \geq &\sum_{t=t_j^\star}^{s_i-1}\|\epsilon_t\|_2^2 + c|s_i - \tau_j^\star|\sqrt{\|\widehat{\Delta}^L\|_F^2 + \|\widehat{\Delta}^S\|_F^2}\left( \sqrt{\|\widehat{\Delta}^L\|_F^2 + \|\widehat{\Delta}^S\|_F^2} - (\frac{1}{2}+\frac{1}{c})\sqrt{\lambda_i^2 d_{j+1}^\star + \mu_i^2 r_{j+1}^\star} \right)  \nonumber \\
        \overset{\text{(iii)}}{\geq} &\sum_{t=t_j^\star}^{s_i-1}\|\epsilon_t\|_2^2 - c^{\prime\prime}|s_i - \tau_j^\star| \left\{ \frac{d_{j+1}^\star \log p + r_{j+1}^\star p}{s_i - s_{i-1}} + (M_S^2 d_{\max}^{\star^2}d_{j+1}^\star + \alpha_L^2 r_{\max}^\star r_{j+1}^\star)\left(\frac{T\xi_T}{s_{i} - s_{i-1}}\right)^2 \right\}  \nonumber \\
        - &2c^{\prime\prime}|s_i - \tau_j^\star|\left( M_S d_{\max}^\star d_{j+1}^\star\sqrt{\frac{\log p}{s_i - s_{i-1}}} + \alpha_L \sqrt{r_{\max}^\star} r_{j+1}^\star\sqrt{\frac{p}{s_i - s_{i-1}}}\right)\frac{T\xi_T}{s_i - s_{i-1}}, \nonumber
    \end{align}
    where $c, c^\prime, c^{\prime\prime}$ are large enough positive constants which can be determined based on the tuning parameter rates. In \eqref{eq:25}, (i) holds because of the selection of tuning parameters; (ii) holds by using the result from \eqref{eq:22}; (iii) holds due to the verified error bound in \eqref{eq:24} and the selection of the tuning parameters. Based on the results for $|s_i-\tau_{j+1}^\star| \leq T\xi_T$ and $|s_{i-1}-\tau_j^\star|\leq T\xi_T$, then we get that $T\xi_T / (s_i - s_{i-1}) \to 0$ as $T \to +\infty$; the latter together with Assumption H3' imply
    \begin{equation*}
        \frac{d_{\max}^\star\log p + r_{\max}^\star p}{s_i-s_{i-1}} \geq \left(M_S^2d_{\max}^{\star^3} + \alpha_L^2r_{\max}^{\star^2}\right)\left(\frac{T\xi_T}{s_i - s_{i-1}}\right)^2.
    \end{equation*}
    It can be verified by some algebraic rearrangements that
    \begin{align*}
        &\frac{d_{\max}^\star\log p + r_{\max}^\star p}{s_i-s_{i-1}} \geq \left(M_S^2d_{\max}^{\star^3} + \alpha_L^2r_{\max}^{\star^2}\right)\left(\frac{T\xi_T}{s_i - s_{i-1}}\right)^2 \\
        \iff &\left(\frac{s_i - s_{i-1}}{T\xi_T}\right)\frac{d_{\max}^\star\log p + r_{\max}^\star p}{T\xi_T} \geq M_S^2d_{\max}^{\star^3} + \alpha_L^2r_{\max}^{\star^2},
    \end{align*}
    which can be directly derived from Assumption H3'. Similarly, we can prove the following fact:
    \begin{align*}
        &\frac{d_{\max}^\star\log p + r_{\max}^\star p}{s_i-s_{i-1}} \geq \left( M_S d_{\max}^\star d_{j+1}^\star\sqrt{\frac{\log p}{s_i - s_{i-1}}} + \alpha_L \sqrt{r_{\max}^\star} r_{j+1}^\star\sqrt{\frac{p}{s_i - s_{i-1}}}\right)\frac{T\xi_T}{s_i - s_{i-1}} \\
        \iff &\left(\frac{d_{\max}^\star\log p + r_{\max}^\star p}{T\xi_T}\right)^2 \geq \left( M_S d_{\max}^\star d_{j+1}^\star\sqrt{\frac{\log p}{s_i - s_{i-1}}} + \alpha_L \sqrt{r_{\max}^\star} r_{j+1}^\star\sqrt{\frac{p}{s_i - s_{i-1}}}\right)^2,
    \end{align*}
    the right-hand side being upper bounded by:
    \begin{equation*}
        \begin{aligned}
            &\left( M_S d_{\max}^\star d_{j+1}^\star\sqrt{\frac{\log p}{s_i - s_{i-1}}} + \alpha_L \sqrt{r_{\max}^\star} r_{j+1}^\star\sqrt{\frac{p}{s_i - s_{i-1}}}\right)^2 \\
            &\leq \left(M_S^2d_{\max}^{\star^3} + \alpha_L^2r_{\max}^{\star^2}\right)\left(\frac{d_{\max}^\star\log p + r_{\max}^\star p}{s_i - s_{i-1}}\right).
        \end{aligned}
    \end{equation*}
    Substituting the upper bound into the inequality above, the fact can be verified. Therefore, \eqref{eq:25} can be further lower bounded by:
    \begin{equation}\label{eq:26}
        \sum_{t=t_j^\star}^{s_i-1}\|X_t - (\widehat{L}_i + \widehat{S}_i)X_{t-1}\|_2^2 
        \geq \sum_{t=t_j^\star}^{s_i-1}\|\epsilon_t\|_2^2 - c^{\prime\prime}\left(d_{j+1}^\star\log p + r_{j+1}^\star p\right).
    \end{equation}
    
    Next, we consider the misspecified model in the interval $[s_{i-1}, \tau_j^\star)$, which satisfies the condition $|\tau_j^\star - s_{i-1}| \leq T\xi_T$, by using the notation $\widetilde{\Delta}^L$ and $\widetilde{\Delta}^S$ as previously defined. Then,
    \begin{equation}
        \label{eq:27}
        \begin{aligned}
            &\sum_{t=s_{i-1}}^{\tau_j^\star-1}\|X_t - (\widehat{L}_i + \widehat{S}_i)X_{t-1}\|_2^2 \geq \sum_{t=s_{i-1}}^{\tau_j^\star-1}\|\epsilon\|_2^2 - c^\prime\left( \sqrt{T\xi_T\log p}\|\widetilde{\Delta}^S\|_1 + \sqrt{T\xi_T p}\|\widetilde{\Delta}^L\|_* \right) \\
            \geq &\sum_{t=s_{i-1}}^{\tau_j^\star-1}\|\epsilon\|_2^2 - c^\prime\left\{ \sqrt{T\xi_T\log p}\left(\|\widehat{\Delta}^S\|_1 + \|S_{j+1}^\star - S_j^\star\|_1\right) + \sqrt{T\xi_T p}\left(\|\widehat{\Delta}^L\|_* + \|L_{j+1}^\star - L_j^\star\|_*\right) \right\} \\
            \geq &\sum_{t=s_{i-1}}^{\tau_j^\star-1}\|\epsilon_t\|_2^2 - c^\prime\sqrt{T\xi_T\log p}\|\widehat{\Delta}^S\|_1 - c^\prime\sqrt{T\xi_T p}\|\widehat{\Delta}^L\|_* - c_1^\prime d_{\max}^\star\sqrt{T\xi_T\log p} - c_2^\prime\sqrt{r_{\max}^\star}\sqrt{T\xi_T p} \\
            \overset{\text{(i)}}{\geq} &\sum_{t=s_{i-1}}^{\tau_j^\star - 1}\|\epsilon_t\|_2^2 - \sqrt{T\xi_T(s_i - s_{i-1})}\mu_i\mathcal{Q}(\widehat{\Delta}^L, \widehat{\Delta}^S) - c_1^\prime d_{\max}^\star\sqrt{T\xi_T\log p} - c_2^\prime\sqrt{r_{\max}^\star}\sqrt{T\xi_T p} \\
            \overset{\text{(ii)}}{\geq} &\sum_{t=s_{i-1}}^{\tau_j^\star-1}\|\epsilon_t\|_2^2 - c_0^\prime \sqrt{T\xi_T(s_i - s_{i-1})}\left(\lambda_i^2 d_{j+1}^\star + \mu_i^2 r_{j+1}^\star \right) - c_1^\prime d_{\max}^\star\sqrt{T\xi_T\log p} - c_2^\prime\sqrt{r_{\max}^\star}\sqrt{T\xi_T p} \\
            \overset{\text{(iii)}}{\geq} &\sum_{t=s_{i-1}}^{\tau_j^\star-1}\|\epsilon\|_2^2 - c_1^\prime d_{\max}^\star\sqrt{T\xi_T \log p} - c_2^\prime \sqrt{r_{\max}^\star} \sqrt{T\xi_T p} - 16c_0^\prime\left(d_{j+1}^\star \log p + r_{j+1}^\star p\right)\sqrt{\frac{T\xi_T}{s_i - s_{i-1}}} \\
            - &c_0^{\prime\prime}\left(\frac{(T\xi_T)^{\frac{3}{2}}}{\sqrt{s_i - s_{i-1}}}\right)\left((M_S^2 d_{\max}^{\star^3}+ \alpha_L^2 r_{\max}^{\star^2})\left(\frac{T\xi_T}{s_{i} - s_{i-1}}\right) + 2\left(M_Sd_{\max}^{\star^2}\sqrt{\frac{\log p}{s_i - s_{i-1}}} + \alpha_Lr_{\max}^{\star^{\frac{3}{2}}}\sqrt{\frac{p}{s_i-s_{i-1}}}\right)\right),
        \end{aligned}
    \end{equation}
    where $c^\prime$, $c_0^\prime$, $c_1^\prime$ and $c_2^\prime$ are large enough positive constants. Note that (i) holds because of the selection of the tuning parameters and the relationships between the $\ell_1$, $\ell_2$ and nuclear norms on the true transition matrices;
    (ii) holds because of the upper bound of the weighted penalty term $\mathcal{Q}$ derived in \eqref{eq:22}; (iii) holds because of the selection of the tuning parameters. Similar to \eqref{eq:25}, we need to find a further lower bound for \eqref{eq:27}. 
    Let us first establish the following two facts:
    \begin{equation*}
        16c_0^\prime\left(d_{j+1}^\star \log p + r_{j+1}^\star p\right)\sqrt{\frac{T\xi_T}{s_i - s_{i-1}}} \geq c_0^{\prime\prime}(M_S^2d_{\max}^{\star^3} + \alpha_L^2r_{\max}^{\star^2})\frac{(T\xi_T)^{\frac{5}{2}}}{(s_i - s_{i-1})^{\frac{3}{2}}},
    \end{equation*}
    and
    \begin{equation*}
        \begin{aligned}
            &16c_0^\prime\left(d_{j+1}^\star \log p + r_{j+1}^\star p\right)\sqrt{\frac{T\xi_T}{s_i - s_{i-1}}} \\
            &\geq 2c_0^{\prime\prime}\left(\frac{(T\xi_T)^{\frac{3}{2}}}{\sqrt{s_i-s_{i-1}}}\right)\left(M_Sd_{\max}^{\star^2}\sqrt{\frac{\log p}{s_i - s_{i-1}}} + \alpha_Lr_{\max}^{\frac{3}{2}}\sqrt{\frac{p}{s_i-s_{i-1}}}\right).
        \end{aligned}
    \end{equation*}
    The first inequality can be rearranged as:
    \begin{equation*}
        16c_0^\prime\left(\frac{d_{j+1}^\star\log p + r_{j+1}^\star p}{T\xi_T}\right)\left(\frac{s_i-s_{i-1}}{T\xi_T}\right) \geq c_0^{\prime\prime}(M_S^2d_{\max}^{\star^3} + \alpha_L^2r_{\max}^{\star^2}),
    \end{equation*}
    which can be directly verified by using Assumption H3'. The right-hand side of the second inequality is upper bounded by the Cauchy-Schwarz inequality:
    \begin{equation*}
        M_Sd_{\max}^{\star^2}\sqrt{\frac{\log p}{s_i - s_{i-1}}} + \alpha_Lr_{\max}^{\frac{3}{2}}\sqrt{\frac{p}{s_i-s_{i-1}}} \leq \left(M_S^2d_{\max}^{\star^3} + \alpha_Lr_{\max}^{\star^2}\right)^{\frac{1}{2}}\left(\frac{d_{\max}^\star\log p + r_{\max}^\star p}{s_i - s_{i-1}}\right)^{\frac{1}{2}},
    \end{equation*}
    and after substituting back into the second inequality, we need to establish:
    \begin{align*}
        &16c_0^\prime\left(d_{\max}^\star \log p + r_{\max}^\star p\right)\sqrt{\frac{T\xi_T}{s_i - s_{i-1}}} \\
        &\geq 2c_0^{\prime\prime}\left(\frac{(T\xi_T)^{\frac{3}{2}}}{\sqrt{s_i-s_{i-1}}}\right)\left(M_S^2d_{\max}^{\star^3} + \alpha_Lr_{\max}^{\star^2}\right)^{\frac{1}{2}}\left(\frac{d_{\max}^\star\log p + r_{\max}^\star p}{s_i - s_{i-1}}\right)^{\frac{1}{2}} \\
        \iff &64c_0^\prime\left(\frac{d_{\max}^\star \log p + r_{\max}^\star p}{T\xi_T}\right)\left(\frac{s_i-s_{i-1}}{T\xi_T}\right) \geq c_0^{\prime\prime}\left(M_S^2d_{\max}^{\star^3} + \alpha_Lr_{\max}^{\star^2}\right),
    \end{align*}
    which has already been proven. The following facts are consequences of Assumption H3':
    \begin{equation*}
        d_{\max}^\star\sqrt{T\xi_T\log p} \geq d_{\max}^\star\log p\sqrt{\frac{T\xi_T}{s_i-s_{i-1}}}\ \text{and}\ \sqrt{r_{\max}^\star}\sqrt{T\xi_T p} \geq r_{\max}^\star p\sqrt{\frac{T\xi_T}{s_i-s_{i-1}}}.
    \end{equation*}
    Therefore, we can derive a further lower bound for \eqref{eq:27}:
    \begin{equation}\label{eq:28}
        \sum_{t=s_{i-1}}^{\tau_j^\star-1}\|X_t - (\widehat{L}_i + \widehat{S}_i)X_{t-1}\|_2^2 
        \geq \sum_{t=s_{i-1}}^{t_j^\star-1}\|\epsilon_t\|_2^2 - c_1^\prime d_{\max}^\star\sqrt{T\xi_T\log p} - c_2^\prime \sqrt{r_{\max}^\star}\sqrt{T\xi_T p}.
    \end{equation}
    
    Combining the results from \eqref{eq:26} and \eqref{eq:28}, we get:
    \begin{equation}\label{eq:29}
        \sum_{t=s_{i-1}}^{s_i-1}\|X_t - (\widehat{L}_i + \widehat{S}_i)X_{t-1}\|_2^2 \geq \sum_{t=s_{i-1}}^{s_i-1}\|\epsilon_t\|_2^2 - c_2\left( d_{\max}^\star\sqrt{T\xi_T\log p} + \sqrt{r_{\max}^\star}\sqrt{T\xi_T p} \right). 
    \end{equation}
    
    For case (c), we firstly assume that $s_{i-1} < \tau_j^\star < s_i$, $|s_{i-1} - \tau_j^\star| > \Delta_T/4$ and $|s_i - \tau_j^\star| > \Delta_T/4$, respectively. Therefore, the interval $[s_{i-1}, \tau_j^\star)$where the model is misspecified is not negligible as compared to the other interval $[\tau_j^\star, s_i)$; hence, we can not obtain the convergence rate of $\widehat{L}_i$ and $\widehat{S}_i$ on the whole interval $[s_{i-1}, s_i)$. By using a similar procedure as in \eqref{eq:21}, similar results can be derived as long as we choose the tuning parameters as proposed in case (c) of Assumption H6:
    \begin{equation*}
        \mathcal{Q}\left( \widehat{\Delta}^L, \widehat{\Delta}^S\right) \leq 4\mathcal{Q}\left( \widehat{\Delta}^L\vert_A, \widehat{\Delta}^S\vert_{\mathcal{I}} \right),\quad 
        \mathcal{Q}\left( \widetilde{\Delta}^L, \widetilde{\Delta}^S \right) \leq 4\mathcal{Q}\left( \widetilde{\Delta}^L\vert_A, \widetilde{\Delta}^S\vert_{\mathcal{I}} \right).
    \end{equation*}
    
    Next, by following the same procedure as in the proof of case (b), we separately consider two intervals: $[s_{i-1}, \tau_j^\star)$ and $[\tau_j^\star, s_i)$ as follows: \\
    For the interval $[s_{i-1}, \tau_j^\star)$, we adopt the same notation as before and obtain:
    \begin{align}\label{eq:30}
        &\sum_{t=s_{i-1}}^{\tau_j^\star-1} \|X_t - (\widehat{L}_i + \widehat{S}_i)X_{t-1}\|_2^2 \nonumber \\
        \geq &\sum_{t=s_{i-1}}^{\tau_j^\star-1}\|\epsilon_t\|_2^2 + c|\tau_j^\star - s_{i-1}|\|\widetilde{\Delta}^L + \widetilde{\Delta}^S \|_F^2 - c^\prime\sqrt{|\tau_j^\star - s_{i-1}|\log p}\|\widetilde{\Delta}^S\|_1 - c^\prime\sqrt{|\tau_j^\star - s_{i-1}|p}\|\widetilde{\Delta}^L\|_* \nonumber \\
        \overset{\text{(i)}}{\geq} &\sum_{t=s_{i-1}}^{\tau_j^\star-1}\|\epsilon_t\|_2^2 + c|\tau^\star_j - s_{i-1}|\left(\|\widetilde{\Delta}^L\|_F^2 + \|\widetilde{\Delta}^S\|_F^2 - \frac{3\mu_i}{2}\mathcal{Q}(\widetilde{\Delta}^L, \widetilde{\Delta}^S) - \frac{(d_{\max}^\star+\sqrt{r_{\max}^\star})\|\widetilde{\Delta}^S\|_1 + \|\widetilde{\Delta}^L\|_*}{s_i-s_{i-1}} \right) \nonumber \\
        \geq &\sum_{t=s_{i-1}}^{\tau^\star_j -1}\|\epsilon_t\|_2^2 + c|\tau_j^\star-s_{i-1}|\sqrt{\|\widetilde{\Delta}^L\|_F^2 + \|\widetilde{\Delta}^S\|_F^2}\left(\sqrt{\|\widetilde{\Delta}^L\|_F^2 + \|\widetilde{\Delta}^S\|_F^2} - \sqrt{\lambda_i^2 d_{\max}^\star + \mu_i^2 r_{\max}^\star}\right),
    \end{align}
    where (i) can be verified by using a similar procedure as in \eqref{eq:38} and \eqref{eq:39} in the proof of Theorem 1; the last inequality is a direct consequence of the Cauchy-Schwarz inequality for the upper bound of $\mathcal{Q}(\widetilde{\Delta}^L, \widetilde{\Delta}^S)$ and Assumption H3' on the minimum spacing, which leads to the vanishing of the last term. 
    
    On the other hand, for the interval $[\tau_j^\star, s_i)$, we employ a similar procedure as in \eqref{eq:31} to derive the following result:
    \begin{align}\label{eq:31}
        &\sum_{t=\tau_j^\star}^{s_i-1}\|X_t - (\widehat{L}_i + \widehat{S}_i)X_{t-1}\|_2^2 \nonumber \\
        \geq &\sum_{t=\tau^\star_j}^{s_i -1}\|\epsilon_t\|_2^2 + c|s_i-\tau^\star_j|\sqrt{\|\widehat{\Delta}^L\|_F^2 + \|\widehat{\Delta}^S\|_F^2}\left(\sqrt{\|\widehat{\Delta}^L\|_F^2 + \|\widehat{\Delta}^S\|_F^2} - \sqrt{\lambda_i^2 d_{\max}^\star + \mu_i^2 r_{\max}^\star}\right).
    \end{align}
    {
    According to Assumption H1', either $\|S_{j+1}^\star - S_j^\star\|_2 \geq v_S > 0$ holds or $\|L_{j+1}^\star - L_j^\star\|_2 \geq v_L > 0$ holds. By defining $\widetilde{v} = \min_{1\leq j\leq m_0}\{v_{j,S}^2 + v_{j,L}^2\}$, it is not difficult to see that either $\|\widehat{\Delta}^S\|_F^2 \geq \widetilde{v}/4$ or $\|\widetilde{\Delta}^S\|_F^2 \geq \widetilde{v}/4$, or $\|\widehat{\Delta}^L\|_F^2 \geq \widetilde{v}/4$ or $\|\widetilde{\Delta}^L\|_F^2 \geq \widetilde{v}/4$. Without loss of generality, we assume that $\|\widetilde{\Delta}^L\|_F \geq \widetilde{v}/4$ and $\|\widetilde{\Delta}^S\|_F \geq \widetilde{v}/4$, respectively. We can then obtain a further lower bound for \eqref{eq:30} as follows:
    \begin{align}\label{eq:32}
        \sum_{t=s_{i-1}}^{\tau_j^\star-1}\|X_t - (\widehat{L}_i + \widehat{S}_i)X_{t-1}\|_2^2 
        &\geq \sum_{t=s_{i-1}}^{\tau_j^\star-1}\|\epsilon_t\|_2^2 \nonumber \\
        &+ \frac{\sqrt{2}}{4}c|\tau_j^\star - s_{i-1}|\sqrt{\widetilde{v}}\left( \frac{\sqrt{2}}{4}\sqrt{\widetilde{v}} - \left(2 + \frac{4c^\prime}{c}\right)\sqrt{\lambda_i^2d_{\max}^\star + \mu_i^2 r_{\max}^\star} \right)  \nonumber \\
        &\geq \sum_{t=s_{i-1}}^{\tau_j^\star-1}\|\epsilon_t\|_2^2 + c_1\widetilde{v}\Delta_T.
    \end{align}}
    For the second interval $[\tau_j^\star, s_i)$, we obtain a lower bound for \eqref{eq:31}:
    \begin{align}\label{eq:33}
        \sum_{t=\tau_j^\star}^{s_i-1}\|X_t - (\widehat{L}_i + \widehat{S}_i)X_{t-1}\|_2^2 
        &\geq \sum_{t=\tau_j^\star}^{s_i-1}\|\epsilon_t\|_2^2 - c_2|s_i - \tau_j^\star|\left( \lambda_i^2 d_{\max}^\star + \mu_i^2 r_{\max}^\star \right) \nonumber \\
        &\geq \sum_{t=\tau_j^\star}^{s_i-1}\|\epsilon_t\|_2^2 - c_2\bigg(d_{\max}^\star\log(p\vee \Delta_T) + r_{\max}^\star(p\vee \log \Delta_T)\bigg).
    \end{align}
    
    Next, by combining the lower bounds \eqref{eq:32} and \eqref{eq:33}, we obtain:
    \begin{equation}\label{eq:34}
        \sum_{t=s_{i-1}}^{s_i-1}\|X_t - (\widehat{L}_i + \widehat{S}_i)\|_2^2 \geq \sum_{t=s_{i-1}}^{s_i-1}\|\epsilon_t\|_2^2 + c_1\widetilde{v}\Delta_T - c_2\bigg(d_{\max}^\star\log(p\vee \Delta_T)+ r_{\max}^\star(p\vee \log \Delta_T)\bigg).
    \end{equation}
    Notice that another scenario might arise in the (c) case: namely, $s_{i-1} < \tau_j^\star < s_i$, $T\xi_T < |s_i - \tau_j^\star| \ll \Delta_T$ and $T\xi_T < |\tau_j^\star - s_{i-1}| \ll \Delta_T$. By using analogous calculations as in case (b), we finally establish:
    \begin{equation}\label{eq:35}
        \sum_{t=s_{i-1}}^{s_i-1}\|X_t - (\widehat{L}_i + \widehat{S}_i)\|_2^2 \geq \sum_{t=s_{i-1}}^{s_i-1}\|\epsilon_t\|_2^2 - c_2^\prime T\xi_T\left(d_{\max}^{\star^2} + r_{\max}^{\star^{\frac{3}{2}}}\right).
    \end{equation}
    Combining \eqref{eq:19}, \eqref{eq:29}, \eqref{eq:34}, and \eqref{eq:35} and summing up all $m+1$ intervals leads to the final result. 
\end{proof}

To verify the theoretical properties for the surrogate weakly sparse model, the following lemmas are required. 
\begin{lemma}
\label{lemma:7}
Consider the single change point scenario in Proposition 2 and also assume $\tau>\tau^\star$. For the misspecified model in the interval $[1,\tau)$ with tuning parameters provided in \eqref{eq:12}, and given that the search domain $\mathcal{T}^w$ satisfies Assumption W2, as $T \succsim \log p$, we obtain:\\
(1) with probability at least $1-c_1p^{-1}$:
{\small
\begin{equation*}
    \left\| \frac{1}{\tau-1}\sum_{t=1}^{\tau-1}X_{t-1}(X_t - A_1^\star X_{t-1})^\prime \right\|_\infty \leq \frac{\lambda_{1,\tau}^w}{2} + c_0M_S\frac{(\tau-\tau^\star)_+}{\tau-1}R_q\eta_{\min}^{-q},
\end{equation*}}
(2) with probability at least $1-c_2p^{-1}$:
{\small
\begin{equation*}
    \left\| \frac{1}{T-\tau}\sum_{t=\tau}^{T-\tau}X_{t-1}(X_t - A_2^\star X_{t-1})^\prime \right\|_\infty \leq \frac{\lambda_{2,\tau}^w}{2} + c_0M_S\frac{(\tau^\star-\tau)_+}{T-\tau}R_q\eta_{\min}^{-q},
\end{equation*}}
where $c_0,c_1,c_2$ are some large positive constants. 
\end{lemma}

\begin{proof}[Proof of Lemma \ref{lemma:7}]
    This proof is similar to the proof of Lemma \ref{lemma:3}. The key step is to measure the deviations for the misspecified model in the posited interval. In this case, since we assume that $\tau > \tau^\star$,  the deviation on the interval $[\tau^\star, \tau)$ is upper bounded by $c_0M_S\frac{(\tau-\tau^\star)_+}{\tau-1}|\mathcal{J}(\eta_{j})|$ for some large constant $c_0>0$. Then, substituting the upper bound of $|\mathcal{J}(\eta_j)|$ by $R_q\eta_j^{-q}$ implies the final result. 
\end{proof}
{
\begin{lemma}
\label{lemma:8}
Under Assumptions W1-W3, for a set of estimated change points $(s_1,s_2,\cdots, s_m)$ with $m < m_0$ and for the minimum spacing $\Delta_T$, and jump size $v_A = \min_{1\leq j\leq m_0}\|A_{j+1}^\star - A_j^\star\|_2$, there exist universal positive constants $c_1, c_2 > 0$ such that:
\begin{equation*}
    \mathbb{P}\left(\min_{(s_1,\dots,s_m)}\mathcal{L}_T(s_1,\dots,s_m; \bm{\lambda}^w) > \sum_{t=1}^T\|\epsilon_t\|_2^2 + c_1v_A\Delta_T - c_2mT\xi_TR_q^2\left(\frac{\log (p\vee T)}{T}\right)^{-q} \right) \to 1.
\end{equation*}
\end{lemma}}

\begin{proof}[Proof of Lemma \ref{lemma:8}]
    This lemma is proved in a similar manner as Lemma 4 in \cite{safikhani2017joint}. 
\end{proof}


\newpage
\section{Proofs of Main Theorems}\label{appendix:E}
\begin{proof}[Proof of Theorem \ref{thm:1}]
    Let $\widehat{\tau}$ be the estimated change point obtained by solving optimization problem \eqref{eq:4} addressed in the main text. Based on Algorithm~\ref{algo:1}, we use the following objective function $\mathcal{L}(\tau)$ similar to \eqref{eq:3} in the main text, which can be written as:   
    \begin{equation}\label{eq:36}
        \mathcal{L}(\tau) = \sum_{t=1}^{\tau-1}\|X_t - (\widehat{L}_{1,\tau} + \widehat{S}_{1,\tau})X_{t-1}\|_2^2 + \sum_{t=\tau}^{T-1}\|X_t - (\widehat{L}_{2,\tau} + \widehat{S}_{2,\tau})X_{t-1}\|_2^2\overset{\text{def}}{=} I_1 + I_2,
    \end{equation}
    wherein $\widehat{L}_{j,\tau}$ and $\widehat{S}_{j,\tau}$ for $j=1,2$ are the optimizers of the convex programs \eqref{eq:5} in the main. Note that the estimated low rank and sparse components are functions of $\tau$. Without loss of generality, we assume that $\tau > \tau^\star$, and the length of misspecified interval $(\tau - \tau^\star)$ is large enough. 
    
    Denote by $(L_1^\star, S_1^\star)$ and $(L_2^\star, S_2^\star)$ the true components in the interval $[1,\tau^\star)$ and $[\tau^\star, T)$, respectively. Further, we use the notation $\widehat{\Delta}^L_{j,\tau} = \widehat{L}_{j,\tau} - L_j^\star$, $\widehat{\Delta}^S_{j,\tau} = \widehat{S}_{j,\tau} - S_j^\star$ for $j=1,2$. Then, by using the tuning parameters defined in \eqref{eq:6} in the main text, we obtain the corresponding lower bound of $I_2$:
    \allowdisplaybreaks
    \begin{align}\label{eq:37}
        I_2
        &= \sum_{t=\tau}^{T-1}\|X_t - (\widehat{L}_{2,\tau} + \widehat{S}_{2,\tau})X_{t-1}\|_2^2 \nonumber \\
        &\geq \sum_{t=\tau}^{T-1}\|\epsilon_t\|_2^2 + \sum_{t=\tau}^{e-1}\|X_{t-1}(\widehat{\Delta}^L_{2,\tau} + \widehat{\Delta}^S_{2,\tau})\|_2^2 - 2\left| \sum_{t=\tau}^{T-1}X_{t-1}^\prime(\widehat{\Delta}^L_{2,\tau} + \widehat{\Delta}^S_{2,\tau})^\prime \epsilon_t \right| \nonumber \\
        &\overset{\text{(i)}}{\geq} \sum_{t=\tau}^{T-1}\|\epsilon_t\|_2^2 + c^\prime|T-\tau|\|\widehat{\Delta}^L_{2,\tau} + \widehat{\Delta}^S_{2,\tau}\|_F^2 \nonumber \\
        &- c^{\prime\prime}|T-\tau|\left( \sqrt{\frac{\log p + \log(T-\tau)}{T-\tau}}\|\widehat{\Delta}^S_{2,\tau}\|_1 + \sqrt{\frac{p + \log(T-\tau)}{T-\tau}}\|\widehat{\Delta}^L_{2,\tau}\|_*\right) \nonumber \\
        &\overset{\text{(ii)}}{\geq} \sum_{t=\tau}^{T-1}\|\epsilon_t\|_2^2 + c^\prime|T-\tau|\left( \|\widehat{\Delta}^L_{2,\tau}\|_F^2 + \|\widehat{\Delta}^S_{2,\tau}\|_F^2 - \frac{\mu_{2,\tau}}{2}\mathcal{Q}(\widehat{\Delta}^L_{2,\tau}, \widehat{\Delta}^S_{2,\tau}) \right) - c^{\prime\prime}|T-\tau|\mu_{2,\tau}\mathcal{Q}(\widehat{\Delta}^L_{2,\tau}, \widehat{\Delta}^S_{2,\tau}) \nonumber \\
        &\overset{\text{(iii)}}{\geq} \sum_{t=\tau}^{T-1}\|\epsilon_t\|_2^2 + c^\prime|T-\tau|\sqrt{\|\widehat{\Delta}^L_{2,\tau}\|_F^2 + \|\widehat{\Delta}^S_{2,\tau}\|_F^2}\left( \sqrt{\|\widehat{\Delta}^L_{2,\tau}\|_F^2 + \|\widehat{\Delta}^S_{2,\tau}\|_F^2} \right. \nonumber \\
        &\left.- \left(\frac{1}{2} + \frac{c^{\prime\prime}}{c^\prime}\right)\sqrt{\lambda_{2,\tau}^2 d_{\max}^\star + \mu_{2,\tau}^2 r_{\max}^\star} \right) \nonumber \\
        &\geq \sum_{t=\tau}^{T-1}\|\epsilon_t\|_2^2 - c^{\prime}|T-\tau|\left( \lambda_{2,\tau}^2d_{\max}^\star + \mu_{2,\tau}^2 r_{\max}^\star \right) \nonumber \\
        &\geq \sum_{t=\tau}^{T-1}\|\epsilon_t\|_2^2 - c^\prime\bigg(d_{\max}^\star\log p + r_{\max}^\star p + (d_{\max}^\star + r_{\max}^\star)\log(T-\tau)\bigg),
    \end{align}
    where (i) holds based on Lemma \ref{lemma:3}; (ii) is derived based on the following result together with Assumption H3 on the size of the search domain $\mathcal{T}$:
    \begin{align}
        \label{eq:38}
        \|\widehat{\Delta}^L_{2,\tau} + \widehat{\Delta}^S_{2,\tau}\|_F^2 
        &\geq \|\widehat{\Delta}^L_{2,\tau}\|_F^2 + \|\widehat{\Delta}^S_{2,\tau}\|_F^2 -2\bigg| \langle \widehat{\Delta}^L_{2,\tau}, \widehat{\Delta}^S_{2,\tau}\rangle\bigg| \nonumber \\
        &\geq \|\widehat{\Delta}^L_{2,\tau}\|_F^2 + \|\widehat{\Delta}^S_{2,\tau}\|_F^2 - 2\|\widehat{\Delta}^L_{2,\tau}\|_\infty\|\widehat{\Delta}^S_{2,\tau}\|_1 \nonumber \\
        &\geq \|\widehat{\Delta}^L_{2,\tau}\|_F^2 + \|\widehat{\Delta}^S_{2,\tau}\|_F^2 - \frac{2\alpha_L}{p}\|\widehat{\Delta}^S_{2,\tau}\|_1 \nonumber \\
        &\geq \|\widehat{\Delta}^L_{2,\tau}\|_F^2 + \|\widehat{\Delta}^S_{2,\tau}\|_F^2 - \lambda_{2,\tau}\|\widehat{\Delta}^S_{2,\tau}\|_1 - \mu_{2,\tau}\|\widehat{\Delta}^L_{2,\tau}\|_*;
    \end{align}
    (iii) holds because of an application of the Cauchy-Schwarz inequality to the result of Lemma \ref{lemma:5}.
    
    Next, we derive a lower bound for $I_1$. Before stating the results, we first define the \textit{misspecified} error terms $\widetilde{\Delta}^L_{1/2,\tau} = \widehat{L}_{1,\tau} - L_2^\star$ and $\widetilde{\Delta}^S_{1/2,\tau} = \widehat{S}_{1,\tau} - S_2^\star$, then we similarly obtain that:
    {
    \begin{align}
        \label{eq:39}
        I_1
        &= \sum_{t=1}^{\tau-1}\|X_t - (\widehat{L}_{1,\tau} + \widehat{S}_{1,\tau})X_{t-1}\|_2^2 \nonumber \\
        &\geq \sum_{t=1}^{\tau-1}\|\epsilon_t\|_2^2 + \sum_{t=1}^{\tau^\star-1}\|X_{t-1}(\widehat{\Delta}^L_{1,\tau} + \widehat{\Delta}^S_{1,\tau})\|_2^2 + \sum_{t=\tau^\star}^{\tau-1}\|X_{t-1}(\widetilde{\Delta}^L_{1/2,\tau} + \widetilde{\Delta}^S_{1/2,\tau})\|_2^2 \nonumber \\
        &- 2\left| \sum_{t=1}^{\tau^\star-1}X_{t-1}^\prime(\widehat{\Delta}^L_{1,\tau} + \widehat{\Delta}^S_{1,\tau})^\prime \epsilon_t \right| - 2\left| \sum_{t=\tau^\star}^{\tau-1}X_{t-1}^\prime(\widetilde{\Delta}^L_{1/2,\tau} + \widetilde{\Delta}^S_{1/2,\tau})^\prime \epsilon_t \right| \nonumber \\
        &\overset{(i)}{\geq} \sum_{t=1}^{\tau-1}\|\epsilon_t\|_2^2 + c|\tau^\star-1|\|\widehat{\Delta}^L_{1,\tau} + \widehat{\Delta}^S_{1,\tau}\|_F^2 + c^\prime|\tau-\tau^\star|\|\widetilde{\Delta}^L_{1/2,\tau} + \widetilde{\Delta}^S_{1/2,\tau}\|_F^2 \\
        &- c_1|\tau^\star-1|\left(\sqrt{\frac{\log p + \log(\tau-1)}{\tau-1}}\|\widehat{\Delta}^S_{1,\tau}\|_1 + \sqrt{\frac{p + \log(\tau-1)}{\tau-1}}\|\widehat{\Delta}^L_{1,\tau}\|_* \right) \nonumber \\
        &- c_1^\prime|\tau-\tau^\star|\left( \left(\sqrt{\frac{\log p + \log(\tau-1)}{\tau-1}} + \frac{M_S\vee\alpha_L}{\tau-1}(d_{\max}^\star+\sqrt{r_{\max}^\star})\right)\|\widetilde{\Delta}^S_{1/2,\tau}\|_1 \right. \nonumber \\
        &+\left. \left(\sqrt{\frac{p + \log(\tau-1)}{\tau-1}}+\frac{M_S\vee\alpha_L}{\tau-1}\right)\|\widetilde{\Delta}^L_{1/2,\tau}\|_* \right) \nonumber \\
        &\overset{\text{(ii)}}{\geq} \sum_{t=1}^{\tau-1}\|\epsilon_t\|_2^2 + c|\tau^\star-1|\left( \|\widehat{\Delta}^L_{1,\tau}\|_F^2 + \|\widehat{\Delta}^S_{1,\tau}\|_F^2 - \frac{3\mu_{1,\tau}}{2}\mathcal{Q}(\widehat{\Delta}^L_{1,\tau}, \widehat{\Delta}^S_{1,\tau}) \right) \nonumber \\
        &+ c^\prime|\tau-\tau^\star|\left( \|\widetilde{\Delta}^L_{1/2,\tau}\|_F^2 + \|\widetilde{\Delta}^S_{1/2,\tau}\|_F^2 - \frac{3\mu_{1,\tau}}{2}\mathcal{Q}(\widetilde{\Delta}^L_{1/2,\tau}, \widetilde{\Delta}^S_{1/2,\tau}) \right) \nonumber \\
        &- c^{\prime\prime}|\tau-\tau^\star|\frac{(d_{\max}^\star+\sqrt{r_{\max}^\star})\|\widetilde{\Delta}^S_{1/2,\tau}\|_1 + \|\widetilde{\Delta}^L_{1/2,\tau}\|_*}{\tau-1} \nonumber \\
        &\overset{\text{(iii)}}{\geq} \sum_{t=1}^{\tau-1}\|\epsilon_t\|_2^2 + c^\prime_1|\tau-\tau^\star|\left(v_S^2 + v_L^2\right) - c_2^\prime\left(d_{\max}^\star\log(p\vee T) + r_{\max}^\star(p\vee T)\right), \nonumber
    \end{align}}
    \noindent
    where $C_0 \geq (M_S\vee\alpha_L)$ is some large constant. Inequality (i) is derived by the deviation bound in Lemma \ref{lemma:3}; (ii) holds due to the definition of the weighted regularizer $\mathcal{Q}$ and \eqref{eq:38}; {(iii) is derived by substituting the differences of sparse components $\|\widehat{\Delta}^S_{1,\tau}\|_2$ and $\|\widetilde{\Delta}^S_{1/2,\tau}\|_2$ by $v_S$ and the differences of low rank components $\|\widehat{\Delta}^L_{1,\tau}\|_2$ and $\|\widetilde{\Delta}^L_{1/2,\tau}\|_2$ by $v_L$, respectively.} 
    
    Since we can not verify the RE condition on the misspecified interval $[1,\tau)$ and due to  $\|S_2^\star - S_1^\star\|_2 \geq v_S >0$ and $\|L_2^\star - L_1^\star\|_2 \geq v_L > 0$, then we have either $\|\widehat{\Delta}^S_{1,\tau}\|_2 \geq v_S/4$ or $\|\widetilde{\Delta}^S_{1/2,\tau}\|_2 \geq v_S/4$, and either $\|\widehat{\Delta}^L_{1,\tau}\|_2 \geq v_L/4$ or $\|\widetilde{\Delta}^L_{1/2,\tau}\|_2 \geq v_L/4$. Assume that $\|\widetilde{\Delta}_{1/2,\tau}^S\|_2 \geq v_S/4$ and $\|\widetilde{\Delta}^L_{1/2,\tau}\|_2 \geq v_L/4$, then based on Assumptions H2 and H3, it implies that
    \begin{equation*}
        \frac{(d_{\max}^\star+\sqrt{r_{\max}^\star})\|\widetilde{\Delta}^S_{1/2,\tau}\|_1+\|\widetilde{\Delta}^L_{1/2,\tau}\|_*}{\tau-1} \to 0,
    \end{equation*}
    hence, for some constants $c_1$, $c_2>0$ we get: 
    \begin{equation}\label{eq:40}
        I_1 \geq \sum_{t=1}^{\tau-1}\|\epsilon_t\|_2^2 + c_1|\tau-\tau^\star| - c_2\bigg(d_{\max}^\star \log p + r_{\max}^\star p + (d^\star_{\max} + r_{\max}^\star)\log(\tau-1)\bigg).
    \end{equation}
    Combining \eqref{eq:38} and \eqref{eq:40} establishes that the objective function $\mathcal{L}(\tau)$ satisfies:
    {
    \begin{equation}\label{eq:41}
        \mathcal{L}(\tau) \geq \sum_{t=1}^{T-1}\|\epsilon_t\|_2^2 + K_1\left(v_S^2 + v_L^2\right)|\tau-\tau^\star| - K_2\bigg(d_{\max}^\star\log(p \vee T) + r_{\max}^\star(p \vee \log T)\bigg).
    \end{equation}}
    
    Next, we prove the upper bound of $\mathcal{L}(\tau^\star)$. For some constant $K>0$,
    \begin{equation}\label{eq:42}
        \mathcal{L}(\tau^\star) \leq \sum_{t=1}^{T-1}\|\epsilon_t\|_2^2 + K\bigg( d_{\max}^\star\log(p\vee T) + r_{\max}^\star(p \vee \log T) \bigg).
    \end{equation}
    To see this result, by using a similar procedure, we have:
    \begin{equation*}
        \mathcal{L}(\tau^\star) = \sum_{t=1}^{\tau^\star-1}\|X_t - (\widehat{L}_{1,\tau} + \widehat{S}_{1,\tau})X_{t-1}\|_2^2 + \sum_{t=\tau^\star}^{T-1}\|X_t - (\widehat{L}_{2,\tau} + \widehat{S}_{2,\tau})X_{t-1}\|_2^2\overset{\text{def}}{=} J_1 + J_2.
    \end{equation*}
    Then, we obtain:
    \begin{align}\label{eq:43}
        J_1 
        &\leq \sum_{t=1}^{\tau^\star-1} \|\epsilon_t\|_2^2 + 2c|\tau^\star-1|\left(\|\widehat{\Delta}^L_{1,\tau^\star}\|_F^2 + \|\widehat{\Delta}^S_{1,\tau^\star}\|_F^2 + c^\prime\sqrt{\frac{\log p+\log(\tau^\star-1)}{\tau^\star-1}}\|\widehat{\Delta}^S_{1,\tau^\star}\|_1 \right. \nonumber \\
        &+\left. c^\prime\sqrt{\frac{p+\log(\tau^\star-1)}{\tau^\star-1}}\|\widehat{\Delta}^L_{1,\tau^\star}\|_*\right) \nonumber \\
        &\leq \sum_{t=1}^{\tau^\star-1}\|\epsilon_t\|_2^2 + K_1\bigg( d_{\max}^\star\log(p\vee T) + r_{\max}^\star(p \vee \log T) \bigg),
    \end{align}
    and similarly we have:
    \begin{equation}\label{eq:44}
        J_2 \leq \sum_{t=\tau^\star}^{T-1}\|\epsilon_t\|_2^2 + K_2\left( d_{\max}^\star \log(p\vee T) + r_{\max}^\star(p \vee \log T) \right).
    \end{equation}
    Hence, combining inequalities \eqref{eq:43} and \eqref{eq:44} leads to the fact \eqref{eq:41}.
    
    Based on \eqref{eq:41} and \eqref{eq:42}, and using the fact that $\widehat{\tau}$ is the minimizer of optimization program (4) in the main text, we get that with high probability:
    {
    \begin{align}\label{eq:45}
        &\sum_{t=1}^{T-1}\|\epsilon_t\|_2^2 + K_1(v_S^2 + v_L^2)|\widehat{\tau} - \tau^\star| - K_2\bigg( d^\star_{\max}\log(p\vee T) + r_{\max}^\star(p\vee \log T) \bigg) \leq \mathcal{L}(\widehat{\tau}) \nonumber \\
        \leq &\mathcal{L}(\tau^\star) \leq \sum_{t=1}^{T-1}\|\epsilon_t\|_2^2 + K\bigg( d_{\max}^\star\log(p\vee T) + r_{\max}^\star(p\vee \log T) \bigg).
    \end{align}
    Therefore, with high probability, for some large enough constant $K_0>0$, the following holds
    \begin{equation}\label{eq:46}
        |\widehat{\tau} - \tau^\star| \leq K_0\frac{ d_{\max}^\star\log(p\vee T) + r_{\max}^\star(p\vee \log T) }{v_S^2 + v_L^2},
    \end{equation}
    which concludes the proof of the Theorem.}
\end{proof}

{
\begin{remark}
\label{remark:4}
The consistency rate for a single change point derived in Theorem 1 is optimal up to a logarithmic factor $d_{\max}^* \log p$ for the sparse component with an additional term $r_{\max}^* p$ due to the low rank component. A closer look at the proof of Theorem 1 reveals that the inclusion of these two factors in the consistency rate is due to the unknown auto-regressive parameters to the left and the right of the change point which need to be estimated simultaneously, while the location of the change point is being estimated. As stated in \cite{csorgo1997limit}, the optimal consistency rate for locating a single change point in a family of low-dimensional linear models is $\mathcal{O} \left( 1/ \| A_2^* - A_1^*  \|_2^2 \right) $. Hence, it is of theoretical interest to examine what (possibly more stringent) conditions on the model parameters and/or model dimensions could remove these extra factors in the consistency rate obtained in Theorem 1.
\end{remark}
}
{
\begin{remark}
\label{remark:5}
The asymptotic framework considered in Theorem 1 is that of ``increasing domain asymptotics" in which the sampling rate is fixed and physical time (number of observations) increases. Further, note that we consider an \textit{offline change point detection setting} in which all the time points are present and one is primarily interested in estimating their unknown number and locations based on the observed data. Thus, an increasing domain asymptotics regime is meaningful for the detection framework under consideration. Such a modeling and asymptotic framework has been stated/utilized in the literature for univariate AR processes (\cite{davis2006structural}) as well as high-dimensional VAR processes (\cite{wang2019localizing,safikhani2017joint}). As mentioned in \cite{davis2006structural}, the detection problem can be seen as segmenting the series into blocks of different autoregressive (AR) processes with the objective of obtaining the ``best-fitting" model
from the class of piecewise AR processes.
\end{remark}
}

\begin{proof}[Proof of Theorem \ref{thm:2}]
    This proof is similar to the proof of Proposition 4.1 in \cite{basu2015regularized}. The key steps in the proof that require verification are (a) the restricted strong convexity condition and (b) the deviation bound condition (see Appendix A) for the intervals $[1,\tau^\star-R)$ and $[\tau^\star+R, T)$, respectively, for the radius $R$. For (a), analogous arguments as in the proof of Theorem 4 in \cite{safikhani2017joint} establish the result. Further (b) follows from the result established in Lemma \ref{lemma:1}. 
\end{proof}

\begin{proof}[Proof of Theorem \ref{thm:3}]
    We first establish the following fact: suppose $(m_0, \widehat{\tau}_i, i=1,2,\dots, m_0)$ is a subset of the candidate set $\widetilde{\mathcal{S}}$, which satisfies $\max_{1\leq i \leq m_0}|\widehat{\tau}_i - \tau_i^\star| \leq T\xi_T$. Then, we can obtain the upper bound for $\mathcal{L}_n(\widehat{\tau}_1, \cdots, \widehat{\tau}_{m_0}; \bm{\lambda}, \bm{\mu})$:
    \begin{equation}\label{eq:47}
        \mathcal{L}_T(\widehat{\tau}_1, \dots, \widehat{\tau}_{m_0}; \bm{\lambda}, \bm{\mu}) \leq \sum_{t=1}^T\|\epsilon_t\|_2^2 + Km_0T\xi_T (d_{\max}^{\star^2} + r_{\max}^{\star^{\frac{3}{2}}}),
    \end{equation}
    where $K>0$ is a large enough constant.
    
    To prove \eqref{eq:47}, we can focus on the estimated interval $(\widehat{\tau}_{i-1}, \widehat{\tau}_i)$ and corresponding estimates: $\widehat{L}_i$ and $\widehat{S}_i$. Suppose there is a true change point $\tau_j^\star$ such that: $\widehat{\tau}_{i-1} < \tau_j^\star < \widehat{\tau}_i$ with $|\tau_j^\star - \widehat{\tau}_{i-1}| \leq T\xi_T$. Similar to the proof in case (b) in Lemma \ref{lemma:6}, for the interval $[\tau^\star_j, \widehat{\tau}_i)$, by choosing the same tuning parameters as in case (b) of Assumption H6, we have:
    \allowdisplaybreaks
    \begin{align}\label{eq:48}
        &\sum_{t=\tau_j^\star}^{\widehat{\tau}_i-1}\|X_t - (\widehat{L}_i + \widehat{S}_i)X_{t-1}\|_2^2 \nonumber \\
        \leq &\sum_{t=t_j^\star}^{\widehat{t}_i-1}\|\epsilon_t\|_2^2 + c_3|\widehat{\tau}_i - \tau_j^\star|\|\widehat{\Delta}^L + \widehat{\Delta}^S\|_F^2 + c_3^\prime\left( \sqrt{|\widehat{\tau}_i - \tau_j^\star|p}\|\widehat{\Delta}^L\|_* + \sqrt{|\widehat{\tau}_i - \tau_j^\star|\log p}\|\widehat{\Delta}^S\|_1 \right) \nonumber \\
        \leq &\sum_{t=\tau_j^\star}^{\widehat{\tau}_i-1}\|\epsilon_t\|_2^2 + 2c_3|\widehat{\tau}_i - \tau_j^\star|(\|\widehat{\Delta}^L\|_F^2 + \|\widehat{\Delta}^S\|_F^2) + c_3^\prime|\widehat{\tau}_i - \tau_j^\star|\left(\sqrt{\frac{p}{|\widehat{\tau}_i - \tau_j^\star|}}\|\widehat{\Delta}^L\|_* + \sqrt{\frac{\log p}{|\widehat{\tau}_i - \tau_j^\star|}}\|\widehat{\Delta}^S\|_1 \right) \nonumber \\
        \overset{\text{(i)}}{\leq} &\sum_{t=\tau_j^\star}^{\widehat{\tau}_i-1}\|\epsilon_t\|_2^2 + 2c_3|\widehat{\tau}_i - \tau_j^\star|(\|\widehat{\Delta}^L\|_F^2 + \|\widehat{\Delta}^S\|_F^2) + c_3^\prime|\widehat{\tau}_i - \tau_j^\star|\frac{1}{2}{\mu_i}{\mathcal{Q}}(\widehat{\Delta}^L, \widehat{\Delta}^S) \nonumber \\
        = &\sum_{t=\tau_j^\star}^{\widehat{\tau}_i - 1}\|\epsilon_t\|_2^2 + c_3|\widehat{\tau}_i - \tau_j^\star|\left( 2(\|\widehat{\Delta}^L\|_F^2 + \|\widehat{\Delta}^S\|_F^2) + \frac{c_3^\prime}{2c_3}{\mu_i}\mathcal{Q}(\widehat{\Delta}^L, \widehat{\Delta}^S) \right) \nonumber \\
        \overset{\text{(ii)}}{\leq} &\sum_{t=\tau_j^\star}^{\widehat{\tau}_i - 1}\|\epsilon_t\|_2^2 + c_3|\widehat{\tau}_i - \tau_j^\star|\sqrt{\|\widehat{\Delta}^L\|_F^2 + \|\widehat{\Delta}^S\|_F^2}\left(2\sqrt{\|\widehat{\Delta}^L\|_F^2 + \|\widehat{\Delta}^S\|_F^2} + \frac{c_3^\prime}{2c_3}\sqrt{{\lambda_i}^2 d_{\max}^\star + {\mu_i}^2 r_{\max}^\star}\right) \nonumber \\
        \overset{\text{(iii)}}{\leq} &\sum_{t=\tau_j^\star}^{\widehat{\tau}_i-1}\|\epsilon_t\|_2^2 + \mathcal{O}_p\left( T\xi_T\left(d_{\max}^{\star^2} + r_{\max}^{\star^{\frac{3}{2}}}\right) \right),
    \end{align}
    where (i) holds because of the selection of the tuning parameters (see Lemma \ref{lemma:6} case (b)); (ii) holds since we can derive the error bound together with the upper bound of the weighted regularizer $\mathcal{Q}$; (iii) holds because of Assumptions H3' and H6. 
    
    For the other interval $[\widehat{\tau}_{i-1}, \tau_j^\star)$ we also get:
    \begin{align}\label{eq:49}
        &\sum_{t=\widehat{\tau}_{i-1}}^{\tau_j^\star-1}\|X_t - (\widehat{L}_i + \widehat{S}_i)X_{t-1}\|_2^2 \nonumber \\
        \leq &\sum_{t=\widehat{\tau}_{i-1}}^{\tau_j^\star-1}\|\epsilon_t\|_2^2 + c_3|\tau_j^\star - \widehat{\tau}_{i-1}|\|\widetilde{\Delta}^L + \widetilde{\Delta}^S\|_F^2 + c_3^\prime\left( \sqrt{|\tau_j^\star - \widehat{\tau}_{i-1}|p}\|\widetilde{\Delta}^L\|_* + \sqrt{|\tau_j^\star - \widehat{\tau}_{i-1}|\log p}\|\widetilde{\Delta}^S\|_1 \right) \nonumber \\
        \leq &\sum_{t=\widehat{\tau}_{i-1}}^{\tau_j^\star-1}\|\epsilon_t\|_2^2 + 2c_3|\tau_j^\star - \widehat{\tau}_{i-1}|\left(\|\widehat{\Delta}^L\|_F^2 + \|\widehat{\Delta}^S\|_F^2 + \|L_{j+1}^\star - L_j^\star\|_F^2 + \|S_{j+1}^\star - S_j^\star\|_F^2\right) \nonumber \\
        + &c_3^\prime\left( \sqrt{|\tau_j^\star - \widehat{\tau}_{i-1}|p}(\|\widehat{\Delta}^L\|_* + \|L_{j+1}^\star - L_j^\star\|_*) + \sqrt{|\tau_j^\star - \widehat{\tau}_{i-1}|\log p}(\|\widehat{\Delta}^S\|_1 + \|S_{j+1}^\star - S_j^\star\|_1) \right) \nonumber \\
        \leq &\sum_{t=\widehat{\tau}_{i-1}}^{\tau_j^\star-1}\|\epsilon_t\|_2^2 + \mathcal{O}_p\left( T\xi_T\left( d_{\max}^{\star^2} + r_{\max}^{\star^{\frac{3}{2}}}) \right) \right).
    \end{align}
    Combining \eqref{eq:48} and \eqref{eq:49} and adding all $m_0+1$ intervals lead to \eqref{eq:47}.
    
    Next, in order to prove the consistency of the number of estimated change points, we need to prove that: (a) $\mathbb{P}(\widehat{m} < m_0) \to 0$; and (b) $\mathbb{P}(\widehat{m} > m_0) \to 0$, respectively. To prove (a), we apply the result from Lemma \ref{lemma:6}, which leads to:
    \begin{align}\label{eq:50}
        &\text{IC}(\widehat{\tau}_1,\dots, \widehat{\tau}_{\widehat{m}}; \bm{\lambda}, \bm{\mu}, \omega_T) \nonumber = \mathcal{L}_T(\widehat{\tau}_1, \dots, \widehat{\tau}_{\widehat{m}}; \bm{\lambda}, \bm{\mu}) + \widehat{m}\omega_T \nonumber \\
        \overset{\text{(i)}}{>} &\sum_{t=1}^T\|\epsilon_t\|_2^2 + c_1\widetilde{v}\Delta_T - c_2\widehat{m}T\xi_T(d_{\max}^{\star^2} + r_{\max}^{\star^{\frac{3}{2}}}) + \widehat{m}\omega_T \nonumber \\
        \geq &\mathcal{L}_T(\widehat{\tau}_1,\dots,\widehat{\tau}_{m_0}; \bm{\lambda}, \bm{\mu}) + m_0\omega_T + c_1\widetilde{v}\Delta_T - c_2m_0T\xi_T(d_{\max}^{\star^2} + r_{\max}^{\star^{\frac{3}{2}}}) - (m_0 - \widehat{m})\omega_T \nonumber \\
        \overset{\text{(ii)}}{\geq} &\mathcal{L}_T(\widehat{\tau}_1, \dots, \widehat{\tau}_{m_0}; \bm{\lambda}, \bm{\mu}) + m_0\omega_T,
    \end{align}
    where (i) holds because of Lemma \ref{lemma:6}; (ii) holds because of Assumption H5. The result in \eqref{eq:50} shows that $(\widehat{\tau}_1, \cdots, \widehat{\tau}_{\widehat{m}})$ is not the optimal solution to minimize IC function defined in (10) in the main; hence, we conclude that $\mathbb{P}(\widehat{m} < m_0) \to 0$. To prove (b), we assume that $(\widehat{\tau}_1, \cdots, \widehat{\tau}_{\widehat{m}})$ are the estimated change points with $\widehat{m} > m_0$. Then, similarly we get:
    \begin{equation}\label{eq:51}
        \mathcal{L}_T(\widehat{\tau}_1, \dots, \widehat{\tau}_{\widehat{m}}; \bm{\lambda}, \bm{\mu}) \geq \sum_{t=1}^T \|\epsilon_t\|_2^2 - c_2^\prime \widehat{m}T\xi_T (d_{\max}^{\star^2} + r_{\max}^{\star^{\frac{3}{2}}}).
    \end{equation}
    Next, choose a subset $\{\widehat{\tau}_{i_1}, \cdots, \widehat{\tau}_{i_{m_0}}\}$ from $\{\widehat{\tau}_1,\cdots, \widehat{\tau}_{\widehat{m}}\}$ such that $\max_{1\leq j\leq m_0}|\widehat{\tau}_{i_j} - \tau_j^\star| \leq T\xi_T$. Then, based on the definitions for $\text{IC}(\widehat{\tau}_1,\cdots, \widehat{\tau}_{\widehat{m}}; \bm{\lambda}, \bm{\mu}, \omega_T)$ and $\text{IC}(\widehat{\tau}_{i_1},\cdots,\widehat{\tau}_{i_{m_0}}; \bm{\lambda}, \bm{\mu}, \omega_T)$ and using \eqref{eq:51} we obtain:
    \begin{align}\label{eq:52}
        \sum_{t=1}^T\|\epsilon_t\|_2^2 - c_2^\prime \widehat{m}T\xi_T(d_{\max}^{\star^2} + r_{\max}^{\star^2}) + \widehat{m}\omega_T
        &\leq \text{IC}(\widehat{\tau}_1, \dots, \widehat{\tau}_{\widehat{m}}; \bm{\lambda}, \bm{\mu}, \omega_T) \nonumber \\
        &\leq \text{IC}(\widehat{\tau}_{i_1}, \dots, \widehat{\tau}_{i_{m_0}}; \bm{\lambda}, \bm{\mu}, \omega_T) \nonumber \\
        &\leq \sum_{t=1}^T\|\epsilon_t\|_2^2 + Km_0T\xi_T(d_{\max}^{\star^2} + r_{\max}^{\star^{\frac{3}{2}}}) + m_0\omega_T,
    \end{align}
    which leads to:
    \begin{equation}\label{eq:53}
        (\widehat{m} - m_0)\omega_T \leq (Km_0 + c_2^\prime \widehat{m})T\xi_T(d_{\max}^{\star^2} + r_{\max}^{\star^{\frac{3}{2}}}).
    \end{equation}
    Assumption $m_0T\xi_T(d_{\max}^{\star^2} + r_{\max}^{\star^{\frac{3}{2}}})/\omega_T \to 0$ implies that $m_0 < \widehat{m} \leq m_0$, which is a \textit{contradiction}. Thus, we have established case (b) that $\mathbb{P}(\widehat{m} > m_0) \to 0$. Hence, we successfully prove that $\mathbb{P}(\widehat{m} = m_0) \to 1$.
    {
    The second part of Theorem \ref{thm:3} follows directly from the first part. By using similar arguments as in the proof of Theorem 1, it shows that for any estimated change point $\widehat{\tau}_j$, and corresponding true change point $\tau_j^\star$ such that:
    \begin{equation*}
        \sum_{t=1}^T\|\epsilon_t\|_2^2 + c_1\widetilde{v}|\widehat{\tau}_j - \tau^\star_j| - c_2m_0T\xi_T\left(d_{\max}^{\star^2} + r_{\max}^{\star^{\frac{3}{2}}}\right)
        \leq \sum_{t=1}^T\|\epsilon_t\|_2^2 + Km_0T\xi_T\left(d_{\max}^{\star^2} + r_{\max}^{\star^{\frac{3}{2}}}\right),
    \end{equation*}
    which implies that
    \begin{equation*}
        \max_{1\leq j\leq m_0}|\widehat{\tau}_j - \tau^\star_j| \leq Bm_0T\xi_T\frac{d_{\max}^{\star^2} + r_{\max}^{\star^{\frac{3}{2}}}}{\min_{1\leq j\leq m_0}\{v_{j,S}^2 + v_{j,L}^2\}},
    \end{equation*}
    where $B>0$ is a large enough constant.
    }
    
\end{proof}

\begin{proof}[Proof of Corollary \ref{cor:1}]
    This proof is similar to the proof of Theorem 4 in \cite{safikhani2017joint}. We first remove the $R$-radius neighborhoods for each estimated change points $\widehat{\tau}_i$, we thus obtain the stationary segments $I_i \overset{\text{def}}{=}[\widehat{\tau}_i-R, \widehat{\tau}_i+R]$ for $i=1,2,\dots,m_0$. Then, let $N_i$ be the length of the $i$-th segment, the two key aspects that need to be verified are (a) the restricted strong convexity condition; (b) the deviation bound condition. 
    
    For each estimated segment $I_i$, the result of Theorem \ref{thm:3} suggests that $N_i = \mathcal{O}(\Delta_T)$; therefore, sufficiently large sample sizes are available to verify the RSC condition and the deviation bounds in each segment. The verification is similar to Proposition 4.1 in \cite{basu2015regularized}. 
    
    Therefore, by using the tuning parameters selected and the result in Proposition 1(a) in \cite{basu2019low}, the final result follows.
\end{proof}

{
\begin{proof}[Proof of Corollary \ref{cor:2}]
    This proof is similar to the proof of Theorem \ref{thm:3} and Theorem \ref{thm:1}. By using the conclusion in Theorem \ref{thm:3}, we have $\mathbb{P}(\widehat{m} = m_0) \to 1$. Since we are using the similar procedure as singel change point detection proposed in Theorem \ref{thm:1}, the estimated change points $\widetilde{\tau}_j$ satisfy the similar results as the proof of Theorem \ref{thm:1}. Hence, for the $j$th refined change point:
    \begin{equation*}
        |\widetilde{\tau}_j - \tau^\star_j| \leq K_0\frac{d_{j}^\star\log(p\vee h) + r_{j}^\star(p\vee \log h)}{v_{j,S}^2 + v_{j,L}^2},
    \end{equation*}
    then combining all $\widehat{m}$ refined change points leads to the final result.
\end{proof}}

\begin{proof}[Proof of Corollary \ref{cor:3}]
    This proof is similar to the proof of Corollary 3 in \cite{negahban2012unified}. The main idea is to find an upper bound for the pseudo-sparsity level and an upper bound for the $\ell_1$ norm of the true model parameter for the complementary sparse support set $\mathcal{J}(\eta_j)$, which have been already derived in the proof of Lemma~\ref{lemma:7}. 
    
    The RSC condition can be verified as well for each estimated segment by using the same procedure as in the proof of Lemma \ref{lemma:2}. Applying Theorem 1 in \cite{negahban2012unified} to the specific segment leads to the result.
    
    Specifically, according to Theorem 1 in \cite{negahban2012unified}, with suitable selected tuning parameters, the error bound for the estimated model parameters is given by:
    \begin{equation*}
        \|\widehat{A}_j^w - A^\star\|_F^2 \leq c_1\lambda_j^{w^2}|\mathcal{J}(\eta_j)| + c_2\lambda_j^w\left(c_3\frac{\log p}{N_j}\|A^\star\|_{1,\mathcal{J}^c(\eta_j)}^2 + 4\|A^\star\|_{1,\mathcal{J}(\eta_j)^c}\right);
    \end{equation*}
    therefore, by substituting the results of (20), we obtain
    \begin{equation*}
        \|\widehat{A}_j^w - A^\star\|_F^2 \leq c_1\lambda_j^{w^{2-q}}R_q + c_2\left(\lambda_j^{w^{2-q}}R_q\right)^2\frac{\log p}{\lambda_j^wN_j} \leq C_0R_q\left(\frac{\log p}{N_j}\right)^{1-\frac{q}{2}},
    \end{equation*}
    where $c_1$, $c_2$, $c_3$, and $C_0$ are universal positive constants.
\end{proof}

\begin{proof}[Proof of Proposition \ref{prop:1}]
    The result can be directly established by using the definition of the Hausdorff distance and the rolling-window mechanism provided in Algorithm~\ref{algo:1}. Based on Assumption H4, the number of candidate change points $\widetilde{m}$ obtained by the rolling-window strategy satisfies $\widetilde{m} > T/\Delta_T > m_0$. Therefore, we get that $\mathbb{P}(\widetilde{m} \geq m_0) = 1$. 
    
    Based on the result of Theorem 1, for any true change point $\tau^\star_j$, once the window includes $\tau^\star_j$, there exists an estimated change point $\widehat{\tau}_i$ satisfying with high probability:
    {
    \begin{equation*}
        |\widehat{\tau}_i-\tau^\star_j| \leq K\frac{d_{\max}^\star\log(p\vee T) + r_{\max}^\star(p \vee \log T)}{v_{j,S}^2+v_{j,L}^2}
    \end{equation*}}
    for some large enough positive constant $K$. Combining all $m_0$ change points, we obtain the final result.
\end{proof}

\begin{proof}[Proof of Proposition \ref{prop:2}]
    Suppose that $A = L+S$ is one of the transition matrices in model (1). Further, suppose  $A$ is in the given $\ell_q$-ball and the support set of the sparse component $S$ is denoted by $\mathcal{I}$, and $|\mathcal{I}| = d^\star$. We can then get:
    \begin{equation}\label{eq:54}
        \begin{aligned}
            \sum_{i=1}^p\sum_{j=1}^p|A_{ij}|^q
            &= \sum_{(i,j)\in \mathcal{I}}|L_{ij}+S_{ij}|^q + \sum_{(i,j)\in \mathcal{I}^c}|L_{ij}+S_{ij}|^q \\
            &= \sum_{(i,j) \in \mathcal{I}}|L_{ij}+S_{ij}|^q + \sum_{(i,j)\in \mathcal{I}^c}|L_{ij}|^q \overset{\text{def}}{=} J_1 + J_2.
        \end{aligned}
    \end{equation}
    
    First, a Singular Value Decomposition of matrix $L$ yields: $L = UDV^\prime$, where $U = [u_1, \dots, u_p] \in \mathbb{R}^{p\times p}$, $V = [v_1, \dots, v_p] \in \mathbb{R}^{p \times p}$ are orthonormal matrices (i.e., for any $u_i$ or $v_j$, $\|u_i\| = \|v_j\| = 1$), and $D = \text{diag}(\sigma_1, \dots, \sigma_r, 0, \dots, 0)$, where $\sigma_k$ is the $k$-th largest singular value of $L$, and $r$ is the rank of $L$. We can then obtain:
    \begin{equation}\label{eq:55}
        \begin{aligned}
            J_2 &\leq \sum_{(i,j)\in \mathcal{I}^c}\left|\sum_{k=1}^r\sigma_k u_{ik}v_{jk} \right|^q \leq \sum_{(i,j)\in \mathcal{I}^c}\left|\left(\sum_{k=1}^r\sigma_k u_{ik}^2\right)^\frac{1}{2} \left(\sum_{k=1}^r \sigma_k v_{jk}^2\right)^\frac{1}{2} \right|^q \\
            &\leq \sum_{(i,j)\in \mathcal{I}^c}\left| \sigma_1 \right|^q = |\sigma_1|^q(p^2 - d^\star).
        \end{aligned}
    \end{equation}
    
    Next, due to the fact that $|L_{ij} + S_{ij}|^q \leq |L_{ij}|^q + |S_{ij}|^q$, we can obtain the following result
    \begin{equation}\label{eq:56}
        J_1 \leq \sum_{(i,j)\in \mathcal{I}}|L_{ij}|^q + \sum_{(i,j)\in \mathcal{I}}|S_{ij}|^q \leq d^\star\left\{ \left(\frac{\alpha_L}{p}\right)^q + M_S^q\right\}.
    \end{equation}
    
    Combining the results \eqref{eq:55} and \eqref{eq:56} leads to the final result:
    \begin{equation}\label{eq:57}
        \sum_{i=1}^p\sum_{j=1}^p |A_{ij}|^q \leq d^\star\left(\left(\frac{\alpha_L}{p}\right)^q + M_S^q\right) + (p^2 - d^\star)|\sigma_1|^q.
    \end{equation}
\end{proof}

\begin{proof}[Proof of Proposition \ref{prop:3}]
    Let the transition matrices $A_1^\star$ and $A_2^\star \in \mathbb{B}_q(R_q)$, with a fixed $q \in (0, 1]$, and $R_q$ satisfying the condition proposed in Proposition 1. Also, assume that the associated true change point satisfies $\tau^\star \in [1,T)$. To establish the result, we follow a similar strategy as in the proof of Theorem 1. First, we establish:
    \begin{equation}\label{eq:58}
        \ell(\tau^\star) \leq \sum_{t=1}^{T-1}\|\epsilon_t\|_2^2 + c_0 T^{\frac{q}{2}}R_q\left(\log p\right)^{1-\frac{q}{2}}.
    \end{equation}
    Split the objective function $\ell(t)$ as follows:
    \begin{equation*}
        \ell(\tau^\star) = \sum_{t=1}^{\tau^\star-1}\|X_t - \widehat{A}_{1,\tau^\star}X_{t-1}\|_2^2 + \sum_{t=\tau^\star}^{T-1}\|X_t - \widehat{A}_{2,\tau^\star}X_{t-1}\|_2^2 \equiv I_1 + I_2.
    \end{equation*}
    Then, based on the definition $\ell_q$ norm, we are able to obtain that:
    \begin{align}\label{eq:59}
        I_1
        &= \sum_{t=1}^{\tau^\star-1}\|X_t - \widehat{A}_{1,\tau^\star}X_{t-1}\|_2^2  \nonumber \\
        &\leq \sum_{t=1}^{\tau^\star-1}\|\epsilon_t\|_2^2 + c_1|\tau^\star-1|\|\widehat{A}_{1,\tau^\star} - A_1^\star\|_2^2 + c_1^\prime\sqrt{|\tau^\star-1|(\log p+\log(\tau^\star-1))}\|\widehat{A}_{1,\tau^\star} - A_1^\star\|_1 \nonumber \\
        &\leq \sum_{t=1}^{\tau^\star-1}\|\epsilon_t\|_2^2 + c_1|\tau^\star-1|\|\widehat{A}_{1,\tau^\star} - A_1^\star\|_2\left( \|\widehat{A}_{1,\tau^\star} - A_1^\star\|_2 + \frac{c_1^\prime}{c_1}\sqrt{R_q}\left(\frac{\log p+\log(\tau^\star-1)}{\tau^\star-1}\right)^{\frac{1}{2}(1-\frac{q}{2})} \right) \nonumber \\
        &+ c_1^\prime R_q\left(\frac{\log p+\log(\tau^\star-1)}{\tau^\star-1}\right)^{1-\frac{q}{2}} \nonumber \\
        &\leq \sum_{t=1}^{\tau^\star-1}\|\epsilon_t\|_2^2 + c_1|\tau^\star - 1|\|\widehat{A}_{1,\tau^\star} - A_1^\star\|_2^2 + c_1^\prime R_q\left(\frac{\log p + \log(\tau^\star-1)}{\tau^\star-1}\right)^{1-\frac{q}{2}} \nonumber \\
        &\leq \sum_{t=1}^{\tau^\star-1}\|\epsilon_t\|_2^2 + c_1 T^{\frac{q}{2}}R_q\left(\log p+\log T\right)^{1-\frac{q}{2}}.
    \end{align}
   Analogously, we can get for $I_2$:
    \begin{equation}\label{eq:60}
        I_2 \leq \sum_{t=\tau^\star}^{T-1}\|\epsilon_t\|_2^2 + c_2 T^{\frac{q}{2}}R_q\left(\log p + \log T\right)^{1-\frac{q}{2}}.
    \end{equation}
    Combining \eqref{eq:59} and \eqref{eq:60} leads to the result in \eqref{eq:57}. Next, we prove that for any fixed time point $\tau \in \mathcal{T}$, there exists some large enough constants $c_1, c_2>0$, together with jump size $v_A \overset{\text{def}}{=} \|A_2^\star - A_1^\star\|_2$ such that the lower bound for $\ell(\tau)$ is given by:
    \begin{equation}\label{eq:61}
        \ell(\tau) \geq \sum_{t=1}^{T-1}\|\epsilon_t\|_2^2 - c_1 T^{\frac{q}{2}}R_q(\log p + \log T)^{1-\frac{q}{2}} + c_2v_A^2|\tau-\tau^\star|. 
    \end{equation}
    Consider the interval $[1,\tau)$ and $[\tau,T)$ separately. Notice that, in this situation, we might have a misspecified model in that interval. Specifically, let us assume $\tau > \tau^\star$; then, the interval with a misspecified model corresponds to $[\tau^\star, \tau)$. We then have:
    \begin{equation*}
        \ell(\tau) = \sum_{t=1}^{\tau-1}\|X_t - \widehat{A}_{1,\tau}X_{t-1}\|_2^2 + \sum_{t=\tau}^{T-1}\|X_t - \widehat{A}_{2,\tau}X_{t-1}\|_2^2 \equiv I_1 + I_2,
    \end{equation*}
   and for $I_1$:
   {
   \begingroup
    \allowdisplaybreaks
    \begin{align}\label{eq:62}
        I_1 &= \sum_{t=1}^{\tau^\star-1}\|X_t - \widehat{A}_{1,\tau}X_{t-1}\|_2^2 + \sum_{t=\tau^\star}^{\tau-1}\|X_t - \widehat{A}_{1,\tau}X_{t-1}\|_2^2 \nonumber \\
        &\geq \sum_{t=1}^{\tau^\star-1}\|\epsilon_t\|_2^2 + c|\tau^\star-1|\|\widehat{A}_{1,\tau} - A_1^\star\|_2^2 - c^\prime \sqrt{|\tau^\star-1|(\log p+\log(\tau-1))}\|\widehat{A}_{1,\tau} - A_1^\star\|_1 \nonumber \\
        &+ \sum_{t=\tau^\star}^{\tau-1}\|\epsilon_t\|_2^2 + \widetilde{c}|\tau-\tau^\star|\|\widehat{A}_{1,\tau} - A_2^\star\|_2^2 - \widetilde{c}^\prime \sqrt{|\tau-\tau^\star|(\log p+\log(\tau-1))}\|\widehat{A}_{1,\tau} - A_2^\star\|_1 \nonumber \\
        &- \widetilde{c}^{\prime\prime}|\tau-\tau^\star|\frac{M_SR_q}{\tau-1}\left(\frac{\log p + \log(\tau-1)}{\tau-1}\right)^{-\frac{q}{2}}\|\widehat{A}_1 - A_2^\star\|_1 \nonumber \\
        &\overset{\text{(i)}}{\geq} \sum_{t=1}^{\tau-1}\|\epsilon_t\|_2^2 + c|\tau^\star-1|\|\widehat{A}_{1,\tau} - A_1^\star\|_2\left( \|\widehat{A}_{1,\tau} - A_1^\star\|_2 - \frac{c^\prime}{c}\sqrt{R_q}\left(\frac{\log p + \log(\tau-1)}{\tau-1}\right)^{\frac{1}{2}(1-\frac{q}{2})} \right) \nonumber \\
        &+ \widetilde{c}|\tau-\tau^\star|\|\widehat{A}_{1,\tau} - A_2^\star\|_2\left( \|\widehat{A}_{1,\tau} - A_2^\star\|_2 - \frac{\widetilde{c}^\prime}{\widetilde{c}}\sqrt{R_q}\left(\frac{\log p+\log(\tau-1)}{\tau-1}\right)^{\frac{1}{2}(1-\frac{q}{2})}\right) \nonumber \\
        &- 4c R_q\left(\frac{\log p+\log(\tau-1)}{\tau-1}\right)^{1-\frac{q}{2}} - 4\widetilde{c}R_q\left(\frac{\log p+\log(\tau-1)}{\tau-1}\right)^{1-\frac{q}{2}} \nonumber \\
        &\overset{\text{(ii)}}{\geq} \sum_{t=1}^{\tau-1}\|\epsilon_t\|_2^2 - c_1 T^{\frac{q}{2}}R_q(\log p+\log T)^{1-\frac{q}{2}} + c_2v_A^2|\tau-\tau^\star|.
    \end{align}
    \endgroup}
    (i) holds due to Assumption W2 on the search domain $\mathcal{T}^w$; (ii) holds due to assuming that $\|\widehat{A}_1 - A_2^\star\|_2 \geq v_A/4 > 0$.
    
    Analogously, we can derive a lower bound for $I_2$:
    \begin{equation}\label{eq:63}
        I_2 \geq \sum_{t=\tau}^{T-1}\|\epsilon_t\|_2^2 - c_1 T^{\frac{q}{2}}R_q(\log p + \log T)^{1-\frac{q}{2}}.
    \end{equation}
    Hence, we proved the conclusion in \eqref{eq:61}. Next, by using \eqref{eq:62} and \eqref{eq:63}, we obtain that with high probability the following holds:
    {
    \begin{equation}\label{eq:64}
        \begin{aligned}
            &\sum_{t=1}^{T-1}\|\epsilon_t\|_2^2 - c_1 h^{\frac{q}{2}}R_q(\log p+\log T)^{1-\frac{q}{2}} + c_2v_A^2|\widehat{\tau}-\tau^\star| \\
            &\leq \ell(\widehat{\tau}) \leq \ell(t_j^\star) \leq \sum_{t=1}^{T-1}\|\epsilon_t\|_2^2 + c_0 T^{\frac{q}{2}}R_q\left(\log p+\log T\right)^{1-\frac{q}{2}}.
        \end{aligned}
    \end{equation}
    Thus, the error bound for $|\widehat{\tau} - \tau^\star|$ is given by
    \begin{equation}\label{eq:65}
        |\widehat{\tau} - \tau^\star| \leq \frac{c_0+c_1}{c_2} \frac{T^{\frac{q}{2}}R_q\left(\log(p\vee T)\right)^{1-\frac{q}{2}}}{v_A^2},
    \end{equation}
    for some constants $c_0, c_1, c_2>0$ and $v_A = \|A_2^\star - A_1^\star\|_2$. }
\end{proof}

\begin{proof}[Proof of Proposition \ref{prop:4}]
    The proof is analogous to that of Proposition 1. Based on the rolling-window mechanism and the result of Proposition 2, we can verify the result. In this case, we just need to replace the sample size $T$ by the window-size $h$. 
\end{proof}

\begin{proof}[Proof of Proposition \ref{prop:5}]
    This proof follows in a similar manner to that of Theorem \ref{thm:3}. By following the arguments in Lemma \ref{lemma:8}, we firstly verify the upper bound of $\mathcal{L}_T^w(\widehat{\tau}_1, \cdots, \widehat{\tau}_{m_0}; \bm{\lambda}^w)$ with respect to the set of estimated change points $(\widehat{\tau}_1,\cdots, \widehat{\tau}_{m_0})$. The latter satisfy $\max_{1\leq i\leq m_0}|\widehat{\tau}_i - \tau^\star_i| \leq T\xi_T^w$. Similar to the proof of Theorem 2, we obtain that
    \begin{equation*}
        \mathcal{L}_T(\widehat{\tau}_1,\dots,\widehat{\tau}_{m_0}; \bm{\lambda}^w) \leq \sum_{t=1}^T\|\epsilon_t\|_2^2 + K^\prime m_0T\xi_T^wR_q^2\left(\frac{\log(p\vee T)}{T}\right)^{-q},
    \end{equation*}
    where $K^\prime > 0$ is a large enough constant and $\Delta_T = \min_{1\leq i \leq m_0-1}|\tau_i^\star - \tau_{i+1}^\star|$.
    
    Next, we establish: (a) $\mathbb{P}(\widehat{m} < m_0) \to 0$; (b) $\mathbb{P}(\widehat{m} > m_0) \to 0$, respectively. For (a), we have that: denote $\widetilde{v}_A = \min_{1\leq j\leq m_0} v_{j,A}$, where $v_{j,A} \overset{\text{def}}{=} \|A_{j+1}^\star - A_j^\star\|_2$, then
    \begin{equation*}
        \begin{aligned}
            &\text{IC}^w(\widehat{\tau}_1,\dots,\widehat{\tau}_{\widehat{m}}; \bm{\lambda}^w, \omega_T^w) = \mathcal{L}_T^w(\widehat{\tau}_1, \dots, \widehat{\tau}_{\widehat{m}}; \bm{\lambda}^w) + \widehat{m}\omega_T^w \\
            > &\sum_{t=1}^T\|\epsilon_t\|_2^2 + c_1\widetilde{v}_A^{2}\Delta_T - c_2\widehat{m}T\xi_T^wR_q^2\left(\frac{\log(p\vee T)}{T}\right)^{-q} + \widehat{m}\omega_T^w \\
            \geq &\mathcal{L}_T(\widehat{\tau}_1,\dots,\widehat{\tau}_{m_0}; \bm{\lambda}^w) + m_0\omega_T^w + c_1\widetilde{v}_A^2\Delta_T - c_2m_0T\xi_T^wR_q^2\left(\frac{\log(p\vee T)}{T}\right)^{-q} - (m_0 - \widehat{m})\omega_T^w \\
            \geq &\mathcal{L}_T(\widehat{\tau}_1,\dots,\widehat{\tau}_{m_0}; \bm{\lambda}^w) + m_0\omega_T^w,
        \end{aligned}
    \end{equation*}
    which implies that the set of estimated change points $(\widehat{\tau}_1,\cdots,\widehat{\tau}_{\widehat{m}})$ are not the optimal solution for minimizing $\text{IC}^w$. Hence, we conclude that $\mathbb{P}(\widehat{m} < m_0) \to 0$. To prove (b), suppose the set of estimated change points $(\widehat{\tau}_1,\cdots,\widehat{\tau}_{\widehat{m}})$ satisfy $\widehat{m} > m_0$; hence, we similarly obtain the following result:
    \begin{equation*}
        \mathcal{L}_T^w(\widehat{\tau}_1,\dots,\widehat{\tau}_{\widehat{m}}; \bm{\lambda}^w) \geq \sum_{t=1}^T\|\epsilon_t\|_2^2 - c_2^\prime\widehat{m}T\xi_T^wR_q^2\left(\frac{\log(p\vee T)}{T}\right)^{-q}.
    \end{equation*}
    Choose a subset of $(\widehat{\tau}_1,\cdots, \widehat{\tau}_{\widehat{m}})$ with $m_0$ elements, such that $\max_{1\leq i\leq m_0}|\widehat{\tau}_{k_i} - \tau^\star_i| \leq T\xi_T^w$. We then have:
    \begin{equation*}
        \begin{aligned}
            &\sum_{t=1}^T\|\epsilon_t\|_2^2 - c_2^\prime\widehat{m}T\xi_T^wR_q^2\left(\frac{\log(p\vee T)}{T}\right)^{-q} + \widehat{m}\omega_T^w \\
            \leq &\text{IC}^w(\widehat{\tau}_1,\dots,\widehat{\tau}_{\widehat{m}}; \bm{\lambda}^w, \omega_T^w) \leq \text{IC}^w(\widehat{\tau}_{k_1},\dots,\widehat{\tau}_{k_{m_0}}; \bm{\lambda}^w, \omega_T^w) \\
            \leq &\sum_{t=1}^T\|\epsilon_t\|_2^2 + K^\prime m_0T\xi_T^wR_q^2\left(\frac{\log(p\vee T)}{T}\right)^{-q},
        \end{aligned}
    \end{equation*}
    which implies that $m_0 < \widehat{m} \leq m_0$, which is a contradiction. Therefore, we have $\mathbb{P}(\widehat{m} = m_0) \to 1$. The error bound is established by similar arguments as in Lemma~\ref{lemma:6}.
\end{proof}


Under the additional Assumptions (W5a)--(W5c), Proposition \ref{prop:6} can be verified by the following proof.
\begin{proof}[Proof of Proposition \ref{prop:6}]
    According to Proposition \ref{prop:4} and Proposition \ref{prop:1}, we separately obtain that:
    \begin{equation*}
        d_H(\widetilde{\mathcal{S}}_w, \mathcal{S}^\star) = K^\prime \frac{h^{\frac{q}{2}}R_q\left(\log (p\vee h)\right)^{1-\frac{q}{2}}}{\min_{1\leq j\leq m_0}v_{j,A}^2},
    \end{equation*}
    and
    \begin{equation*}
        d_H(\widetilde{\mathcal{S}}, \mathcal{S}^\star) = K\frac{d_{\max}^\star\log(p\vee h) + r_{\max}^\star(p\vee \log h)}{\min_{1\leq j\leq m_0}\{v_{j,S}^2+v_{j,L}^2\}},
    \end{equation*}
    for some constants $K, K^\prime > 0$. 
    {
    Since we have $v_{j,A} = \|A_{j+1}^\star - A_j^\star\|_2 = \|(S_{j+1}^\star-S_j^\star)+(L_{j+1}^\star - L_j^\star)\|_2$ for $j=1,2,\dots,m_0$, hence, we have $2v_{j,A}^2 \geq (v_{j,S}^2 + v_{j,L}^2)$. Then, suppose $p \succsim h$ we establish the left-hand side:
    \begin{align}
        \label{eq:66}
        \frac{d_H(\widetilde{\mathcal{S}}_w, \mathcal{S}^\star)}{d_H(\widetilde{\mathcal{S}}, \mathcal{S}^\star)} 
        &\geq \frac{K^\prime}{K}\frac{\min_j\{v_{j,S}^2+v_{j,L}^2\}}{\min_j v_{j,A}^2}\frac{h^{\frac{q}{2}}R_q(\log(p\vee h))^{1-\frac{q}{2}}}{d_{\max}^\star\log(p\vee h) + r_{\max}^\star(p\vee \log h)} \nonumber \\
        &\geq \frac{K^\prime}{2K} \frac{c_0^{\frac{q}{2}}(\log T)^{\frac{q}{2}}R_q(\log p)^{1-\frac{q}{2}}}{\left(d_{\max}^\star\log p + r_{\max}^\star p\right)^{1-\frac{q}{2}}} \nonumber \\
        &\geq \frac{K^\prime}{2K} \frac{c_0^\prime (\log T)^{\frac{q}{2}}\left\{d_{\max}^\star\left(\frac{\alpha_L^q}{p^q} + M_S^q\right) + (p^2 - d_{\max}^\star)|\sigma_1|^q \right\}(\log p)^{1-\frac{q}{2}} }{\left(d_{\max}^\star\log p + r_{\max}^\star p\right)^{1-\frac{q}{2}}} \nonumber \\
        &\geq \frac{c_0^\prime K}{2K^\prime} \frac{ d_{\max}^\star\left(\frac{\alpha_L^q}{p^q} + M_S^q\right) + (p^2 - d_{\max}^\star)|\sigma_1|^q }{\left( d_{\max}^\star + r_{\max}^\star\frac{p}{\log p} \right)^{1-q}} \\
        &\geq \frac{c_0^\prime K}{2K^\prime} \frac{d_{\max}^\star\frac{\alpha_L^q}{p^q} + (p^2 - d_{\max}^\star)(\frac{\alpha_L}{p})^q}{\left(d_{\max}^\star + r_{\max}^\star\frac{p}{\log p} \right)^{1-q}} \geq \frac{c_0^{\prime\prime} K}{2K^\prime} \frac{p^{2-q}}{\left(d_{\max}^\star + r_{\max}^\star\frac{p}{\log p} \right)^{1-q}} \geq 1 \nonumber
    \end{align}}
    
    {
    Next, we determine an upper bound for the ratio of estimation errors. 
    \begin{equation}
        \label{eq:67}
        \begin{aligned}
            \frac{d_H(\widetilde{\mathcal{S}}_w, \mathcal{S}^\star)}{d_H(\widetilde{\mathcal{S}}, \mathcal{S}^\star)} 
            &= \frac{K}{K^\prime}\frac{\min_j\{v_{j,S}^2+v_{j,L}^2\}}{\min_j v_{j,A}^2}\frac{h^{\frac{q}{2}}R_q(\log p)^{1-\frac{q}{2}}}{d_{\max}^\star\log p + r_{\max}^\star p} \\
            &= \frac{K}{K^\prime}\frac{\min_j\{v_{j,S}^2+v_{j,L}^2\}}{\min_j v_{j,A}^2}\frac{c_0^{\frac{q}{2}}(\log T)^{\frac{q}{2}}R_q(\log p)^{1-\frac{q}{2}} }{\left( d_{\max}^\star\log p + r_{\max}^\star p \right)^{1-\frac{q}{2}}} \\
            &\overset{\text{(i)}}{\leq} \frac{K}{K^\prime} \frac{\min_j\{v_{j,S}^2+v_{j,L}^2\}}{\min_j v_{j,A}^2}\frac{c_0^\prime (d_{\max}^\star + r^\star_{\max})^{1-\frac{q}{2}}p^{2-q}\max\left\{\alpha_L, M_S\right\}^q }{\left(d_{\max}^\star + r_{\max}^\star \right)^{1-\frac{q}{2}}} \\
            &= \frac{\min_j\{v_{j,S}^2+v_{j,L}^2\}}{\min_j v_{j,A}^2}c_0^{\prime\prime}p^{2-q}(\log T)^{\frac{q}{2}},
        \end{aligned}
    \end{equation}}
    where $c_0^\prime$ and $c_0^{\prime\prime}$ are some large enough universal constants, and (i) holds due to Assumption (W5c). Since we only consider the case that the information ratios $0 < \gamma_j < p$, which indicates that the sparse components are dominating as well as the jump size of $A_j^\star$'s are lower bounded by a small enough constant, then we have the last equation in \eqref{eq:67} is bounded by $c_0^{\prime\prime}p^{2-q}(\log T)^{\frac{q}{2}}$. Therefore, combining the results in \eqref{eq:66} and \eqref{eq:67} leads to the desired outcome. 
\end{proof}

\newpage
\section{Additional Numerical Experiments}\label{appendix:F}
Table \ref{tab:1} summarizes all the parameter settings for all scenarios introduced in Section 5.1. 
\begin{table}[!ht]
    \spacingset{1}
    \centering
    \caption{Model parameters for different settings considered.}
    \label{tab:1}
    \resizebox{0.64\textwidth}{!}{%
    \begin{tabular}{c|c|c|c|c|c|c|c}
        \hline\hline
          & $p$ & $T$ & $\tau^\star / T$ & $(r^\star_1, r^\star_2)$ & $v_L$ & $v_S$ & $(\gamma_1, \gamma_2)$  \\
        \hline
        A.1 & 20 & 300 & $0.500$ & $(1,3)$ & 0.10 & 1.5 & $(0.25, 0.25)$ \\
        A.2 & 20 & 300 & $0.500$ & $(1,3)$ & 0.25 & 1.5 & $(0.25, 0.25)$ \\
        A.3 & 20 & 300 & $0.500$ & $(1,3)$ & 0.50 & 1.5 & $(0.25, 0.25)$ \\ 
        \hline
        B.1 & 20 & 300 & $0.500$ & $(1,2)$ & 0.25 & 2.0 & $(2.0, 2.0)$ \\
        B.2 & 20 & 300 & $0.500$ & $(1,2)$ & 0.50 & 2.0 & $(2.0, 2.0)$ \\
        B.3 & 20 & 300 & $0.500$ & $(1,2)$ & 0.75 & 2.0 & $(2.0, 2.0)$ \\
        \hline
        C.1 & 20 & 300 & $0.500$ & $(1,2)$ & 0.25 & 2.0 & $(1.75, 2.0)$ \\
        C.2 & 20 & 300 & $0.500$ & $(1,2)$ & 0.25 & 2.0 & $(1.25, 2.0)$ \\
        C.3 & 20 & 300 & $0.500$ & $(1,2)$ & 0.25 & 2.0 & $(1.0, 2.0)$ \\
        C.4 & 20 & 300 & $0.500$ & $(1,2)$ & 0.25 & 2.0 & $(0.5, 2.0)$ \\
        \hline
        D.1 & 20 & 300 & $0.500$ & $(1,2)$ & 3.0 & 0.75 & $(1.5, 1.5)$ \\
        D.2 & 20 & 300 & $0.500$ & $(1,2)$ & 3.5 & 0.75 & $(1.5, 1.5)$ \\
        D.3 & 20 & 300 & $0.500$ & $(1,2)$ & 4.0 & 0.75 & $(1.5, 1.5)$ \\
        \hline
        E.1 & 20 & 300 & $0.500$ & $(1,3)$ & 2.5 & 0.15 & $(0.25, 0.25)$ \\
        E.2 & 20 & 300 & $0.500$ & $(1,3)$ & 3.0 & 0.15 & $(0.25, 0.25)$ \\
        E.3 & 20 & 300 & $0.500$ & $(1,3)$ & 4.5 & 0.15 & $(0.25, 0.25)$ \\
        \hline
        F.1 & 20 & 300 & $0.500$ & $(1,2)$ & 2.5 & 0.25 & $(0.5, 0.45)$ \\
        F.2 & 20 & 300 & $0.500$ & $(1,2)$ & 2.5 & 0.25 & $(0.5, 0.75)$ \\
        F.3 & 20 & 300 & $0.500$ & $(1,2)$ & 2.5 & 0.25 & $(0.5, 0.95)$ \\
       \hline\hline
    \end{tabular}}
\end{table}
{Table~\ref{tab:2} presents the extra settings for scenario G. Precisely, we investigate the scenario with the high dimensional model parameters:
\begin{itemize}
    \item[(G)] {In this setting, we investigate high-dimensional scenarios with $p=80$, $T=200$ and a single change point. Other model parameters, including $\gamma_j$ and jump sizes $v_S$ and $v_L$ are similar to those in setting A. Note that in this scenario the number of effective parameters (i.e., $d_j^\star\log p + r^\star_jp$) is in the range $[100, 300]$, thus corresponding to a high dimensional setting. The specific settings are listed in Table~\ref{tab:result-G}.} 
\end{itemize}
The results of those settings are presented in Table~\ref{tab:3}. We can easily observe that the accuracy of the estimated change points, as well as the transition matrices are satisfactory under the high dimensional setting. The estimated model parameters are also highly satisfactory based on the sensitivity and specificity metrics. }
\begin{table}[]
    \spacingset{1}
    \centering
    \caption{{Model parameters for the high dimensional scenario G.}}
    \label{tab:result-G}
    \resizebox{0.64\textwidth}{!}{%
    \begin{tabular}{c|c|c|c|c|c|c|c}
        \hline\hline
          & $p$ & $T$ & $\tau^\star / T$ & $(r^\star_1, r^\star_2)$ & $v_L$ & $v_S$ & $(\gamma_1, \gamma_2)$  \\
        \hline
        G.1 & 80 & 200 & 0.500 & (1,3) & 0.20 & 0.75 & (0.25, 0.25) \\
        G.2 & 80 & 200 & 0.500 & (1,3) & 0.40 & 0.75 & (0.25, 0.25) \\
        G.3 & 80 & 200 & 0.200 & (1,3) & 0.20 & 0.75 & (0.25, 0.25) \\
        G.4 & 80 & 200 & 0.800 & (1,3) & 0.20 & 0.75 & (0.25, 0.25) \\
        G.5 & 80 & 200 & 0.500 & (3,1) & 0.20 & 0.75 & (0.25, 0.25) \\
        G.6 & 80 & 200 & 0.500 & (3,3) & 0.20 & 0.75 & (0.25, 0.25) \\
        G.7 & 80 & 200 & 0.500 & (5,3) & 0.20 & 0.75 & (0.25, 0.25) \\
        G.8 & 50 & 200 & 0.500 & (1,3) & 0.45 & 0.40 & (0.75, 0.75) \\
       \hline\hline
    \end{tabular}
    }
\end{table}
\begin{table}[!ht]
    \spacingset{1}
    \centering
    \caption{ {Performance of the full L+S model under the simulation setting G.}}
    \label{tab:3}
    \resizebox{0.9\textwidth}{!}{%
    \begin{tabular}{c|c|c|c|c|c|c|c}
        \hline\hline
          & mean & sd & $\widehat{r}_1$ & $\widehat{r}_2$ & SEN & SPC & Total RE/ Sparse RE / Low-rank RE  \\
        \hline
        G.1 & 0.528 & 0.077 & $1.000$ & $2.800$ & $(0.946, 0.959)$ & $(0.925, 0.907)$ & $(0.635, 0.655)/(0.639, 0.762)/(0.773, 0.781)$ \\
        G.2 & 0.499 & 0.003 & $1.000$ & $3.300$ & $(0.984, 0.924)$ & $(0.934, 0.952)$ & $(0.599, 0.684)/(0.601, 0.770)/(0.686,0.808)$ \\
        G.3 & 0.203 & 0.016 & $1.051$ & $2.773$ & $(0.887, 0.945)$ & $(0.985, 0.957)$ & $(0.782, 0.635)/(0.692, 0.503)/(0.897, 0.704)$ \\
        G.4 & 0.822 & 0.035 & $1.000$ & $1.885$ & $(0.967, 0.918)$ & $(0.955, 0.932)$ & $(0.603, 0.745)/(0.531, 0.691)/(0.688,0.880)$ \\
        G.5 & 0.515 & 0.056 & 2.750 & 1.050 & $(0.965, 0.955)$ & $(0.928, 0.987)$ & $(0.606, 0.775)/(0.608, 0.742)/(0.776, 0.763)$ \\
        G.6 & 0.514 & 0.056 & 3.000 & 3.250 & $(0.963, 0.977)$ & $(0.926, 0.979)$ & $(0.607, 0.853)/(0.608, 0.771)/(0.776, 0.820)$ \\
        G.7 & 0.502 & 0.006 & 5.900 & 4.025 & $(0.991, 0.987)$ & $(0.928, 0.939)$ & $(0.575, 0.850)/(0.567, 0.669)/(0.826, 0.969)$ \\
        G.8 & 0.539 & 0.031 & 0.855 & 2.335 & $(0.925, 0.902)$ & $(0.865, 0.899)$ & $(1.002, 0.975)/(1.200, 1.004)/(0.827, 0.927)$ \\
      \hline\hline
    \end{tabular}}
\end{table}

\subsection{Performance of the Surrogate Weakly Sparse Model for the Detection of a Single Change Point}
Table \ref{tab:4} summarizes the results of the surrogate weakly sparse model. Analogously to the results for the low-rank plus sparse model, under settings A and D the estimates of the change point are highly accurate. In settings B and E, the surrogate model performs worse than the full model, since the difference in the norm of the transitions matrices is rather small. Specifically, under setting B, the low-rank components contribute most of the ``signal", even though their changes before and after the change point are rather small, thus effectively not satisfying Assumption W1. A similar reasoning justifies the rather poor performance of the surrogate model under setting E. In settings C and F, we investigate the case of different information ratios, covered in the second part of assumption W1. It can be seen that performance gradually improves by enlarging the differences between the information ratios. Estimation of the transition matrices is analogous to that under the full model; when the sparse component contributes most of the ``signal" as in settings A, E and F, the relative error of is good and comparable to that of the full model. On the other hand, the relative error becomes worse than that obtained by the full model.
\begin{table}[ht]
    \spacingset{1}
    \centering
    \caption{Performance of the surrogate model under different simulation settings.}
    \label{tab:4}
    \resizebox{0.8\textwidth}{!}{%
    \begin{tabular}{c|c|c|c|c|c|c|c}
        \hline\hline
          & mean & sd & RE & & mean & sd & RE \\
        \hline
        A.1 & 0.498 & 0.002 & $(0.188, 0.201)$ & D.1 & 0.502 & 0.025 & $(0.766, 0.829)$ \\
        A.2 & 0.498 & 0.002 & $(0.190, 0.200)$ & D.2 & 0.498 & 0.012 & $(0.763, 0.743)$ \\
        A.3 & 0.498 & 0.002 & $(0.190, 0.206)$ & D.3 & 0.498 & 0.011 & $(0.762, 0.691)$ \\ 
        \hline
        B.1 & 0.538 & 0.125 & $(0.788, 0.814)$ & E.1 & 0.525 & 0.154 & $(0.199, 0.234)$ \\
        B.2 & 0.538 & 0.125 & $(0.787, 0.815)$ & E.2 & 0.510 & 0.104 & $(0.200, 0.248)$ \\
        B.3 & 0.539 & 0.124 & $(0.787, 0.814)$ & E.3 & 0.518 & 0.060 & $(0.198, 0.285)$ \\
        \hline
        C.1 & 0.550 & 0.112 & $(0.765, 0.814)$ & F.1 & 0.456 & 0.200 & $(0.431, 0.352)$ \\
        C.2 & 0.515 & 0.076 & $(0.737, 0.798)$ & F.2 & 0.470 & 0.098 & $(0.411, 0.458)$ \\
        C.3 & 0.494 & 0.041 & $(0.682, 0.783)$ & F.3 & 0.475 & 0.095 & $(0.415, 0.571)$ \\
        C.4 & 0.501 & 0.008 & $(0.370, 0.775)$ & & & & \\
       \hline\hline
    \end{tabular}}
\end{table}

\subsection{Performance of Multiple Change Points Detection}
Table \ref{tab:multi-summary} below, summarizes the parameter settings for each scenario considered. 
\begin{table}[!ht]
    \spacingset{1}
    \centering
    \caption{Model parameters under different multiple change points scenario settings.}
    \label{tab:multi-summary}
    \resizebox{\textwidth}{!}{%
    \begin{tabular}{c|c|c|c|c|c|c|c}
        \hline\hline
          & $p$ & $T$ & $\tau_j^\star / T$ & ranks & $\Delta L_j$ & $\Delta S_j$ & $\gamma_j$  \\
        \hline
        L.1 & 20 & 1200 & $(0.167, 0.333, 0.500, 0.667, 0.833)$ & $(1,1,1,1,1,1)$ & 0.10 & 1.50 & $0.25$ \\
        L.2 & 20 & 1800 & $(0.100, 0.250, 0.400, 0.600, 0.800)$ & $(3,3,3,3,3,3)$ & 0.10 & 1.50 & $0.25$ \\
        L.3 & 20 & 2400 & $(0.100, 0.300, 0.500, 0.700, 0.900)$ & $(1,2,3,3,2,1)$ & 0.10 & 1.50 & $0.25$ \\ 
        \hline
        M.1 & 100 & 1200 & $(0.3333, 0.6667)$ & $(1,1,1)$ & 0.25 & 1.50 & $0.25$ \\
        M.2 & 125 & 1800 & $(0.3333, 0.6667)$ & $(1,1,1)$ & 0.30 & 1.50 & $0.25$ \\
        \hline
        N.1 & 20 & 300 & $(0.3333, 0.6667)$ & $(1,3,2)$ & $(0.35,0.25)$ & $(2.50, 3.00)$ & $0.25$ \\
        N.2 & 20 & 300 & $(0.1667, 0.8333)$ & $(1,3,2)$ & $(0.35,0.25)$ & $(2.50, 3.00)$ & $0.25$ \\
        N.3 & 20 & 300 & $(0.3333, 0.6667)$ & $(1,3,2)$ & $(0.50,0.50)$ & $(3.00, 3.00)$ & $0.25$ \\
       \hline\hline
    \end{tabular}}
\end{table}

{
Figure \ref{fig:random-structure} depicts the random structure investigated in scenario N.
\begin{figure}[!ht]
    \centering
    \includegraphics[trim = {0 4cm 0 4cm}, clip, scale=.35]{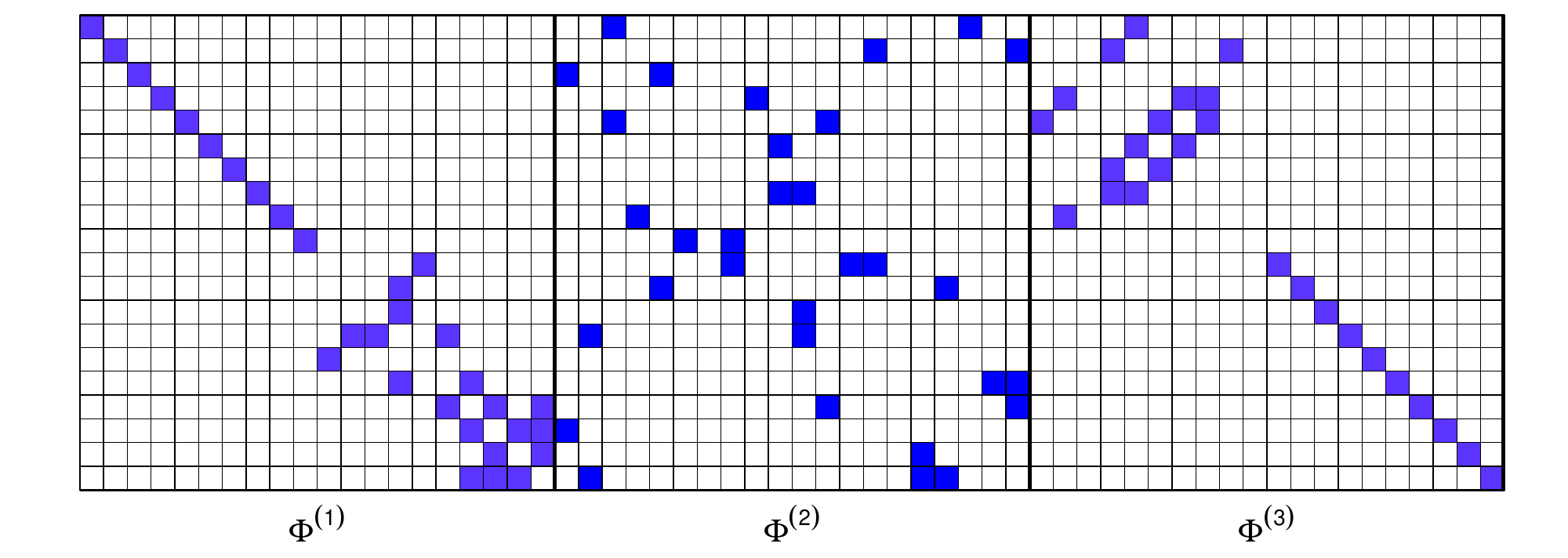}%
    \includegraphics[trim = {0 4cm 0 4cm}, clip, scale=.35]{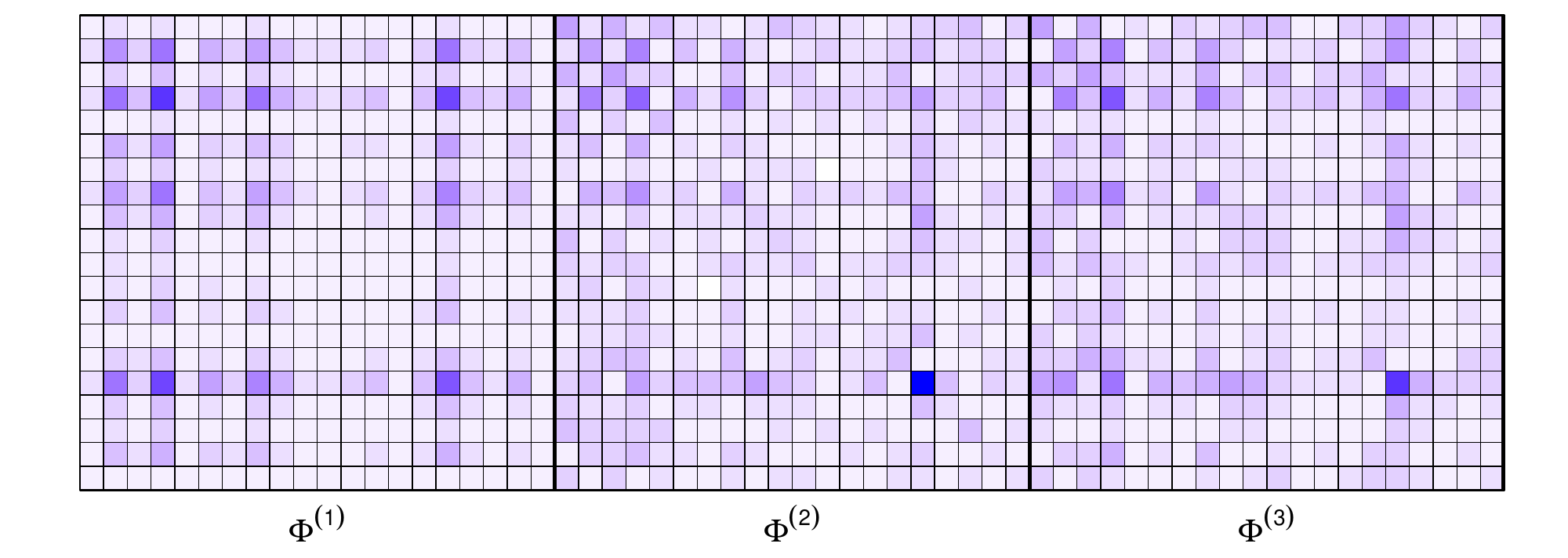}
    \caption{Left: random sparse pattern used in scenario N; Right: low rank pattern in scenario N, the ranks are 1, 3, and 2, respectively.}
    \label{fig:random-structure}
\end{figure}
}

Next, we present the performance of the multiple change points detection algorithm in Figure~\ref{fig:multi-boxplots} over 50 replications under setting N. 
\begin{figure}[!ht]
    \centering
    \includegraphics[scale=.32]{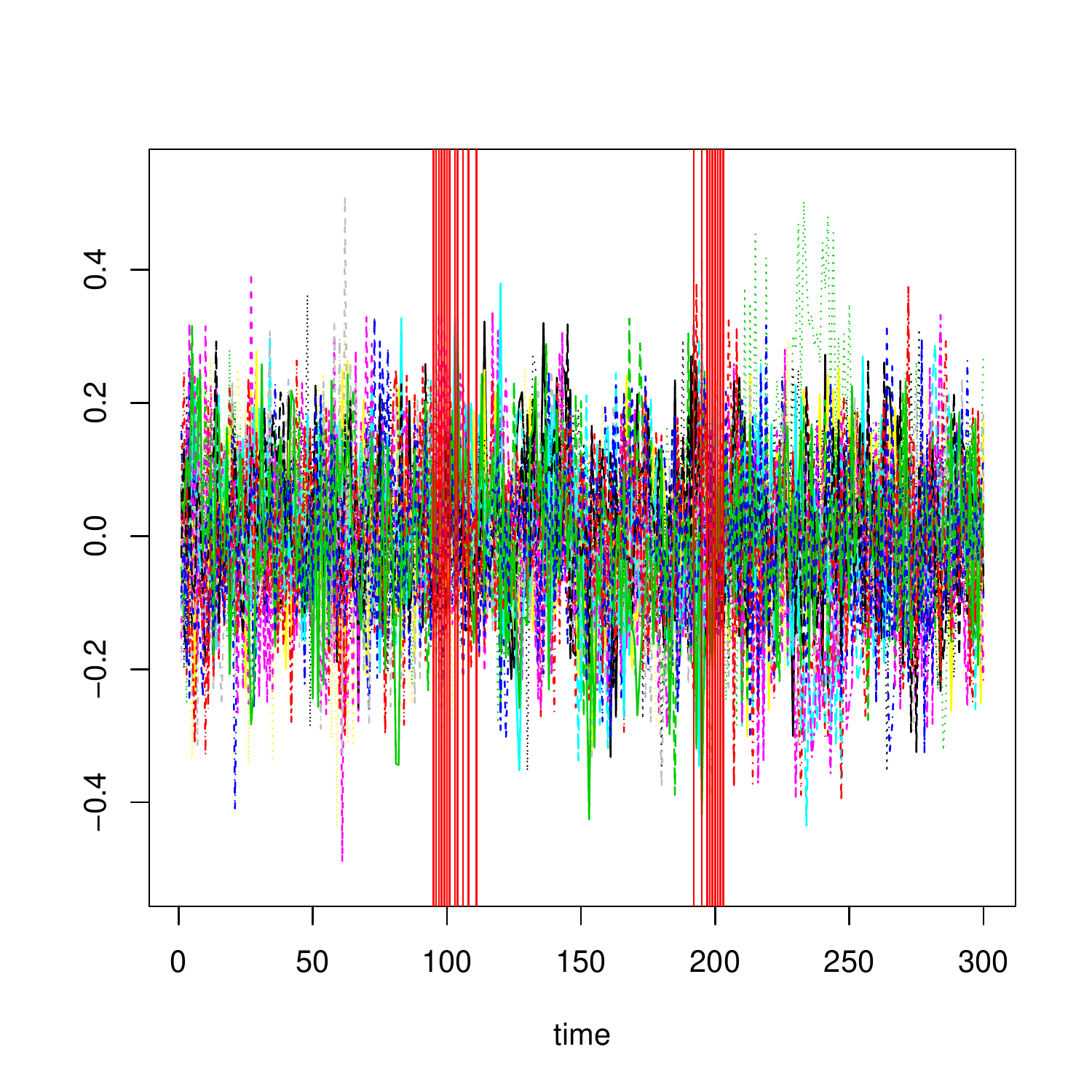}
    \includegraphics[scale=.32]{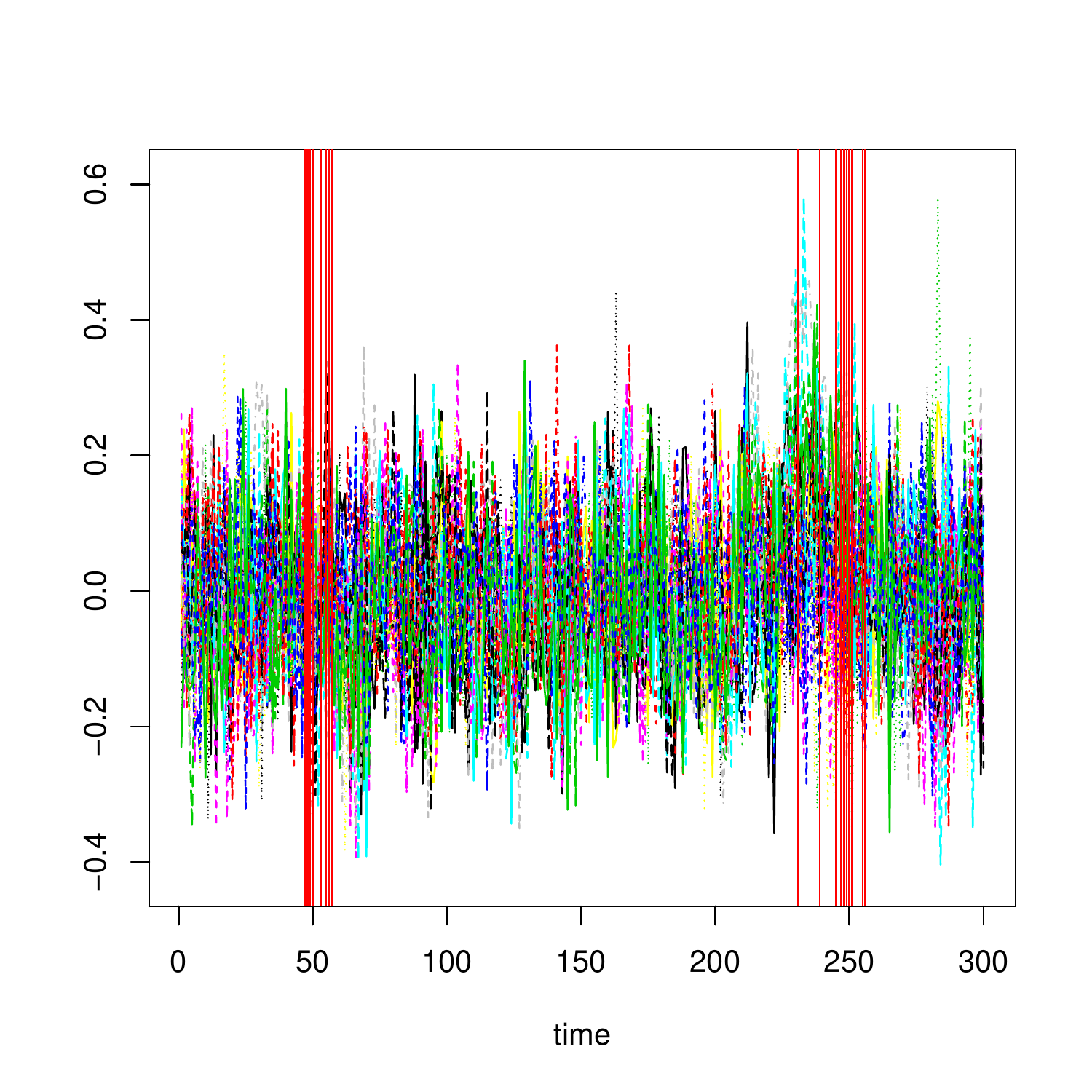}%
    \includegraphics[scale=.32]{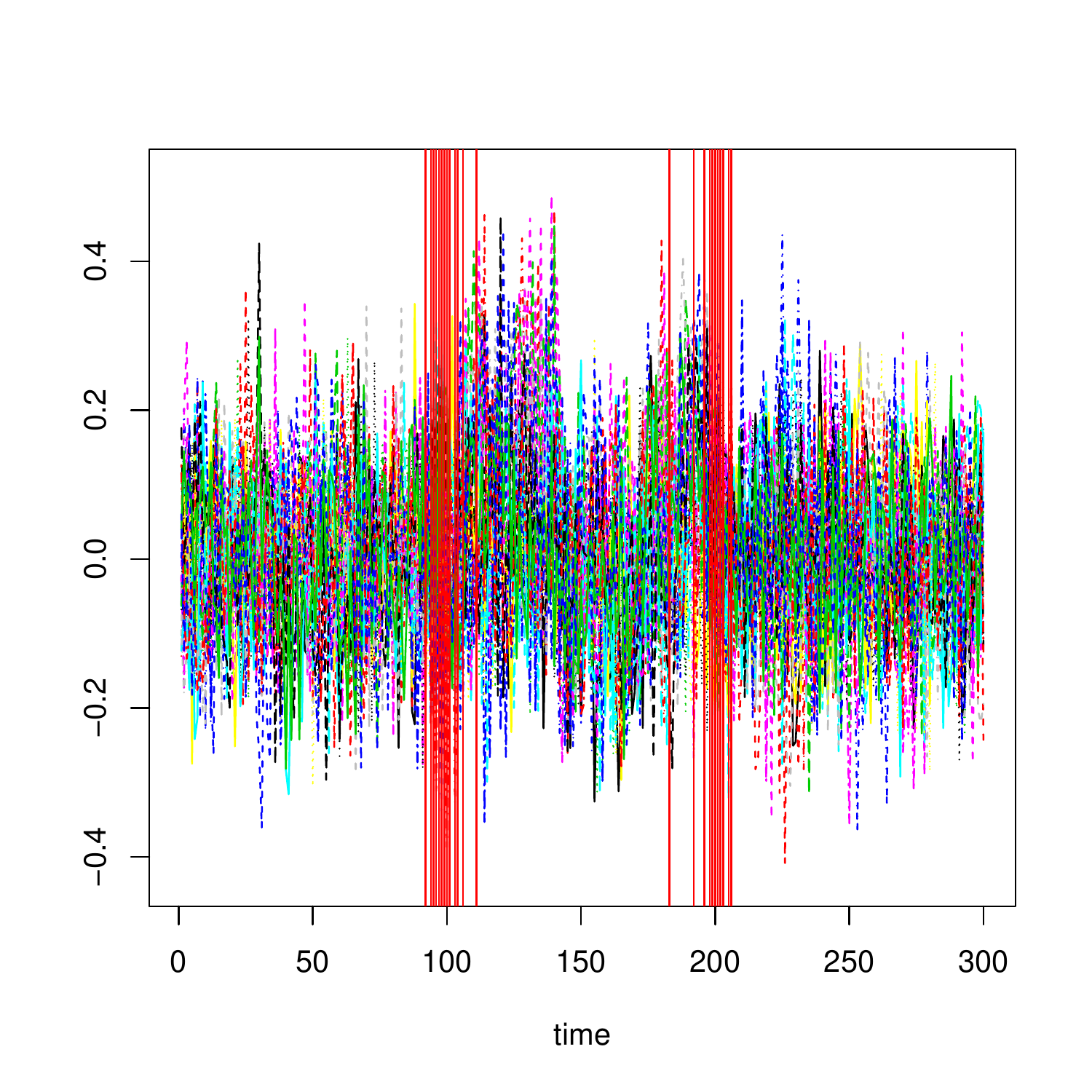}%
    \vfill
    \includegraphics[scale=.18]{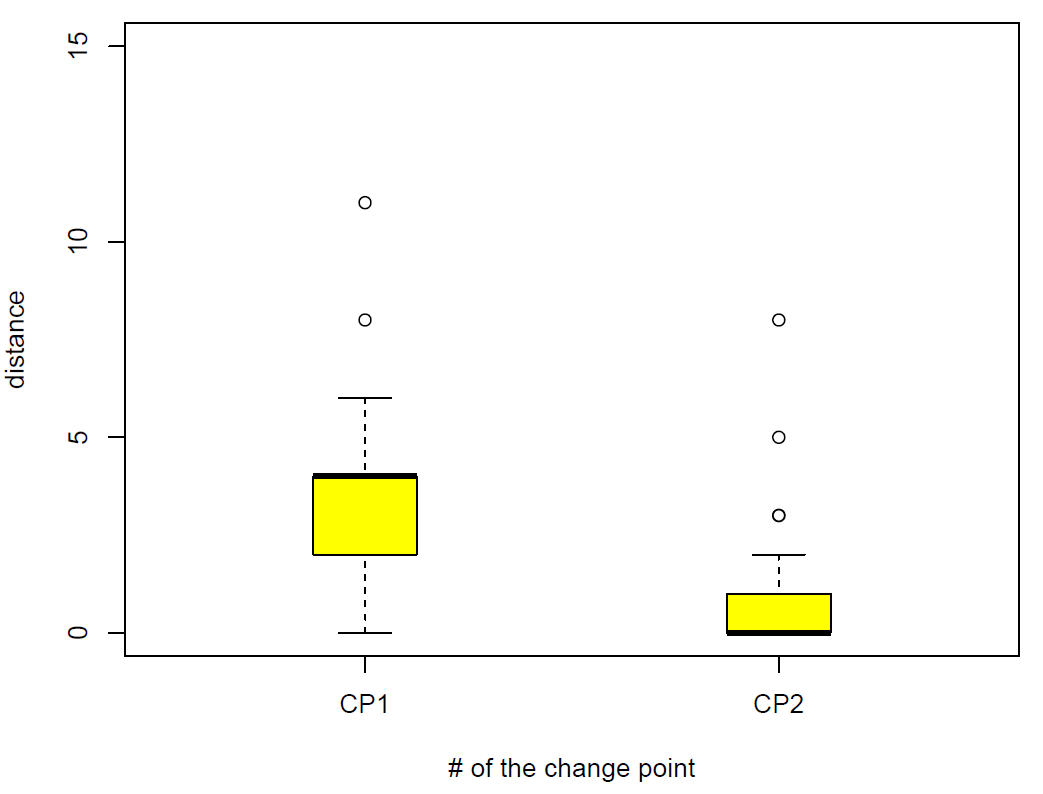}%
    \includegraphics[scale=.18]{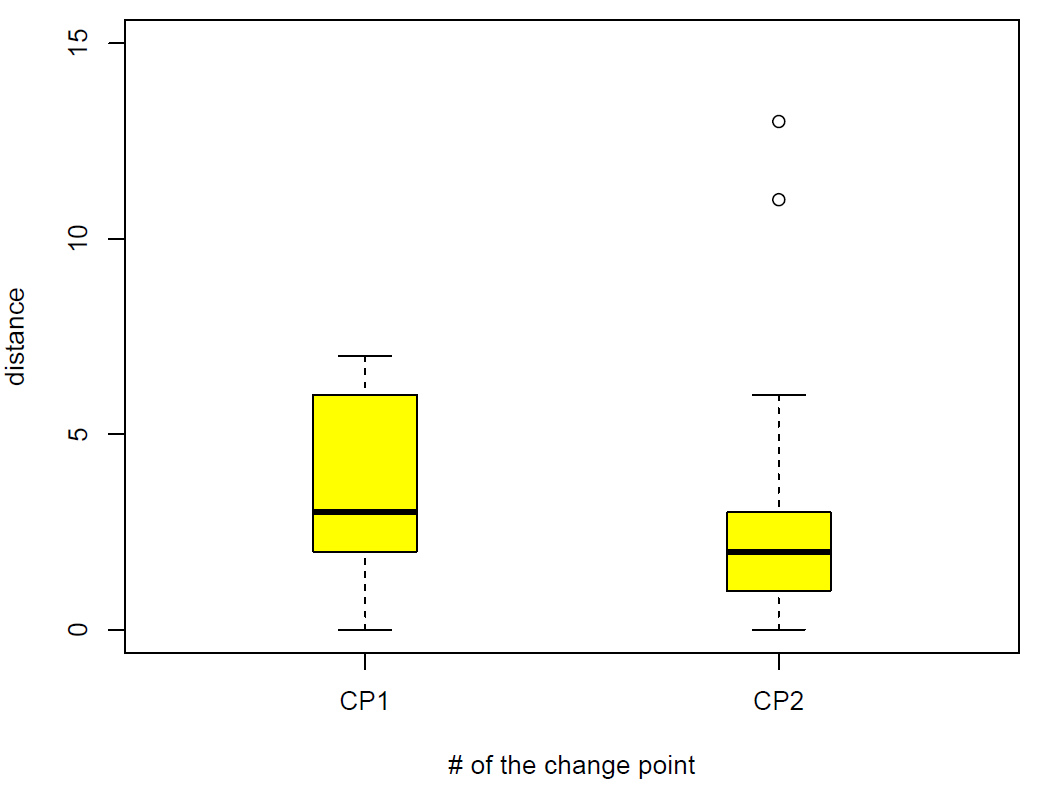}%
    \includegraphics[scale=.18]{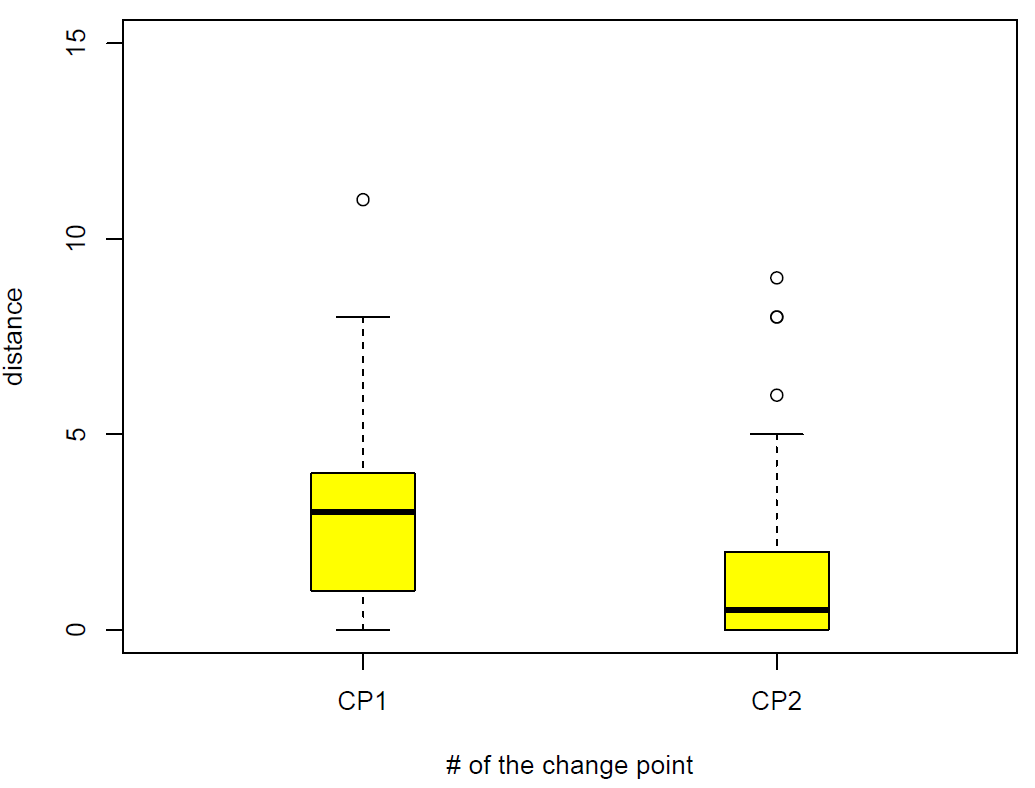}
    \caption{Final selected change points (red lines) by using two-step algorithm and boxplots for $|\widehat{t} - t^\star|$ under different scenario N settings.} 
    \label{fig:multi-boxplots}
\end{figure}

\subsection{Investigating the Impact of Signal-to-Noise Ration on Detection Power}
{

We design a series of simulation experiments to examine the minimum sample size $T$ needed to identify change points with a minimum selection rate, say $80\%$. The setting is as follows: $p=20$; 5 change points located at $\lfloor T/6 \rfloor$, $\lfloor 2T/6 \rfloor, \cdots, \lfloor 5T/6 \rfloor$, respectively. The jump sizes $v_S$ and $v_L$ for each change point are fixed. Specifically, the total jump size is chosen as 0.4, 0.8, and 1.6, respectively, and we examine the following sample sizes: $T=50, 55, 60, 150, 300$ to evaluate the detection power (selection rate). The detection power is calculated by averaging the detection rate for all 5 change points over 50 simulation replications. The following Figure \ref{fig:2} illustrates the detection power and Table \ref{tab:detect-power} presents the specific values. 
\begin{table}[!ht]
    \centering
    \spacingset{1}
    \caption{Averaged detection power for different sample sizes and signals.}
    \label{tab:detect-power}
    \begin{tabular}{c|c|c}
    \hline\hline
        jump size & sample size & detect power \\
    \hline
        \multirow{6}*{$v=0.4$} & 50 & 3\% \\
                              & 55 & 16\% \\
                              & 60 & 24\% \\
                              & 150 & 44\% \\
                              & 300 & 62\% \\
                              & 600 & 92\% \\
    \hline
        \multirow{6}*{$v=0.8$} & 50 & 5\% \\
                              & 55 & 16.8\% \\
                              & 60 & 26\% \\
                              & 150 & 66\% \\
                              & 300 & 94\% \\
                              & 600 & 100\% \\
    \hline
        \multirow{6}*{$v=1.6$} & 50 & 14\% \\
                              & 55 & 51.2\% \\
                              & 60 & 82.8\% \\
                              & 150 & 100\% \\
                              & 300 & 100\% \\
                              & 600 & 100\% \\
    \hline\hline
    \end{tabular}
\end{table}

\begin{figure}[!ht]
    \centering
    \includegraphics[scale=.3]{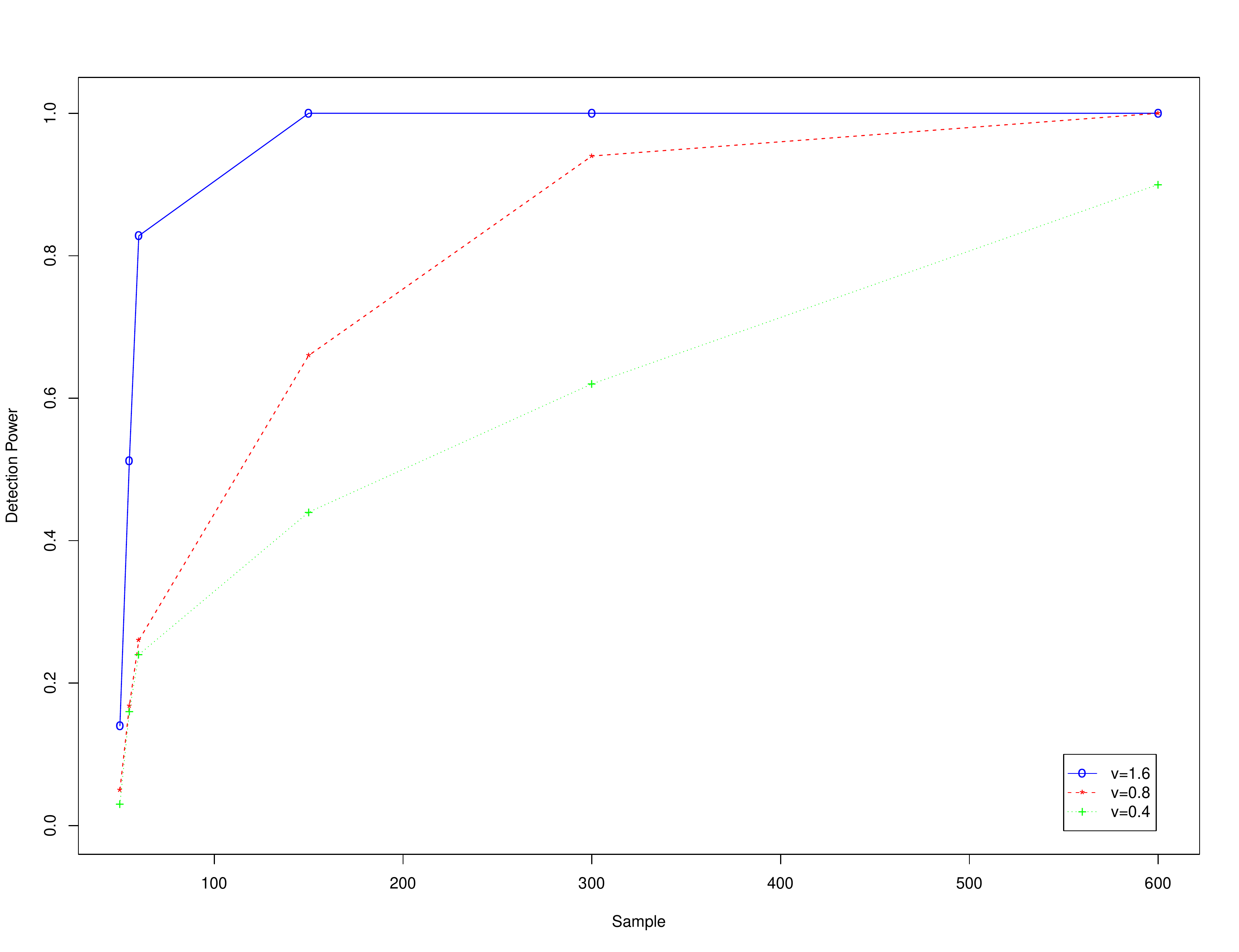}
    \caption{Averaged detection rate over change points for different sample sizes.}
    \label{fig:2}
\end{figure}
As expected, as the jump size $v$ increases together with he sample size, the selection rate consistently exceeds 90\%. These results provide guidelines for practitioners as well. For example, $v$ can easily exceed 0.8, whenever the rank changes, since the low rank component is a dense matrix.

In the EEG data application, the estimated jump size (calculated based on the estimated transitions matrices) is approximately 2.75. In the macroeconomics data application, the estimated jump size is approximately 4.20.
Hence, based on the number of break points identified, the length of the time series, the estimated jump sizes $v$, and the results presented in Table \ref{tab:detect-power}, it is reasonable to presume that the algorithm identifies correctly the underlying break points. Recall that for the EEG application, the design of the experiment corroborates the correctness of the results, whereas for the macroeconomics data, the corroborating evidence comes from important economic events and shocks that the literature recognizes as important drivers to induce break points.
}

\subsection{A Comparison of Run Times between the Low Rank plus Sparse and the Surrogate Weakly Sparse Models}
We undertake such a comparison for settings A, C and D presented in Section 5 in the main context. The results averaged over 50 replicates indicate that Algorithm 1 for the full model takes approximately 900 secs per replicate, while the surrogate model less than 200 secs. For the two-step Algorithm 2, the average run time for the full model takes approximately 3.5 hours per replicate, while that for the surrogate model approximately 20 minutes per replicate. The results are plotted in the following Figure~\ref{fig:running-time}.
\begin{figure}[!ht]
    \centering
    \resizebox{6in}{2.25in}{
    \begin{subfigure}
        \centering
        \includegraphics[scale=.45]{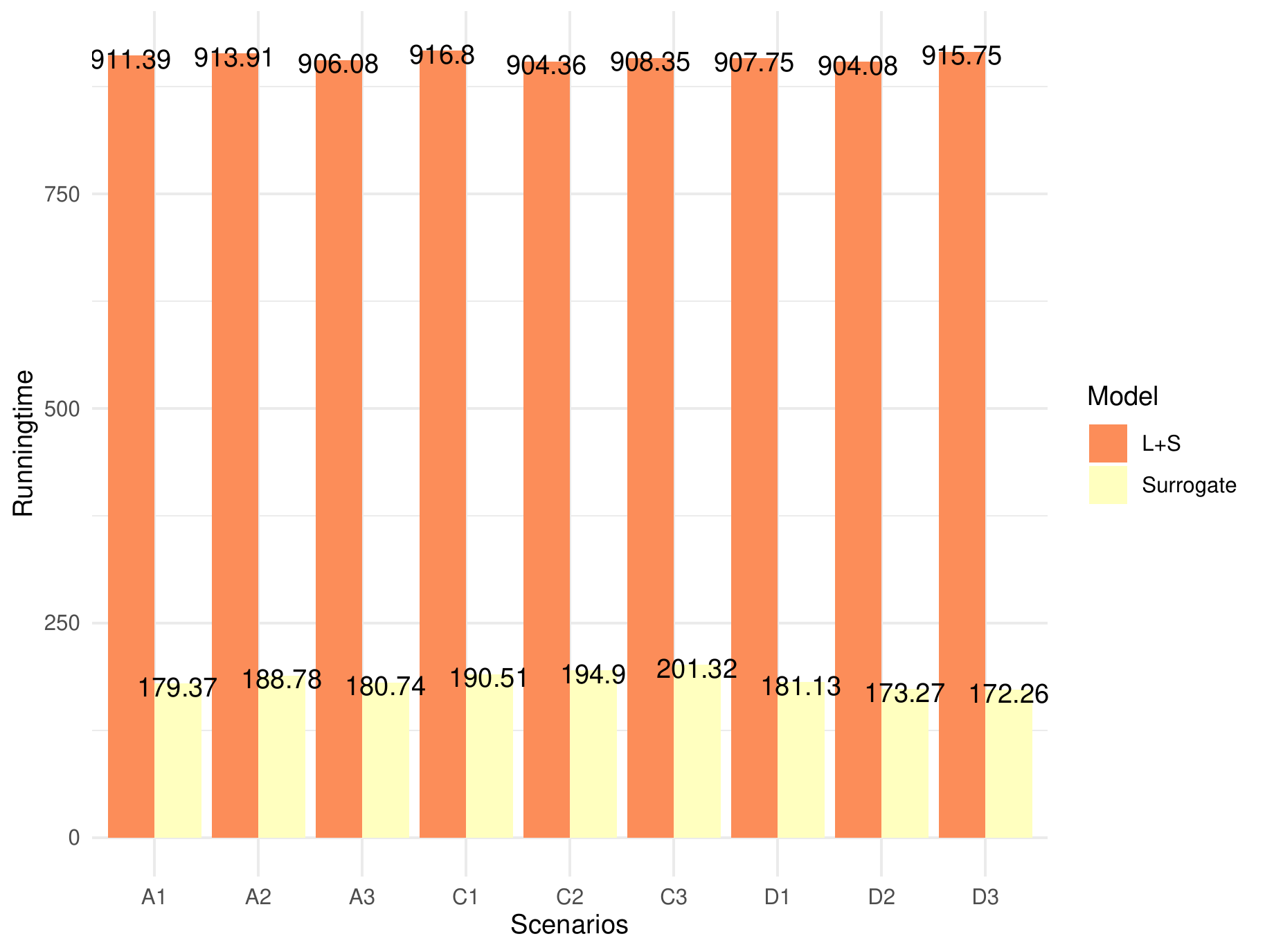}%
        \hfill
        \includegraphics[scale=.45]{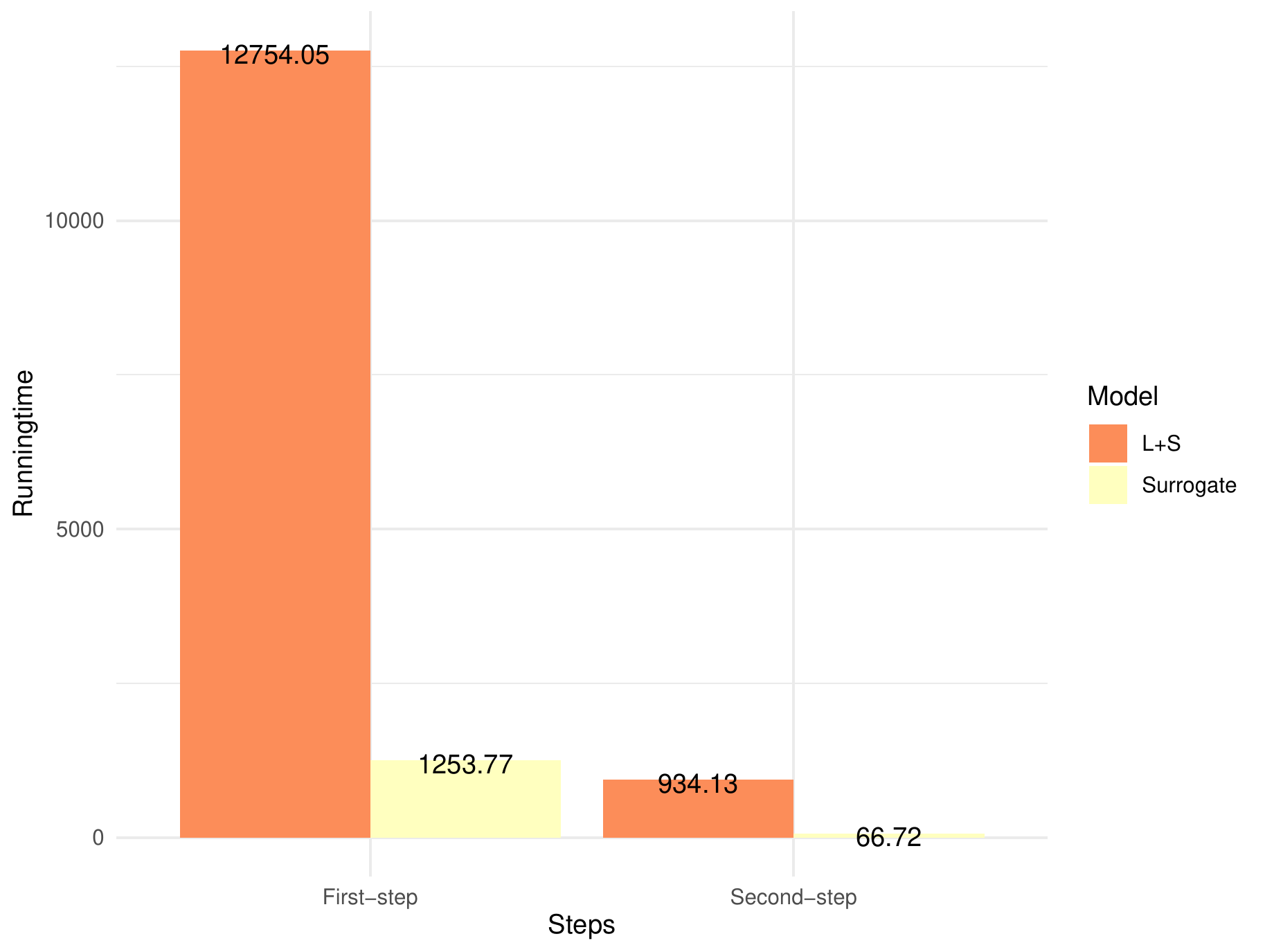}
    \end{subfigure}}
    \caption{Comparison the run times for the full low-rank plus sparse and alternative weakly sparse models; Left panel: single change point detection in settings A, C, and D; Right panel: multiple change point detection in setting L.1.
    \label{fig:running-time}}
\end{figure}
The high computational cost of the exhaustive search procedure for the full model is apparent and is due to performing multiple SVDs, while the surrogate model provides significant computational savings and hence justify its use when suitable, based the theoretical developments and guarantees presented in Section 4. 

\subsection{Comparisons between the Low Rank plus Sparse VAR Model and a Factor Model}
\label{sec:compare-factor}
\subsubsection{A comparison with the factor-based model under  scenarios L.1 and L.2}
{In \cite{barigozzi2018simultaneous}, the authors investigate detection of multiple change points in a static factor model. Note that a factor model
assumes that the data exhibit low rank structure in the 
contemporaneous dependence (correlation) structure, as opposed to their lead-lag (autocorrelation) structure as in a VAR model. Next, we  provide results for scenarios L.1 and L.2 in Table~\ref{tab:compare-factor},
respectively.}
\begin{table}[!ht]
    \spacingset{1}
    \centering
    \caption{Results for change point selection by low rank plus sparse VAR  model and a factor-based model.}
    \label{tab:compare-factor}
    \resizebox{\textwidth}{!}{%
    \begin{tabular}{c|c|c|c|c|c|c|c|c|c|c|c|c}
        \hline\hline
                          & Model & points & truth & mean & sd & selection rate & Model & points & truth & mean & sd & selection rate \\
        \hline
        \multirow{5}*{L.1} & \multirow{5}*{L+S VAR model} & 1 & 0.1667 & 0.1667 & 0.0004 & 1.00 & \multirow{5}*{Factor model} & 1 & 0.1667 & 0.1657 & 0.0061 & {0.78} \\
                         & & 2 & 0.3333 & 0.3333 & 0.0003 & 1.00 & & 2 & 0.3333 & 0.3340 & 0.0057 & 0.44 \\
                         & & 3 & 0.5000 & 0.4999 & 0.0003 & 1.00 & & 3 & 0.5000 & 0.5005 & 0.0065 & 0.56 \\
                         & & 4 & 0.6667 & 0.6665 & 0.0004 & 1.00 & & 4 & 0.6667 & 0.6706 & 0.0065 & 0.46 \\
                         & & 5 & 0.8333 & 0.8335 & 0.0004 & 1.00 & & 5 & 0.8333 & 0.8361 & 0.0073 & 0.70 \\
        \hline
        \multirow{5}*{L.2} & \multirow{5}*{L+S VAR model} & 1 & 0.1000 & 0.0999 & 0.0002 & 1.00 & \multirow{5}*{Factor model} & 1 & 0.1000 & 0.1000 & 0.0046 & {0.94} \\
                         & & 2 & 0.2500 & 0.2500 & 0.0000 & 1.00 & & 2 & 0.2500 & 0.2528 & 0.0041 & {0.84} \\
                         & & 3 & 0.4000 & 0.3999 & 0.0002 & 1.00 & & 3 & 0.4000 & 0.4011 & 0.0040 & {0.94} \\
                         & & 4 & 0.6000 & 0.6000 & 0.0000 & 1.00 & & 4 & 0.6000 & 0.6038 & 0.0028 & 0.70 \\
                         & & 5 & 0.8000 & 0.7999 & 0.0001 & 1.00 & & 5 & 0.8000 & 0.8020 & 0.0029 & 0.94 \\
        \hline\hline
    \end{tabular}}
\end{table}

{It can be seen that the accuracy of the location of the detected change points by the factor-based model is high; nevertheless, the selection rate (\# of times that it correctly identifies the right number of change points) is significantly lower than that of the VAR model. To illustrate further the latter point, the mean/median Hausdorff distance between the estimated change points set $\widetilde{\mathcal{S}}$ and the true change points set $\mathcal{S}^\star$ is tabulated in Table~\ref{tab:hausdorff-dist}. The result is not particularly surprising, since the true data generating mechanism is according to the posited low rank plus sparse VAR model.}
\begin{table}[!ht]
\spacingset{1}
\centering
\caption{Hausdorff distance $d_H(\widetilde{\mathcal{S}}, \mathcal{S}^\star)$ comparison with factor change point model.}
\label{tab:hausdorff-dist}
\begin{tabular}{c|c|ccc}
  \hline\hline
  & Model & mean($d_H(\widetilde{\mathcal{S}}, \mathcal{S}^\star)$) & std($d_H(\widetilde{\mathcal{S}}, \mathcal{S}^\star)$) & median($d_H(\widetilde{\mathcal{S}}, \mathcal{S}^\star)$) \\
  \hline
  \multirow{2}*{L.1} & Our model & \textbf{1.46} & \textbf{2.6358} & \textbf{0.00} \\
  \cline{2-5}
  & Factor model & 27.26 & 18.6808 & 19.00 \\
  \hline
  \multirow{2}*{L.2} & Our model & \textbf{1.45} & \textbf{2.3774} & \textbf{0.00} \\
  \cline{2-5}
  & Factor model & 33.34 & 23.4082 & 30.00 \\ 
  \hline\hline
\end{tabular}
\end{table}

\subsubsection{A comparison with the factor-based model under a dynamical factor model (DFM) generating mechanism}
{
For a further comparison between the detection strategy based on the factor model (\cite{barigozzi2018simultaneous}) and the two step rolling window strategy based on the posited VAR model, we employed a \textit{dynamical factor model} to generate the data. Hence, the data are generated according to:
\begin{equation*}
    X_t = \Lambda_jF_t + e_t, \quad F_t = \Psi F_{t-1} + \epsilon_t,
\end{equation*}
where $\Lambda_j, j=1,2,\dots, m_0+1$ are factor loadings and $\Psi$ is a diagonal transition matrix of the latent process $\{F_t\}$ with independent and identically distributed standard normal entries. Both error terms $e_t$ and $\epsilon_t$ are independent and identically normally distributed with mean zero and variance $0.01 \mathbf{I}$. \\
 The dimension of the time series is set to $p=20$, the sample size to $T=300$, the locations of the change points at $t_1^\star=100$ and $t_2^\star=200$, and the dimension of the latent factor process to $r=5$. The loadings matrices $\Lambda_j$ are generated at random, with varying magnitudes across different stationary segments.}

{Table~\ref{tab:dfm} tabulates the mean and standard deviation of the relative location of the detected change points, as well as the selection rate for the two-step procedure for the VAR model and the procedure based on the factor model over 50 replications. 
}
\begin{table}[!ht]
    \centering
    \spacingset{1}
    \caption{Comparison of the two-step strategy for the VAR model and the strategy based on factor model under a DFM data generating mechanism.}
    \label{tab:dfm}
    \resizebox{\textwidth}{!}{%
    \begin{tabular}{c|c|c|c|c|c}
        \hline\hline
        Method & CP & Truth & Mean & Sd & Selection rate \\
        \hline
        \multirow{2}*{Two-step strategy for a L+S VAR model} & 1 & 0.333 & 0.342 & 0.050 & 0.80 \\
        \cline{2-6}
         & 2 & 0.667 & 0.647 & 0.041 & 0.66 \\
         \hline
        \multirow{2}*{Factor model based strategy ( \cite{barigozzi2018simultaneous})} & 1 & 0.333 & 0.232 & 0.107 & 0.06 \\
        \cline{2-6}
        & 2 & 0.667 & 0.801 & 0.118 & 0.10 \\
        \hline\hline
    \end{tabular}}
\end{table}

{It can be seen that the two-step strategy based on the posited low-rank plus sparse VAR model exhibits a significantly higher selection rate than the strategy based on the static factor model. Further, the former provides much more accurate estimates of the locations of the underlying change points.}

{Note that both models \textit{misspecify} the true data generating mechanism. The factor model assumes a static factor structure (no autoregressive dynamics in the latent factor), whereas the VAR model assumes autoregressive dynamics on the observed data. The inferior performance of the strategy based on the factor model may be due to the detection mechanism used in  \cite{barigozzi2018simultaneous}, which first extracts principal components of the data across the whole observation interval and then leverages a binary segmentation algorithm to identify the change points. 
}

\subsection{A comparison with the TSP Algorithm}
{
To compare with the TSP algorithm proposed in \cite{bai2020multiple}, we use the settings in scenario B.1. in the Performance Evaluation Section in \cite{bai2020multiple}, wherein $T=300$, $p=20$, with two change points $t_1=100$ and $t_2=200$. Further, we specify a \textit{fixed} (not changing) low rank component of rank 5, and time-varying sparse components with 1-off diagonal structure and magnitudes equal to -0.75, 0.8, and -0.7, respectively. The performance of the fused lasso based algorithm in \cite{bai2020multiple} and the two-step algorithm in the current manuscript is compared in Table \ref{tab:compare-tsp}:
\begin{table}[!ht]
    \centering
    \spacingset{1}
    \caption{Performance of comparison between Two-step L+S and \cite{bai2020multiple}}
    \label{tab:compare-tsp}
    \begin{tabular}{c|c|c|c|c|c}
    \hline\hline
         & points & truth & mean & sd & selection rate \\
    \hline
        \multirow{2}*{Two-step L+S} & 1 & 0.3333 & 0.3327 & 0.0014 & 1.00 \\
                                    & 2 & 0.6667 & 0.6669 & 0.0012 & 1.00 \\
    \hline
        \multirow{2}*{\cite{bai2020multiple}} & 1 & 0.3333 & 0.3413 & 0.0234 & 0.98 \\
                                              & 2 & 0.6667 & 0.6665 & 0.0087 & 1.00 \\
    \hline\hline
    \end{tabular}
\end{table}

It can be seen that the newly developed Two-step L+S algorithm matches the performance of the algorithm in \cite{bai2020multiple}, while at the same time it can handle the much more challenging setting wherein \textit{both} the low rank and the sparse components of the VAR transition matrices can change.
}

\subsection{A Comparison of the Two-step Strategy with a Dynamic Programming (DP) Algorithm}\label{sec:dp}
Next, prompted by a comment from a reviewer,  we investigate a popular and generally applicable strategy for detecting multiple change points, namely one based on a dynamic program (see \cite{friedrich2008complexity} for a similar algorithm for detecting multiple change points in a \textit{pure sparse} VAR model). The algorithmic details are given in Appendix B Algorithm \ref{algo:3}. {As is well-known, the time complexity of a dynamic programming algorithm is $\mathcal{O}(T^2C(T))$, where $C(T)$ denotes the computational cost of estimating the parameters over the entire observation sequence.} Contrary, as mentioned earlier, {the total computational time complexity of the two step rolling window strategy is $\mathcal{O}(TC(T))$.}

The setting under consideration is as follows: the data are generated according to the posited low rank plus sparse VAR model with
$p=20$, $T=240$ and two change points located at $t_1^\star=80$ and $t_2^\star=160$. The ranks for each stationary segment remain the same $r_j \equiv 1$ and the jump sizes are set to $v_L = 0.1$ and $v_S = 1.5$. The information ratio $\gamma_j = 0.25$ for all three stationary segments.
\begin{table}[!ht]
    \centering
    \spacingset{1}
    \caption{Comparison of Proposed Two-step Algorithm with DP Algorithm.}
    \label{tab:dp-results}
    \begin{tabular}{c|c|c}
    \hline\hline
    & Model & Case (DP)  \\
    \hline
    \multirow{3}*{Running time (sec)} & Two-step L+S VAR & 942.11 \\
    \cline{2-3}
     & Two-step Surrogate & 193.20 \\
     \cline{2-3}
     & DP L+S VAR & {1417.71} \\
     \hline
     \multirow{3}*{No. of estimated change points} & Two-step L+S VAR & 2 \\
     \cline{2-3}
     & Two-step Surrogate & 2 \\
     \cline{2-3}
     & DP L+S VAR & {2} \\
     \hline
    \multirow{3}*{Estimated change points: $\widehat{t}_j/T$} & Two-step L+S VAR & $(0.3334_{0.0005}, 0.6665_{0.0003})$ \\
    \cline{2-3}
     & Two Step Surrogate & $(0.3294_{0.0025}, 0.6708_{0.0103})$ \\
     \cline{2-3}
     & DP L+S VAR & {$(0.3300_{0.0000}, 0.6611_{0.0004})$} \\
    \hline\hline
    \end{tabular}
\end{table}

{
The results presented in Table~\ref{tab:dp-results} are based on the following evaluation metrics: (1) the \textit{running time} for the low rank plus sparse model and the surrogate weakly sparse model using the two step rolling window strategy, and the low rank plus sparse model using the DP algorithm; (2) the \textit{number of estimated change points}; and (3) the mean
and standard deviation of the \textit{estimated change points}.
}

{
It can be seen that the running time for the DP Algorithm is 1.5 times longer than the two-step rolling window algorithm for the full low rank plus sparse model. The two step strategy for the surrogate model requires only a fraction of time compared to that of the full model. Further, both strategies accurately estimate \textit{both the number and the relative locations} of the true change points under this setting.}

\newpage
\section{Additional Results for Applications}\label{appendix:G}
\subsection{Guidelines for Applying the Methods to Data and Tuning Parameters Selection}
Recall that the information ratio plays a key role in the identifiability of the change points under both the full and the surrogate model. 
Since the information ratio is unknown in practice, it becomes unclear whether the surrogate model is capable of detecting the underlying change points, even when the full model clearly can. However, the computational savings of the former make it an attractive candidate for a first pass at obtaining candidate change points.
To that end, we outline below a strategy for deciding on the question of applicability of the surrogate model.

\begin{itemize}
    \item[Step 1:] Use the surrogate model and obtain candidate change points (after applying the screening step): $\widetilde{\tau}_1, \widetilde{\tau}_2, \dots, \widetilde{\tau}_m$. 
    \item[Step 2:] Let $\widetilde{\mathcal{I}}_j \overset{\text{def}}{=} |\widetilde{\tau}_{j+1} - \widetilde{\tau}_j|$, for $j=0,1,\dots, m$, where $\widetilde{\tau}_0 = 1$ and $\widetilde{\tau}_{m+1} = T$. Then, apply the full model on each selected time segment $\widetilde{\Delta}_j$. Suppose in the $j$-th segment $\widetilde{\mathcal{I}}_{j}$, we have estimated change points $\widehat{\tau}_1^{(j)}, \dots, \widehat{\tau}_{T_{j}}^{(j)}$ obtained by the full L+S model, where $T_j = |\widetilde{\tau}_{j+1} - \widetilde{\tau}_j|$. Then the final estimated change points set is given by:
    \begin{equation*}
        \left(\bigcup_{j=0}^m\{\widehat{\tau}_1^{(j)}, \dots, \widehat{\tau}_{T_j}^{(j)}\}\right) \cup \{\widetilde{\tau}_1, \dots, \widetilde{\tau}_m\}.
    \end{equation*}
\end{itemize}

Next, we discuss how the following tuning parameters are selected.
\begin{itemize}
    \item[$\circ$] For the full low-rank plus sparse model, we need to select the tuning parameters $\lambda_j$ and $\mu_j$;
    \item[$\circ$] For the surrogate weakly sparse model, the tuning parameter $\eta_j$ needs to be selected;
    \item[$\circ$] For the selection step for candidate change points, a proper window size is the key factor impacting the accuracy and speed of the algorithm. Further, the screening step requires the penalization parameter $\omega_n$ to be specified. 
\end{itemize}

Our recommendations are summarized next:
\begin{itemize}
    \item[$(\lambda_j, \mu_j)$]: {We use the same selection procedure as discussed in Section 5 in the main manuscript; specifically, we use the theoretical values provided and select the constants $c_0, c_0^\prime$ instead. By using a grid search, we simultaneously pick these two tuning parameters for each specified segment. }
    \item[$\eta_j$]: According to \cite{negahban2012unified}, we adopt the theoretical assumption on $\eta_j$: $\eta_j \propto \xi_j$, where $\xi_j$ is the corresponding lasso penalization parameter selected by \texttt{glmnet} and \texttt{sparsevar}. We typically choose $\eta_j \in [0.01\xi_j, 0.1\xi_j]$. 
    {
    \item[$\alpha_L$]: In practice, we choose $\alpha_L$ based on the goal of the application. We empirically choose $\alpha_L$ based on the theoretical value $cp\sqrt{\frac{\log(pT)}{T}}$, and choose the constant $c \in [0.1,1]$ in order to obtain a satisfactory estimation of the sparse components. 
    }
    \item[$h$]: For window-size $h$, a feasible selection method is introduced in Section 5.
    \item[$\omega_n$]: The idea is to first finish the backward elimination algorithm (BEA) until no break points are left. Then, we cluster the \emph{jumps} in the objective function $ \mathcal{L}_T $ into two subgroups, small and large. Intuitively, if removing a break point leads to a small jump in $ \mathcal{L}_T $, then the break point is likely redundant. In contrast, larger jumps correspond to true break points. The smallest jump in the second group is thus a reasonable candidate for $ \omega_n $. The proposed algorithm is summarized as follow:
    \begin{itemize}
        \item[(i)] Apply the BEA algorithm to the set $ \widetilde{\mathcal{S}} $ until no break points are left. Denote the ordered deleted break points as $ \widetilde{t}_{i_1}, \widetilde{t}_{i_2}, \ldots, \widetilde{t}_{i_{\widetilde{m}}} $.
        \item[(ii)] For each  $ k = 1, 2, \ldots, \widetilde{m} $, set $ v_k = \left|    \mathcal{L}_T(\widetilde{t}_{i_k}, \ldots, \widetilde{t}_{i_{\widetilde{m}}}; \bm{\lambda}, \bm{\mu}) - \mathcal{L}_T(\widetilde{t}_{i_{k-1}}, \ldots, \widetilde{t}_{i_{\widetilde{m}}}; \bm{\lambda}, \bm{\mu})\right| $. Define $V = \left \lbrace v_1, v_2, \ldots, v_{\widetilde{m}}  \right \rbrace $.
        \item[(iii)] Apply k-means clustering algorithm (\cite{hartigan1979algorithm}) to the set $ V $ with two centers. Denote the subset with smaller center as the small subgroup, $ V_S $, and the other subset as the large subgroup, $ V_L $.
        \item[(iv)]
        \begin{itemize}
            \item[(a)] If $ \left( \mbox{between-group SS/total SS} \right) $ in (iii) is high, set $ \omega_n = \min V_L $.
            \item[(b)] If $ \left( \mbox{between-group SS/total SS} \right) $ in (iii) is low, set $ \omega_n = \max V $.
        \end{itemize}  
    \end{itemize}
\end{itemize}

\subsection{Detailed Results for the EEG Dataset}\label{sec:eeg-application}
The time series for all 21 EEG channels are shown in the left plot of Figure \ref{fig:application-data}. By examining the time series data, it can be seen that the signal changes significantly. To speed up computations, we select one data point every 1/16 seconds and reduce the total time points to $T = 4376$.
\begin{figure}[!ht]
    \centering
    \begin{subfigure}
        \centering
        \includegraphics[width=.48\linewidth, height=.4\linewidth]{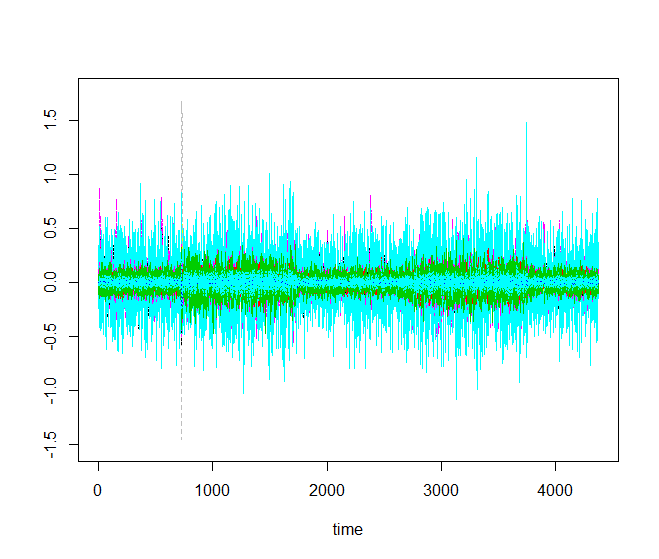}%
        \includegraphics[width=.48\linewidth, height=.4\linewidth]{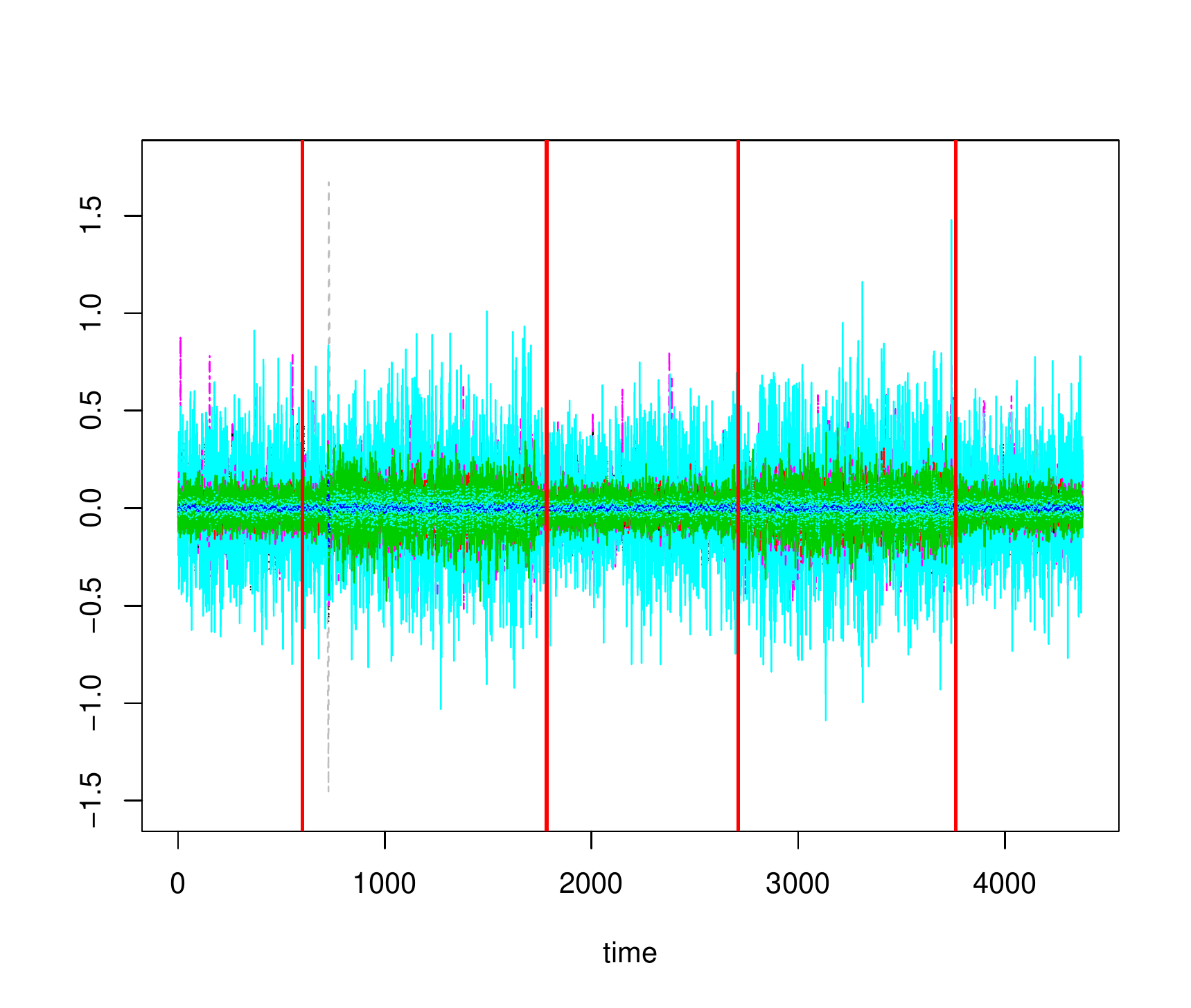}
    \end{subfigure}
    \vspace{-10pt}
    \caption{Left: Selected 21 EEG channels reduced time series data; Right: Estimated 4 change points.}
    \label{fig:application-data}
\end{figure}
We also estimate the structured transition matrices in each estimated segment. The following Figure \ref{fig:patterns} shows the estimated sparsity patterns and low-rank patterns. It can be seen that there are obvious similarities among segments 1, 3, and 5, where the subject had eyes closed, and segments 2 and 4, where eyes were open. Further, the estimated ranks for these 5 segments are: 5, 12, 6, 12, and 4, respectively. 
\begin{figure}[!ht]
    \centering
    \resizebox{\textwidth}{!}{%
    \begin{subfigure}
        \centering 
        \includegraphics[scale=.31]{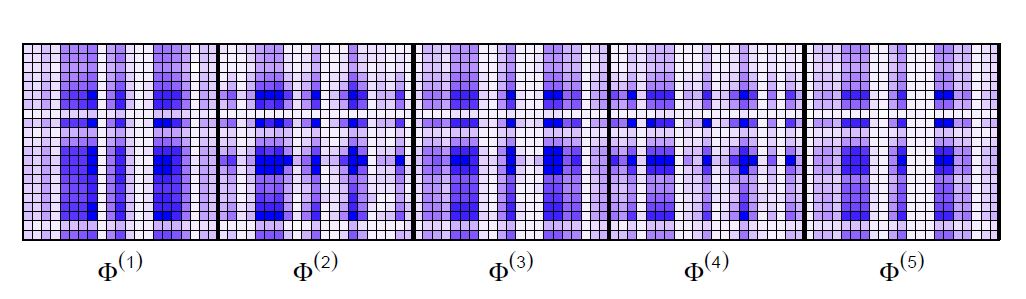}
        \includegraphics[scale=.31]{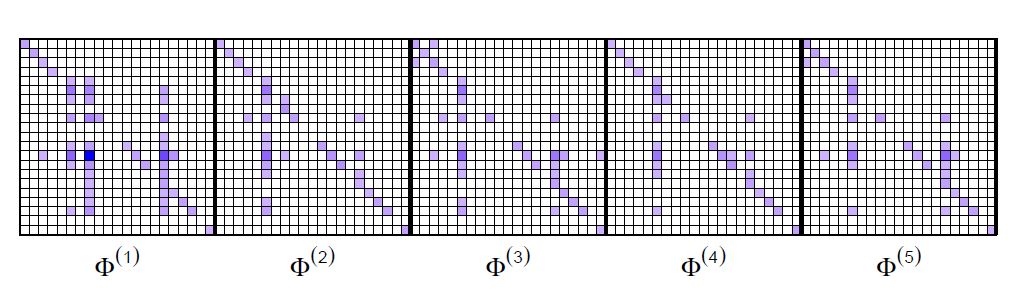}
    \end{subfigure}}
    \caption{Left: Estimated low-rank pattern for 5 segments; Right: Estimated sparse pattern for 5 segments.}
    \label{fig:patterns}
\end{figure}
Furthermore, the connectivity patterns in the sparse components are depicted in Figure \ref{fig:networks}. The name of the EEG channel denotes the node and entries of the transition matrices with magnitude larger than 0.1 are shown.
\begin{figure*}[!ht]
    \centering
    \begin{subfigure}
        \centering
        \includegraphics[width=0.32\linewidth]{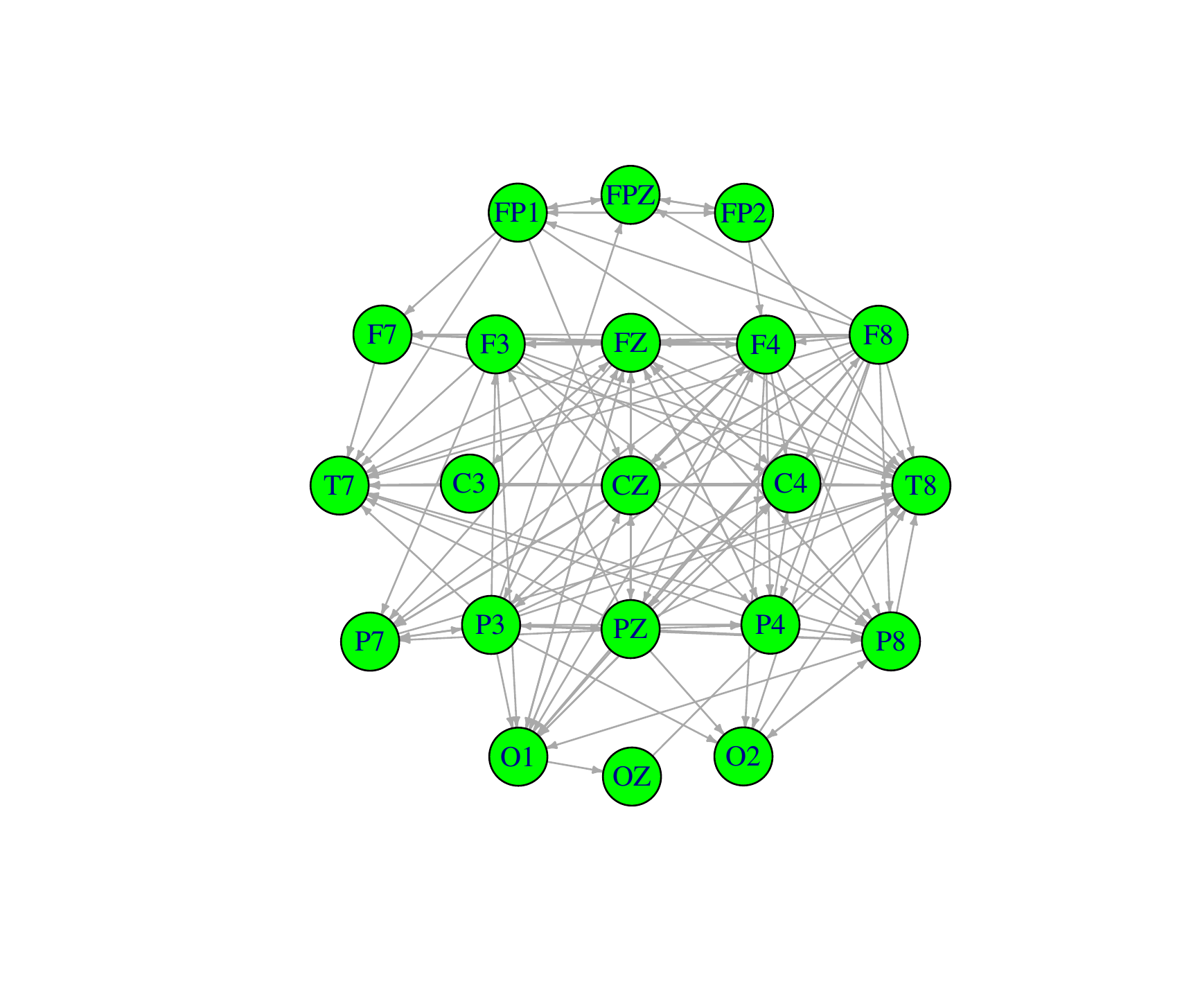}
        \includegraphics[width=0.32\linewidth]{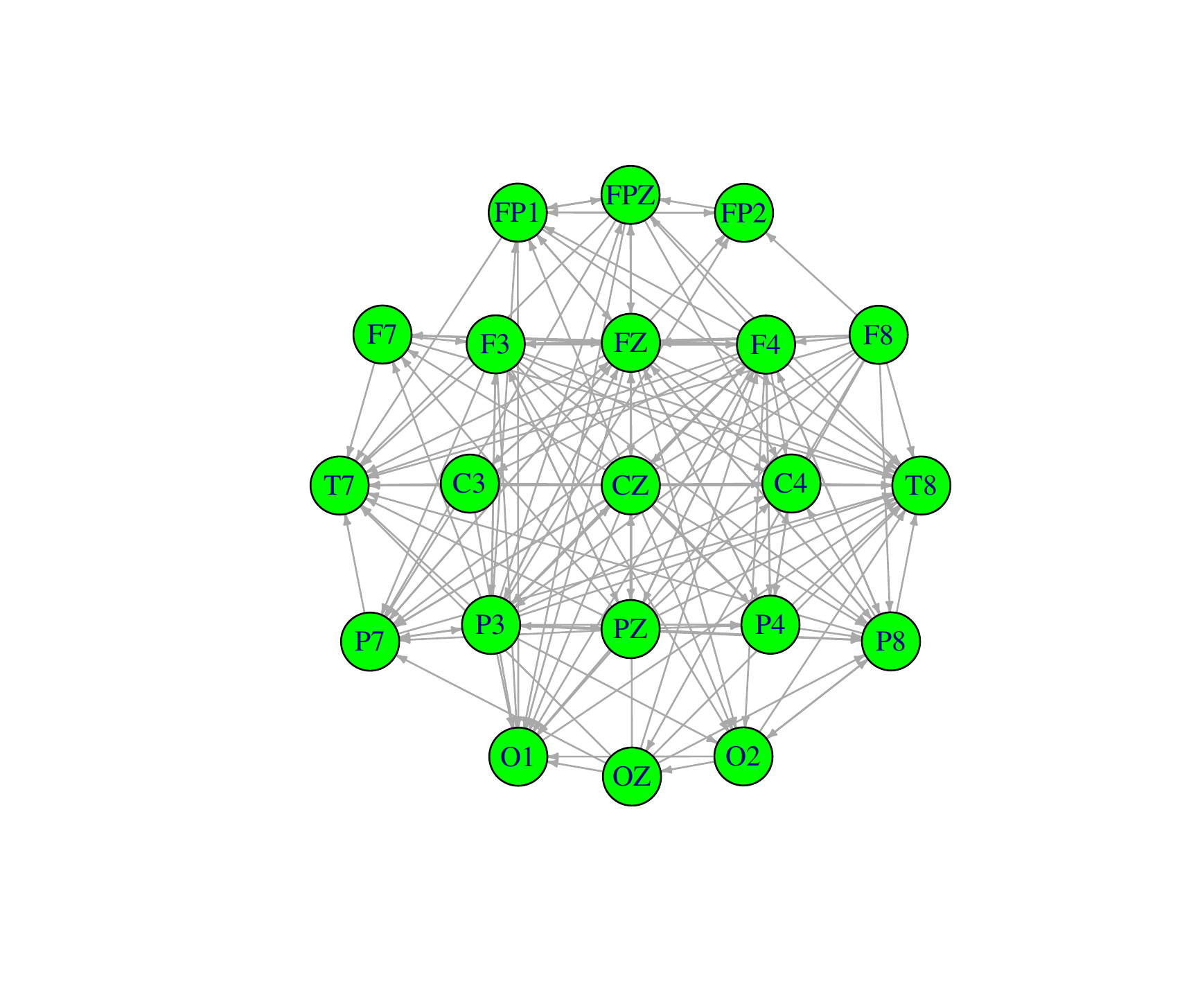}
        \includegraphics[width=0.32\linewidth]{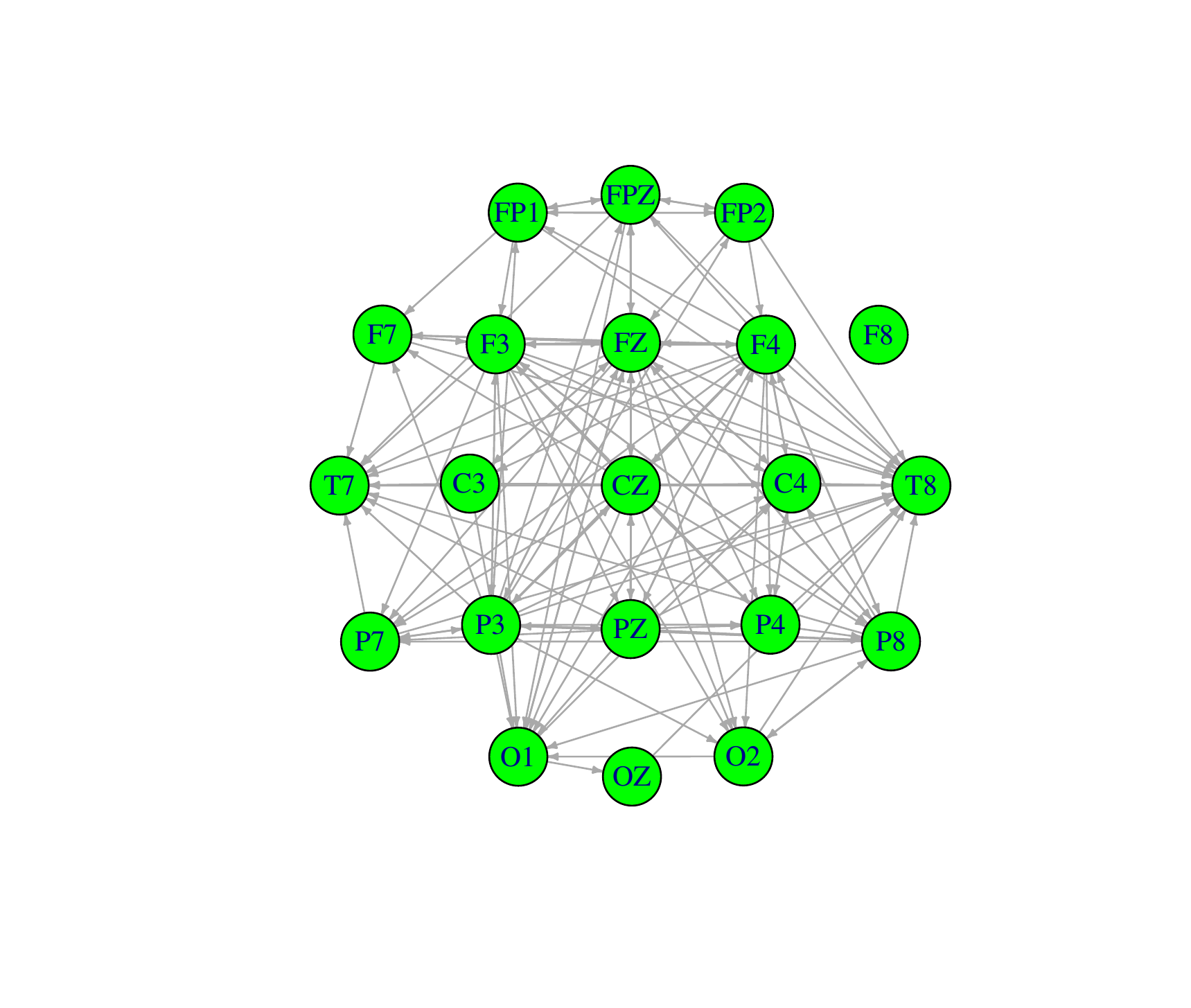}
        \includegraphics[width=0.32\linewidth]{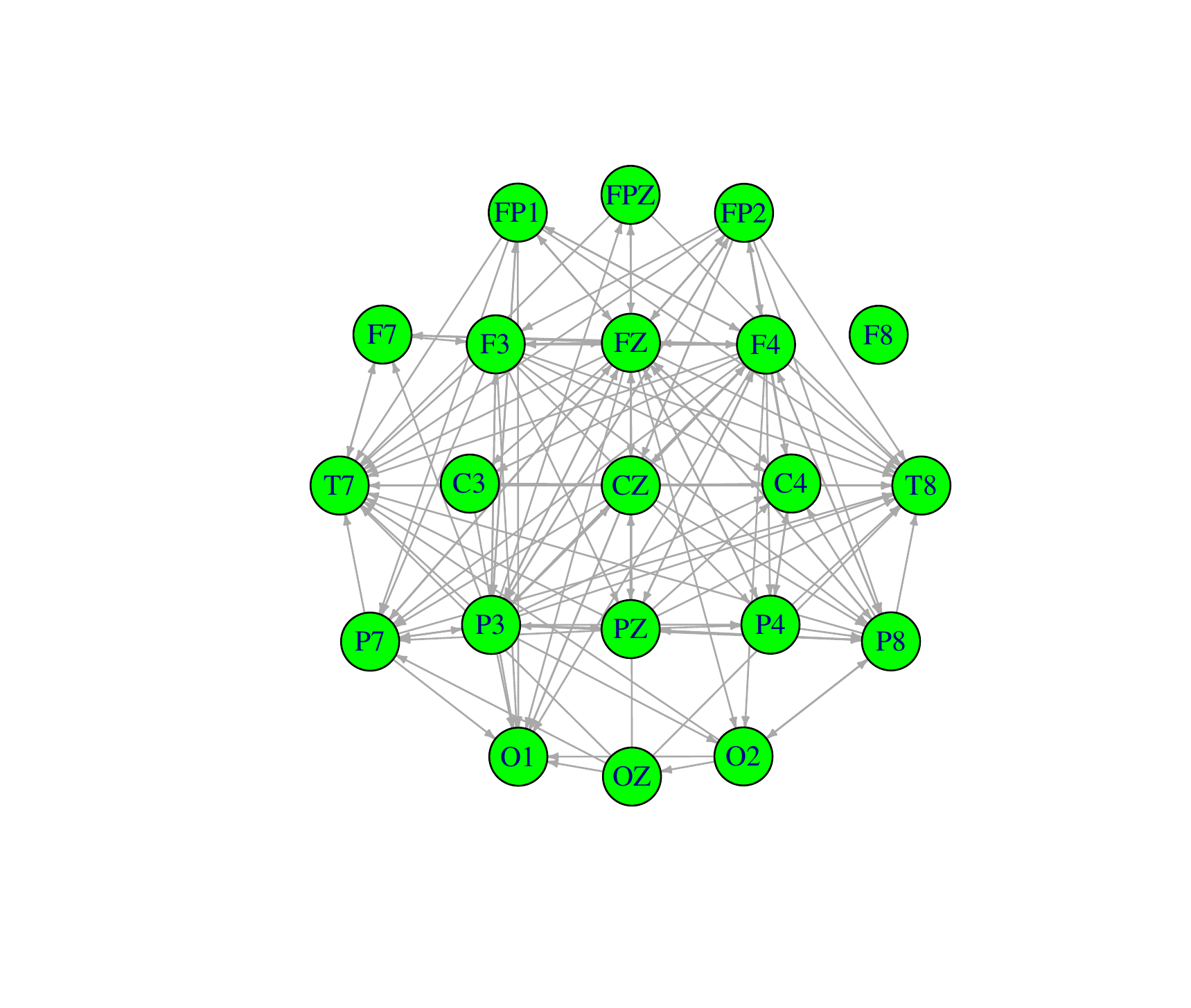}
        \includegraphics[width=0.32\linewidth]{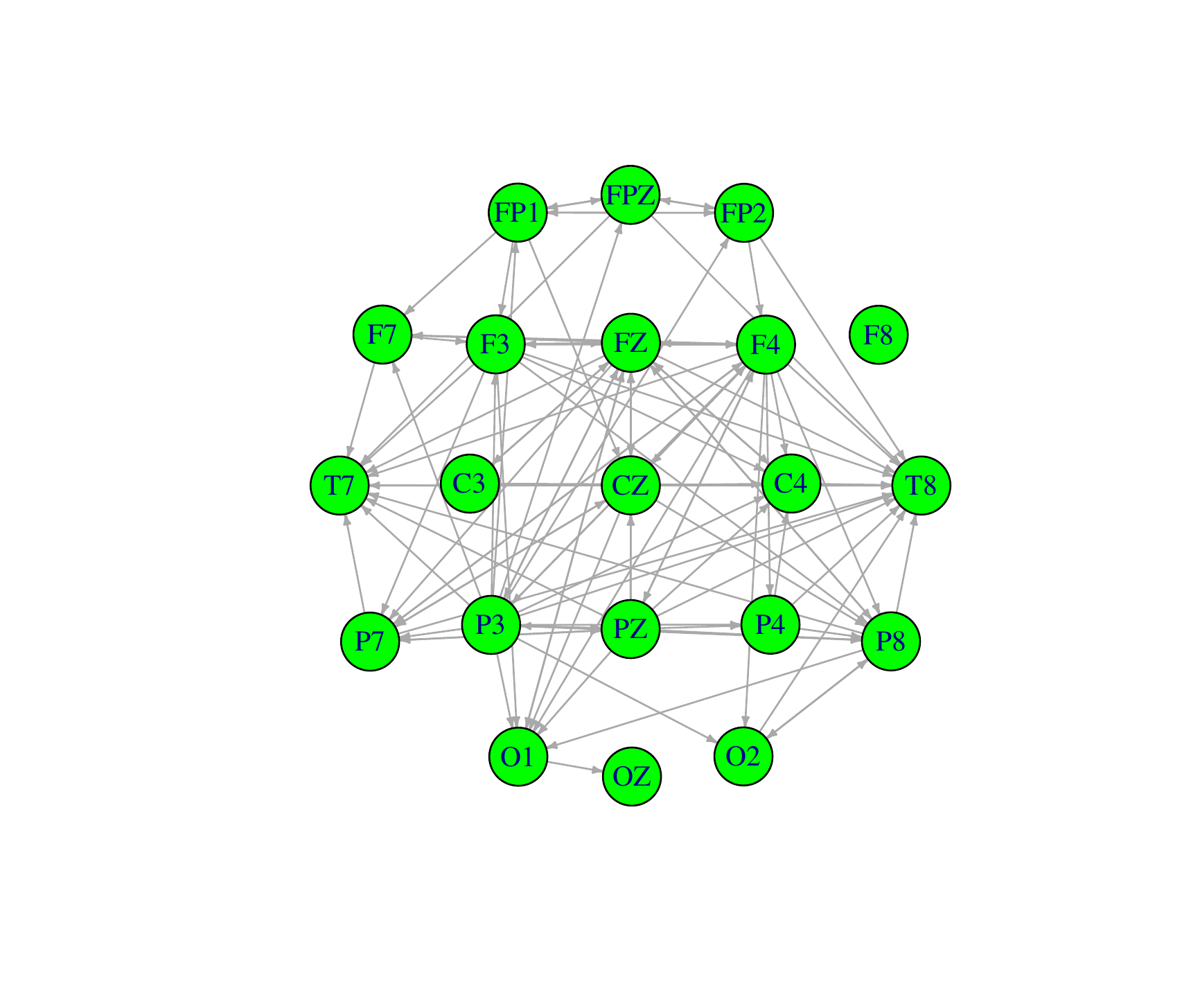}
    \end{subfigure}
    \caption{From top left to bottom right, we present the networks of sparse components for the segment 1-5.} 
    \label{fig:networks}
\end{figure*}

\subsection{Detailed Results for the Macroeconomics Data}
The modified macroeconomics data together with all 6 estimated change points is provided in Figure \ref{fig:macro}.
\begin{figure}[!ht]
    \centering
    \includegraphics[scale=.5]{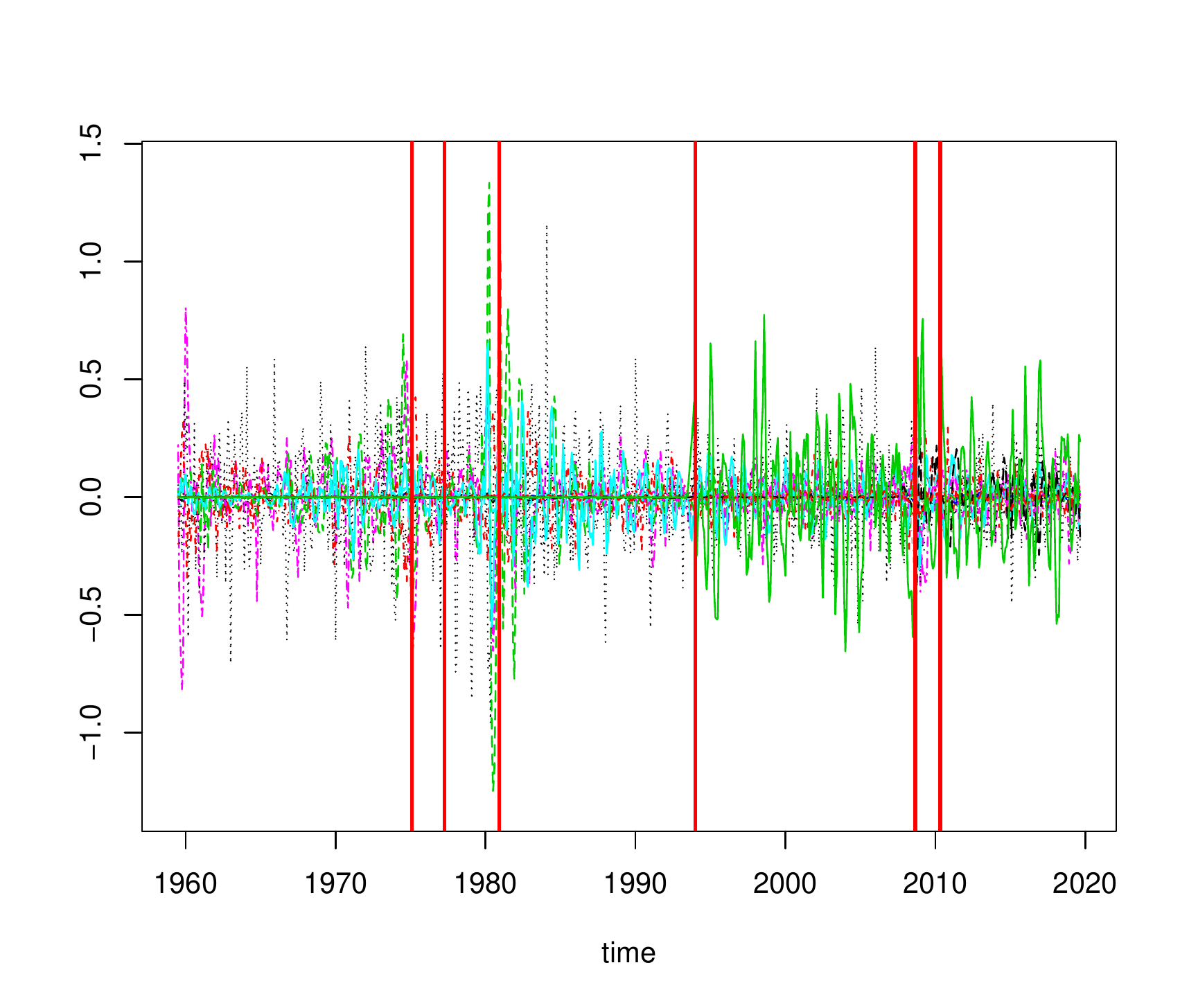}
    \vspace{-10pt}
    \caption{Macroeconomic indicators for the 1959-2019 period with all 6 estimated change points (red lines).}
    \label{fig:macro}
\end{figure}

Next, we show again the identified change points in Table \ref{tab:cp-events-addition} below.
\begin{table}[!ht]
    \spacingset{1}
    \centering
    \caption{Estimated Change Points and Candidate Related Events.}
    \label{tab:cp-events-addition}
    \resizebox{0.8\textwidth}{!}{%
    \begin{tabular}{c|l}
    \hline\hline
        Date (mm/dd/yyyy) & Candidate Related Events\\
        \hline
        02/01/1975 & Aftermath of 1973 oil crisis \\
        \hline
        04/01/1977 & Rapid build-up of inflation expectations \\
        \hline
        12/01/1980 & Rapid increase of interest rates by the Volcker Fed \\
        \hline
        01/01/1994 & Multiple events - see Appendix G.3 \\
        \hline
        09/01/2008 & Recession following collapse of Lehman Brothers\\
        \hline
        05/01/2010 & Recovery from the Great Financial crisis of 2008  \\
        \hline\hline
    \end{tabular}}
\end{table}
The first change point corresponds to the aftermath of the first oil crisis in 1973 and the collapse of the post-war Bretton Woods system of monetary management of commercial and financial relations among the leading western economics \citep{bordo2007retrospective}, that led to low growth and sustained inflation.
The second change point identified, marks the rapid build-up of inflation expectations \citep{kareken1978inflation} that led the Federal Reserve Board under Chairman Volcker to pursue a contractionary monetary policy through doubling the federal funds rate to 20\% to fight-off persisting inflation expectations \citep{orphanides2004monetary}. The next change point is associated with multiple events, including the Republican Party controlling the US House of Representatives for the first time since 1952 with a business and markets friendly agenda, and the ratification of the North American Free Trade Agreement. The last two change points are associated with the onset and exit of the Great Financial Crisis of 2008 that led to a deep recession, collapse and/or bailouts of various financial institutions, liquidity crunches and a debt crisis in peripheral countries in the Eurozone that exhibited a negative feedback to the US economy \citep{eichengreen2014hall}.

\end{appendices}
\end{document}